\pdfoutput=1
\documentclass[10pt,final]{book}

\newcommand{\MyTitle}{On Time in Quantum Mechanics}
\title{\MyTitle}
\newcommand{\MyAuthor}{Nicola Vona}
\author{\MyAuthor}
\newcommand{\MySubject}{Time in Quantum Mechanics}
\newcommand{\MyPdfKeywords}{Time, Quantum Mechanics, Bohmian Mechanics, Uncertainty relations}  

\usepackage{letltxmacro,ifthen,ifdraft,etoolbox}

\usepackage[ngerman,english]{babel}
\usepackage[latin1]{inputenc} 

\usepackage{color,enumitem}    
\newcommand{\red}[1]{\textcolor{red}{#1}}

\usepackage{graphicx,subfig}   
\graphicspath{{./Figures/Book.Chapter/}{./Figures/PRL/}{./Figures/Experiment/}{./Figures/ET/}{./Figures/Scattering.Bounds/}} 
\usepackage[percent]{overpic} 
\usepackage{pdfpages}  
\usepackage{relsize} 
\usepackage{tikz}
	\usetikzlibrary{positioning}		
	\usetikzlibrary{arrows}
	\usetikzlibrary{shapes.multipart}
\newcommand{\blackdiamond}
	{\protect\tikz \protect\fill (0ex,0.5ex) -- (0.5ex,1ex) -- (1ex,0.5ex) -- (0.5ex,0ex) ;}

\usepackage{natbib}
\LetLtxMacro\cite\citep

\usepackage{microtype,caption}   
\newcommand\blfootnote[1]{%
  \begingroup
  \renewcommand\thefootnote{}\footnote{#1}%
  \addtocounter{footnote}{-1}%
  \endgroup
}

\usepackage{versions}

\usepackage{ifdraft,letltxmacro}
\usepackage[colorinlistoftodos,
			\ifoptionfinal{disable}{}
		]{todonotes} 

\LetLtxMacro\oldlistoftodos\listoftodos
\renewcommand{\listoftodos}{
\makeatletter
\let\oldchapter\chapter
\let\chapter\section
\protect\hypertarget{todolist}{} 
\oldlistoftodos
\let\chapter\oldchapter
\makeatother}

\LetLtxMacro\oldtodo\todo

\makeatletter
\newcommand{\smalllinkedtodo}[2][]{\oldtodo[caption={#2}, #1]{%
\if@todonotes@inlinenote
#2
\else
\renewcommand{\baselinestretch}{0.85}\selectfont\footnotesize#2\par
\fi
\hfill\hyperlink{todolist}{$\uparrow$}
}}
\makeatother

\renewcommand{\todo}[2][]{\smalllinkedtodo[#1]{#2}} 

\definecolor{defaultmarkcolor}{rgb}{0,0.5,0}
\renewcommand{\markversion}[2][defaultmarkcolor]
	{
	\ifoptionfinal {\includeversion{#2}} {	
	 \newenvironment{#2}
			{\bigskip\bigskip\bigskip
				\todo[inline,color=#1,caption={Begin #2}]{\color{white}Begin #2}
					\color{#1!50!black} 
				 }
			{
				\todo[inline,color=#1,nolist
						]{\color{white}End #2}
				\bigskip\bigskip}
	}}

\usepackage{hyperref} 

\usepackage{amsmath,amssymb,amsthm,braket,mathtools,bm,units}
\numberwithin{equation}{section}

\newcommand{\povm}{\textsc{povm}}

\newcommand{\of}[1]{(#1)}
\newcommand{\vect}[1]{\mathbf{#1}}
\newcommand{\syst}[1]{\left\{ \begin{aligned} #1 \end{aligned}  \right.}    
\newcommand{\de}{d}
\newcommand{\thalf}[1][1]{\tfrac{#1}{2}} 
\newcommand{\abs}[1]{\lvert#1\rvert} 
\newcommand{\Abs}[1]{\left \lvert#1\right \rvert}  
\newcommand{\intdef}[2]{\int_{#1}^{#2}} 
\newcommand{\intdefde}[3]{\intdef{#1}{#2} \de #3}   
\newcommand{\1}{\mathbf{1}}  
\newcommand{\N}{\mathbb{N}}    
\newcommand{\R}{\mathbb{R}}    
\newcommand{\C}{\mathbb{C}}    
\newcommand{\Prob}{\mathbb{P}}    
\newcommand{\PP}{\Prob}
\newcommand{\norm} [1]{\| #1 \|} 
\newcommand{\intinde}[1]{\int \de #1}  
\newcommand{\der}[3]{\frac{\de^{#1}#2}{\de {#3}^{#1}}}   
\newcommand{\Der}[2]{\frac{\de^{#1}}{\de {#2}^{#1}}}    
\DeclareMathOperator{\erf}{Erf}   

\renewcommand\Re{\operatorname{Re}}
\renewcommand\Im{\operatorname{Im}}
\let\tempepsilon\varepsilon
\let\varepsilon\epsilon
\let\epsilon\tempepsilon

\newcommand{\cpc}{\ifmmode C^+\else$C^+$\fi}
\newcommand{\cpcI}{\ifmmode\cpc_I\else$\cpc_I$\fi}

\newcommand{\dkp}[1]{\dot{#1}} 
\newcommand{\ddkp}[1]{\ddot{#1}} 
\newcommand{\dddkp}[1]{\dddot{#1}} 
\newcommand{\dr}[1]{{#1}'} 
\newcommand{\ddr}[1]{{#1}''} 
\DeclareMathOperator{\supp}{supp}

\newcommand{\mean}[1]{\langle  #1  \rangle}
\newcommand{\meano}[1]{\langle  #1  \rangle_{\oldstylenums{0}}}
\newcommand{\var}[1]{\Var}
\newcommand{\Var}{\mathrm{Var}\,}
\newcommand{\Varo}{\mathrm{Var}_{\oldstylenums{0}}\,}
\newcommand{\Pio}{\Pi^{\oldstylenums{0}}}
\newcommand{\Po}{P_{\oldstylenums{0}}}
\newcommand{\tPo}{\tilde{P}_{\oldstylenums{0}}}
\newcommand{\talpha}{{\tilde\alpha}}
\newcommand{\tbeta}{{\tilde\beta}}
\newcommand{\CV}{{\mathcal C_V}}

\newtheoremstyle{lemma}
	{\topsep}
	{\topsep}
	{\itshape}
	{0pt}
	{\bfseries\sffamily}
	{.}
	{5pt plus 1pt minus 1pt}
	{\thmname{#1}\thmnumber{ #2}\thmnote{ (#3)}}

\theoremstyle{lemma}
\newtheorem{definition}{Definition}[chapter]
\newtheorem{theorem}{Theorem}[chapter]
\newtheorem{lemma}{Lemma}[chapter]
\newtheorem{corollary}{Corollary}[chapter]
\newtheorem{hypothesis}{Hypothesis}[chapter]
\newtheorem{remark}{Remark}[chapter]

\LetLtxMacro\oldproof\proof
\newcommand{\newproof}[1][]{\ifstrempty{#1}{\oldproof}{\oldproof[Proof (#1)]}}
\LetLtxMacro\proof\newproof

\newcommand{\QED}{}

\usepackage[
	verbose,
	driver=none,
	a5paper,
	hmargin={19.3mm,27.29mm},
	vmargin={19.3mm,27.29mm},
	heightrounded, 
	footskip=10mm,
	nohead,
	marginparsep={.1cm},
	marginparwidth={1.7cm}
	]{geometry}

\usepackage{fancyhdr}
\fancypagestyle{plain}{
\fancyhf{} 
\fancyfoot[LE,RO]{\thepage}}
\usepackage{floatpag,rotating}
\floatpagestyle{plain}
\rotfloatpagestyle{plain}

\usepackage{titlesec}
\titleformat{\chapter}[display]
	{\filleft\sffamily\Large}
	{\itshape\chaptertitlename{} \thechapter}
	{1ex}
	{
	\bfseries\textsf}
	[]
\titlespacing{\chapter}{3cm}{0pt}{7\baselineskip}[0pt]
\titleformat{\section}[hang]{\bfseries\sffamily\large}{\thesection}{2ex}{}[]
\titleformat{\subsection}[hang]{\sffamily\bfseries}{}{0pt}{}[]
\titleformat{\subsubsection}[hang]{\sffamily
\itshape}{}{0pt}{}[]
\titlespacing*{\subsubsection}{0pt}{\baselineskip}{0pt}[0pt]
\titleformat{\paragraph}[runin]{\itshape}{}{0pt}{}[.]

\usepackage{titletoc}
\newlength{\pagespacing}
\titlecontents{chapter}
[2em] 
{\vspace*{3ex}\bfseries\sffamily} 
{\contentslabel{2em}} 
{\hspace*{-2em}} 
{\mbox{\hspace*{\pagespacing}
\itshape\thecontentspage}} 

\titlecontents{section}
[4em] 
{\vspace*{.5ex}\sffamily} 
{\contentslabel{2em}} 
{\hspace*{-2em}} 
{\mbox{\hspace*{\pagespacing}
\itshape\thecontentspage}} 

\newcommand{\question}[1]{\begin{center}\emph{#1}\end{center}}

\PassOptionsToPackage{hyphens}{url}%
\usepackage{hyperref}   
\hypersetup{colorlinks, citecolor=black, filecolor=black, linkcolor=black, urlcolor=black, 
bookmarksopen=true, 
bookmarksnumbered=true, 
pdftitle={\MyTitle}, pdfauthor={\MyAuthor}, pdfsubject={\MySubject}, pdfkeywords={\MyPdfKeywords} 
}
\expandafter\def\expandafter\UrlBreaks\expandafter{\UrlBreaks
  \do\a\do\b\do\c\do\d\do\e\do\f\do\g\do\h\do\i\do\j%
  \do\k\do\l\do\m\do\n\do\o\do\p\do\q\do\r\do\s\do\t%
  \do\u\do\v\do\w\do\x\do\y\do\z\do\A\do\B\do\C\do\D%
  \do\E\do\F\do\G\do\H\do\I\do\J\do\K\do\L\do\M\do\N%
  \do\O\do\P\do\Q\do\R\do\S\do\T\do\U\do\V\do\W\do\X%
  \do\Y\do\Z}

\usepackage[T1]{fontenc}
\usepackage{libertine}
\usepackage{fourier}
\DeclareMathAlphabet{\mathpzc}{OT1}{pzc}{m}{it}
\let\mathcal\mathpzc

\DeclareMathSymbol{\Gamma}{\mathalpha}{letters}{"00}
\DeclareMathSymbol{\Delta}{\mathalpha}{letters}{"01}
\DeclareMathSymbol{\Theta}{\mathalpha}{letters}{"02}
\DeclareMathSymbol{\Lambda}{\mathalpha}{letters}{"03}
\DeclareMathSymbol{\Xi}{\mathalpha}{letters}{"04}
\DeclareMathSymbol{\Pi}{\mathalpha}{letters}{"05}
\DeclareMathSymbol{\Sigma}{\mathalpha}{letters}{"06}
\DeclareMathSymbol{\Upsilon}{\mathalpha}{letters}{"07}
\DeclareMathSymbol{\Phi}{\mathalpha}{letters}{"08}
\DeclareMathSymbol{\Psi}{\mathalpha}{letters}{"09}
\DeclareMathSymbol{\Omega}{\mathalpha}{letters}{"0A}

\ifoptionfinal	{\excludeversion{VersionPERSONALNOTES}}  
				{\markversion{VersionPERSONALNOTES}}

\begin{document}

\selectlanguage{english}


\begin{titlepage}
\pagestyle{empty}%
\includepdf{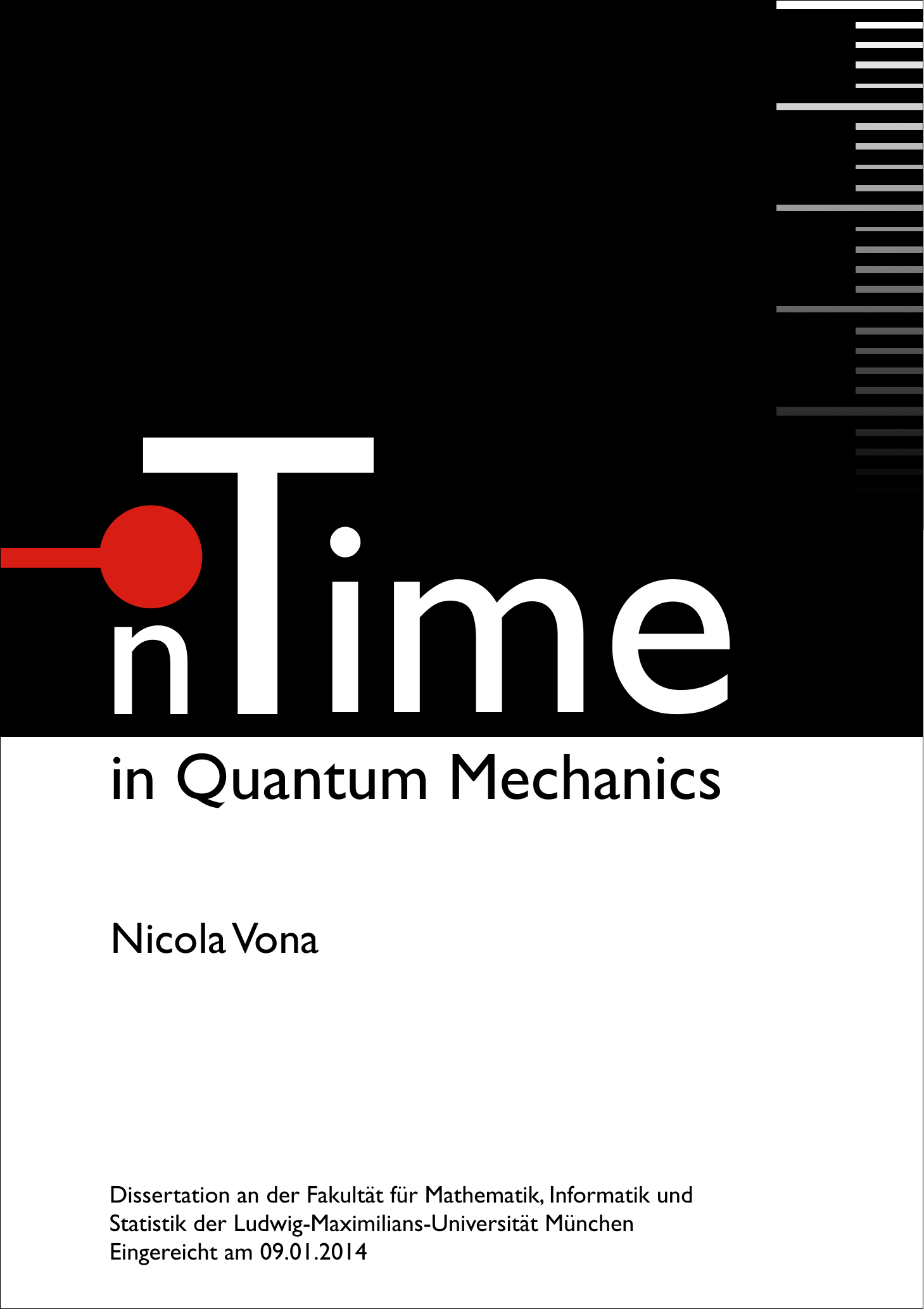}%
\cleardoublepage%
\setcounter{page}{1}%
\sffamily%
\begin{center}%
\mbox{}\\[3\baselineskip]
{\large Nicola Vona}\\[\baselineskip]
{\LARGE \textbf{\red On Time in Quantum Mechanics}}\\[1.62\baselineskip]
Dissertation an der Fakultät für\\
Mathematik, Informatik und Statistik der\\
Ludwig-Maximilians-Universität München\\[2\baselineskip]
\vfill
\begin{minipage}[t]{.8\textwidth}
\mbox{}\\
{\newlength\rowsep
\setlength\rowsep{\the\baselineskip}
\begin{tabular}{@{} l @{\hspace{.5em}} l @{\quad} l @{\hspace{.5em}} l @{}}
Submission: & 9 January 2014\hspace*{2em}& 
Defense:  & 7 February 2014
\\[1.62\rowsep]
\makebox[0pt][l]{Supervisor and First Reviewer:} & & \makebox[0pt][l]{Prof.\,Dr.\ Detlef Dürr} &\\[\rowsep]
\makebox[0pt][l]{Second Reviewer:} & &  \makebox[0pt][l]{Prof.\,Dr.\ Peter Pickl}&\\[\rowsep]
\makebox[0pt][l]{External Reviewer:} & &
\makebox[0pt][l]{%
\begin{minipage}[t][0pt][t]{\textwidth}
Prof.\,Dr.\ Stefan Teufel\\
Fachbereich Mathematik,\\
Universität Tübingen
\end{minipage}}&
\end {tabular}
}%
\end{minipage}
\end{center}

%

\clearpage\mbox{}
\end{titlepage}

\selectlanguage{ngerman}
\begin{center}
{\large \textbf{Eidesstattliche Versicherung}}\\
(Siehe Promotionsordnung vom 12.07.11, \S8, Abs.\ 2 Pkt.\ 5.)\\[2\baselineskip]
\end{center}

Hiermit erkläre ich an Eidesstatt, dass die Dissertation von mir selbstständig, ohne unerlaubte Beihilfe angefertigt ist. 
Kapitel 2, 3, 5 und 6 enthalten Resultate, die aus der Zusammenarbeit mit verschiedenen Koautoren entstammen (für Details siehe den Anfang des entsprechenden Kapitels).
\\[2\baselineskip]
Name, Vorname:\quad Vona, Nicola\\[2\baselineskip]
München, 12. Februar 2014

\selectlanguage{english}

\cleardoublepage

\section*{Abstract}
Although time measurements are routinely performed in laboratories,  their theoretical description is still an open problem.
Similarly, also the validity and the status of the energy-time uncertainty relation is unsettled.

In the first part of this work the necessity of positive operator valued measures (\povm{}) as descriptions of every quantum experiment is reviewed, as well as the suggestive role played by the  probability current in time measurements.
Furthermore, it is shown that no \povm{} exists, which approximately  agrees with the  probability current on a very natural set of wave functions; nevertheless, the choice of the set is crucial, and on more restrictive sets the probability current does provide a good arrival time prediction.
Some ideas to experimentally detect quantum effects in time measurements are discussed.

In the second part of the work the energy-time uncertainty relation is considered, in particular for a model of alpha decay for which the variance of the energy can be calculated explicitly, and the variance of time can be estimated.
This estimate is tight for systems with long lifetimes, in which case the uncertainty relation is shown to be satisfied.
Also the  linewidth-lifetime relation is shown to hold, but contrary to the common expectation, it is found that the two relations behave independently, and therefore it is not possible to interpret one as a consequence of the other.

To perform the mentioned analysis quantitative scattering estimates are necessary.
To this end, bounds of the form $\|\1_Re^{-iHt}\psi\|_2^2 \leq C t^{-3}$ have been derived, where $\psi$ denotes the initial state, $H$ the Hamiltonian,  $R$ a positive constant,  and $C$ is explicitly known.
As intermediate step, bounds on the derivatives of the $S$-matrix in the form $\|\1_K S^{(n)}\|_\infty  \leq  C_{n,K} $ have been established, with $n=1,2,3$,  and the constants $C_{n,K}$ explicitly known.

\cleardoublepage

\selectlanguage{ngerman}
\section*{Zusammenfassung}
Obwohl Zeitmessungen täglich in vielen Laboren durchgeführt werden, ist ihre theoretische Beschreibung noch unklar.
Gleichermaßen sind Gültigkeit und Bedeutung der Energie-Zeit-Unschärfe ungeklärt.

Der erste Teil dieser Arbeit diskutiert die Notwendigkeit von \emph{\foreignlanguage{english}{positive operator valued measures}} (\povm{}) 
zur Beschreibung
von allen Quantenexperimenten, sowie die bedeutende Rolle des Wahrscheinlichkeitsstroms in Zeitmessungen.
Außerdem, wird 
gezeigt, dass kein \povm{} existiert, der den Wahrscheinlichkeitsstrom jeder Wellenfunktion in einer natürlichen Menge annähert.
Die Wahl dieser Menge ist aber entscheidend, und auf beschränkten Mengen ist der Wahrscheinlichkeitsstrom eine gute Vorhersage für Zeitmessungen.
Einige Ideen sind diskutiert, wie man Zeitexperimente durchführen kann, um Quanteneffekten zu detektieren.

Der zweite Teil dieser Arbeit beschäftigt sich mit der Energie-Zeit-Unschärfe, insbesondere für ein Modell von Alpha-Zerfall, wobei man die Energievarianz explizit 
berechnen
kann, und die Zeitvarianz 
abschätzt.
Diese Abschätzung ist für Systeme mit langen Lebensdauern gut, und in diesem Fall wird gezeigt, dass die Energie-Zeit-Unschärfe gilt.
Ebenso wird gezeigt, dass die \emph{\foreignlanguage{english}{linewidth-lifetime relation}}  gilt.
Im allgemein wird angenommen, dass diese zwei Relationen dieselben sind. 
Im Gegensatz dazu, 
wird in der Dissertation aber gezeigt, dass sie sich unabhängig voneinander verhalten.

Für diese Resultate, braucht man quantitative Streuabschätzungen.
Zu diesem Zweck werden Schranken in der Form $\|\1_Re^{-iHt}\psi\|_2^2 \leq C t^{-3}$ in der Dissertation gezeigt, wo $\psi$ der Anfangszustand ist, $H$ der Hamiltonoperator, $R$ eine positive Konstante, und $C$ explizit bekannt ist.
Als Zwischenschritt werden Schranken für die Ableitungen der $S$-Matrix in der Form $\|\1_K S^{(n)}\|_\infty  \leq  C_{n,K} $ bewiesen, wobei $n=1,2,3$,  und die Konstanten  $C_{n,K}$  explizit bekannt sind.

\selectlanguage{english}

\cleardoublepage

\tableofcontents

\cleardoublepage

\chapter{Preface}

\question{What is the right statistics for the measurements\\of arrival time of a quantum particle?}

\question{Is the energy-time uncertainty relation valid?}

\question{What is the origin of the linewidth-lifetime\\relation for metastable states?}

These three questions are the topic of this thesis.
They are very basic and  everybody would reasonably expect their answers to be completely clear.
Surprisingly, the opposite is true: the description of time measurements in quantum mechanics is an open problem that attracts the interest of the scientific community since decades.
This situation is even more surprising if one takes into account that time measurements are routinely performed in laboratories.
The trick used to describe actual experiments despite the mentioned theoretical difficulties is to note that usual experiments are performed in far-field regime, for which a semiclassical analysis suffices.
Nevertheless, the developments in detector technology promise to soon allow for near-field investigations \citep[for example for photons see][]{ZhangSlyszVerevkin2003,PearlmanCrossSlysz2005,RenHofmann2011}: the old problem of the description of time measurements in quantum mechanics becomes now a timely  issue in foundational research. 

The outcomes of any quantum measurement are described by a positive operator valued measure (\povm{}), therefore this must be the case for time measurements too.
On the other  hand, many otherwise respectable physicists would bet that, at least in most cases, the arrival time statistics of a quantum particle is given by the flux of its probability current, which is not a \povm.   
Chapter~\ref{ch:BookChapter} extensively illustrates  how the \povm{}s arise in the description of every quantum experiment, as well as the suggestive role played by the quantum flux in time measurements.
Furthermore, Chapter~\ref{ch:PRL} investigates the question whether a \povm{} exists, which agrees approximately with the quantum flux values on a reasonable set of wave functions. 
The answer to this question is negative for a very natural set of wave functions, but it is important to remark that the choice of the set is crucial, and that the use of a different set changes the picture drastically.
On a more restrictive set the quantum flux might provide a good arrival time prediction.
For example, it is possible to find a \povm{} that agrees with the quantum flux  on the scattering states, and it is conjectured that the same is true for the wave functions with high energy. 
Some numerical evidence is provided that supports this conjecture.
Besides, Chapter~\ref{ch:Experiment} presents some preliminary ideas about how to devise an  experiment able to detect quantum effects in time measurements.

Getting the right object to describe time measurements is of course relevant to understand if the energy-time uncertainty relation holds.
Although the first question has not yet been clarified, many results on the validity of the uncertainty relation already exist, based on general properties.
Nevertheless, they are far from providing a comprehensive framework \citep[see for example][]{Busch1990,MugaMayatoEgusquiza2008,MugaMayatoEgusquiza2009}.
The alpha decay of a radioactive nucleus is a case study for the uncertainty relation because of the uncertainty on the energy of the emitted alpha particle and on the instant of emission.
In Chapter~\ref{ch:ET} a model for this process is considered, for which the variance of the energy can be calculated explicitly, and the variance of time can be estimated.
The result will be that, whenever the error on this estimate is small enough to be successfully used, then the uncertainty relation is satisfied.
Unfortunately, the error bound becomes small enough only for systems with lifetimes much longer than any physical element, but this circumstance does not prevent from using the result on a principle level.

The fact that the status of the energy-time uncertainty relation is still an open question is at odd with its ubiquity.
For instance, consider the well known  linewidth-lifetime relation
\begin{align}
  \Gamma\tau=\hbar,
\end{align}
that connects the full width at half maximum $\Gamma$ of the energy distribution of the decay product with the lifetime $\tau$ of the unstable state that produces it.
This relation is very often explained as an instance of the energy-time uncertainty relation \citep[see for example][]{Rohlf}, although \citet{KrylovFock1947} provided some arguments against this explanation.
In Chapter~\ref{ch:ET} the quantities $\Gamma$ and $\tau$ are estimated for the mentioned model of alpha decay.
It is shown that, by adjusting the potential and the initial state, it is possible to make the product $\Gamma\tau$ arbitrarily close to $\hbar$, while at the same time the product of the energy and time variances  gets arbitrarily large.
This explicitly confirms the thesis that the two relations are indeed independent and one can not be interpreted  as a consequence of the other.

For the analysis of Chapter~\ref{ch:ET} one needs control over the whole time evolution of the metastable state considered in the model.
For a long time span, the decay approximately follows an exponential law, for which explicit bounds are known \citep{Skibsted86}.
At later times the exponential behavior gets superseded by the power like decay typical of the scattering regime, for which many results are available, for example in the form of dispersive estimates \citep{JensenKato1979,JSS,Rauch1978,Schlag}.
In the easiest case, denoting by $\psi$ the initial state, by $H$ the Hamiltonian, and by $R$ and $C$ two positive constants, these estimates can be written as
\begin{equation}\label{eq:main_acIntroPhD}
\|\1_Re^{-iHt}\psi\|_2^2
	\leq C t^{-3} .
\end{equation}
Little is known quantitatively about the constant $C$. 
Chapter~\ref{ch:Scattering} is devoted to proving quantitative bounds of this form, in terms of the initial wave function, the potential, and the spectral properties of the Hamiltonian.
This goal will be achieved first by proving bounds on the derivatives of the $S$-matrix in the form
\begin{equation}
\|\1_K S^{(n)}\|_\infty  \leq  C_{n,K}  ,
\end{equation}
with $n=1,2,3$,  and the constants $C_{n,K}$ 
explicitly known.
These bounds are obtained by means of the theory of entire functions.
Then, they are used in the usual stationary phase approach to get the bound~\ref{eq:main_acIntroPhD}.

\newpage

\section*{Acknowledgments}
\addcontentsline{toc}{section}{Acknowledgments} 

A large part of this thesis is the fruit of a daily teamwork with Robert Grummt. I enjoyed very much this collaboration, from which I have learned the important truth  that
\[
1+1>2.
\]
My advisor Detlef Dürr deserves all of my gratitude for having donated me \emph{a wonderfully clear world outlook}, and for having shown me that it is \emph{still} possible to do science with the high critical standards everybody would expect from it.
I also wish to thank Günter Hinrichs for his patience and for his constant will to help.
The many conversations I had with Catalina Curceanu, Shelly Goldstein, Arun Ravishankar, Harald Weinfurter, and Nino Zanghì have been very important to me, as well as those with all the members of the research group.
In particular, I learned from Dirk Deckert the importance of having a goal clearly formulated since the beginning of the work, and from Martin Kolb the importance of reading a lot.
I thank V.~Bach, M.~Goldberg, L.~Rozema, W.~Schlag, E.~Skibsted, and M.~Zworski  for helpful correspondence, and gratefully acknowledge the financial support of the Elite Network of Bavaria, of the \textsc{daad stibet} Program, of the \textsc{lmu} Graduate Center, and of the \textsc{cost} action ``Fundamental Problems in Quantum Physics".

When the project behind this thesis 
came out, it was my lovely wife Giorgia who convinced me of its potentiality.
It is thanks to her if I started it and to her endless support if I could complete it.
She gave me precious advices, she constantly listened to me -- no matter how boring that can be -- but most of all \emph{she is my family}.

\subsection*{Notice}
Chapter~\ref{ch:BookChapter} will appear as  \citep{VonaDurr2013};
Chapter~\ref{ch:PRL} has been published as \citep{VonaHinrichsDurr2013};
Chapters~\ref{ch:ET} and~\ref{ch:Scattering} present the results of the preprints \citep{GrummtVona2014,GrummtVona2014a}.

The calculation of the trajectories shown in Figs.~\ref{fig:NegativeCurrent} and~\ref{fig:InterferenceSumDifference} is based on the code by Klaus von Bloh available at 
\emph{\href{http://demonstrations.wolfram.com/CausalInterpretationOfTheDoubleSlitExperimentInQuantumTheory/}{\nolinkurl{http://demonstrations.wolfram.com/CausalInterpretationOfTheDoubleSlitExperimentInQuantumTheory/}}}.

\chapter[The Role of the Probability Current\\for Time Measurements]{The Role of the Probability Current for Time Measurements}
\label{ch:BookChapter}
\blfootnote{This chapter will appear as \citep{VonaDurr2013}.}
\section{Introduction}
Think of a very simple experiment, in which a particle is sent towards a detector.
\question{When will the detector click?}
Imagine to repeat the experiment many times, starting a stopwatch at every run.
The instant at which the particle hits the detector will be different each time, forming a statistics of \emph{arrival times}.
Experiments of this kind are routinely performed in almost any laboratory, and are the basis of many common techniques, collectively known as time-of-flight methods (\textsc{tof}).
In spite of that, how to theoretically  describe an arrival time measurement is a very debated topic since the early days of quantum mechanics \citep{Pauli1958}.
It is legitimate to wonder why it is so easy to speak about a position measurement at a fixed time, and so hard to speak about a time measurement at a fixed position.
An overview of the main attempts and a discussion of the several difficulties they involve can be found in \citep{MugaLeavens2000,MugaMayatoEgusquiza2008,MugaMayatoEgusquiza2009}.

In the following, we will discuss the theoretical description of time measurements with particular emphasis on the role of the probability current.

\section{What is a Measurement?}
We will start recalling the general description of a measurement in quantum mechanics in terms of \emph{positive operator valued measures} (\povm s).
This framework is less common than the one based on self-adjoint operators, but is more general and more explicit than the latter.

\subsection{Linear Measurements -- \povm s}\label{sec:Povms}
When we speak about a \emph{measurement}, what are we speaking about?\\
A measurement is a situation in which a physical system of interest interacts with a second physical system, the apparatus, that is used to inquire into the former.
In general, we are  interested in those cases in which the experimental procedure is fixed and independent of the state of the system to be measured given as input; these cases are called \emph{linear measurements}.
The meaning of this name will be clarified in the following.
The analysis of the general properties of a linear measurement, and of the general mathematical description of such a process, has been carried out mostly by \citet{Ludwig1983a}, and finds a natural completion within Bohmian mechanics \citep{DuerrGoldsteinZanghi2004}.
In the following, we will present a simplified form of this analysis \citep{DurrTeufel2009,NielsenChuang2000}.

We will denote by $x$ the configuration of the system and by $\psi_0$ its initial state, element of the Hilbert space $L^2(\R^{3n})$, while we will use  $y$ for the configuration of the apparatus and $\Phi_0\in L^2(\R^{3N})$ for  its ready state; moreover, we will denote by $(0,T)$ the interval during which the interaction constituting the measurement takes place.
The evolution of the composite system is a usual quantum process, so the state at time $T$ is 
\begin{equation}
\Psi_T \of{x,y}  =  (U_T\,  \Psi_0)\of{x,y}  =  U_T\, (\psi_0\,\Phi_0) \of{x,y}   ,
\end{equation}
where $U_T$ is a  unitary operator on $L^2(\R^{3(N+n)})$.
We call such an interaction a \emph{measurement} if for every initial state $\psi_0$ it is possible to write the final state $\Psi_T$ as
\begin{equation}
\Psi_T \of{x,y}  =  \sum_\alpha  \psi_\alpha\of x \, \Phi_\alpha\of y  ,
\end{equation}
with the states $\Phi_\alpha$ normalized and clearly distinguishable, i.e.\ with supports $G_\alpha = \{ y \,|\, \Phi_\alpha\of y \neq 0 \}$ macroscopically separated.
This means that after the interaction it is enough to ``look'' at the position of the apparatus pointer to know the state of the apparatus.
Each support $G_\alpha$ corresponds to a different result of the experiment, that we will denote by  $\lambda_\alpha$.
One can imagine each support to have a label with the value $\lambda_\alpha$ written on it: if the position of the pointer at the end of the measurement is inside the region $G_\alpha$, then the result of the experiment is $\lambda_\alpha$.
The probability of getting the outcome $\lambda_\alpha$ is
\begin{align}\label{eq:ProbAlpha}
\mathbb{P}_\alpha 
	&= \intinde{x} \intdefde{G_\alpha}{}{y}  | \Psi_T\of{x,y} |^2
	\nonumber\\
	&= \intinde{x} \intdefde{G_\alpha}{}{y}  | \psi_\alpha\of x\,  \Phi_\alpha\of y |^2
	= \int   | \psi_\alpha\of{x} |^2\,  \de{x}  ,
\end{align}
indeed $\Phi_{\alpha'}\of{y} =0$ $\forall y\in G_\alpha$, $\alpha'\neq \alpha$, and  the $\Phi_\alpha$ are normalized.
Consider now the projectors $P_\alpha$ that act on the Hilbert space  $L^2(\R^{3(N+n)})$
 of the composite system  and project to the subspace {$L^2(\R^{3n}\times G_\alpha)$}  corresponding to the pointer in the position $\alpha$, i.e.\ in particular
\begin{equation}
P_\alpha \Psi_T =  \psi_\alpha\,  \Phi_\alpha  .
\end{equation}
Through the projectors $P_\alpha$ we can  define the operators $R_\alpha$ such that
\begin{equation}
P_\alpha \Psi_T =  \psi_\alpha\,  \Phi_\alpha  =  (R_\alpha \psi_0) \; \Phi_\alpha  ,
\end{equation}
that means $R_\alpha \psi_0 = \psi_\alpha$.
Finally, we can also define the operators $O_\alpha =  R_\alpha^\dag R_\alpha$.
These operators are directly connected to the probability \eqref{eq:ProbAlpha} of getting the outcome $\alpha$
\begin{equation}\label{eq:ProbO}
\PP_\alpha = \norm{\psi_\alpha}^2  = \braket{\psi_0 | O_\alpha \psi_0} .
\end{equation}
Therefore, \emph{the operators $O_\alpha$ together with the set of  values $\lambda_\alpha$ are sufficient to determine any statistical quantity related to the experiment}.
The fact that any experiment of the kind we have considered can be completely described by a set of linear operators, explains the origin of the name \emph{linear measurement}.
Equation\eqref{eq:ProbO} implies also that the operators $O_\alpha$ are \emph{positive}, i.e.\
\begin{equation}
\braket{\psi_0 | O_\alpha \psi_0}  \geq 0
	\qquad \forall \psi_0 \in L^2(\R^{3n})  .
\end{equation}
In addition, they constitute a \emph{decomposition of the unity}, i.e.\ 
\begin{equation}
\sum_\alpha O_\alpha = \1  ,
\end{equation}
as a consequence of the unitarity of $U_T$ and of the orthonormality of the states $\Phi_\alpha$, that imply
\begin{align}
1 = \norm{\psi_0\Phi_0}^2  
	&= \norm{\Psi_T}^2  
\nonumber\\	
	&= \sum_\alpha \norm{\psi_\alpha}^2
	= \sum_\alpha \braket{\psi_0 | O_\alpha \psi_0}
	\quad \forall \psi_0 \in L^2(\R^{3n})  .
\end{align}
A set of operators with these features is called discrete \emph{positive operator valued measure}, or simply \povm.  It is a measure on the discrete set of values $\lambda_\alpha$. In case the value set is a continuum, the \povm\  is a Borel-measure on that continuum, taking values in the set of positive linear operators.

It is important to note that in the derivation of the \povm{} structure the orthonormality of the states $\Phi_\alpha$ and the unitarity of the overall evolution play a crucial role, while in general, the states $\psi_\alpha$ do not need to be neither orthogonal nor distinct.

In case the operators $O_\alpha$ happen to be orthogonal projectors, then the usual measurement formalism of standard quantum mechanics is recovered by defining  the selfadjoint operator $\hat A = \sum_\alpha \lambda_\alpha \, O_\alpha$.
Physically, this condition is achieved  for example in a reproducible measurement, i.e.\ one in which the repetition of the measurement using the final state $\psi_\alpha$ as input, gives the result $\alpha$ with certainty.

We remark that calculating the action of a \povm{} on a given initial state requires that that state is evolved for the duration of the measurement together with an apparatus, and therefore its evolution in general differs from the evolution of the system alone.
This circumstance is evident if one thinks that the state of the system after the measurement will depend on the measurement outcome.%
\footnote{It will be an eigenstate of the selfadjoint operator corresponding to the measurement, in case it exists.}
Usually, if the measurement is not explicitly modeled, this evolution is considered as a black box that takes a state as input and gives an outcome and another state as output.
It is important to keep in mind that the measurement formalism always entails such a departure from the autonomous evolution of the system, even if not explicitly described.

\subsection{Not only \povm s\label{sec:NotPOVMs}}
Although a linear measurement is a very general process, there are many quantities that are not measurable in this sense.
An easy example is the probability distribution of the position $|\psi|^2$.
Indeed, suppose to have a device that shows the result $\lambda_1$ if the input is a particle in a state for which the position is distributed according to $|\psi_1|^2$, and $\lambda_2$ if it is in a state with distribution $|\psi_2|^2$.
If the process is described by a \povm, the linearity of the latter requires that when the state $\psi_1+\psi_2$ is given as input, the result is \emph{either $\lambda_1$ or $\lambda_2$}, as for example the result of a measurement of spin on the state $\ket{\mathrm{up}}+\ket{\mathrm{down}}$ is either ``up'' or ``down''.
On the contrary, if the device was supposed to measure the probability distribution of the position, the result had to be $\lambda_+$, corresponding to $|\psi_1 + \psi_2|^2$,  possibly distinct both from $\lambda_1$ and from $\lambda_2$.

To overcome a limitation of this kind, the only possibility is to give up on linearity, accepting as measurement also processes different than the one devised in the previous section.
These processes use additional information about the $x$-system, for example giving a result dependent on previous runs, or adjusting the interaction according to the state of the $x$-system.
In particular, to measure the probability distribution of the position one exploits the fact that $|\psi\of{x}|^2 = \braket{\psi | O_x | \psi}$, where $O_x = \ket x \bra x$ is the density of the \povm{} corresponding to a position measurement.
Instead of measuring directly $|\psi|^2$, one measures $x$, and repeats the measurement on many systems prepared in the same state $\psi$.
The distribution $|\psi|^2$ is then recovered from the statistics of the results of the position measurements.
The additional information needed in this case is that all the $x$-systems used as input were prepared in the same state.
The outcome shown by the apparatus depends then on the preparation procedure of the input state: if we change it, we have to notify the change to the apparatus, that needs to know how to collect together the single results to build the right statistics.

For other physical quantities not linearly measurable, like for example the wave function, a similar, but more refined strategy is required.
This strategy is known as \emph{weak measurement} \citep{AharonovAlbertVaidman1988}.
An apparatus to perform a weak measurement is characterized first of all by having a very weak interaction with the $x$-system; loosely speaking, we can say that the states $\psi_\alpha$ are very close to the initial state $\psi_0$.
As a consequence of such a small disturbance, the information conveyed to the $y$-system by the interaction is very little.
The departure from linearity is realized in a way similar to that of the measurement of $|\psi|^2$: the single run does not produce any useful information because of the weak coupling, therefore the experiment is repeated many times on many $x$-systems prepared in the same initial state $\psi_0$; the result of the experiment is recovered from a statistical analysis of the collected data.

The advantage of this arrangement is that the output state $\psi_\alpha$ can be used as input for a following  linear measurement of usual kind (\emph{strong}), whose reaction is almost as if its input state was directly $\psi_0$.
In this case the experiment yields a joint statistics for the two measurements, and it is especially interesting to \emph{postselect} on the value of the strong measurement, i.e.\ to arrange the data in sets depending on the result of the strong measurement and to look at the statistics of the outcomes of the weak measurement inside each class.
For example, a weak measurement of position followed by a strong measurement of momentum, postselected on the value zero for the momentum, allows to measure the wave function \citep{LundeenSutherlandPatel2011}.

The nonlinear character of weak measurements becomes apparent if one understands the many repetitions they involve in terms of a calibration.
Indeed, one can think of the last run as the actual measurement, and of all the previous runs as a way for the apparatus to collect information about the $x$-system used in the last run, profiting from the knowledge that it was prepared exactly as the $x$-systems of the previous runs. 
The $x$-systems used in the preliminary phase can be then considered part of the apparatus, used to build the joint statistics needed to decide which outcome to attribute to the last strong measurement.
For example, the result of the experiment could be the average of the previous weak measurements postselected  on the strong value obtained in the last run.
If we then change the initial wave function $\psi_0$ to some $\psi
\mathrlap
{\smash{'}}_0$, the calibration procedure has to be repeated.
In this case, the apparatus itself depends on the state of the $x$-system to be measured, breaking linearity.

\newpage

\section{Time Statistics}\label{sec:TimeStat}
Now we finally come to our topic: time measurements.
At first, we have to note that there are several different experiments that can be called time measurements: measurements of dwell times, sojourn times, and so on.
We will refer in the present discussion exclusively to \emph{arrival times}, although it is possible to recast everything to fit any other kind of time measurement.
More precisely, we will consider the situation described at the beginning: a particle is prepared in a certain initial state and a stopwatch is set to zero; the particle is left evolving in presence of a detector at a fixed position; the stopwatch is read when the detector clicks.
The time read on the stopwatch is what we call \emph{arrival time}.

A measurement of this kind is necessarily linear, and we can ask for the statistics of its outcomes given the initial state of the particle.
If, for example, we measure the position at the fixed time $t$, then we can predict the statistics of the results by calculating the quantity
\begin{equation}
\braket{ \psi_t | x }  \braket{ x | \psi_t }  .
\end{equation}
Which calculation do we have to perform to predict the statistics of the stopwatch readings with the detector at a fixed position?

\subsection{The Semiclassical Approach}
Arrival time measurements are routinely performed in actual experiments, and they are normally treated semiclassically: essentially, they are interpreted as momentum measurements.
The identification with momentum measurements is motivated by the fact that the detector is at a distance $L$ from the source usually much bigger than the uncertainty on the initial position of the particles, so one can assume that each particle covers the same length $L$.
Hence, the randomness of the arrival time must be a consequence of the uncertainty on the momentum, and the time statistics must be given by the momentum statistics.
For a free particle in one dimension, the connection between  time and momentum is provided by the classical relation $p\of{t}  = m \,L/t$.
By a change of variable, this relation implies that the probability density of an arrival at time $t$ is
\begin{equation}\label{eq:TraditionalDensity}
\frac{1}{\hbar}\Abs{\der{}{p\of{t}}{t}}
	\left|\tilde\psi\left(\frac{p(t)}{\hbar}\right)\right|^2  
=  \frac{m L}{\hbar t^2}   
		\left|\tilde\psi\left(\frac{p(t)}{\hbar}\right)\right|^2  ,
\end{equation}
where $\tilde\psi$ is the Fourier transform of the wave function $\psi$.

This semiclassical approach is justified by the distance $L$ being very big, that is true for most experiments so far performed.
On the other hand, we tacitly assumed that the particle moves on a straight line with constant velocity $v$, whose ignorance is the source of the arrival time randomness: such a classical picture is inadequate to describe the behavior of a quantum particle in general conditions, and is expected to fail in future, near-field experiments.
A deeper analysis is needed.

\subsection{An Easy but False Derivation}\label{sec:EasyDerivation}

Consider that the particle crosses the detector at time $t$ with certainty.
This implies that the particle is on one side of the detector before $t$, and on the other side after $t$. 
One can therefore think that it is possible to connect the statistics of arrival time to the probability that the particle is on one side of the detector at different times. Because the latter is known, this seems like a good strategy.

For simplicity we will consider only the one dimensional case, that already entails all the relevant features that we want to discuss.%
\footnote{The same treatment is possible in three dimensions, provided that the detector is sensitive only to the arrival time and not to the arrival position, and that the detecting surface divides the whole space in two separate regions (i.e.\ it is a closed surface or it is unbounded).}
The detector is located at the origin; we will assume the evolution of the particle in presence of the detector to be very close to that of the particle alone.
We consider the easiest possible case: a free particle, initially placed on the negative half-line and moving towards the origin, i.e.\ prepared in a state $\psi_0$ such that
\begin{align}
\psi_0\of x &\approx 0  \qquad  \forall x\geq0; \label{eq:NegativeX}
\\
\tilde\psi\of p &=0  \qquad  \forall p\leq0  ,
\end{align}
where $\tilde\psi$ denotes the Fourier transform of $\psi_0$, and eq.\ \eqref{eq:NegativeX} is a shorthand for $\intdef{0}{\infty}\psi_0\of x \de x \ll 1$.
The particle can only have positive momentum, therefore it will get at some time to the right of the origin and thus it has to cross the detector from the left to the right.

One might think that the probability to have a crossing at a time $\tau$ later than $t$ is equal to the probability that the particle at $t$ is still in the left region,
\begin{equation}\label{eq:LeftProb}
\PP (\tau\geq t)
= \PP (x\leq 0;t)
= \intdefde{-\infty}{0}{x}  \abs{\psi_t\of{x}}^2  .
\end{equation}
Conversely, the probability that the particle arrived at the detector position before $t$ is
\begin{equation}
\PP (\tau<t)
= 1 - \PP (\tau\geq t)
= \intdefde{0}{\infty}{x}  \abs{\psi_t\of{x}}^2  .
\end{equation}
Therefore, the probability density $\Pi\of{t}$ of a crossing at $t$ is
\begin{equation}
\Pi\of t = \Der{}{t}  \PP (\tau<t)
= \intdefde{0}{\infty}{x} \partial_t \abs{\psi_t\of{x}}^2  .
\end{equation}
We can now make use of the continuity equation for the probability
\begin{equation}
\partial_t\, \bigl(|\psi_t\of x |^2\bigr) +  \partial_x  j\of{ x, t}  = 0  ,
\end{equation}
that is a consequence of the Schrödinger equation, with the probability current
\begin{equation}\label{eq:J}
j\of{ x, t} \coloneqq  \tfrac{\hbar}{m}\,  \Im \psi_t^*\of x\ \partial_x\psi_t\of x  .
\end{equation}
Substituting,
\begin{equation}\label{eq:QFlux1D}
\Pi\of{t}
=
-\intdefde{0}{\infty}{x} \partial_x  j\of{ x, t} 
=
 j\of{ x=0, t}   .
\end{equation}
Thus,  the probability density $\Pi\of{t}$ of an arrival at the detector at time $t$ is equal to the probability current $j\of{ x=0, t}$, \emph{provided everything so far has been correct}. Well, it hasn't. Eq.\ \eqref{eq:LeftProb} is problematical. It is only correct if the right hand side is a monotonously decreasing function of time, or, equivalently, if the current in \eqref{eq:QFlux1D} is always positive. But that is in general not the case and it is most certainly not guaranteed by asking that the momentum be positive. Indeed, even considering only free motion and positive momentum, there are states for which the current is not always positive, a circumstance  known as \emph{backflow} (for an example, see the appendix).
But a probability distribution must necessarily be positive, hence, \emph{the current can not be equal to the statistics of the results of any linear measurement}, i.e.\ there is no \povm{} with density $O_t$ such that
\begin{equation}
\braket{ \psi_0 | O_t | \psi_0 } = j\of{ x=0, t}  .
\end{equation}

This problem is well known \citep{Allcock1969b} and has given rise to a long debate, aiming at finding a quantum prediction for the arrival time distribution with the needed \povm{} structure \citep{MugaLeavens2000}.
%

One might wonder: How can it be that the momentum is only positive, and yet the probability that the particle is in the left region is not necessarily decreasing?
A state with only positive momentum is such that, if we \emph{measure} the momentum, then we find a positive value with certainty.
This is not the same as saying that the particle \emph{moves} only from the left to the right when we do not measure it.
Actually, in strict quantum-mechanical terms, it does not even make sense to speak about the momentum of the particle when it is not measured, as it does not make sense to speak about its position if we do not measure it, and therefore there is no way of conceiving how the particle moves in this framework.
Think for example of a double slit setting: we can speak about the position of the particle at the screen, but we can not say through which slit the particle went.

Although the quantum-mechanical momentum is only positive, the conclusion that the particle  moves only once from the left to the right is unwarranted.
Even more: it simply does not mean anything.

\subsection{The Moral}
The problem with the simple derivation of the arrival time statistics is quite instructive, indeed it forces us to face the fact that quantum mechanics is really about measurement outcomes, and therefore it is a mistake to think of quantum-mechanical quantities as of quantities intrinsic to the system under study and independent of the measurement apparatus.

\newpage

\section{The Bohmian View}\label{sec:BohmView}
Bohmian mechanics is a theory of the quantum phenomena alternative to quantum mechanics, but giving the same empirical predictions \citep[see][]{DurrTeufel2009,DurrGoldsteinZanghi2013}.
The two theories share at their foundation the Schrödinger equation.
Quantum mechanics complements it by some further axioms like the collapse postulate, and describes all the objects around us only in terms of wave functions.
On the contrary, according to Bohmian mechanics the world around us is composed by actual point particles moving on continuous paths, that are determined by  the  wave function.
The Schrödinger equation is in this case supplemented by a guiding equation that specifies the relation between the wave function and the motion of the particles.
The usual quantum mechanical formalism is recovered in Bohmian mechanics as an effective description of measurement situations \citep[see][]{DurrTeufel2009}.

The main difference between quantum and Bohmian mechanics is that the first one is concerned only with measurement outcomes, while the second one gives account of the \emph{physical reality} in any situation.
Although every linear experiment corresponds to a \povm{} according to quantum mechanics as well as to Bohmian mechanics \citep{DuerrGoldsteinZanghi2004}, for the former \povm s are the fundamental objects the theory is all about, while for the latter they are only very convenient tools that occur when the theory is used to make predictions.

We saw already how interpreting quantum-mechanical quantities as intrinsic properties of a system is mistaken, and how the framework of quantum mechanics is limited to measurement outcomes.
In Bohmian mechanics the particle has a definite trajectory, so it makes perfectly sense to speak about its position or velocity also when they are not measured, and it is perfectly meaningful to argue about the way the particle \emph{moves}.
In doing so, one has just to mind the difference between the outcomes of hypothetical  (quantum) measurements, and actual (Bohmian) quantities.

\subsection{The Easy Derivation Again\dots}
Let's review the derivation of section \ref{sec:EasyDerivation} from the point of view of Bohmian mechanics.

To find out the arrival time of a Bohmian particle it is sufficient to literally follow its motion and to register the instant when it actually arrives at the detector position.
A Bohmian trajectory $Q\of t$ is determined by the wave function through the equation
\begin{equation}\label{eq:Bohm}
\dot Q \of t = \frac { j \of{Q\of t, t}  }    { | \psi\of{Q\of t, t}|^2 } ,
\end{equation}
with $j$  defined in eq.\ \eqref{eq:J}.
Hence, the Bohmian velocity, that is the actual velocity with which the Bohmian particle moves, is not directly related to the quantum-mechanical momentum, that rather encodes only information about the possible results of a hypothetical momentum measurement.
Even if the probability of finding a negative momentum in a measurement is  zero, the Bohmian particle can still  have negative velocity and arrive at the detector from behind, or even cross it more than once.%
\footnote{Note that the notion of \emph{multiple crossings} of the same trajectory is genuinely Bohmian, with no analog in quantum mechanics.}
It is in these cases that the current becomes negative.

\begin{figure}
\centering \includegraphics[width=.8\textwidth]{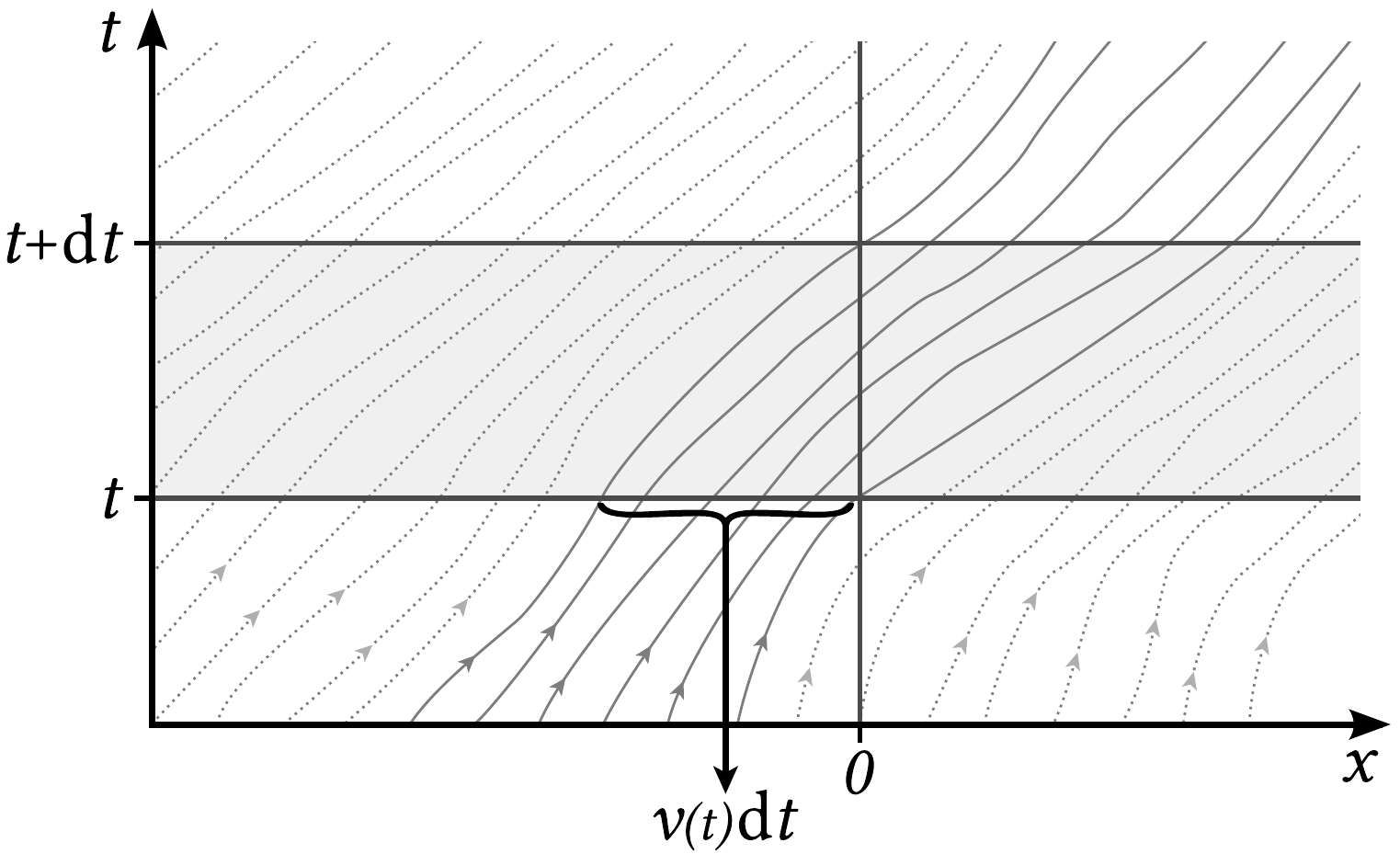}  
\caption{Bohmian trajectories in the vicinity of the detector, placed at $x=0$.
The trajectories, that cross the detector between the times $t$ and $t+\de t$, are those that at time $t$ have a distance from the detector smaller than the distance they cover during the interval $\de t$, that is $v\of t \, \de t$.}
\label{fig:BoltzmannCylinder}
\end{figure}

We can now repeat the derivation of section \ref{sec:EasyDerivation} using the Bohmian velocity instead of the quantum-mechanical momentum.
We consider again an initial state $\psi_0$ such that $\psi_0\of x \approx 0$ if $x\geq0$, but we do not ask anymore the momentum to be positive: we rather ask the Bohmian velocity to stay positive for every time after the initial state is prepared.
The particle crosses the detector between the times $t$ and $t+\de t$ if at time $t$ they are separated by a distance less than $v\of{x=0, t} \,\de t$ (cf.\ fig.\ \ref{fig:BoltzmannCylinder}).
The probability that at time $t$ the particle is in this region is $ v\of{0, t} \, |\psi\of{0, t}|^2\, \de t$, thus the probability density of arrival times is simply 
\begin{equation}
\Pi\of{t} =  v\of{0, t}\, |\psi\of{0, t}|^2 = j\of{0,t}.
\end{equation}

If  the velocity does not stay positive, it is still true that the particle crosses the detector during $(t, t+\de t)$ if at $t$ they are closer than $v\of{x=0, t} \,\de t$, but now this distance can also be negative.
In this case the current $j\of{0,t}$ still entails information about the crossing probability, but it also contains information about the direction of the crossing.
To get a probability distribution from the current we have to clearly specify how to handle the crossings from behind the detector and the multiple crossings of the same trajectory.
For example, one can count only the first time that every trajectory reaches the detector position, disregarding any further crossing, getting the so-called \emph{truncated current} \citep{DaumerDurrGoldstein1997,GrublRheinberger2002}.

The Bohmian analysis is readily generalized to three dimensions with an arbitrarily shaped detector, in which case also the arrival position is found.
More complicated situations, like the presence of a potential, or an explicit model for the detector, can be easily handled too.
Note that the presence of the detector can in principle be taken into account by use of the so-called \emph{conditional wave function} \citep{DuerrGoldsteinZanghi1992,PladevallOriolsMompart2012}, that allows to calculate the actual Bohmian arrival time in exactly the same way as described in this section, although the apparatus needs to be explicitly considered.

\subsection{Is the Bohmian Arrival Time Measurable\\in an Actual Experiment?}

Any distribution calculated from the trajectories conveys some aspects of the actual motion of the Bohmian particle.
Such a distribution does not need in principle to have any connection with the results of a measurement, similarly to the Bohmian velocity that is not directly connected to the results of a momentum measurement.
The Bohmian level of the description is the one we should refer to when arguing about intrinsic properties of the system rather than measurement outcomes.
Since, in the framework of Bohmian mechanics, an intrinsic arrival time exists, namely that of the Bohmian particle, one should ask the intrinsic question that constitutes the title of this section rather than asking the apparatus dependent question
\question{When will the detector click?}
We do not mean that the latter question is irrelevant, to the contrary, it points towards the prediction of experimental results, that is of course of high value.
We shall continue the discussion of the latter topic in section \ref{sec:WhenClicks}.

\subsubsection{Linear measurement of the Bohmian arrival time}\label{sec:BohmPOVM}
We now ask if a linear measurement exists, such that its outcomes are the first arrival times of a Bohmian particle. 
For sure, this can not be exactly true, indeed, if this was the case, then the outcomes of such an experiment would be distributed according to the truncated current, that depends explicitly on the trajectories and is not sesquilinear with respect to the initial wave function as needed for a \povm{}.

However, it is reasonable to expect it to be approximately correct for some set of ``good'' wave functions.
That is motivated by the following considerations.
A typical position detector is characterized by a set of sensitive regions $\{ A_i \subset\R \}_{i=0,\dots,N}$, each triggering a different result.
If the measurement is performed at a fixed time $t$, and  if we get the answer $i$, then the Bohmian particle is at that time somewhere inside the region $A_i$.
A time measurement is usually performed with a very similar set up: one uses a position detector with just one sensitive region $A_0$ (in our case located around the origin) and waits until it fires.
In the ideal case, the reaction time of the detector is very small, and we can consider that the click occurs right after the Bohmian particle entered the sensitive region.
As a consequence, if the Bohmian trajectories cross the detector region only once and do not turn back in its vicinity, then we can expect the response of the actual detector to be very close to the quantum current.
This puts forward the set of wave functions such that the Bohmian velocity stays positive as a natural candidate for the set of good wave functions.
Surprisingly, it can be shown that \emph{there exists no} \povm{} \emph{which approximates the Bohmian arrival time statistics on all functions in this set} (see Chapter~\ref{ch:PRL}).

On the other hand, it is easy to see that the Bohmian arrival time is approximately given by a measurement of the momentum for all \emph{scattering states}, i.e.\ those states that reach the detector only after a very long time, so that they are well approximated by local plane waves.
Numerical evidence for a similar statement for the states with positive Bohmian velocity and high energy was also produced (see Chapter~\ref{ch:PRL}), but a precise determination of the set of good wave functions on which the Bohmian arrival time can be measured is still missing.

An explicit example of a model detector whose outcomes in appropriate conditions approximate the Bohmian arrival time can be found in \citep{DamboreneaEgusquizaHegerfeldt2002}.

\subsubsection{Nonlinear measurement}
An alternative to a linear measurement that directly detects the arrival time of a Bohmian particle is the reconstruction of its statistics from a set of measurements by a nonlinear procedure.

A first possibility in this direction starts by rewriting  the probability current \eqref{eq:J} as
\begin{equation}
 j\of{ x, t} =     
 \bra{ \psi_t}  \  
 	\tfrac{1}{2m}\,   \bigl(  \ket x \bra x \hat p + \hat p \ket x \bra x  \bigr)
\ \ket{\psi_t}  ,
\end{equation}
where $\hat p = -i\hbar\,\partial_x$ is the momentum operator.
The operator 
\begin{equation}
\hat j  \coloneqq  \tfrac{1}{2m}\,   \bigl(  \ket x \bra x \hat p + \hat p \ket x \bra x  \bigr)
\end{equation}
is selfadjoint, therefore it could be possible to measure the current at the position $x$ and at time $t$ by measuring the average value at time $t$ of the operator $\hat j$.
Unfortunately, the operational meaning of this operator is unclear.

A viable solution is offered by weak measurements.
As showed by \citet{Wiseman2007}, it is possible to measure the Bohmian velocity, and therefore the current, by a sequence of two position measurements, the first weak and the second strong, used for postselection.
Wiseman's proposal has been implemented with small modifications in an experiment with photons%
\footnote{This experiment did not, of course, show the existence of a pointlike particle actually moving on the detected paths, but only the measurability of the Bohmian trajectories for a quantum system.}
\citep{KocsisBravermanRavets2011}.
A detailed analysis of the weak measurement of the Bohmian velocity and of the quantum current has been carried out by  \citet{TraversaAlbaredaDi-Ventra2012}.

It is worth noting that the weak measurement of the Bohmian velocity, if intended as a calibration of a non-linear measurement as explained in sec.\ \ref{sec:NotPOVMs}, gives rise to a genuine measurement, i.e.\ one whose outcome reveals the actual velocity possessed by the particle in that run
\citep{DurrGoldsteinZanghi2009}.

\section{When will the Detector Click?}\label{sec:WhenClicks}

We still have to answer the question we  posed at the beginning:
\question{When will the detector click?}

Surely, for any given experiment there is a \povm{} that describes the statistics of its outcomes.
Such an object will depend on the details of the specific physical system and of the measurement apparatus used for the experiment.
That is true not only for time measurements, but for any measurement, and for quantum mechanics as for Bohmian mechanics.
Yet, we can speak for example of the position measurement in general terms, with no reference to any specific setting, as it was disclosing  an intrinsic property of the system.
How can that be?

One can speak of the position measurement and of its \povm{} in general terms because a \povm{} happens to exist, that has all the symmetry properties expected for a position measurement and that does not depend on any external parameter.
That suggests that some kind of intrinsic position exists independently of the measurement details.
Recalling how the \povm{}s have been introduced in sec.\ \ref{sec:Povms}, it is readily clear that they inherently involve an external system (the apparatus) in addition to the system under consideration, and therefore they encode the results of an interaction rather than the values of an intrinsic property.
We also saw in sec.\ \ref{sec:EasyDerivation} how interpreting quantum-mechanical statistics as intrinsic objects leads to a mistake.
It is therefore very important to keep in mind that all \povm{}s describe the interaction with an apparatus.
Having this clear, it still makes sense to look for a \povm{} that does not explicitly depend on any external parameter, meaning with this simply that one does not want to give too much importance to the details of the apparatus.
Such a \povm{} may be regarded for example as the limiting element of a sequence of finer and finer devices, and it does not necessarily correspond to any realizable experiment.
Nevertheless, the fortunate circumstance that occurs for position measurements, for which such an idealized \povm{} exists, does not need to come about for all physical quantities one can think of.

For the arrival time it is possible to show that some \povm{}s exist that have the transformation properties expected for a time measurement \citep{Ludwig1983a}, but in three dimensions it is not possible to arrive at a unique expression in the general case, i.e.\ to something independent of any external parameter.
To do so, one needs to restrict the analysis to detectors shaped as infinite planes, or similarly to restrict the problem to one dimension \citetext{\citealp{Kijowski1974,Werner1986}; see also \citealp{Giannitrapani1997,EgusquizaMuga1999,MugaLeavens2000}}.
In this case, for arrivals at the origin, one finds the \povm{} 
\begin{align}
&K \of{t_1, t_2} 
=
\sum_{\alpha=\pm 1}  \intdefde{t_1}{t_2}{T} \ket{T,\alpha} \bra{T,\alpha}  ,
\\
&\text{with}\quad
\braket{p | T,\alpha}
=
\sqrt{ \frac{|p|}{m h} }\ 
	\theta\of{\alpha p}\ 
	e^{ \frac{i p^2 T} {2m\hbar} }    ,
\end{align}
that corresponds to the probability density of an arrival at time $t$
\begin{equation}\label{eq:KijowskiProbDensity}
\sum_{\alpha=\pm 1}  
| \braket{t,\alpha | \psi_0} |^2 
=
\sum_{\alpha=\pm 1}  
\frac{1} {m h}\,
\left|  
\intdefde{0}{\alpha\infty}{p}
\smash{\sqrt{ |p| }}\ 
\braket{ p | \psi_t}
\right |^2   .
\end{equation}
Note that $K$ is not a projector valued measure because $\braket{T,+ | T,-}\neq0$.
For scattering states $K$ becomes proportional to the momentum operator, and the density \eqref{eq:KijowskiProbDensity} gets well approximated by the probability current \citep{Delgado1998}.
The general conditions under which this approximation holds are still not clear.

\subsection{The Easy Derivation, Once Again}

The analysis of sec.\ \ref{sec:BohmPOVM} of the measurability of the Bohmian arrival time translates quite easily in an approximate derivation of the response of a detector: essentially what we tried to do in sec.\ \ref{sec:EasyDerivation}, just right.

Consider again the setting described in sec.\ \ref{sec:EasyDerivation}, but with an initial state such that the Bohmian velocity stays positive.
That is equivalent to ask that the probability current stays positive, and therefore that the probability that the particle is on the left of the detector decreases monotonically in time.
As described in sec.\ \ref{sec:BohmPOVM}, thinking of the arrival time detector as of a position detector with only one sensitive region $A_0$ around the origin, it is reasonable to expect that for some set of good wave functions the detector will click right when the particle enters $A_0$.
Hence, the probability of a click at time $t$ is approximately equal to the increase of the probability that the particle is inside $A_0$ at that time, i.e.\ to the probability current through the detector.
Therefore, for the good wave functions, the probability current  is expected to be a good approximation of the statistics of the clicks of an arrival time detector.
As remarked in sec.\ \ref{sec:BohmPOVM} the set of the good wave functions is not exactly known, although it is clear that the scattering states are among its elements, and possibly also the states with positive probability current and high energy.

%
%
%
%

\section*{Appendix: Example of Backflow}
\addcontentsline{toc}{section}{Appendix: Example of Backflow}

\begin{figure}
\begin{minipage}{\textwidth/2}%
\centering
\subfloat[\label{subfig:RhoTofX}]{\includegraphics[width=\textwidth]
{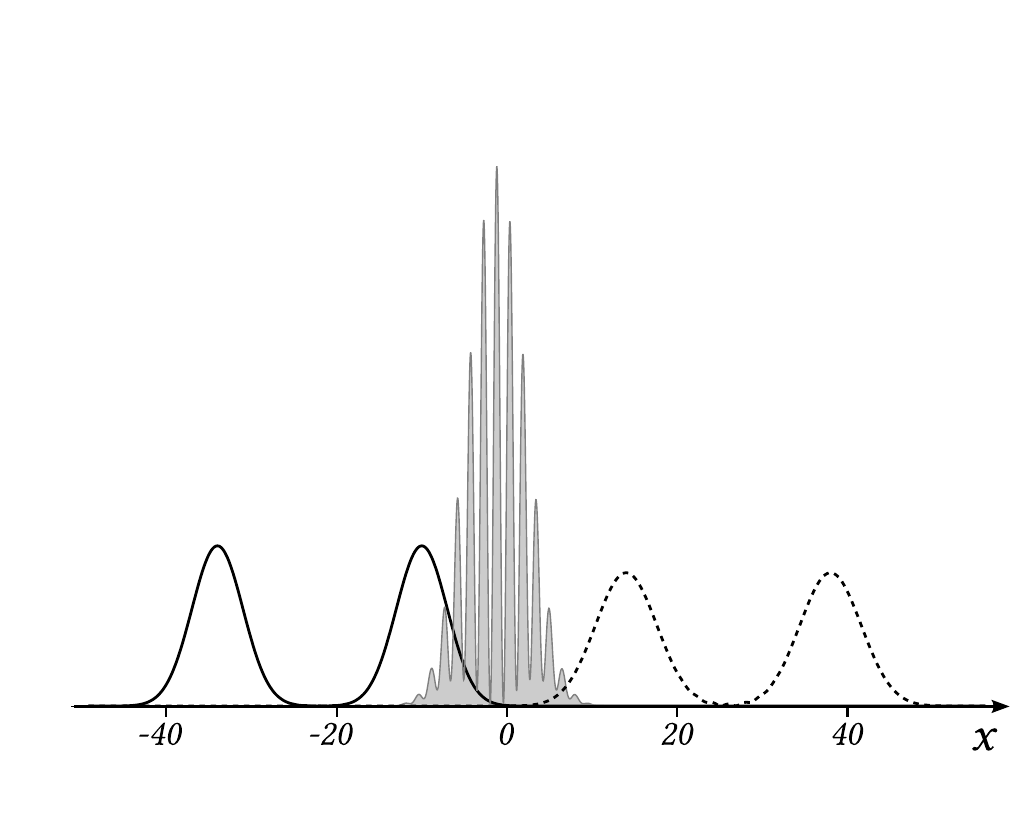}}%
\\
\subfloat[\label{subfig:RhoofTandX}]{\includegraphics[width=\textwidth]{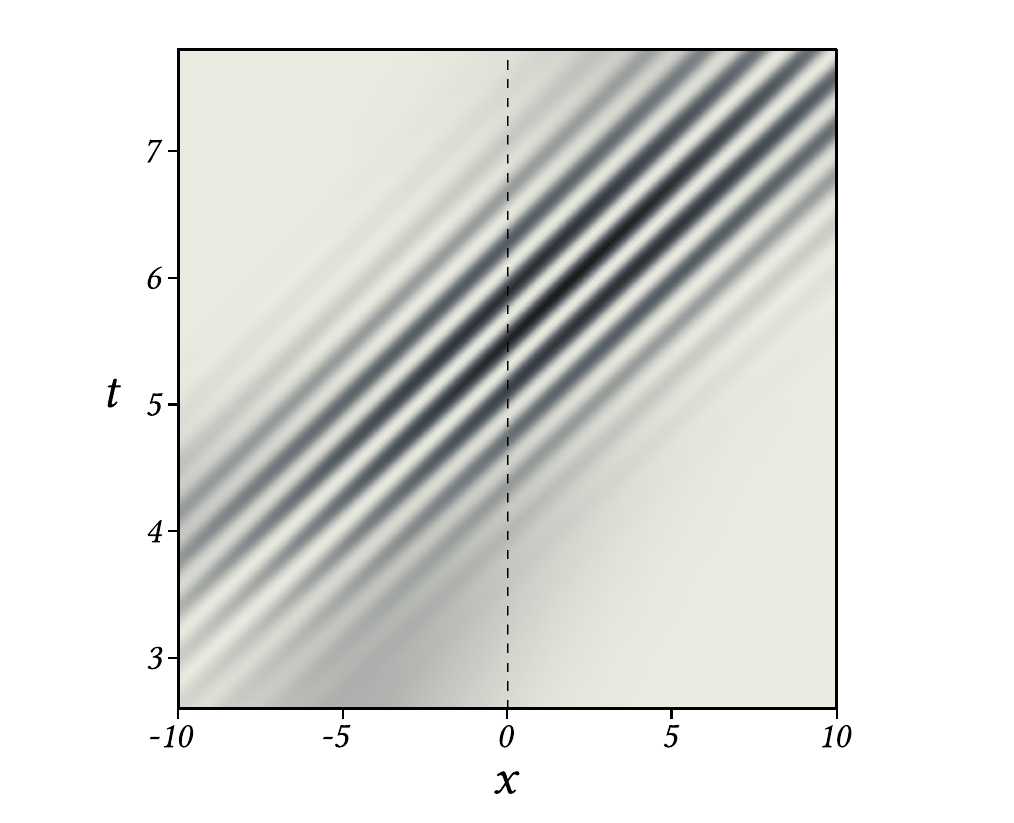}} 
\\
\subfloat[\label{subfig:CurrentAtScreen}]{\includegraphics[width=\textwidth]{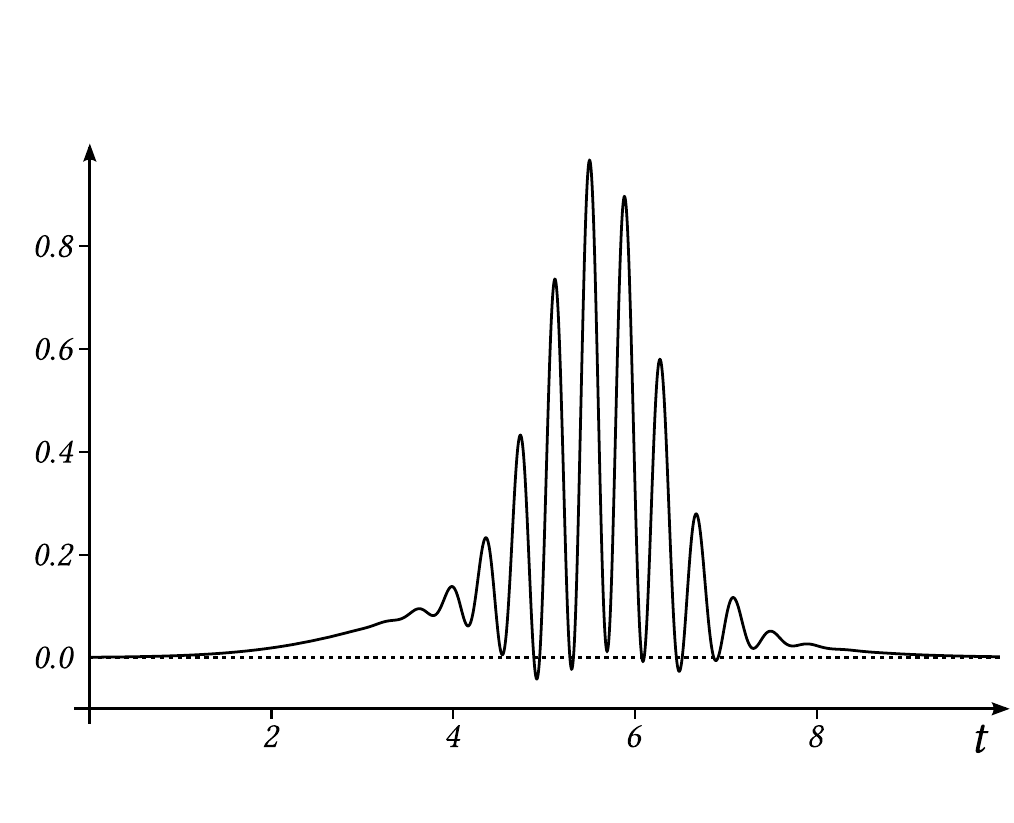}} 
\end{minipage}%
\centering
\begin{minipage}{\textwidth/2}%
\subfloat[\label{subfig:Traj}]{\includegraphics[width=\textwidth]{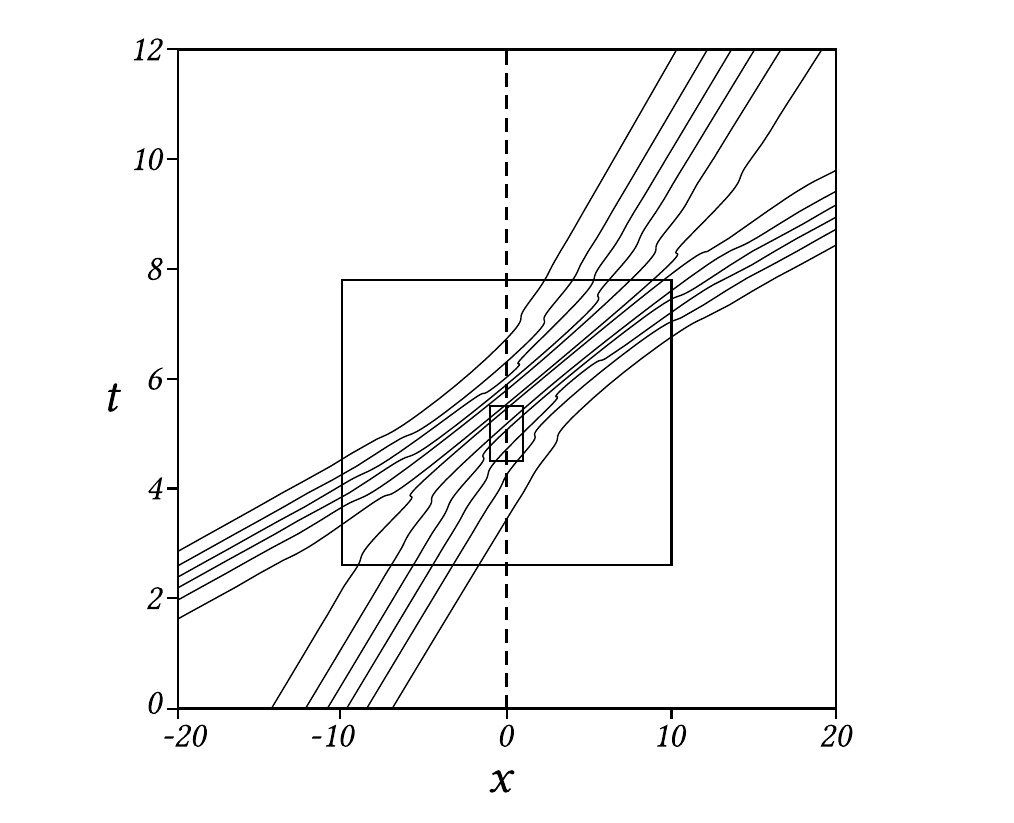}}%
\\
\subfloat[\label{subfig:TrajCloseUp}]{\includegraphics[width=\textwidth]{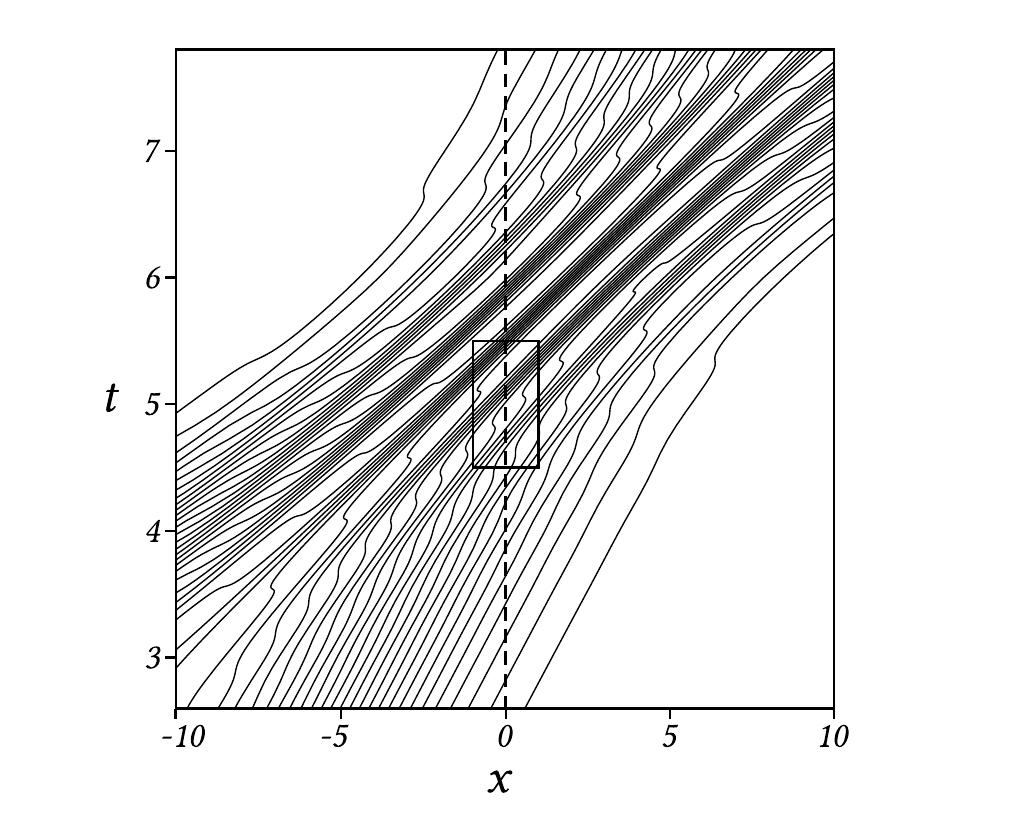}} 
\\
\subfloat[\label{subfig:TrajVeryCloseUp}]{\includegraphics[width=\textwidth]{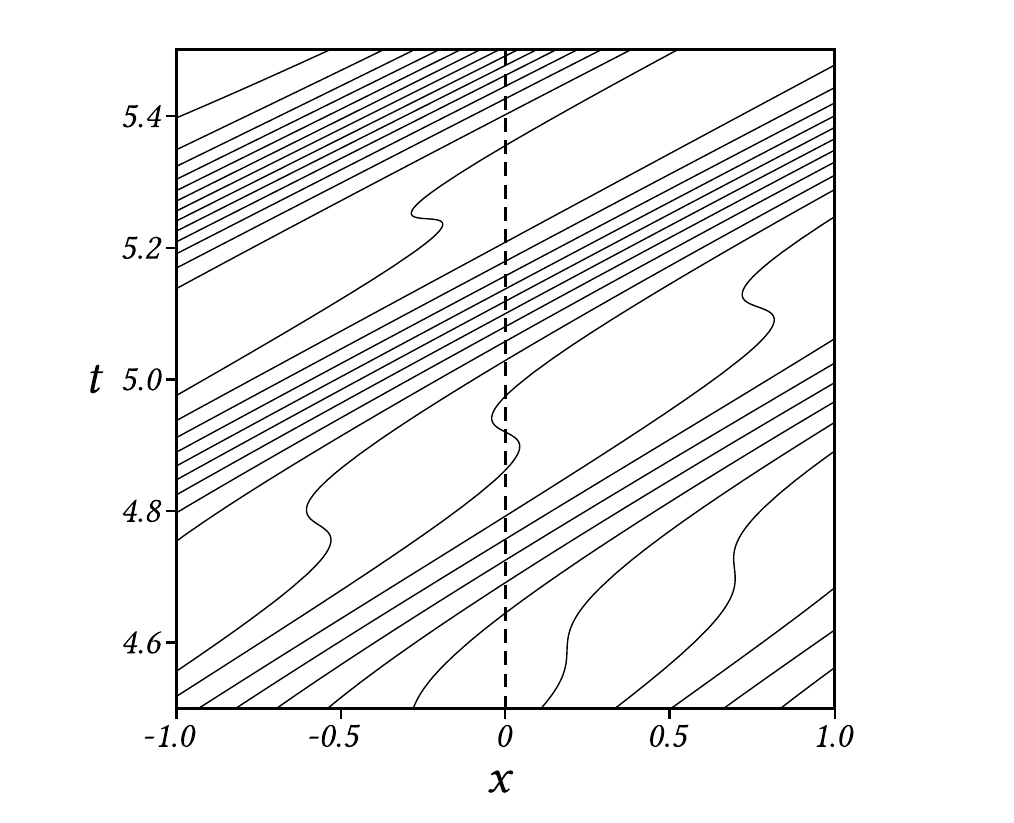}} 
\end{minipage}%
\caption[]{
{Example of backflow: superposition of two Gaussian packets (for the parameters see text).
The dashed line always represents the detector.
\subref{subfig:RhoTofX} 
Probability density of the position at time $t=0$ (solid), $t=5.2$ (filled gray), and $t=12$ (dotted).
\subref{subfig:RhoofTandX} 
Probability density of the position as a function of position and time. 
\subref{subfig:CurrentAtScreen} 
Probability current at the screen as a function of time. 
\subref{subfig:Traj} 
Overall structure of the Bohmian trajectories; the rectangles are magnified in \subref{subfig:TrajCloseUp} and \subref{subfig:TrajVeryCloseUp}.
\thisfloatpagestyle{empty}}}
\label{fig:NegativeCurrent}
\end{figure}

We mentioned that, even for states freely evolving and with support only on positive momenta, the quantum current can become negative.
We provide now a simple example of this circumstance, depicted in fig.\ \ref{fig:NegativeCurrent}.
We use units such that $\hbar=1$, and choose the mass to be one.

We consider the superposition of two Gaussian packets, both with initial standard deviation of position equal to $3$, corresponding to a standard deviation of momentum of $1/6$.
The first packet is  initially centered in $x=-10$ and moves with average momentum $p=2$, while the second packet is centered in $x=-34$ and has  momentum $p=6$.
The probability of negative momentum is in this case negligible.
The second packet overcomes the first when they are both in the region around the origin, where the detector is placed.
In this area the two packets interfere, but then they separate again (cf.\ fig.\ \ref{subfig:RhoTofX}).

In fig.\ \ref{subfig:Traj} the Bohmian trajectories are shown on a big scale.
One can see that they never cross, but rather switch from one packet to the other. 
Moreover, they are almost straight lines, except for the interference region.
In that region, it is interesting to look at a higher number of trajectories, making apparent that the trajectories bunch together, resembling the interference fringes (cf.\ fig.\ \ref{subfig:RhoofTandX} and \ref{subfig:TrajCloseUp}).

Looking at the trajectories more in detail (fig.\ \ref{subfig:TrajVeryCloseUp}), one can see that they suddenly jump from one fringe to the next, somewhen even inverting the direction of their motion.
In this case, it can happen that the particle crosses the detector backwards, leading to a negative current, as shown in fig.\ \ref{subfig:CurrentAtScreen}.

One could argue that Gaussian packets always entail negative momenta, and that this could be the cause of the negative current.
To show that this is not the case, we can compare the probability to have negative momentum 
\begin{equation}
\PP (p<0) = \intdef{-\infty}{0} |\tilde\psi\of p|^2 \, \de p \approx 10^{-33}
\end{equation}
with  the probability to have a negative Bohmian velocity 
\begin{equation}
\PP(v\of {t} <0) = \int_{K_t} \rho\of{x,t} \,\de x  ,
\end{equation}
where $K_t \coloneqq \{  x \in \R  |   j\of{ x, t}  < 0 \}$.
For instance, at time $t=5.2$ this probability is $0.008$ (numerically calculated), therefore the negative current can not be caused by the negative momenta.

\chapter[What Does One Measure When One Measures\\the Arrival Time of a Quantum Particle?]{What Does One Measure When One Measures the Arrival Time of a Quantum Particle?}
\label{ch:PRL}
\blfootnote{This chapter has been published as \citep{VonaHinrichsDurr2013}.}
\section{Introduction}

\subsection{Measurement of Time in Quantum Mechanics} 

Consider the following experiment: a one particle wave function is prepared at time zero in a certain bounded region  $G$ of space; the wave evolves freely, and  around that region are particle detectors waiting for the particle to arrive.
The times and locations at which detectors click are random, without doubts. 
We ask: What is the distribution of these random events?

The measurement of time in quantum mechanics is an old and recurrent theme, mostly because no time observable as self adjoint operator exists \citep{Pauli1958,EgusquizaMuga1999}. Time is therefore not observable in the orthodox quantum mechanical sense, but since clocks exist and  time measurements are routinely done in quantum mechanical experiments, the situation draws attention. 
We stress that the time measurements we discuss in the present chapter are really meant as clock readings triggered by the click of a waiting particle detector.
Usual experiments of this kind are performed in far-field regime, where a semi-classical analysis that connects the arrival time to the momentum operator is sufficient \cite[see for example][sec 10.1]{MugaLeavens2000}.
However, with faster detectors at hand \citep[for example for photons see][]{ZhangSlyszVerevkin2003,PearlmanCrossSlysz2005,RenHofmann2011} it will be soon possible to investigate the near-field regime, where a deeper analysis is needed.
It is important to remark that what is usually called ``time of flight'' in the context of cold-atoms experiments is not a time measurement in the sense described, but rather a measurement of the position probability density after a time of free evolution.

It follows easily from Born's statistical law that ordinary quantum
  measurements are described by  \povm{s},  positive operator valued measures, \citep{Ludwig1983a, DuerrGoldsteinZanghi2004,DurrGoldsteinZanghi2013}.
This fact motivated a longstanding quest for an arrival time \povm{} derived from first principles and independent of the details of the measurement interaction \citep{Kijowski1974,Werner1986,Werner1987,MugaSalaPalao1998,EgusquizaMuga1999,MugaLeavens2000}.

But what classifies an actual experiment as an arrival time measurement?
Surely not the fact that its outcomes are distributed according to a certain \povm{}, otherwise an appropriate computer program could also be called ``arrival time measurement''.
In fact,  the quest for an arrival time \povm{} cannot be grounded in the belief that there exists some  \emph{true} arrival time, 
{whose distribution is conceived as a \povm{} only because instruments readings are distributed according to a \povm{}. }
Indeed, quantum measurements in general do not  actually measure a preexisting value of an underlying quantity, and outcomes rather result from the interaction of the system with the experimental set-up.

One should rather think that any measurement that one would call arrival time measurement must necessarily satisfy some symmetry requirements, and that these requirements identify a class of \povm{}s \citep{Kijowski1974,Werner1986}.
The elements of this class correspond to different realizations of the measurement interaction, and must be treated on a case-by-case basis.

\subsection{The Integral Flux Statistics}
In the simplified case in which the arrival position is not detected---or, similarly, if we restrict to one dimension---a general and easy analysis is possible for the initial states such that the probability that the particle is inside the region $G$ decreases monotonically with time.
To satisfy this requirement it is sufficient that the wave function $\psi$ of the particle belongs to the set
\begin{equation}
\cpc \coloneqq \{  \psi\:  |\ 
{\vect j}_{\psi}({\vect x},t)  \cdot \de{\vect S} \geq 0,
\ 
\forall {\vect x}\in\partial G, \ \forall t \geq  0  \}  ,
\end{equation}
where
\begin{equation}\label{fluxexpression}
{\vect j}_{\psi}({\vect x},t) 
\coloneqq
 \frac{\hbar}{m}\, \Im \left( \psi^*({\vect x},t)\, \nabla \psi({\vect x},t) \right)   
\end{equation} 
is the probability current, $\partial G$ is the boundary of $G$, and $\de{\vect S}$ is the surface element directed outwards.

In these conditions, the probability that the particle crosses  $\partial G$ later than time $t$ is equal to the probability that the particle is inside $G$ at $t$.
Therefore, the probability for an arrival at $\partial G$ during the time interval $\de t$ is given by the \emph{integral flux statistics}
\begin{equation}\label{eq:IntegralProb}
\PP (\de t) 
=
\frac{\de}{\de t} 
	\left( \int_{G} \abs{\psi(\vect x,t)}^2\,  {\de^3\vect x}  \right)
	 \, \de t
=
\left( \int_{\partial G} {\vect j}_{\psi}({\vect x},t)  \cdot \de{\vect S}\right)\, \de t  .
\end{equation}

The previous analysis, together with the fact that any quantum measurement is described by a \povm{}, raises the following question:
\begin{quote}
\emph{Does there exist a  \povm{} which agrees with the 
integral flux statistics \eqref{eq:IntegralProb} on the set \cpc{}?}
\end{quote}
We answer this question in the next sections.

\subsection{Bohmian arrival times}
The flux statistics is most naturally understood in the context of Bohmian mechanics.

In the experiment introduced above the Bohmian particle moves along the continuous trajectory $\vect X(t)$, and arrives at the detector at the time at which $\vect X(t)$ crosses it, therefore a ``true arrival time'' does exist, namely that of the Bohmian particle.

We recall that the Bohmian trajectories are the flux lines of the probability, i.e.\ 
\begin{equation}\label{BM}
\dot{\vect X}(t)
=\frac{{\vect j}_{\psi}\bm({\vect X}(t),t\bm)}{|\psi\bm({\vect X}(t),t)\bm)|^2}  .
\end{equation}
The particle's wave function in Eq.\ \eqref{BM} can in principle also be the so called conditional wave function, which takes into account the interaction with the detector \citep{DuerrGoldsteinZanghi1992,DurrGoldsteinZanghi2013,PladevallOriolsMompart2012}.
A Bohmian particle can in general cross the surface $\partial G$ several times and the probability for having a first arrival at the surface element $\de \vect S$ during the time interval $\de t$ is
\begin{equation}
\label{POPtruncatedflux}
\PP (\de \vect S, \de t) = 
\tilde{\vect j}_{\psi}  \cdot \de{\vect S}\: \de t  ,
\end{equation}
where 
\begin{equation}
\label{truncatedflux}
 { \tilde{\vect j}}_{\psi}({\vect x},t)
 \coloneqq
 \syst{ & {\vect j}_{\psi}({\vect x},t) && \mbox{if $(t,{\vect x})$ is a first exit from
$G$}  
\\ &0&& \mbox{otherwise} }  
\end{equation}
is the so called truncated current \citep{DaumerDurrGoldstein1997,DurrGoldsteinZanghi2013}.
A first exit event $(t,\vect x)$  from the region $G$ is such that the Bohmian trajectory crosses $\partial G$ for the first time since $t=0$ through the point $\vect x \in \partial G$ at time $t$.

In case each Bohmian trajectory crosses the detector surface only once---i.e.\ the wave function belongs to the set $\cpc$---then every exit is a first exit, and  the first arrival statistics is given by the simpler expression
\begin{equation}
\label{POPRAN}
\PP (\de \vect S, \de t) = 
{\vect j}_{\psi}  \cdot \de{\vect S}\: \de t,
\end{equation}
which we shall call the \emph{flux statistics}.
Note that this gives the statistics for both arrival time and arrival position.

Now one may ask if it is possible to design an experiment whose results disclose the  ``true arrival times''.
The outcomes of such an experiment would be distributed according to Eq.~\eqref{POPtruncatedflux}.
 Unfortunately, this is  impossible, since the truncated flux depends explicitly on the trajectory of the particle, and is not sesquilinear with respect to the wave function as needed for a \povm{} \cite[see also][]{RuggenthalerGrublKreidl2005}.
Hence, according to Bohmian mechanics the ``true arrival time'' exists, but its statistics is not given by a \povm{}, so there is no experiment able to measure it (note that this statement is not in contradiction with the fact that the Bohmian trajectories and the quantum flux are detectable  in weak measurements \citep{Wiseman2007,KocsisBravermanRavets2011,TraversaAlbaredaDi-Ventra2012}, or through deconvolution from absorption signals and other limiting operations \citep{DamboreneaEgusquizaHegerfeldt2002,RuschhauptMugaHegerfeldt2009}).
From this circumstance one may  jump to the conclusion that Bohmian mechanics must be false. 
That conclusion is however unwarranted. 
The measurement analysis in Bohmian mechanics yields straightforwardly that the statistics of measurement outcomes are always given by \povm{}s \citep{DuerrGoldsteinZanghi2004,DurrGoldsteinZanghi2013}. There is no inconsistency here.
Observing that a \povm{} is defined on the whole of the Hilbert space, we see that our previous request of measurability was rather strong, in that we allowed any initial state for the particle, even the very bizarre ones.
As a consequence, it is reasonable to restrict our quest for measurability to a subset of  good wave functions, as for example \cpc{}.
Now we may ask the following question:
\begin{quote}
\emph{Does there exist a  \povm{} which agrees with the 
flux statistics \eqref{POPRAN} on the set \cpc{}?}
\end{quote}
This question slightly generalizes that asked in the previous section.

\section{No Go Theorem for the Arrival Time \textsc{Povm}}
For simplicity we consider a particle moving in one dimension with a detector only at one place.
That restricts our analysis to random times only, and makes \eqref{eq:IntegralProb} and \eqref{POPRAN} equivalent, which is sufficient for the purpose at hand; the generalization to three dimensions is straightforward.
We consider that the detector is placed at $D>0$ and that it is active  during the time interval $I=(0, T)$.
The one particle wave is prepared at time zero well located around the origin.

We introduce the set of wave functions
\begin{equation}
\cpcI \coloneqq \{  \psi\:  |\ 
{ j}_{\psi}(D,t) \geq  0,
\ 
\forall t \in I  \}  .
\end{equation}
On these wave functions the flux statistics is the first arrival time statistics.
We want to find out if a \povm{} density $ O_t$ exists, such that
\begin{equation}\label{eq:JPOVM}
\braket{ \psi |  O_t  |  \psi }
= j_{\psi}(D,t)
\quad \forall t \in I,\ \  
 \forall \psi \in \cpcI  .
\end{equation}
In the following we will use the notation
\begin{equation}
j_+ \coloneqq  j_{\psi+\varphi} ,
\qquad
j_- \coloneqq  j_{\psi-\varphi} .
\end{equation}
By sesquilinearity   of \eqref{fluxexpression} we have
\begin{equation}
\label{eq:JLinearity}
j_\psi + j_\varphi  =  \thalf  \bigl(j_+  +  j_- \bigr)  .
\end{equation}
Similarly,
\begin{multline}
\label{eq:POVMLinearity}
\braket{ \psi |  O_t  |  \psi }  +  \braket{ \varphi |  O_t  |  \varphi } = \\
	=  \thalf \bigl(  \braket{ \psi + \varphi |  O_t  |  \psi + \varphi }   
		+  \braket{ \psi - \varphi |  O_t  |  \psi - \varphi } \bigr)  .
\end{multline}

\begin{sidewaysfigure}
\hfill%
\subfloat[\label{subfig:InterferenceSumDifferencePlus}]{\includegraphics[width=.35\columnwidth]{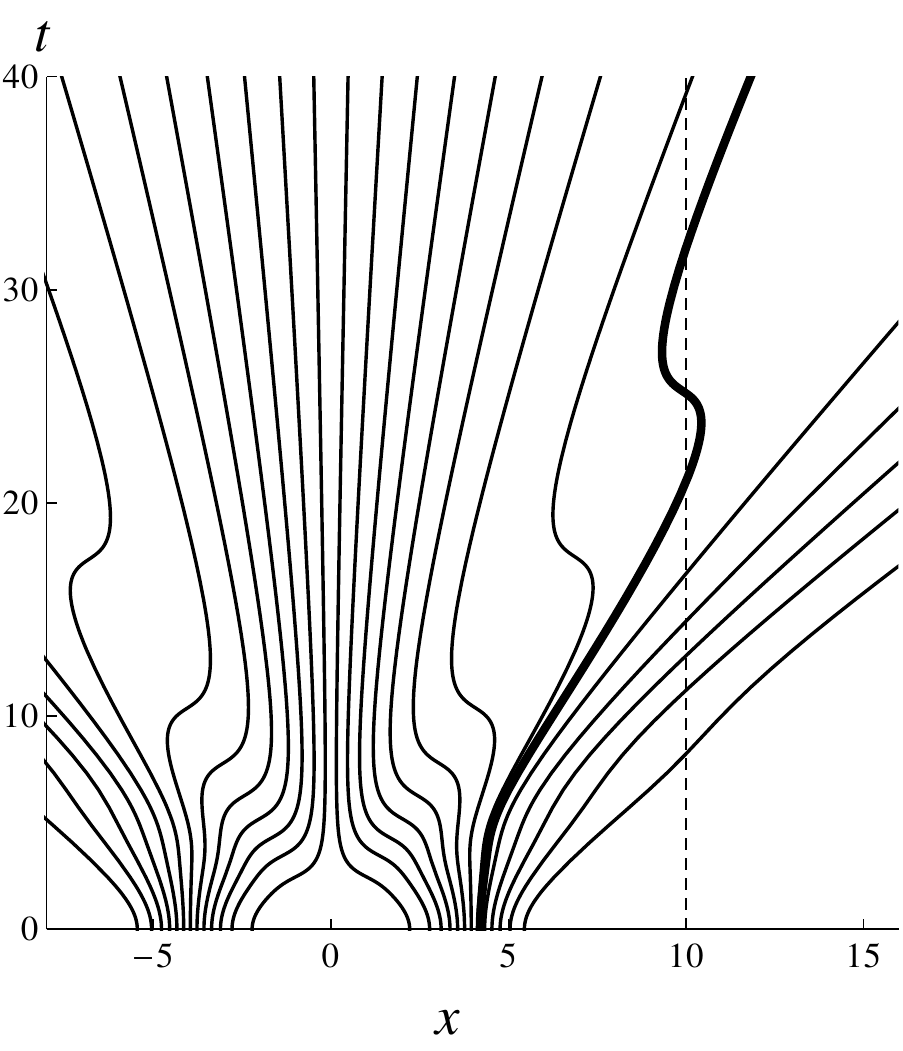}}%
\hfill
\subfloat[\label{subfig:InterferenceSumDifferenceMinus}]{\includegraphics[width=.35
\columnwidth]{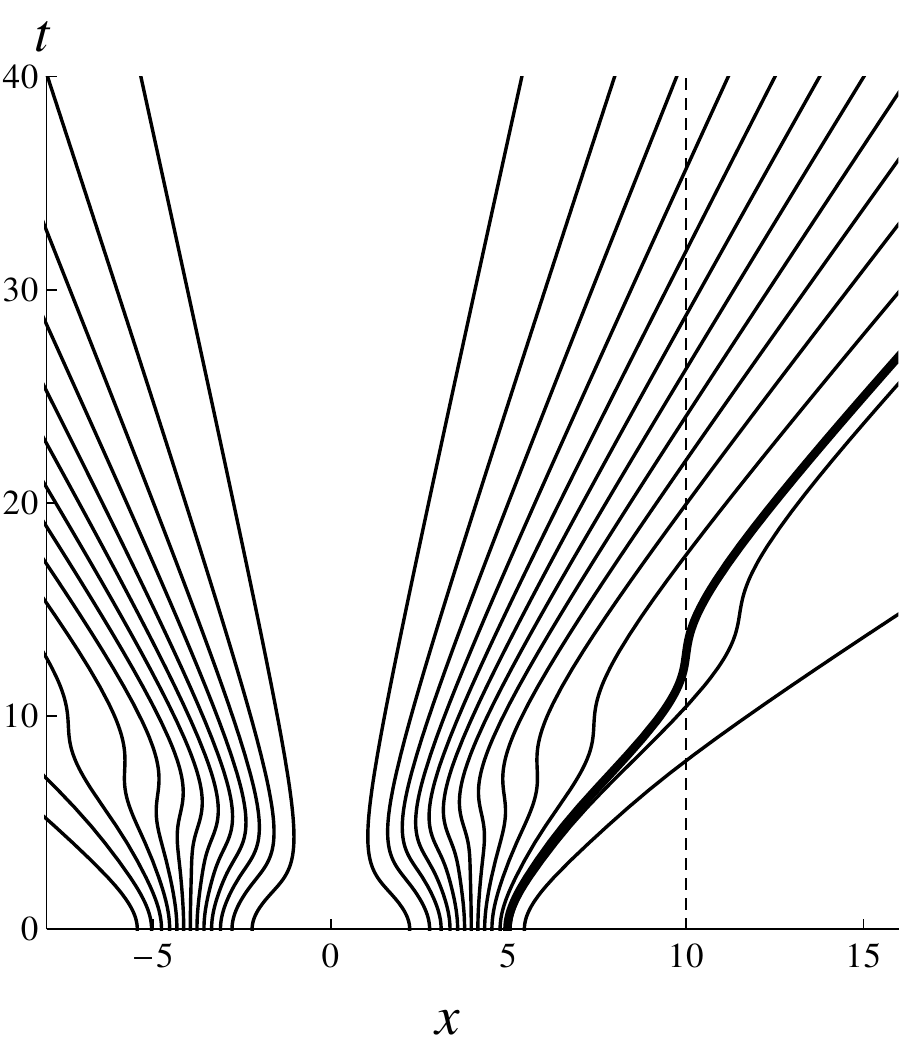}}%
\hfill\mbox{}%
\caption[]{ 
Bohmian trajectories for \subref{subfig:InterferenceSumDifferencePlus} $g_1+g_2$ ($\notin\cpcI$),  and \subref{subfig:InterferenceSumDifferenceMinus} $g_1-g_2$ ($\in\cpcI$), with $g_{1,2}$ Gaussian packets with unitary  position variance, zero mean momentum, and initial mean position equal to $4$ and $-4$, respectively. 
The detector is located at $D=10$ (dashed line).
The initial positions are distributed according to $|\psi|^2$; one further trajectory is shown (bold line), such that it crosses the screen at a minimum of the current.
The units are such that $m=\hbar=1$.
\label{fig:InterferenceSumDifference}}
\end{sidewaysfigure}

Consider now two wave functions $\psi$ and $\varphi$ in $\cpcI$  such that also $\psi+\phi$ is in $\cpcI$, while $j_-(D, t_-)<0 $ for some $t_-\in I$. 
Such functions exist, and an example built with Gaussian wave packets is given in Fig.\ \ref{fig:InterferenceSumDifference}  \citep[see also][for a proposal of a realistic experiment to detect the presence of negative current in similar conditions]{PalmeroTorronteguiMuga2013}.
Requiring \eqref{eq:JPOVM}, we have for every $t$ in $I$ (omitting the argument $D$ in $j$)
\begin{gather}
\braket{ \psi |  O_t  |  \psi } =  j_\psi\of{t},
\qquad
\braket{ \varphi |  O_t  |  \varphi } =  j_\varphi\of{t},
\nonumber\\
\text{and }
\braket{ \psi + \varphi |  O_t  |  \psi + \varphi } =  j_+\of{t}  .
\end{gather}
Substituting in \eqref{eq:JLinearity} and using \eqref{eq:POVMLinearity} we thus get
\begin{equation}
\braket{ \psi - \varphi |  O_t  |  \psi - \varphi }
= j_-\of{t}
\quad \forall t \in I\,.
\end{equation}
But $\braket{ \psi - \varphi |  O_t  |  \psi - \varphi } $ is positive for all $t$ in $I$, while $j_- $ becomes negative at $t_- \in I$, hence a contradiction.
Therefore, \emph{a \povm{} satisfying \eqref{eq:JPOVM} on all functions in \cpcI{} does not exist}.

We can strengthen the result.
Let $O_\psi(t) \coloneqq \braket{ \psi |  O_t  |  \psi }$, 
$\epsilon_\psi\of{t} \coloneqq O_\psi(t) - j_\psi(t)$, 
and let $\epsilon_\pm \coloneqq O_{\psi\pm\varphi}(t) - j_\pm(t)$.
By linearity, i.e.\ subtracting Eq.~\eqref{eq:JLinearity} and \eqref{eq:POVMLinearity},
\begin{equation}
\epsilon_\psi + \epsilon_\varphi  =  \thalf  (\epsilon_+  +  \epsilon_-)  ,
\end{equation}
that implies
\begin{equation}
2\abs{\epsilon_\psi} + 2\abs{\epsilon_\varphi} + \abs{\epsilon_+}
	\geq\abs{2\epsilon_\psi + 2\epsilon_\varphi  - \epsilon_+}
	= \abs{ \epsilon_-}  .
\end{equation}
At a time $t_-$ such that $j_-\of{t_-}<0$, we have $\abs{\epsilon_-\of{t_-}}>\abs{j_-\of{t_-}}$ and thus
\begin{equation}
2\abs{\epsilon_\psi\of{t_-}} + 2\abs{\epsilon_\varphi\of{t_-}} + \abs{\epsilon_+\of{t_-}}
	>\abs{j_-\of{t_-}}  .
\end{equation}
The value $\abs{j_-\of{t_-}}$ is in general not bounded, therefore the error
between any \povm{} and the flux statistics can be arbitrarily large.
 The conclusion is therefore that \emph{there exists no \povm{} which approximates the flux statistics on all functions in \cpcI{}}.

\subsection{The Argument is a Set Argument}
We wish to stress that in the previous section we showed that it is impossible to design an experiment that measures the Bohmian arrival time \emph{on all wave functions in a certain set}, namely \cpcI.
The choice of the set that we consider is crucial, and on a different set our argument may not apply.
To illustrate this point, we present an exaggerated example.
Consider the set of wave functions
\begin{equation}
\cpc_\mathrm G
\coloneqq \{  \psi \in\cpcI \:  |\ 
\psi \text{ is a Gaussian}  \}  .
\end{equation}
For every $\psi$ and $\phi$ in $\cpc_\mathrm G$, neither $\psi + \phi$ nor $\psi - \phi$ is in $\cpc_\mathrm G$, and our argument does not apply.
Of course, the set $\cpc_\mathrm G$ is absolutely artificial and serves only to highlight that our impossibility result depends heavily on the choice of the class of allowed wave functions.

\section{Scattering States}\label{sec:PRL:ScatteringStates}
A class of functions very important from the experimental point of view is that of scattering states, i.e.\ states that reach the detector in far field regime.
These wave functions are particularly important because usual time measurements  are performed in these conditions {\cite[see for example][]{MugaLeavens2000}}.
For these states \citep{BrenigHaag1959,Dollard1969} (in units such that $m=\hbar=1$)
\begin{equation}\label{eq:ScatteringState}
\psi_{\!t}(x)
\approx \frac  {e^{i x^2 /2t}}   {(it)^{1/2}}\ \,
	\tilde\psi \bigl(\tfrac{x}{t} \bigr),
\qquad 
x\approx D,\ 
\forall t \in I  ,
\end{equation}
where $\tilde\psi$ is the Fourier transform of the initial wave function.
As a consequence, it can be shown that \citep{DurrTeufel2009}
\begin{equation}\label{eq:ScatteringCurrent}
j_{\psi}(D,t)
\approx \frac{D}{t^2}\,
	\Abs{\tilde\psi \bigl(\tfrac{D}{t} \bigr)}^2 ,
\qquad \forall t \in I  ,
\end{equation}
and therefore all scattering states are in \cpcI.
A linear combination of scattering states is still a scattering state, indeed Eq.~\eqref{eq:ScatteringState} and \eqref{eq:ScatteringCurrent} apply to the combination as well.
Therefore, no contradiction arises asking for a \povm{} that agrees with the flux statistics on scattering states.
An example of such a \povm, at least approximately, is given by the momentum operator.
This follows from Eq.~\eqref{eq:ScatteringCurrent}, that shows that Bohmian arrival time measurements on scattering states are nothing else than momentum measurements.
In conclusion, our negative result about \cpcI{}  does not forbid to interpret actual, far field time measurements in terms of the flux statistics.

\section{High-Energy Wave Functions}
As already remarked, the set for which we ask accordance between the flux statistics and the \povm{} is crucial.
We found that the set of scattering states presents no problem, but it would be of course much more interesting to identify a subset of \cpcI{}, such that it is possible to measure the Bohmian arrival time also in near field conditions.
We do not have any proof that such a set exists, nevertheless we believe that the subset of \cpcI{} of wave functions with high energy is a good candidate, at least in an approximate sense.

To support our conjecture, we performed some numerical investigations.%
\footnote{See the Appendix at the end of the Chapter for more details.
Our conjecture is supported also by the results obtained by \citet{YearsleyDownsHalliwell2011} for a wide class of clock models.}
We considered as model system the superposition of two Gaussian packets $g_1$ and $g_2$, with equal standard deviation of position  $\sigma$.
If $g_1$ and $g_2$ are both elements of \cpcI, then the eventual negative current of their superpositions must be caused by interference, that is in turn either due to the spreading of the packets, or to their different velocities.

To study the effect of the spreading, we first set the mean momenta of the two packets to be the same, and equal to $k$.
Varying $\sigma$, we found that a threshold $k_\sigma$ exists, such that for $k$ smaller than $k_\sigma$ it is possible that $g_1+g_2$ is in $\cpcI$, but $g_1-g_2$ is not, while for $k$ larger than $k_\sigma$ both $g_1+g_2$ and $g_1-g_2$ are in $\cpcI$.
The threshold $k_\sigma$ increases with decreasing $\sigma$, as expected from the fact that a smaller $\sigma$ means a larger momentum variance, and therefore a larger probability of small momentum.

We examined the effect of a difference in the velocities of the two packets considering the closest packet to the screen to have a fixed momentum $k_1$ well above the value $k_\sigma$, and varying the momentum $k_2$ of the second packet. 
We found that, if $k_2$ is sufficiently far away from $k_1$, then neither $g_1+g_2$ nor $g_1-g_2$ is in \cpcI.
On the contrary, for $k_2$ close to $k_1$, it can happen that the sum is in \cpcI{} and the difference is not, or the other way round.
However, the interval {of $k_2$ values around $k_1$ for } which this happens shrinks (relatively to $k_1$) with growing $k_1$, as well as the maximal value of the negative current.

For the subset of \cpcI{} of  wave functions with high energy it is therefore not true that it is possible to find a \povm{} that agrees \emph{exactly} with the flux statistics, indeed our main argument still applies.
Nevertheless, our numerical study supports the conjecture that it is possible to find  a \povm{} that \emph{approximately} agrees with the flux statistics, with a better agreement for higher energies.

\section{Conclusions}
We showed that  \emph{no \povm{} exists, that approximates the flux statistics on all  functions in \cpcI{}}.
Moreover, the error $\epsilon_\psi$ between a candidate \povm{} and the flux statistic can be very large on any  wave function in \cpcI{}, even for simple states like Gaussians or sum of Gaussians.
{As a consequence, the flux statistics cannot be used to predict the outcomes of an arrival time experiment conducted with wave functions in \cpcI{}.
However, this negative result is very sensitive to the choice of the set and the flux statistic might provide a good arrival time prediction for a more restrictive set of wave functions.
For example, it is indeed possible to find a \povm{} that agrees with the flux statistics on the subset of \cpcI{} composed by scattering states.
Similarly, we conjecture that a \povm{} exists that approximates the flux statistics on the subset of \cpcI{} of wave functions with high energy. 
We produced some numerical evidence to support this conjecture.
}

\newpage
\section*{Appendix: Numerical Investigation}
\addcontentsline{toc}{section}{Appendix: Numerical Investigation} 
We present here the numerical calculations that we performed to investigate if a \povm{} can exist, that agrees with the flux statistics  on the subset of \cpcI{} of wave functions with high energy.
Our model system was the superposition of two Gaussian packets $g_1$ and $g_2$ with equal initial standard deviation of position  $\sigma$.
We used units are such that $m=\hbar=1$, and we considered the time interval $I=(0,T)$, where 
\begin{equation}
T=\frac {x_2-3\sigma}    {k_2+5/12}  ,
\end{equation}
$x_2$ is the initial mean position of the furthest packet from the screen, and $k_2$ is its mean momentum; the term $5/12$ has been inserted by hand to ensure that $T$ is reasonably small also when $k_2$ is zero.

\subsection{Effect of the Spreading}

\begin{figure}[b] 
\centering
\subfloat[\label{subfig:Sigma1}]{\includegraphics[width=.45\columnwidth]{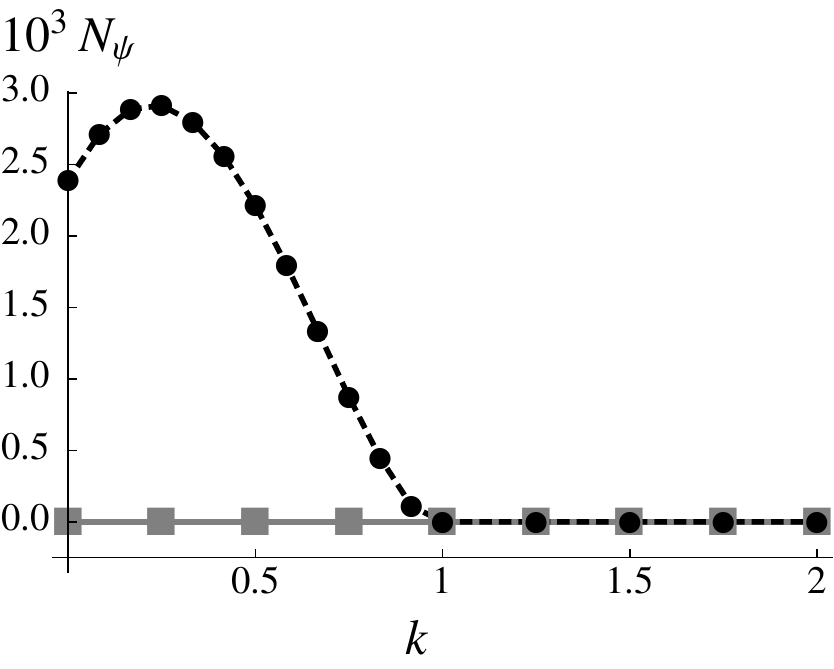}}  
\hfil
\subfloat[\label{subfig:Threshold}]{\includegraphics[width=.45\columnwidth]{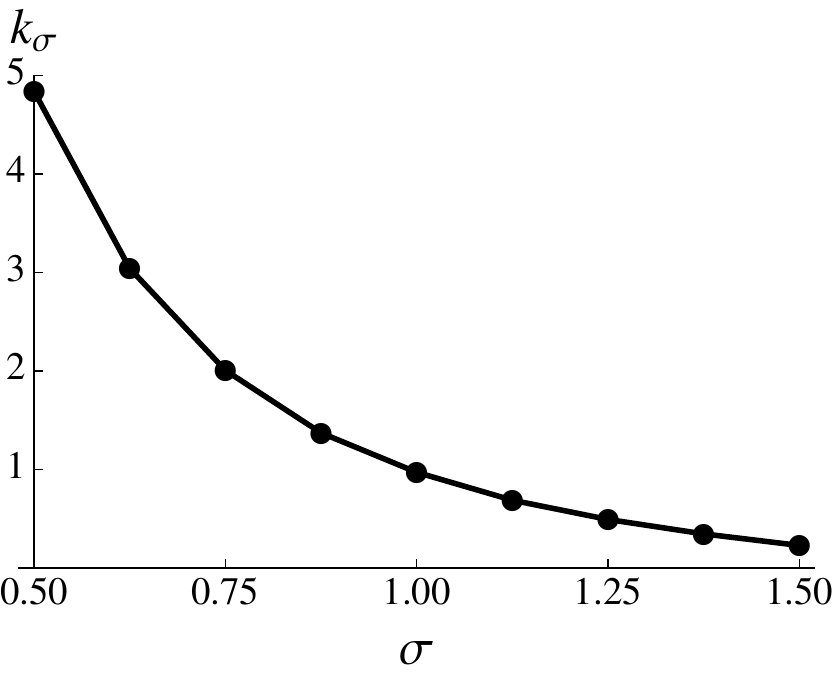}}  
\caption[]{ 
\subref{subfig:Sigma1} 
Numerically integrated negative current $N_\psi$ for $\psi=g_1+g_2$ (circles, dashed black line) and for $\psi=g_1-g_2$ (squares, solid gray line), as functions of the mean momentum $k$ of the two packets.
\subref{subfig:Threshold} 
Threshold $k_\sigma$ as a function of the position variance $\sigma$ of the two packets.
\label{fig:SpreadingIntegral}}
\end{figure}

We studied the effect of the spreading setting the mean momenta of the two packets to be both equal to $k$.
We quantified the total amount of negative current during the interval $I$ by
\begin{align}
N_\psi &\coloneqq
\intdefde{I}{}{t}
| j_{\psi}(D,t)| \
\chi_< (t),
\nonumber\\
&\text{with}\quad\chi_< (t) \coloneqq
\syst{1, &&&j_{\psi}(D,t)<0
\\
0, &&&j_{\psi}(D,t)\geq0  .}  
\end{align}
In Fig.\ \ref{subfig:Sigma1} we plotted $N_{g_1+g_2}$ and $N_{g_1-g_2}$ as functions of $k$, for two Gaussian packets with unitary  position variance, zero mean momentum, and initial mean position equal to $4$ and $-4$, respectively; the detector is located at $D=10$. 
For $k$ bigger than one no negative current is present, and both $g_1+g_2$ and $g_1-g_2$ are in \cpcI{}.

We denoted this threshold value by $k_\sigma$, and we found that it decreases as $\sigma$ increases, as shown in Fig.\ \ref{subfig:Threshold}.

\subsection{Effect of Different Velocities}

\begin{figure}[b] 
\centering%
\includegraphics[width=.72\columnwidth]{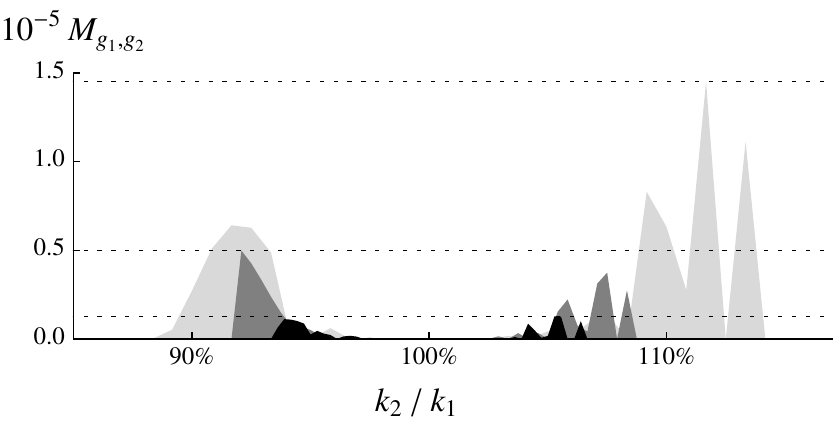}  
\caption{
Numerical values of $M_{g_1,g_2}$  as a function of the ratio of the momenta of the two packets, for  $k_1 = 20$  (light gray), $40$ (gray), and $60$ (black). 
The horizontal dashed lines highlight the maximal values of $M_{g_1,g_2}$.
\label{fig:DifferentVelocitiesIntegral}}
\end{figure}

To study the effect of a difference in the velocities of the two packets we considered $g_1$ and $g_2$ to have again unitary position variance, but different mean momenta $k_1$ and $k_2$, respectively. 
The packet $g_1$ was initially centered around zero, and we considered the values $20$, $40$, and $60$ for its mean momentum $k_1$.
The initial mean position $x_2$ of the packet $g_2$ was such that the two maxima cross in $D+\sigma_t$, where $\sigma_t$ is the position variance at the time of the crossing, and $D=40$ is the detector position.
Consequently, $x_2$ ranged approximately between $-20$ and $-80$, depending on $k_2$.

We studied the quantity
\begin{equation}
M_{g_1,g_2} \coloneqq
(  N_{g_1+g_2} + N_{g_1-g_2}  )\  \chi_0  ,
\end{equation}
where
\begin{equation}
\chi_0 \coloneqq
\syst{0, &&&\text{if either $N_{g_1+g_2}$ or $N_{g_1-g_2}$ is zero,}
\\&&&\text{but not both of them,}
\\
1, &&&\text{otherwise} .}  
\end{equation}
Therefore,  $M_{g_1,g_2}$ is zero when both $g_1+g_2$ and $g_1-g_2$ are in \cpcI{}, as well as when none of them is, while $M_{g_1,g_2}$ is different from zero when one combination is in \cpcI{} and the other one is not.
The results are presented in Fig.\ \ref{fig:DifferentVelocitiesIntegral}, from which it is evident that the interval on which $M_{g_1,g_2}$ is different from zero narrows (relatively to $k_1$) with growing $k_1$, and at the same time the maximal value of $M_{g_1,g_2}$ decreases.

\chapter[How to Reveal the Limits of the\\Semiclassical Approach]{How to Reveal the Limits of the Semiclassical Approach}
\label{ch:Experiment}
The usual experimental conditions in which time measurements are performed  allow to use the semiclassical description borrowed from momentum measurements described in  Sec.~\ref{sec:TimeStat}.
This is due to the fact that the distances involved are such that the detection always happens in far-field regime, and this is true even for \emph{dedicated} time experiments \citep[see for example][]{SzriftgiserGuery-OdelinArndt1996}.
It is then very interesting to devise an experiment in which the limits of the semiclassical approach are clearly exceeded.
This chapter sketches some ideas in this sense; they should be intended as a basis for discussion, not as complete proposals.
The main result of this chapter is the identification of the relevant quantities that have to be considered in the analysis.

\section{Overtaking Gaussians}\label{sec:OvertakingGaussians}


\begin{sidewaysfigure}
\centering
\hfill\subfloat[\label{subfig:DoubleGaussianX0}]{\includegraphics [width=.3\textwidth]{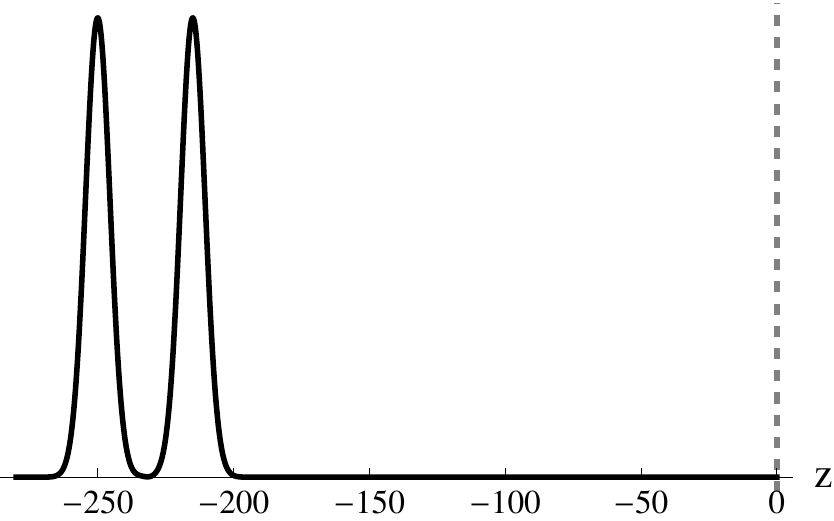}}%
\hfill\subfloat[\label{subfig:DoubleGaussianK}]{\includegraphics [width=.3\textwidth]{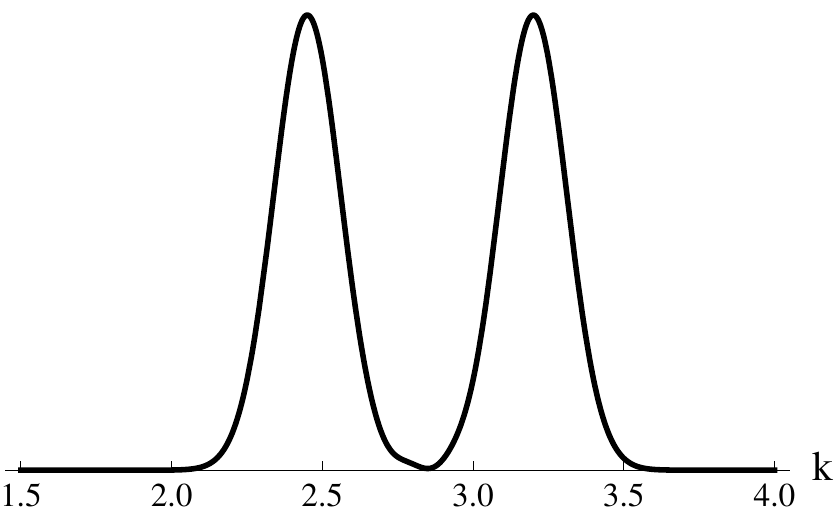}}%
\hfill\subfloat[\label{subfig:DoubleGaussianT}]{\includegraphics [width=.3\textwidth]{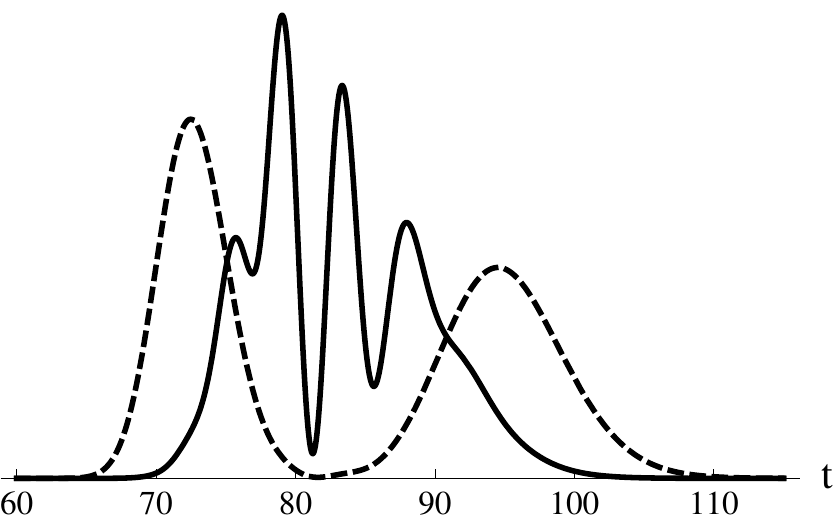}}%
\hfill\mbox{}%
\\
\hfill\subfloat[\label{subfig:DoubleGaussianX1}]{\includegraphics [width=.3\textwidth]{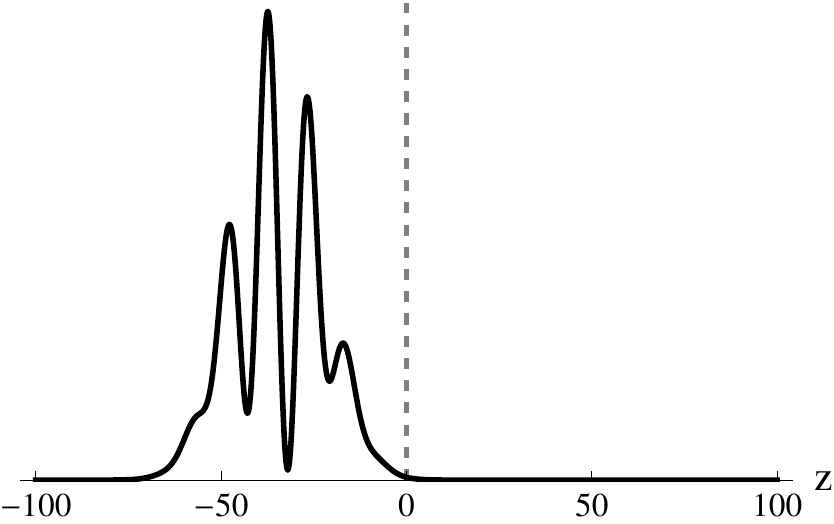}}%
\hfill\subfloat[\label{subfig:DoubleGaussianX2}]{\includegraphics [width=.3\textwidth]{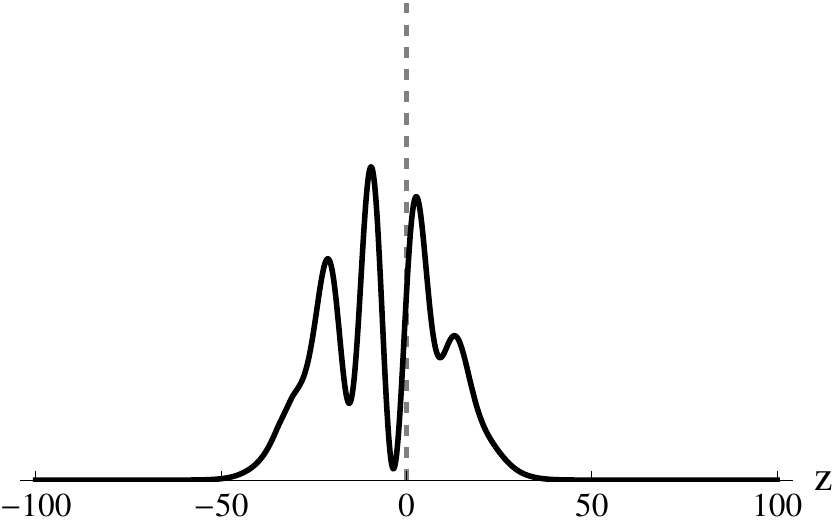}}%
\hfill\subfloat[\label{subfig:DoubleGaussianX3}]{\includegraphics [width=.3\textwidth]{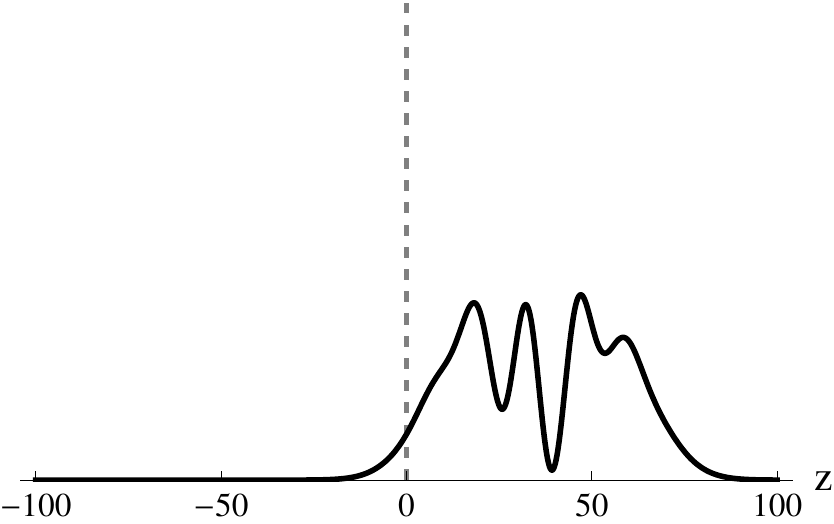}}%
\hfill\mbox{}
\caption[]{
Superposition of two free Gaussian packets in one dimension.
The units are such that $\hbar$ and $m$ are equal to one;
 $\sigma_1=\sigma_2=4.5$, $L=215$, $\delta=35$, $v_1=2.45$, $v_2=3.2$.
The first row shows the probability density of: \subref {subfig:DoubleGaussianX0}~initial position, \subref{subfig:DoubleGaussianK}~momentum, and \subref {subfig:DoubleGaussianT}~arrival time (the solid line is calculated from the quantum current as explained in Secs.~\ref{sec:TimeStat} and~\ref{sec:BohmView}; the dashed line is calculated using Eq.~\eqref{eq:TraditionalDensity}, i.e.\ the semiclassical approach).
The pictures~\subref{subfig:DoubleGaussianX1}, \subref{subfig:DoubleGaussianX2}, and~\subref{subfig:DoubleGaussianX3}~show the the probability density of the position at the times $70$, $80$, and $95$ respectively.
The gray dashed line always represents the detector.
In this example the two packets do not cross each other exactly at the detector position, but at $z=-47$.
}
\label{fig:DoubleGaussian}
\end{sidewaysfigure}

Consider the following experiment.
A particle of mass $m$ moves in one dimension $z$ following the free Schrödinger dynamics; initially, the particle is concentrated on the left of the detector, that is placed at $z=0$ (see Fig.~\ref{fig:DoubleGaussian}) and is able to register the instant of arrival. 
Let the initial wave function of the particle be the sum of two Gaussians, one centered in $z_1=-L$, and the second one in $z_2=-L-\delta$, and let them move with mean velocities such that
\begin{equation}
\frac{v_1}{v_2}  =  \frac{L}{L+\delta}
\end{equation}
i.e.\ the two centers arrive at the detector together at the time $T = L/v_1$.
Properly choosing the parameters of the system we can make sure that the two Gaussians are well separated in momentum and initially in position too.
In this way no interference pattern is present, neither in the momentum distribution nor in the position distribution at time zero.
Then, any semiclassical analysis that makes use of the hypothesis of constant velocity will provide a prediction for the arrival time distribution that shows no interference at all (see Eq.~\eqref{eq:TraditionalDensity}).
On the other hand, in the region in which the fastest packet overtakes the slowest one, the position distribution will surely show interference, and so will an actual measurement of the arrival time.

In order to have two packets well separated both in momentum and in initial position, their standard deviations in position $\sigma_i$ must be such that
\begin{equation}\label{eq:Cond}
\left\{
\begin{aligned}
&3(\sigma_1+\sigma_2) \lesssim \delta
\\
& 3 \left(  \frac{1}{2\sigma_1} + \frac{1}{2\sigma_2} \right)  \lesssim  \abs{k_2 - k_1} 
	= \frac{m}{\hbar} (v_2-v_1)  
	=  \frac{\delta}{\hbar L} m v_1  .
\end{aligned}
\right.  
\end{equation}
One additional condition is that the interference fringes must be small enough to be distinguished if compared to the envelope of the signal.
Assuming $\sigma_1\approx\sigma_2$,  the two Gaussian packets are completely overlapped when they are in  the detector region, therefore the envelope of the signal is about ${6\sigma_1\of{T}}/{ v_1}$.
The interference fringes in position are approximately given by $\cos \abs{k_1-k_2}z$, therefore the interference fringes in time  approximately have period ${2\pi}/({\abs{k_1-k_2} v_1} )$.
To have visible interference, the fringe period must be smaller than the envelope signal, that implies
\begin{equation}
 \abs{k_1-k_2}  \gtrsim \frac{\pi}{3\, \sigma_1\of{T}} ,
\end{equation}
that is automatically realized if Eq.~\eqref{eq:Cond} is fulfilled.
Finally, the fringes must be separated enough to be visible to the detector.
If the detector has time sensitivity $\tau$, then we need that
\begin{equation}
\frac{2\pi}{\abs{k_1-k_2} v_1} 
 \gtrsim\tau,
\end{equation}
This condition and Eq.~\eqref{eq:Cond} with $\sigma_1\approx\sigma_2$ can be written together as
\begin{equation}\label{eq:Cond2}
\left\{
\begin{aligned}
&\sigma_1 \lesssim \frac{\delta}{6} ,
\\
&m v_1 \gtrsim \frac{3\hbar L}{\delta\sigma_1}  ,
\\
&\frac{m v_1^2}{2} \lesssim \frac{\pi\hbar L}{\delta\tau}  .
\end{aligned}
\right.  
\end{equation}
To summarize, the described experiment needs that:
\begin{enumerate}
\item The initial  position distribution is the sum of two well separated pulses, for which no interference is present.
\item The farthest pulse must be faster than the closest one, so that they meet in the detector region. The velocity difference must be big enough to have two well separated pulses in momentum too, so that no interference is present.
\item The interference fringes that arise in the detector region must be separated enough to be resolved in time by the detector.
\end{enumerate}
It is interesting to note that for large times the position probability density $\abs{\psi\of{x,t}}^2$ is essentially determined by the momentum probability density, thanks to well known scattering results.
Therefore, the request of having no interference both in initial position and in momentum can be viewed as the request of having no interference in position both at very short and at  very long times.
This means that the two Gaussian packets must be initially well separated in position, then they have to interfere, and finally they have to separate again.
The detector must be placed in the middle region.

One possible implementation of this experiment is to use photons propagating through a dispersive medium.
Nevertheless, present-day detectors are still too slow to allow for such a simple setting, although the technological progress in this field might soon change the situation  \citep[see for example][]{ZhangSlyszVerevkin2003,PearlmanCrossSlysz2005,RenHofmann2011}.
Massive particles do not seem more promising in this sense, as the last condition in Eq.~\eqref{eq:Cond2} shows that a bigger mass requires a better time sensitivity $\tau$ of the detector.
In any case, stroboscopic techniques might be of help, in which a window is periodically opened in front of the detector. 
Varying the frequency and the phase of the opening function one can recover the period of the interference fringes, even if $\tau$ is very poor.
For particles, this might be carried out using optical mirrors.

\section{Converging Double Slit}
\begin{figure}
\newcommand{\filename}{DoubleSlit}
\centering \includegraphics[width=.75\textwidth]{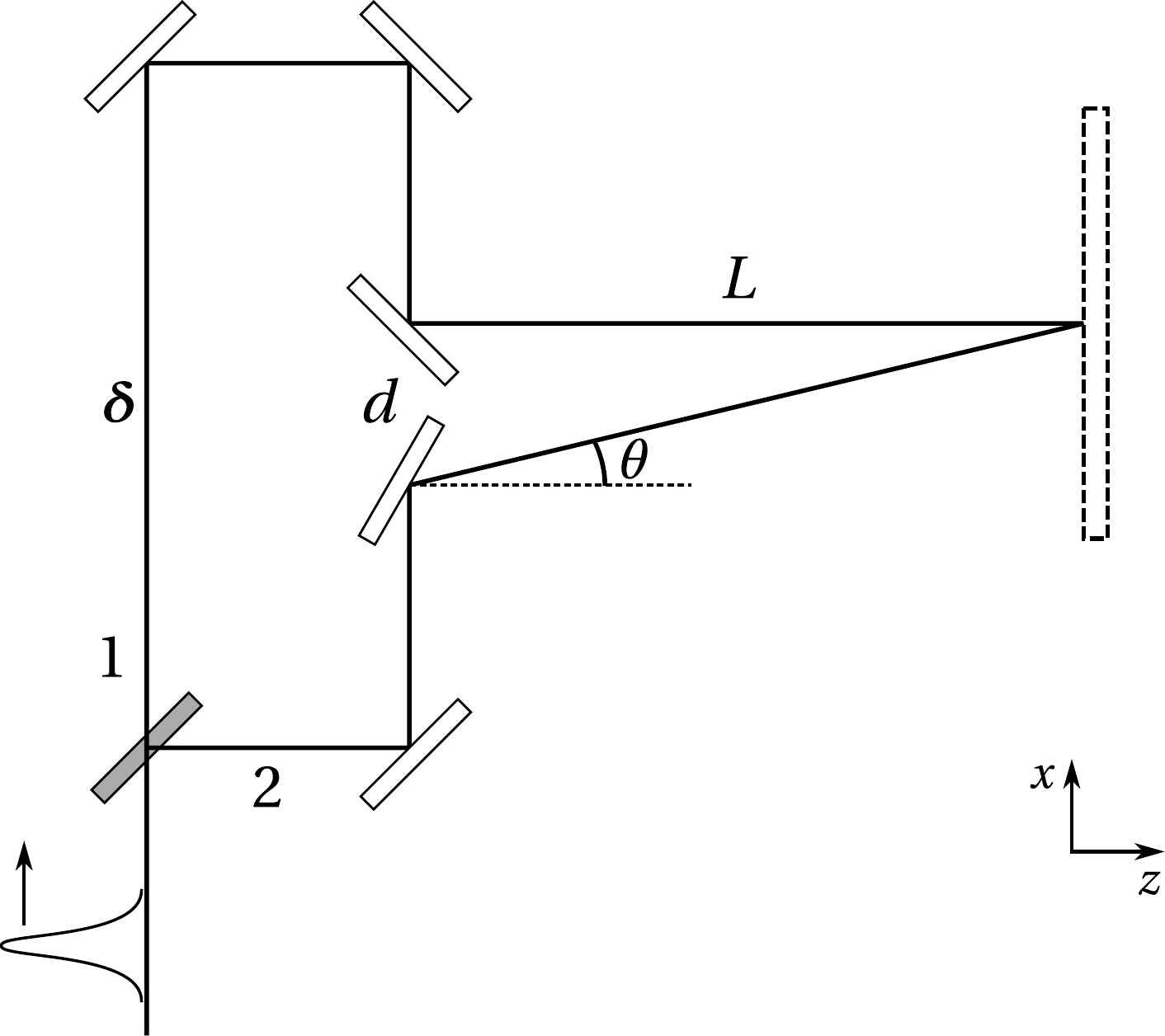}  
\caption{Converging double slit experiment.
A pulse is divided in two by a beam splitter (gray); the first beam covers a longer path and is send towards the detector (dashed rectangle) perpendicularly  to it; the second beam covers a shorter path and arrives at the detector with an angle $\theta$.
We denote by $d$ the separation between the slits, by $L$ the distance between the slits and the detector, and by $\delta$ the difference between the length of the two arms.}
\label{fig:\filename}
\end{figure}

We will now illustrate a setting that is more suited to photons.
The basic idea is to use two packets in two dimensions propagating in non-parallel directions, and then to project them on one dimension.
In this way, the two packets can move with equal velocity, and still have different apparent ones.
Moreover, in two dimensions on can also consider the joint distribution of arrival time and position at the screen.

Figure~\ref{fig:DoubleSlit} shows the arrangement.
A source produces a Gaussian packet that is split in two; each part travels a different length and is then sent towards a screen along non-parallel directions.%
\footnote{The different propagation length and the beam splitter introduce a phase difference between the two beams; nevertheless, this phase can be completely controlled and represents no problem.}
We take as origin of the coordinates $\{x,z\}$ the center of the last mirror of the arm 2, and as time zero that at which the pulse 2 passes at the origin; we consider also the coordinates $\{x',z'\}$ obtained from $\{x,z\}$ by rotation of $\theta$, i.e.\ $x'=\cos \theta \, x - \sin \theta\, z$, and $z'=\cos \theta \, z + \sin \theta\, x$.
We assume that the pulses leaving from the mirrors move in the longitudinal and transverse directions independently, i.e.\ they can be written as products of functions of these coordinates.
Moreover, we assume that when they pass through $z=0$ they have the same total velocity $v$ and the same standard deviations of position $\sigma_{\!l}$ and $\sigma_{\!t}$ in the longitudinal and transversal direction respectively.
The arm~1 is longer than the arm~2 by $\delta$, hence the Gaussian~1 will start in $(x=d, z=0)$ at the time $t_1=\delta/v$.
Let $v,m,\sigma>0$, 
\begin{align}
A_t (\sigma)&\coloneqq \sigma\,  \left(1 +  \frac{i\hbar t}{2m\sigma^2}  \right),
\\
g_{\sigma}\of{y,t}
&\coloneqq
\frac1{ ( 2\pi)^{1/4} A_t^{1/2}}\
\exp\left[{-\frac{y^2}{(4 \sigma A_t) }}\right],
\\
G_{\sigma}\of{y,t}
&\coloneqq
\exp\left[{ \frac{im v}{2\hbar}  \left( 2y-  v  t\right)}\right]
g_{\sigma}\of{y-vt,t}.
\end{align}
The function $g_\sigma$ describes a freely evolving Gaussian packet with zero initial velocity, and $G_\sigma$ one with velocity $v$.
Then, the initial wave function of the particle is 
\begin{align}
\Psi\of{x,z,t} \coloneqq& g_{\sigma_t}(x-d,t-t_1) G_{\sigma_l}(z,t-t_1)
\\
	&+ g_{\sigma_t}(x',t) G_{\sigma_l}(z',t).
\end{align}

Consider now that the detector is not sensitive to the arrival position $x$, but only to the arrival time.
Then, the semiclassical approach prescribes to use Eq.~\eqref{eq:TraditionalDensity} for  the probability density of the arrival time.
Denoting by $\tilde \Psi\of{k_x,k_z}$ the Fourier transform of $\Psi\of{x,z,0}$ in both space variables, we get as probability density of the arrival time
\begin{equation}
\frac{mL}{\hbar t^2}
\int_{-\infty}^{\infty}
\left|\tilde \Psi\left(\frac{p_x}{\hbar},\frac{mL}{\hbar t}\right)\right|^2
\de {p_x}.
\end{equation}

This quantity has to be compared to a probability density derived taking into account the complete quantum nature of the phenomenon. 
We can consider the flux of the quantum probability current 
\begin{equation}
\vec j\of{ x,z, t} \coloneqq  \frac{\hbar}{m}\,  \Im \Psi^*(x,z,t)\ \nabla\Psi(x,z,t)  
\end{equation}
through the detector surface (see Secs.~\ref{sec:TimeStat} and~\ref{sec:BohmView}), that is
\begin{equation}\label{eq:DoubleSlitCurrentFlux}
\intdef {z=L}{}  \vec j\of{ x,z, t}  \cdot \de \vec \sigma
=
\intdef{-\infty}{\infty}    j_z\of{x,z=L, t}\   \de x.
\end{equation}
In Chapter~\ref{ch:PRL} it is shown that the use of the quantum flux requires some caution; nevertheless, in the present application we need only a quantity that reproduces the quantum interference to contrast the semiclassical result, and no special accuracy is needed.

If the detector is also able to register the arrival position $x$, then a more refined analysis is needed.
We need the probability density of getting a click at a given time and at a given position along $x$.
If the detector is very far away from the slits, then different arrival times correspond to different initial momenta in the $z$ direction, while different arrival positions $x$ correspond to different angles of the initial momentum, i.e.\ to different $x$-components of the momentum.
Then, letting
\begin{align}
k(y,t) &\coloneqq \frac{my}{\hbar t},
\end{align}
the probability density of an arrival at time $t$ at the position $x$ is given by
\begin{equation}\label{eq:FirstTraditionalDistrXT}
\left|\tilde \Psi\left(k(x,t),k(L,t)\right)\right|^2
\left| \frac{d k(L,t)}{dt}\, 
\frac{d k(x,t)}{d x} \right| .
\end{equation}
Alternatively, one can maintain that the probability density of having an arrival at time $t$ at the position $x$ is the joint probability of having the momentum ${mL}/{t}$ along $z$ and of being at the position $x$ at time $t$, that gives, denoting by  $\mathcal F_z$  the Fourier transform in the variable $z$ alone,
\begin{equation}\label{eq:SecondTraditionalDistrXT}
\left| \mathcal F_z [\Psi]\of{x,k(L,t),t}  \right|^2 
\left|  \frac{d k(L,t)}{d t} \right| .
\end{equation}
The densities \eqref{eq:FirstTraditionalDistrXT} and \eqref{eq:SecondTraditionalDistrXT} have to be compared to the component $j_z\of{x,z=L, t} $ of the quantum  current along the direction $z$ at time $t$ and at the position $x$.


\section{Pauli Birefringence}

A further possibility to let a wave packet overtake another one on a small distance, is to exploit birefringence.%
\footnote{I owe to Harald Weinfurter some of the ideas presented in this section.}
An example is the propagation of a polarized photon through a birefringent crystal, or of a neutron in an homogeneous magnetic field; in both cases the two polarization/spin channels propagate at different speeds.
We consider the crystal/field to be infinitely extended, i.e.\ we neglect the change of refraction index before and after it.

We model the situation by means of a Pauli particle of mass $m$ evolving in presence of a constant magnetic field $B$ parallel to the propagation direction $z$.
The Hamiltonian of the system is
\begin{equation}
H
= \frac{1}{2 m} p^2 - B \sigma_z\,,
\qquad
\sigma_z
=
\begin{pmatrix}
1 & 0\\
0 & -1
\end{pmatrix},
\end{equation}
and the two $z$-spin channels evolve according to the two independent Schrödinger equations
\begin{equation}
i\hbar\, \partial_t \psi_{z,s} \of{z,t} = - \left(  \frac{\hbar^2}{2 m}\, \Delta + s B  \right)\, \psi_{z,s} \of{z,t} ,
\qquad s=\pm.
\end{equation}
In momentum representation, the time evolution of each $z$-spin channel is simply givrn by the multiplication by a phase factor.
Letting $\tilde\psi$ be the Fourier transform of $\psi$ with respect to $z$, we have
\begin{equation}
\tilde\psi_{z,s}\of{k,t}
=
e^{- \frac{i}{\hbar} (E(k)+sB) t }\
\tilde\psi_{z,s}\of{k,t=0} ,
\qquad
E(k) \coloneqq \frac{\hbar^2k^2}{2m}  .
\end{equation}

Suppose now that before the detector a spin filter is placed, that lets through only the $x,+$ spin component.
Then, the detector will interact with the wave function
\begin{equation}\label{eq:PsiXPlus}
\psi_{x,+} \of{z,t}  =  \frac{1}{\sqrt 2} \left( \psi_{z,+} \of{z,t} +\psi_{z,-} \of{z,t}   \right).
\end{equation}
Note that $\psi_{x,+}$ does not fulfill an autonomous Schrödinger equation, because some probability gets exchanged between the $x,+$ and the $x,-$ spin channels.
The wave functions $\psi_{z,+}$ and $\psi_{z,-}$ contained in $\psi_{x,+}$ can give rise to interference, but it should be noted that the interference here is of a ``different kind''  than that in the case illustrated in Sec.~\ref{sec:OvertakingGaussians}: the interference between $\psi_{z,+}$ and $\psi_{z,-}$ does not suppress the probability in the minima and enhance it in the maxima, but rather the probability gets moved between the $x,+$ and the $x,-$ channels.
Of course, this does not exclude that each of the functions $\psi_{z,+}$ and $\psi_{z,-}$ singly is a sum of packets giving rise to ``normal'' interference.

The wave function $\psi_{x,+}$ in momentum representation reads
\begin{equation}
\tilde\psi_{x,+} \of{k,t}  =  \frac{1}{\sqrt 2} \left( \tilde\psi_{z,+} \of{k,t} +\tilde\psi_{z,-} \of{k,t}   \right),
\end{equation}
therefore, the probability density of having in this spin channel the momentum $p=\hbar k$  is
\begin{equation}
\Pi_{x,+}\of{p,t}
\coloneqq
\frac{1}{2\hbar}  \left |   
	e^{- \frac{i}{\hbar} B t } \tilde\psi_{z,+}\of{k,0}
	+
	e^{ \frac{i}{\hbar} B t } \tilde\psi_{z,-}\of{k,0}
	\right|^2.
\end{equation}
Note that in the general case $\Pi_{x,+}$ oscillates in time and is therefore difficult to apply the usual semiclassical argument, in which it is assumed that the particle moves with constant velocity.
Such an argument makes sense only if time enters in $\Pi_{x,+}$ exclusively through $p$, by virtue of the classical relation $p\of{t}  = m \,L/t$.
Nevertheless, if in analogy to Sec.\ref{sec:OvertakingGaussians} we assume that $\tilde\psi_{z,+}\of{k,0}$ and $\tilde\psi_{z,+}\of{k,0}$ have well separated $k$-supports, then
\begin{equation}
\Pi_{x,+}\of{p,t}
\approx
\frac{1}{2\hbar}   \left(
	\left |  \tilde\psi_{z,+}\of{k,0} \right|^2
	+  \left | \tilde\psi_{z,-}\of{k,0} \right|^2
	\right).
\end{equation}
Hence, the semiclassical analysis is again applicable, giving rise to the probability density of arrival time 
\begin{equation}\label{eq:TDensityPauliSemi}
\frac{mL}{2\hbar^2 t^2}
\left(
	\left |  \tilde\psi_{z,+}\left(\frac{mL}{\hbar t},0\right) \right|^2
	+  \left | \tilde\psi_{z,-}\left(\frac{mL}{\hbar t},0\right) \right|^2
	\right) .
\end{equation}
This density shows no interference pattern.
If in addition $\psi_{z,+}\of{z,0}$ and $\psi_{z,+}\of{z,0}$ have well separated $z$-supports, then we can also be sure that no analogous argument which exchanges the role of position and momentum can give rise to an interference pattern in the time density.

The probability density~\eqref{eq:TDensityPauliSemi} can then be contrasted with the probability current of the channel $x,+$ at the detector position, that is the probability current of the wave function~\eqref{eq:PsiXPlus}.
If the wave functions $\psi_{z,+}$ and $\psi_{z,-}$ cross at the detector region, an interference pattern will appear in the quantum current.

\chapter{On the Energy-Time Uncertainty Relation}
\label{ch:ET}
\blfootnote{The results presented in this chapter are the product of a teamwork with Robert Grummt \citep{GrummtVona2014}.}
\section{Introduction}
The \emph{energy-time uncertainty relation}
\begin{equation}\label{eq:ET}
\Var E \, \Var T \geq \frac{\hbar^2}{4}
\end{equation}
is one of the most famous formulas of quantum mechanics.
But how can it be that this formula is so widely used if the description of time measurements is still an open problem?
How should $\Var T$ that appears in the relation be understood?

For Schrödinger's wave functions one can write the general uncertainty relations
\begin{equation}\label{eq:GeneralUncertainty}
\Var  A \, \Var  B \geq \frac 1 4 | \langle  [A,B] \rangle   |^2,
\end{equation}
where $A$, $B$ are self-adjoint operators, $\Var A$, $\Var B$ are their variances, and $\mean A$, $\mean B$ their means. 
Nevertheless, Eq.~\eqref{eq:ET} cannot be a consequence of this general formula  as no self-adjoint time operator exists~\citep{Pauli1958}.
Despite the difficulties with the treatment of time measurements, many results on the validity of~\eqref{eq:ET} already exist based on general properties.
For these results a variety of situations and of meanings of  the symbol $\Var T$ is considered, and a comprehensive framework is still missing \citep[see for example][]{Busch1990,MugaMayatoEgusquiza2008,MugaMayatoEgusquiza2009}.
For instance, the results in \cite{SrinivasVijayalakshmi1981,Kijowski1974,Giannitrapani1997,Werner1986}  rest on the assumption that the detection happens on the whole time interval $(-\infty,\infty)$, which is appropriate for describing scattering experiments, but  cannot be applied in general.

In this chapter we study the validity of Eq.~\eqref{eq:ET} for the alpha decay of a radioactive nucleus, that is a case study for the uncertainty relation because of the uncertainty on the energy of the emitted alpha particle and on the instant of emission.
In a typical experiment one has a sample  containing unstable nuclei, surrounded by detectors waiting for the decay product to hit them.
The setting is prepared at time zero and the number of decay events is counted starting at that time, so one can not consider the detection window to extend to $-\infty$.
In this case the mentioned results do not apply, and the uncertainty relation~\eqref{eq:ET} could in principle be violated.
This circumstance is indeed general \citep[see][]{LahtiYlinen1987} and easily understood by looking at a particle in a box in an eigenstate of the momentum, for which  $\Var P=0$, while $\Var X$ cannot exceed the size of the box, thereby violating the position-momentum uncertainty relation.
Nevertheless, the  energy-time uncertainty relation is often used for alpha decay \citep[see for example][]{Rohlf} to connect the energy spread of the alpha particle to the lifetime of the nucleus. 

We will start from Gamow's model~\citep{Gamow1928}, where the alpha particle at time zero is trapped inside a barrier potential but subsequently escapes via tunneling and then hits a detector.
We calculate $\Var E$ exactly,  obtain an approximation for $\Var T$, and estimate the error made with this approximation. For potentials with long lifetimes the error is small enough to check the validity of the energy-time uncertainty relation~\eqref{eq:ET}, and we find that it holds. 

To calculate $\Var T$ we used the flux of the probability current through the detecting surface as probability density function for the arrival time of the alpha particle at the detector. From the discussion of Chapter~\ref{ch:BookChapter} we know that the flux of the probability current in general does not have the needed properties to be a probability density function. 
Nevertheless, its use in this case is justified by the fact that the distance between the detector and the decaying nucleus is much bigger than the nucleus itself, therefore the measurement is practically performed under scattering conditions (cf.\ Sec.~\ref{sec:PRL:ScatteringStates}).

Because we have explicit control over $\Var E$ and $\Var T$, we can try to calculate them for physical systems, but unfortunately, for physically reasonable lifetimes, we get that our error bounds are not good enough. As mentioned above, we determine an approximation of $\Var T$ and calculate the error made with this approximation. The error estimates decrease with growing lifetime and if we calculate it for the longest lived element, i.e.\ Bismuth 209 \citep[$1.16\times 10^{-27}\;\mathrm{s}^{-1}$, see][]{Bismuth}, we get that the error is too big for the approximation on $\Var T$ to be reliable or to check Eq.~\eqref{eq:ET}. However, for even longer lived systems that are not physical  the error becomes small enough for us to check the validity of the energy-time uncertainty relation~\eqref{eq:ET}.
The relation between the error and the lifetime is of technical nature, therefore there is no apparent physical explanation for this.

The circumstance that the estimates of this chapter are useful to show the validity of~\eqref{eq:ET} only for unphysical potentials does not prevent from using the result on a principle level.
In particular, we will compare Eq.~\eqref{eq:ET} with the so-called \emph{linewidth-lifetime relation}, which, letting $\tau$ denote the lifetime of the unstable nucleus and $\Gamma$ the full width at half maximum of the probability density function of the energy of the decay product, reads
\begin{align}\label{eq:LL}
  \Gamma\tau=\hbar.
\end{align}
Since $\Gamma$ expresses an uncertainty on energy and $\tau$ on time, this relation is very often explained as an instance of the energy-time uncertainty relation \citep[see for example][]{Rohlf}, although \citet{KrylovFock1947} provided some arguments against this explanation.
To support this arguments, we will estimate the quantities $\Gamma$ and $\tau$ for our model.
We will find that, by adjusting the potential and the initial state, it is possible to make the product $\Gamma\tau$ arbitrarily close to $\hbar$, while at the same time the product of the energy and time variances  gets arbitrarily large.
This explicitly confirms the thesis that the two relations are indeed independent and one can not be interpreted  as a consequence of the other.

\newpage
\section{Assumptions and Definitions}

Throughout this work we will use units in which $\hbar=1$ and the mass $m=1/2$.
Moreover, for ease of notation we introduce for a function $\phi$
\begin{align}
  \dr\phi(k,r)\coloneqq \partial_r\phi(k,r)
  \quad\text{and}\quad
  \dkp\phi(k,r)\coloneqq \partial_k\phi(k,r).
\end{align}

\subsection{Alpha Decay Model}
The theoretical study of alpha decay goes back to Gamow~\citep{Gamow1928}, whose model is based on the one dimensional Schrödinger equation. 
We will summarize his key insight for the three dimensional Schrödinger equation
\begin{align}
	i\partial_t\Psi=(-\Delta+V)\Psi \eqqcolon \textsf{\textit H} \Psi,
\end{align}
with rotationally symmetric $V$ having compact support in $[0,R_V]$ because in the following we will work in this setting. We will only be concerned with the case of zero angular momentum to avoid the angular momentum barrier potential, which would not have compact support. In this case the three dimensional Schrödinger equation is equivalent to the one dimensional problem
\begin{align}
  i\partial_t\psi=\left(-\partial_r^2+V\right)\psi
  \eqqcolon H \psi
  \qquad\text{with}\quad
  \Psi(r,\theta,\phi)=\frac{\psi(r)}{r}.
\end{align}

Gamow's key insight was that eigenfunctions $f(k_0,r)$ of the stationary Schrödinger equation
\begin{align}\label{eq:SchroedingerGamow}
 \left (-\partial_r^2+V(r)\right)f(k_0,r)=k_0^2f(k_0,r)
\end{align}
that satisfy the boundary conditions $f(k_0,r)=e^{ik_0r}$ for $r\geq R_V$ and $f(k_0,0)=0$ have complex eigenvalues
\begin{align}
  k_0^2=(\alpha_0-i\beta_0)^2=\alpha_0^2-\beta_0^2-i2\alpha_0\beta_0\eqqcolon E-i\frac{\gamma}{2}
\end{align}
for some $\alpha_0,\beta_0>0$. 
So, the function $f(k_0,r)$ yields a solution 
\begin{equation}
f_t(k_0,r) \coloneqq e^{-ik_0^2 t} f(k_0,r)
\end{equation}
of the time-dependent Schrödinger equation
\begin{align}
  i\partial_tf_t(k_0,r)=\left(-\partial_r^2+V(r)\right)f_t(k_0,r)=\left(E-i\frac{\gamma}{2}\right)f_t(k_0,r)
\end{align}
which decays exponentially in time with lifetime $1/\gamma$ since
\begin{align}\label{eq:ExponentialDecay}
  |f_t(k_0,r)|^2=e^{-\gamma t}|f(k_0,r)|^2 .
\end{align}

Clearly, Gamow's description does not immediately connect with quantum mechanics because it contains complex eigenvalues and exponentially increasing eigenfunctions, that are not square integrable. 
Skibsted analyzed in~\citep{Skibsted86}  the sense in which Gamow's model of alpha decay carries over to quantum mechanics.
We summarize his main result in the next Lemma.

\begin{lemma}[Lemma~3.5 of \protect{\citep{Skibsted86}}]\label{lem:SkibstedET}
Let the three dimensional potential $V$ be rotationally symmetric, compactly supported in $[0,R_V]$, with $\|r V(r)\|_1<\infty$, and let it have no bound states. 
Moreover, let $t\geq0$, $R\geq R_V$, $R_2(t)=2\alpha_0 t+R$, and 
\begin{equation}\label{eq:fR}
f_R(r)\coloneqq \1_R f(k_0,r). 
\end{equation}
Then,
\begin{equation}
\|e^{-iHt}f_R-e^{-ik_0^2t}f_{R_2(t)}\|_2\leq  K(\alpha_0,\beta_0,t) \|f_R\|_2 ,\label{eq:SkibstedMain}
\end{equation}
\begin{multline}
\frac{\sqrt\pi}{4}K(\alpha_0,\beta_0,t)=
\\=
  \left(\frac{\beta_0}{\alpha_0}\right)^{\frac{1}{2}}
  +\left(\frac{\beta_0}{\alpha_0}\right)^{\frac{1}{4}}
  \sqrt{\frac{3\pi}{16}\sqrt{\gamma t}
  \left(\frac{1+20\sqrt{{\beta_0}/{\alpha_0}}}{1+10\sqrt{{\beta_0}/{\alpha_0}}}\right)^2
  +\frac{3}{40}}\,.
\end{multline}
\end{lemma}
In the following we will always assume that 
\begin{equation}
R\geq R_V.
\end{equation}
If $\beta_0\ll\alpha_0$, we see that $K(\alpha_0,\beta_0,t)\ll1$ for several lifetimes $1/\gamma$. 
So for this time span Eq.~\eqref{eq:SkibstedMain} implies that
\begin{align}
  e^{-iHt}f_R\approx e^{-ik_0^2t}f_{R_2(t)},
\end{align}
or in other words $e^{-iHt}f_R$ undergoes approximate exponential decay.

We would like to use $f_R$ as model for the decaying state, but $\Var E$ is not defined for it. Let us show why, assuming that $V$ does not have bound states. Consider the mean energy
\begin{align}
  \langle f_R,Hf_R\rangle=\|H^{\frac{1}{2}}f_R\|_2^2=\|k\hat f_R\|_2^2,
\end{align}
where $\hat f_R$ is the generalized Fourier transform of $f_R$. From Lemma~3.2 in~\citep{Skibsted86} we know that 
\begin{align}\label{eq:SkibstedDistribution}
  \hat f_R=-\frac{1}{2}\left[\frac{e^{i(k_0-k)R}}{k-k_0}\bar S(k)+\frac{e^{i(k_0+k)R}}{k+k_0}\right],
\end{align}
where $S$ is the $S$-matrix. Multiplied with $k$ this function is not square integrable and therefore neither the mean energy $\langle f_R,Hf_R\rangle$ nor the energy variance is defined for $f_R.$ In fact this argument shows that $f_R$ is not in the form domain of $H$, because for this to be the case,
$\langle f_R,Hf_R\rangle$ needs to be finite. While this is completely unproblematic for Skibsted in~\cite{Skibsted86}, it presents a problem for us, since we want to calculate $\Var E$.

Clearly the sharp truncation of $f_R$ causes the tails of the generalized Fourier transform to be so slow in decay that $k\hat f_R(k)$ is not square integrable. We can solve the problem by using a Gaussian cutoff, which is why we will work with the initial wave function
\begin{align}\label{eq:InitialState}
  \psi(r) \coloneqq  f(k_0,r)\left[\1_{R}+\1_{[R,\infty)}\exp\left(-\frac{(r-R)^2}{2\sigma^2}\right)\right]
\end{align}
for some $\sigma>0$. Note that we do not normalize the Gaussian, because we want the wave function to be continuous at $r=R$. For notational convenience we introduce
\begin{align}
  g_R(r) \coloneqq f(k_0,r)\1_{[R,\infty)}\exp\left(-\frac{(r-R)^2}{2\sigma^2}\right),
\end{align}
so that
\begin{align}
  \psi(r)=f_R(r)+g_R(r).
\end{align}
Clearly,  for $\sigma$  small enough $\|g_R\|_2$ is small and the result of Lemma~\ref{lem:SkibstedET} carries over to $  e^{-iHt}\psi$, i.e.
\begin{align}\label{eq:SkibstedSmooth}
  e^{-iHt}\psi\approx e^{-ik_0^2t}f_{R_2(t)}
\end{align}
for several lifetimes.

The following Lemma proves that $H\psi\in L^2(\mathbb R^+)$ so that $\Var E$ exists and is finite for the wave function $\psi$.
\begin{lemma}
  Let the three dimensional potential $V$ be rotationally symmetric with $\|V\|_2<\infty$. Then $\psi$ lies in the domain of self-adjointness of $H$.
\end{lemma}

\begin{proof}
  We start by determining the domain of self-adjointness of $H$ via the Kato-Rellich theorem~\cite[Theorem~X.12]{RS2}. For this purpose define
  \begin{align}
	H_0\coloneqq -\frac{d^2}{dr^2}
  \end{align}
  on $\{\phi\in L^2(\mathbb R^+)\,|\,\phi(0)=0\}$ and let $\mathcal D(H_0)$ denote its domain of self-\\adjointness. From~\cite[p.~144]{RS2} we get
  \begin{align}
	\mathcal D(H_0)
	=\big\{\phi\in L^2(\mathbb R^+)\,|\,
	&\phi(0)=0,\dr\phi\in L^2(\mathbb R^+),
	\nonumber\\
	&\dr\phi\text{ abs. continuous},\ddr\phi\in L^2(\mathbb R^+)\big\}.
  \end{align}
  From the proof of Lemma~\ref{lem:Selfadjointness} we see that $\mathcal D(H_0)\subset Q(H_0)$, so that by Eq.~\eqref{eq:DensityArgument2} we have
  \begin{align}
  	\|\phi\|_\infty\leq\sqrt{2\|\phi\|_2\|\dr\phi\|_2}
  \end{align}
  for all $\phi\in\mathcal D(H_0)$. With the help of the fact that for arbitrary $A,B>0$ and all $\epsilon>0$ there is a $c_\epsilon>0$ such that
  \begin{align}\label{eq:proof_domain}
	\sqrt{AB}=A\sqrt{B/A}\leq\epsilon B+c_\epsilon A,
  \end{align}
  we then arrive at
  \begin{align}
	\|\phi\|_\infty\leq\epsilon\|\dr\phi\|_2+c_\epsilon\|\phi\|_2.
  \end{align}
  Using this, Cauchy-Schwarz, and Eq.~\eqref{eq:proof_domain} again, we obtain
  \begin{align}
	\|V\phi\|_2
	&\leq\|V\|_2\|\phi\|_\infty\\
	&\leq\epsilon\|\dr\phi\|_2+c_\epsilon\|\phi\|_2\\
	&\leq\epsilon\sqrt{\|\phi\|_2\|H_0\phi\|_2}+c_\epsilon\|\phi\|_2\\
	&\leq\epsilon\|H_0\phi\|_2+c_\epsilon\|\phi\|_2,
  \end{align}
  thereby proving that $V$ is infinitesimally $H_0$-bounded. The Kato-Rellich theorem~\cite[Theorem~X.12]{RS2} then shows
  that $H$ is self-adjoint on $\mathcal D(H_0)$.
  
  To prove that $\psi\in\mathcal D(H_0)$, recall that $f(k_0,r)$ is the solution of the stationary Schrödinger equation~\eqref{eq:SchroedingerGamow}, which satisfies the boundary conditions $f(k_0,r)=e^{ik_0r}$ for $r\geq R_V$ and $f(k_0,0)=0$ with $k_0=\alpha_0-i\beta_0$ for some $\alpha_0,\beta_0>0$. For notational convenience set
  \begin{align}
	\chi(r)\coloneqq \1_{R}+\1_{[R,\infty)}\exp\left(-\frac{(r-R)^2}{2\sigma^2}\right),
  \end{align}
  so that $\psi(r)=f(k_0,r)\chi(r)$. 
The boundary conditions on $f$ imply that 
\begin{equation}
\psi(0)=f(k_0,0)=0.
\end{equation}
Now,
  \begin{align}
	\dr\psi(r)&=\dr f(k_0,r)\chi(r)+f(k_0,r)\dr\chi(r),\\
	\dr\chi(r)&=-\1_{[R,\infty)}\frac{(r-R)}{\sigma^2}\exp\left(-\frac{(r-R)^2}{2\sigma^2}\right)
  \end{align}
  and from Theorem~XI.57 in~\cite{RS3} we know that $\dr f(k_0,r)$ is continuous in $r$. This and the boundary conditions on $f(k_0,r)$ yield the estimate
  \begin{align}
	\|\dr\psi\|_2
	&\leq\|\dr f(k_0,r)\chi(r)\|_2+\|f(k_0,r)\dr\chi(r)\|_2\\
	&\leq\|\1_R\dr f(k_0,r)\|_\infty R+|k_0|\,\left\|\1_{[R,\infty)}\exp\left(ik_0r-\frac{(r-R)^2}{2\sigma^2}\right)\right\|_2\nonumber\\
	&\qquad+\left\|\1_{[R,\infty)}\frac{(r-R)}{\sigma^2}\exp\left(ik_0r-\frac{(r-R)^2}{2\sigma^2}\right)\right\|_2\\
	&<\infty.
  \end{align}
  In order to show the absolute continuity of $\dr\psi$, it is sufficient to ensure that  for all $r\in\mathbb R^+$
  \begin{align}
    \dr\psi(r)
    =\dr\psi(R)+\int_R^r \ddr\psi(r') \,dr'    .    
  \end{align}
Observe that $\dr f$ and $\ddr f$ exist for all $r\in\R^+$ because $f$ is a solution of the Schrödinger equation in the ordinary sense.
Moreover, $\dr \chi$ exists and is continuous for all $r\in\R^+$, but it is not differentiable in $r=R$, so $\ddr \chi$ and $\ddr\psi$ exist in the weak sense for all $r\in\R^+$ and in the ordinary sense  for $r\neq R$.
Now consider the function
\begin{equation}
\phi(x,r) 
\coloneqq  \dr\psi(x)+\int_{x}^{r} \ddr\psi(r') \,dr'  ,
\end{equation}
defined for $x,r<R$ and for $x,r>R$.
Consider $r>R$, then
\begin{equation}
\phi(x,r) = \dr\psi(r) 
\quad \forall x>R,
\end{equation}
therefore
\begin{equation}
\lim_{x\to R^+}\phi(x,r) = \dr\psi(r) 
\quad \forall r>R.
\end{equation}
Similarly, one gets
\begin{equation}
\lim_{x\to R^-}\phi(x,r) = \dr\psi(r) 
\quad \forall r<R.
\end{equation}
Due to the continuity of $\dr\psi$, from the definition of $\phi$ we have
\begin{equation}
\lim_{x\to R^\pm} \phi(x,r)     =\dr\psi(R)+\int_R^r \ddr\psi(r') \,dr'   ,
\end{equation} 
from which we get the absolute continuity of $\dr\psi$.

It remains to show that $\|\ddr\psi\|_2<\infty$. 
Clearly,
  \begin{align}
	\|\ddr\psi\|_2\leq\|\ddr f(k_0,r)\chi(r)\|_2+2\|\dr f(k_0,r)\dr\chi(r)\|_2+\|f(k_0,r)\ddr\chi(r)\|_2.
  \end{align}
  The same arguments which led to $\|\dr\psi\|_2<\infty$ can be applied to show the square integrability of $\dr f\dr\chi$. 
In the weak sense,
\begin{equation}
\ddr\chi(r) = \1_{[R,\infty)}\left[\frac{(r-R)^2}{\sigma^4}-\frac{1}{\sigma^2}\right]
	\exp\left(-\frac{(r-R)^2}{2\sigma^2}\right),
\end{equation}
that together with the boundary conditions on $f(k_0,r)$ gives 
\begin{equation}
\|f(k_0,r)\ddr\chi(r)\|_2 <\infty .
\end{equation}
To handle $\|\ddr f(k_0,r)\chi(r)\|_2$, we use the fact that $f(k_0,r)$ satisfies the Schrödinger equation~\eqref{eq:SchroedingerGamow}, which gives
  \begin{align}
	\|\ddr f(k_0,r)\chi(r)\|_2
	&\leq\|V(r)f(k_0,r)\chi(r)\|_2+|k_0|^2\|f(k_0,r)\chi(r)\|_2\\
	&\leq \|V\|_2 \|\psi\|_\infty+|k_0|^2\|\psi\|_2\\
	&<\infty.
  \end{align}
Thus we see that $\|\ddr\psi\|_2<\infty$, which finishes the proof.
\QED
\end{proof}

\subsection{Assumptions on the Potential}\label{sec:AssumptionsPotentialET}
Throughout the chapter we require the potential $V$ to satisfy the assumptions stated in Section~\ref{sec:assumptions}. For convenience we repeat them here: we consider a non-zero, three-dimensional, rotationally symmetric potential $V=V(r)$, that is real, with support contained in~$[0,R_V]$, such that $\sup( \supp V) = R_V$, and $\|V\|_1<\infty$ (note that this implies $\|r V(r)\|_1<\infty$).
We also assume that the potential admits the asymptotic expansion
\begin{equation}
V(r)   \sim   \sum_{n=0}^M  d_n  (R_V - r)^{\delta_n},   
	\qquad \text{as }{r\to R_V^-},
\end{equation}
with $0\leq M<\infty$,  $-1<\delta_0<\dots<\delta_N$,  $d_n\in\R$,  and $d_n$ not all zero. 

In addition to the assumptions of Section~\ref{sec:assumptions} we will assume that $\|V\|_2<\infty$ and that the potential has no bound states, nor virtual states, nor a zero resonance.%
\footnote{In presence of bound states, the current can be a constant, and its use as arrival time statistic is not reasonable.
Moreover, if the potential has a zero-resonance, then the probability that the particle is in the interior of the detector surface decays as $t^{-1}$ \citep[see][]{JensenKato1979}, and the probability current through the detector has then no variance nor mean.}
We also assume that among all  resonances $k_n=\alpha_n-i\beta_n$, $k_0$ is such that $\alpha_0$ and $\beta_0$ are the minimal ones.
For notational convenience we introduce
\begin{equation}
	\beta\coloneqq\beta_0,  \quad  \alpha\coloneqq\alpha_0,
	\text{ and}\quad
	\gamma \coloneqq 4\alpha\beta.
\end{equation}

\subsection{Time Distribution}
The time variance will be calculated using the flux of the quantum current through the detector surface, which we consider to be a sphere of radius $R$ around the origin.
Note that the cut-off radius $R$ is equal to the detector radius, that is a good choice to model all experiments in which one starts with a bulk of material, and the only information available is that the decay products did not hit the detector yet.
Setting
\begin{align}
  \Psi_t(r,\theta,\phi)\coloneqq e^{-i\textsf{\textit H}t}\Psi(r,\theta,\phi)
  \quad\text{and}\quad
  \psi_t(r)\coloneqq e^{-iHt}\psi(r),
\end{align}
the probability current is 
\begin{align}
  \vec J(r,t)
 & =\frac{2}{\|\Psi\|_2^2} \Im \left[  \bar\Psi_t(r,\theta,\phi)  
	\nabla \Psi_t(r,\theta,\phi)\right]
	\\
 & =\frac{2}{4\pi\|\psi\|_2^2} \Im \left[  \frac{\bar\psi_t(r)}{r}
	\nabla \left(\frac{\psi_t(r)}{r}\right)\right],
\end{align}
hence it is zero in the angular directions, while in the radial direction
\begin{align}
J_r (r,t)&=\frac{1}{2\pi\|\psi\|_2^2} \Im \left[  \frac{\bar\psi_t(r)}{r}
	\partial_r \left(\frac{\psi_t(r)}{r}\right)\right]
\\
&=  \frac {1} {2\pi \|\psi\|_2^2 r^2}\Im \left(  \bar\psi_t(r)  \partial_r \psi_t(r)   \right)
\nonumber\\
&\qquad\qquad-\frac {1} {2\pi \|\psi\|_2^2 r^3}\Im \left(  |\psi_t(r) |^2  \right)
\\
&=  \frac {1} {2\pi \|\psi\|_2^2 r^2}\Im \left(  \bar\psi_t(r)  \partial_r \psi_t(r)   \right)
.
\end{align}
Let
\begin{equation}
j(r,t)\coloneqq  \frac {2} {\|\psi\|_2^2}   \Im \left(  \bar\psi_t(r)  \partial_r \psi_t(r)   \right),
\end{equation}
then the flux of the probability current $\vec J(r,t)$ through the detector is simply 
\begin{equation}\label{eq:Flux}
4\pi R^2  \, J_r(R,t)  = j(R,t) .
\end{equation}
The arrival time probability density $\Pi_T$, being defined as the flux~\eqref{eq:Flux} through the detector surface normalized to one on the time interval $(0,\infty)$, then reads
\begin{equation}\label{eq:DefArrivalTimeDistribution}
\Pi_T(t)  
=   \frac {j(R,t)}  {\int_0^\infty    j(R,t')  d t'}.
\end{equation}
Now, the mean arrival time is
\begin{align}
  \mean{t} \coloneqq \int_0^\infty t\,\Pi_T(t)\,dt
\end{align}
and the time variance is
\begin{align}\label{eq:VarT}
  \Var{T}\coloneqq\int_0^\infty (t-\mean{t})^2 \Pi_T(t)\,dt.
\end{align}

We further simplify the expressions for $\mean{t}$ and $\Var T$ in the following Lemma, which shows that rather than $\Pi_T$, the relevant object is
\begin{equation}\label{eq:TrueSurvivalProb}
 \frac {\|\1_{R}   e^{-iHt}  \psi\|_2^2}  {\|\1_{R}   \psi\|_2^2},
\end{equation}
which we will call \emph{non-escape probability}.
\begin{lemma}\label{lem:SimpleVariance}
  Let $t>0$, then
  \begin{equation}\label{eq:TDistribution}
	\Pi_T(t)   =    - \partial_t  \frac {\|\1_{R}   e^{-iHt}  \psi\|_2^2}  {\|\1_{R}   \psi\|_2^2}.
  \end{equation}
  Moreover,
  \begin{align}
	\mean{t}
	&= \int_0^\infty \frac {\|\1_{R}   e^{-iHt}  \psi\|_2^2}  {\|\1_{R}   \psi\|_2^2}\,dt,\label{eq:SimpleMean}\\
	\Var{T}
	&=2\int_0^\infty t \frac {\|\1_{R}   e^{-iHt}  \psi\|_2^2}  {\|\1_{R}   \psi\|_2^2}\,dt-\mean{t}^2.\label{eq:SimpleVariance}
  \end{align}
\end{lemma}

\begin{proof}
With the help of the continuity equation for the probability, which reads
\begin{equation}
\partial_r j(r,t)  +  \partial_t  \frac {|\psi_t(r)|^2 }  {\|\psi\|_2^2}   = 0,
\end{equation}
and the fact that $j(0,t)=0$ for all times, the current can be written as
\begin{equation}\label{eq:AuxCurrent}
j(R,t) 
= \int_0^{R}  \partial_{r}  j(r,t) d r
=  - \int_0^{R}  \partial_t  \frac {|\psi_t(r)|^2 }  {\|\psi\|_2^2}  d r
=  - \partial_t  \frac {\|\1_{R}   \psi_t\|_2^2}  {\|\psi\|_2^2}  .
\end{equation}
This together with Theorem~\ref{thm:main_ac} gives
\begin{equation}\label{eq:AuxCurrentNormalization}
\int_0^\infty    j(R,t)  d t
=  -  \left[  \frac {\|\1_{R}   \psi_t\|_2^2}  {\|\psi\|_2^2}   \right]_0^\infty
=  \frac {\|\1_{R}   \psi\|_2^2}  {\|\psi\|_2^2}.
\end{equation}
Plugging Eqs.~\eqref{eq:AuxCurrent} and~\eqref{eq:AuxCurrentNormalization} into Eq.~\eqref{eq:DefArrivalTimeDistribution} for the arrival time probability density, we obtain Eq.~\eqref{eq:TDistribution}.

Using integration by parts we obtain
\begin{align}
  \mean{t} &= \int_0^\infty t\,\Pi_T(t)\,dt
\\
&  =-\left[t\frac {\|\1_{R}   e^{-iHt}  \psi\|_2^2}  {\|\1_{R}   \psi\|_2^2}\right]_0^\infty
  +\int_0^\infty \frac {\|\1_{R}   e^{-iHt}  \psi\|_2^2}  {\|\1_{R}   \psi\|_2^2}\,dt.
\end{align}
The boundary term clearly vanishes for $t=0$ and for $t\to\infty$ it vanishes because of Theorem~\ref{thm:main_ac}, which proves Eq.~\eqref{eq:SimpleMean}.

The variance can be expressed as
\begin{align}
  \Var T=\mean{t^2}-\mean{t}^2.
\end{align}
Using integration by parts we get
\begin{align}
  \mean{t^2}
  &=\int_0^\infty t^2 \Pi_T(t)\,dt\nonumber\\
  &=-\left[t^2\frac {\|\1_{R}   e^{-iHt}  \psi\|_2^2}  {\|\1_{R}   \psi\|_2^2}\right]_0^\infty
  +2\int_0^\infty t \frac {\|\1_{R}   e^{-iHt}  \psi\|_2^2}  {\|\1_{R}   \psi\|_2^2}\,dt,
\end{align}
where the boundary terms vanish for the same reasons as before. This proves Eq.~\eqref{eq:SimpleVariance}.
\QED
\end{proof}

\begin{remark}
For a sample of radioactive matter initially containing $N_0$ atoms, the number of undecayed atoms $N(t)$ in the sample at time $t$ is equal  to $N_0$ times the non-escape probability, i.e.
\begin{equation}
N(t) = N_0 \frac {\|\1_{R}   e^{-iHt}  \psi\|_2^2}  {\|\1_{R}   \psi\|_2^2},
\end{equation}
therefore the activity $-dN/dt$ is equal to $N_0\Pi_T$.
\end{remark}

\newpage

\section{Main Results}

\subsection{Approximate Time Distribution}
Due to Eq.~\eqref{eq:SkibstedSmooth}, an approximate time variance is obtained from  the approximate arrival time density
\begin{align}
  \Pio_T(t)\coloneqq - \partial_t  \frac {\|\1_{R}   e^{-ik_0^2t}f_{R_2(t)}\|_2^2}  {\|\1_{R}   \psi\|_2^2},
\end{align}
that corresponds to the non-escape probability 
\begin{equation}\label{eq:ApproxSurvivalProb}
 \frac {\|\1_{R}   e^{-ik_0^2t}  f_{R_2(t)}\|_2^2}  {\|\1_{R}   \psi\|_2^2}.
\end{equation}
We call the approximate time variance $\Varo{T}$ and the approximate mean time $\meano{t}$. Analogously to Lemma~\ref{lem:SimpleVariance} we get
\begin{align}
	\meano{t}
	&= \int_0^\infty \frac {\|\1_{R}   e^{-ik_0^2t}  f_{R_2(t)}\|_2^2}  {\|\1_{R}   \psi\|_2^2}\,dt ,
\label{eq:SimpleMeano}
\\
  \Varo{T}
	&=2\int_0^\infty t \frac {\|\1_{R}   e^{-ik_0^2t}f_{R_2(t)}\|_2^2}  {\|\1_{R}   \psi\|_2^2}\,dt-\meano{t}^2 .
	\label{eq:SimpleVaro}
\end{align}

To get an estimate on the error that we make by approximating $\Var{T}$ with $\Varo{T}$ we will use Lemma~\ref{lem:SkibstedET}, but this will only work up to several lifetimes.
To control the long-time behavior of the wave function, we will use the quantitive bounds given in the next Corollary.
Since it is simply the application of the general estimates from Theorem~\ref{thm:main_ac} to the particular situation we are looking at right now, we shift its proof to the Appendix. 
To state the Corollary we will at first define some shorthands for certain compositions of the constants given in Section~\ref{sec:MainResult}.



\newpage

\vspace*{-3\baselineskip}

\parbox[t][1pt]{\textwidth}{
\begin{definition}\label{def:CorollaryConstants}
For $K>0$ let
\begin{align}
   M_{K,\infty}(0) &\coloneqq  e^{\beta R}\left[\frac{2}{\alpha}+\frac{\sigma}{\sqrt 2}E_{\beta,\sigma/\sqrt 2}\right] ,\\
     M_{K,\infty}(1)  &\coloneqq  e^{\beta R}\Biggl[\frac{2^2}{\alpha^2}+\frac{2R+C_{1,K}}{\alpha}+\sigma^2\nonumber\\
     &\qquad
     +\left(R+\beta\sigma^2+\frac{C_{1,K}}{2}\right)\frac{\sigma}{\sqrt 2}E_{\beta,\sigma/\sqrt 2}\Biggr], \\
    M_{K,\infty}(2) &\coloneqq  e^{\beta R}\bigg[\frac{2^4}{\alpha^3}
    +\frac{2^2(2R+C_{1,K})}{\alpha^2}+\left(R^2+RC_{1,K}
    +\frac{C_{2,K}}{2}\right)\frac{2}{\alpha}\nonumber\\
   &\qquad +\sigma^2\left(2R+C_{1,K}+\beta\sigma^2\right)\nonumber\\
	&\qquad+\left(\frac{C_{2,K}}{2}+C_{1,K}(R+\beta\sigma^2)+\sigma^2+(R+\beta\sigma^2)^2\right)\frac{\sigma}{\sqrt 2}E_{\beta,\sigma/\sqrt 2}\bigg], 
\\
M_{1}(0) &\coloneqq e^{\beta R}\left[2\log\left(\frac{2}{\beta}\right)+\frac{\pi}{2}+\frac{\pi\sigma}{2^{3/2}}E_{\beta,\sigma/\sqrt 2}\right], \\
M_{1}(1)&\coloneqq e^{\beta R}\Bigg[\left(2\log\left(\frac{2}{\beta}\right)+\frac{\pi}{2}\right)\left(R+\frac{C_1}{2s}\right)+\frac{\pi}{\beta}\nonumber\\
	&\qquad+\frac{\pi\sigma^2}{2}+\left(R+\beta\sigma^2+\frac{C_1}{2s}\right)\frac{\pi\sigma}{2^{3/2}}E_{\beta,\sigma/\sqrt 2}\Bigg],\\
M_{1}(2) &\coloneqq e^{\beta R}\Bigg[\left(2\log\left(\frac{2}{\beta}\right)+\frac{\pi}{2}\right)\left(R^2+\frac{C_1}{s}R+\frac{C_2}{2s^2}\right)
	\nonumber\\
	&\qquad+\frac{\pi}{\beta}\left(2R+\frac{C_1}{s}\right)+\frac{4}{\beta^2}
		+\frac{\pi\sigma^2}{2}\left(2R+\frac{C_1}{s}+\beta\sigma^2\right)\nonumber\\
	&\qquad+\left(\frac{C_2}{2s^2}+\frac{C_1}{s}(R+\beta\sigma^2)+\sigma^2+(R+\beta\sigma^2)^2\right)\frac{\pi\sigma}{2^{3/2}}E_{\beta,\sigma/\sqrt 2}\Bigg], 
\\
  	\tilde{c}_3 &\coloneqq 27\frac{2^{10}}{\alpha^5}M_{K,\infty}^2(0) z_{ac,K}^2(2)
	+23\pi^2\frac{2^6}{\alpha^3}M_{K,\infty}^2(1) z_{ac,K}^2(1)\nonumber\\
	&\quad	+27\frac{2^2}{\alpha}M_{K,\infty}^2(2) z_{ac,K}^2(0),
	\label{eq:DefinitionC3Tilde}
	\\
	\tilde{c}_4 &\coloneqq 276 \frac{M_{1}^2(0)}{s^5}\left(1+\frac{2^4}{\alpha^2}\right)^4\left(z_{ac}^2(2)+s^2z_{ac}^2(1)+s^4z_{ac}^2(0)\right)\nonumber\\
	&\quad+304\frac{M_{1}^2(1)}{s^3}\left(1+\frac{2^4}{\alpha^2}\right)^3\left(z_{ac}^2(1)+s^2z_{ac}^2(0)\right)\nonumber\\
	&\quad+14 \frac{M_{1}^2(2)}{s}\left(1+\frac{2^4}{\alpha^2}\right)^2z_{ac}^2(0).
  \end{align}
\end{definition}
}

\newpage

\begin{corollary}\label{cor:main}
Let $t>0$, $K=\alpha/4$, and $s<K\leq1$. 
Then, for $n=0,1,2$
  \begin{align}
s_K&=1,
\label{eq:SkIsOne}
\\
\|\1_K{\hat \psi^{(n)}}\|_\infty &\leq  M_{K,\infty}(n)\label{eq:PsiHat1},
\\
\|{\hat \psi^{(n)}}w\|_1	&\leq  M_{1}(n),
\label{eq:PsiHat2}
\\
\|\1_Re^{-iHt}\psi\|_2^2
	&\leq \tilde{c}_3t^{-3}+\tilde{c}_4t^{-4} .
	\label{eq:CorollaryBound}
  \end{align}
\end{corollary}
Lemma~\ref{lem:SkibstedET} and Corollary~\ref{cor:main} allow us to estimate the error on the variance of time. 
The result is given in the following Lemma, which is proven in Section~\ref{sec:ETVariances}.
\begin{lemma}\label{lem:ErrorMeanVariance}
Let  $A>0$,
\begin{align}
    E_{\beta,\sigma}  &\coloneqq \sqrt\pi e^{\beta^2\sigma^2}\left(1+\erf(\beta\sigma)\right),
    \label{eq:EBetaSigma}
    \displaybreak[0]\\
\omega_{(0,A)} 
&\coloneqq
 \left( 2 + \sqrt{E_{\beta,\sigma} \beta\sigma } \right)
	\nonumber\\
	&\times\left[  \frac{4\sqrt{54\beta}}{5}   \, A^{5/4}  
	+\left (  \frac {\sqrt 6 \beta^{1/4} } {\sqrt{5\pi}\alpha^{1/4}}  +  \frac {4\sqrt{\beta}} {\sqrt{\pi\alpha}}  
		+ \sqrt{E_{\beta,\sigma} \beta\sigma } \right) A \right]  ,
\label{eq:BoundIntegralZeroA}
\displaybreak[0]\\
\omega_{[A,\infty)} 
&\coloneqq
 2\beta e^{-2\beta R}
	\left(   \frac{\tilde{c}_3}{2} A^{-2} + \frac{\tilde{c}_4}{3} A^{-3}  \right)  +  \frac{e^{-\gamma A} }{\gamma},
\label{eq:BoundIntegralAInfty}
\displaybreak[0]\\
\zeta_{(0,A)} 
&\coloneqq
\left( 2 + \sqrt{E_{\beta,\sigma} \beta\sigma } \right)
	\nonumber\\
	&\times\left[  \frac{4\sqrt {54\beta}}{9}   \, A^{9/4}  
	+  \frac{1}{2}   \left (  \frac {\sqrt 6 \beta^{1/4} } {\sqrt{5\pi}\alpha^{1/4}}  +  \frac {4\sqrt \beta} {\sqrt{\pi\alpha}}  
		+ \sqrt{E_{\beta,\sigma} \beta\sigma } \right) A^2\right]  ,
\label{eq:BoundIntegralZeroAVariance}
\displaybreak[0]\\
\zeta_{[A,\infty)} 
&\coloneqq
2\beta e^{-2\beta R}
	\left(   \tilde{c}_3 A^{-1} + \frac{\tilde{c}_4}{2} A^{-2}  \right)  +  \frac{e^{-\gamma A} }{\gamma^2}   (1+\gamma A),
\label{eq:BoundIntegralAInftyVariance}
\end{align}
and
\begin{gather}
\omega  \coloneqq \omega_{(0,A)}  +   \omega_{[A,\infty)},
\qquad
\zeta   \coloneqq  \zeta_{(0,A)}  +   \zeta_{[A,\infty)},
\\
\epsilon_T \coloneqq 2\zeta  + \omega^2  +  \frac{2}{\gamma}  \omega .
\label{eq:EpsilonT}
\end{gather}
Then, for the wave function $\psi$ the following error estimates hold
\begin{align}%
\left|  \mean{t} -   \meano{t}  \right|   &\leq  \omega ,
\label{eq:ErrorMean}
\\
\left|  \Var{T}- \Varo  T   \right|
	&\leq  \epsilon_T.
\label{eq:ErrorVariance}
\end{align}
\end{lemma}

\subsection{Validity of the uncertainty relation}
We will now see that there are $\beta$ and $\sigma$ values for which the error estimate $\epsilon_T$ is sufficiently small to check if the uncertainty relation holds.
For these values we will find that the uncertainty relation is satisfied. 
We start by defining
\begin{equation}
\Po \coloneqq \Var E\,\Varo T,
\qquad
\epsilon_P \coloneqq \Var E\,\epsilon_T,
\end{equation}
so that
\begin{equation}
|  \Var E\,\Var T  -  \Po | \leq  \epsilon_P.
\end{equation}
Then, we have the following possibilities:
\begin{description}
\item[\normalfont $\Po-\epsilon_P\geq 1/4$:] this implies that $\Var E \, \Var T\geq 1/4$ and we can  state that the uncertainty relation holds;
\item[\normalfont$\Po+\epsilon_P< 1/4$:] this implies that $\Var E \, \Var T< 1/4$ and  we can  state that the uncertainty relation is violated;
\item[\normalfont$1/4\in (\Po-\epsilon_P, \Po+\epsilon_P{]}$:] in this case we are not able to check the validity of the uncertainty relation.
\end{description}
This situation is summarized in the following
\begin{definition}\label{def:ErrorSmallEnough}
We say that the error $\epsilon_P$ on the product $\Var E \, \Var T$ for the wave function $\psi$ is small enough to allow us to make statements on the validity of the uncertainty relation if $\Po-\epsilon_P\geq 1/4  $ or $\Po+\epsilon_P< 1/4$.
\end{definition}

We will need the next hypothesis, whose validity will be discussed in  Section~\ref{sec:DiscussionHyp}.
Recall that $\nu_{\tilde K}$ was introduced in Definition~\ref{def:SKAndKTildeAndS} as the smallest non-negative integer such that $\alpha_n \geq 2 \tilde K=12 \|V\|_1$ for all $n\geq\nu_{\tilde K}$.

\begin{definition}\label{def:CV}
Let $\CV$ be the set of all one-parameter families of potentials $\{V_b\}_{b\in[0,\infty)}$ satisfying the properties:
\begin{enumerate}
\item For every finite $b\geq0$ the potential $V_b$ satisfies the assumptions of Section~\ref{sec:AssumptionsPotentialET}.\label{item:PotRequirements}
\item There are two constants $c_{1,2}>0$ so that $c_1\leq\alpha(b)\leq c_2$ for all $b\geq0$.\label{item:AlphaConstants}
\item $\lim_{b\to\infty}\beta(b)=0$.\label{item:BetoToZero}
\item $r_0(b) = \sum_{n=0}^\infty \frac{5\beta_n(b)}{\alpha_n^2(b)+\beta_n^2(b)} = O(1)$ as $b\to\infty$.\label{item:RZeroBetaOrder}
\item $\nu_{\tilde K} =  O\left(\left(\log{\beta(b)}\right)^2 \right)$ as $b\to\infty$.  \label{item:NuTildeBetaOrder}
\end{enumerate}
\end{definition}
\begin{hypothesis}\label{hyp:BetaLimit}
The set $\CV$ is not empty.
\end{hypothesis}
Physically, the most important thing is Property~\ref{item:BetoToZero} of Definition~\ref{def:CV}, that means that it is possible to consider potentials that give rise to resonances with arbitrary long lifetime.
For simplicity we give also the following 
\begin{definition}\label{def:BetaLimit}
By $\lim_{\beta\to0}$ we denote the following: pick any family of potentials $\{V_b\}_{b\in[0,\infty)}\in\CV$ and calculate the limit $\lim_{b\to\infty}$.
\end{definition}
Using this notation, we can rewrite Property~\ref{item:NuTildeBetaOrder} of Definition~\ref{def:CV} as
\begin{align}
  \nu_{\tilde K} =  O\left(\left(\log{\beta}\right)^2 \right) ,
  \qquad\text{as } \beta\to0.  \label{eq:NTildeET}
\end{align}

We can now state our Theorem.

\begin{theorem}\label{thm:MainUncertainty}
Let the assumptions of Corollary~\ref{cor:main} be satisfied and consider the wave function $\psi$.
\begin{enumerate}
\item \label{statement:UncertaintyHolds}
 Let the error $\epsilon_P$ be small enough to allow us to make statements on the validity of the uncertainty relation (cf.\ Definition~\ref{def:ErrorSmallEnough}), then 
\begin{equation}
\Var E \, \Var T\geq 1/4.
\end{equation}
\item\label{statement:ErrorSomewhenSmall}
 Let moreover $\sigma=\beta$ and Hypothesis~\ref{hyp:BetaLimit} be satisfied, then
\begin{equation}
\lim_{\beta\to 0}  (\Po -\epsilon_P ) = \infty.
\end{equation}
\end{enumerate}
\end{theorem}

The second statement of the Theorem implies that there actually exist values of $\beta$ and $\sigma$ for which $\Po-\epsilon_P\geq 1/4$  and therefore our error estimate is small enough to check the validity of the uncertainty relation. 
Unfortunately, as we mentioned in the Introduction, this range of parameters requires $\beta$ to be smaller than the value corresponding to the longest lived physical element, i.e.~Bismuth (recall that the lifetime is connected with $\beta$ by $1/(4\alpha\beta)$).

\subsection[Linewidth-lifetime relation]{The energy time uncertainty relation and the linewidth-lifetime relation are different}
The linewidth-lifetime relation~\eqref{eq:LL} has been verified in many experiments, and its validity is often explained making reference to the time-energy uncertainty relation \cite[see e.g.][]{Rohlf}).
In the following we will see that, with $\sigma=\beta$, it is possible to find values of $\beta$ such that the product of the linewidth and the lifetime  is arbitrarily close to $1$, while at the same time the product of $\Var E$ and $\Var T$ is arbitrarily large and hence far from $1/4$.
Therefore, the validity of the linewidth-lifetime relation cannot be a consequence of the validity of the time-energy uncertainty relation, as asserted by \citet{KrylovFock1947}.

In order to prove this statement, we have at first to give a precise definition for the lifetime and for the linewidth of a generic state.

\begin{definition}[Lifetime]\label{def:Lifetime}
Let $\Prob(T\leq t)$ be the arrival time cumulative distribution function, i.e.\ the probability that the decay product reaches the detector at a time $T$ not later than $t$, and let it be continuous. 
The lifetime is the time $\tau$ such that
\begin{equation}
\Prob(T\leq \tau)  =  1 - \frac{1}{e}.
\end{equation}
\end{definition}
In other words, the lifetime is the time at which a fraction $1/e$ of the initial sample has decayed.
In the usual case in which $\Prob(T\leq t)=1-e^{-\nu t}$, then $\tau=1/\nu$.

\begin{definition}[Linewidth]\label{def:Linewidth}
Let $\Pi_E$ be the probability density function of the energy of the decay product, let it be continuous, and let $M$ be its maximal value.
The linewidth $\Gamma$ is the distance between those solutions of the equation
\begin{equation}
\Pi_E(E)  =  \frac{M}{2}
\end{equation}
that lie furthest apart.
\end{definition}
If $\Pi_E$ has just one peak, then $\Gamma$ is its full width at half maximum; in particular, if $\Pi_E$ is of Breit-Wigner shape, i.e.
\begin{equation}
\Pi_E \propto \frac{1} {(E-E_0)^2+ G^2}  ,
\end{equation}
then $\Gamma=G$.

With these definitions we can state the following
\begin{theorem}\label{thm:ETvsLL}
Let $\sigma=\beta$, the assumptions of Corollary~\ref{cor:main} and Hypothesis~\ref{hyp:BetaLimit} be satisfied, then for the wave function $\psi$
\begin{equation}
\lim_{\beta\to 0}   \Gamma \tau = 1 ,
\end{equation}
while
\begin{equation}
\lim_{\beta\to 0}  \Var E\, \Var T= \infty  .
\end{equation}
\end{theorem}


\section{Discussion of Hypothesis~\ref{hyp:BetaLimit}}\label{sec:DiscussionHyp}
Hypothesis~\ref{hyp:BetaLimit} holds if the requirements in Definition~\ref{def:CV} are satisfied. 
Therefore, we will now give arguments why there exist potentials that satisfy them.

\subsection{Properties~\ref{item:PotRequirements}-\ref{item:BetoToZero} in Definition~\ref{def:CV}}

\begin{figure}
\centering%
\hfill%
\subfloat[\label{fig:potentialET}]{\includegraphics [width=.45\textwidth]{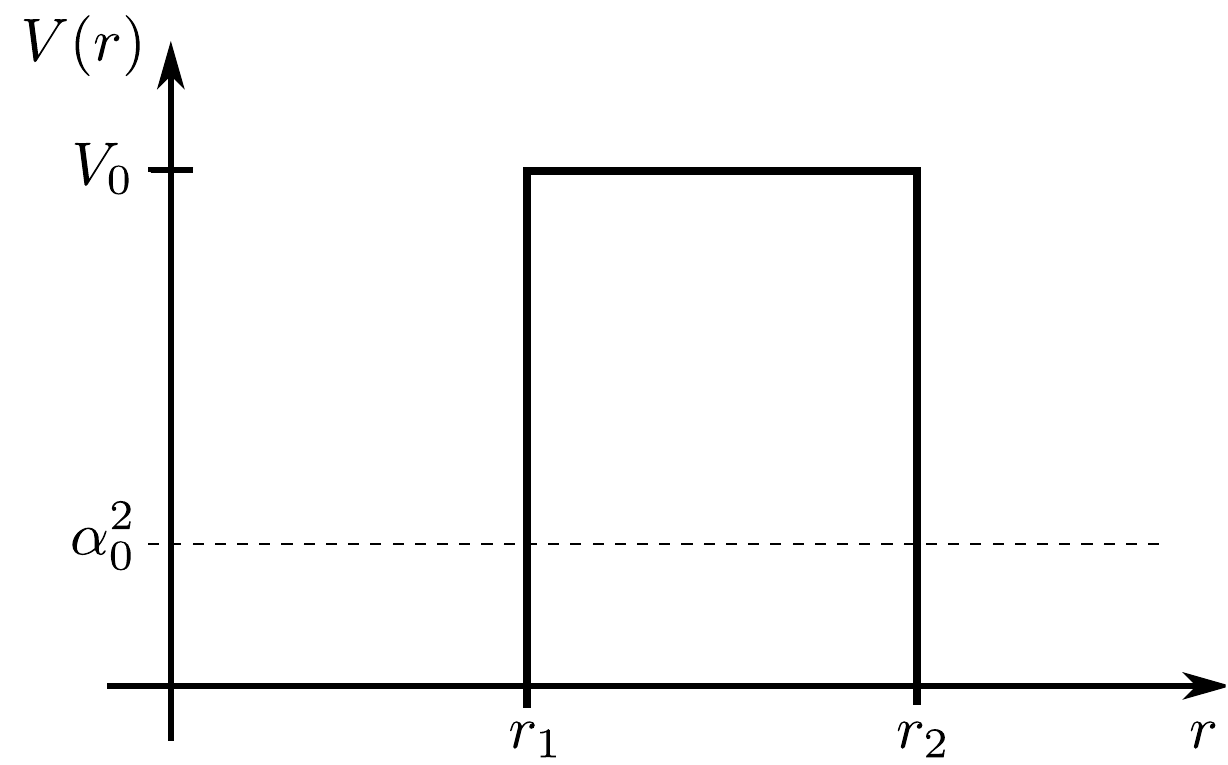}}%
\hfill
\subfloat[\label{fig:ResonanceMotion}]{\begin{overpic}[width=.45\textwidth
]{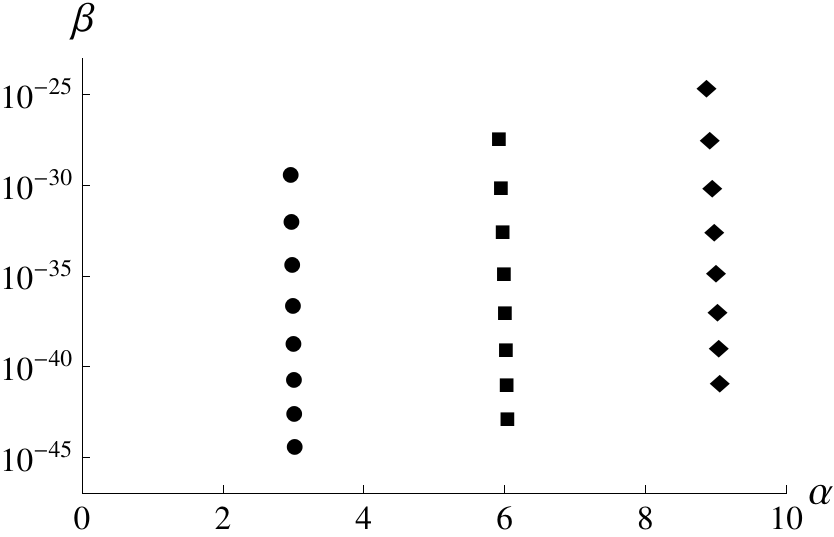}
\put(85,58){\footnotesize\clap{$V_0=230$}}
\put(93,34){\raisebox{-15pt}{\clap{\tikz\draw [->,thick,>=stealth] (0pt,30pt) -- (0pt,0pt);}}}
\put(85,10){\footnotesize\clap{$V_0=580$}}
\end{overpic}}
\hfill\mbox{}%
\caption[]{
\subref {fig:potentialET} {Example of barrier potential.} 
\subref{fig:ResonanceMotion} The plot shows how the first three resonances (\protect\raisebox{-.2ex}{\large$\bullet$},  $\blacksquare$, and $\blackdiamond$) of the barrier potential shown in Fig.~\ref{fig:potentialET} for $r_1=1$ and $r_2=2$ move as $V_0$ is increased from $230$ to $580$ in steps of $50$.}
\label{fig:}
\end{figure}

Consider the barrier potential shown in Fig.~\ref{fig:potentialET} as a family parametrized by $V_0\geq1$; Property~\ref{item:PotRequirements} is then immediate, except for the fact that $\alpha$ and $\beta$ are both minimal for all $V_0\geq1$.

This potential is simple enough to allow us to calculate its Jost function explicitly, that will also be parametrized by $V_0\geq1$ (see Eq.~\eqref{eq:JostBarrier}).
Using this explicit formula we have numerically calculated  the location of the first three resonances for eight increasing values of $V_0$ and  depicted them in Fig.~\ref{fig:ResonanceMotion}. 
We found that their real parts change negligibly, while their imaginary parts tend to zero. Moreover, $k_0(V_0)$ always has the smallest imaginary and real part. 
Thus, Properties~\ref{item:AlphaConstants} and~\ref{item:BetoToZero} of Definition~\ref{def:CV} appear to be fulfilled.
From the physical point of view the reason for this is that when the barrier is high enough then the resonances get close to the bound states of the infinitely high barrier.

\subsection{Property~\ref{item:RZeroBetaOrder} in Definition~\ref{def:CV}}
The scattering length $a$ is defined as   \citep[see][page 136]{RS3}
\begin{equation}
a=\frac{\dkp S(0)}{2iS(0)}  .
\end{equation}
From Eq.~\eqref{eq:ScatteringLength} we see that for potentials without bound and virtual states
\begin{align}\label{eq:r0_ScatteringLength}
  |a|=\left|R_V-\frac{2}{5}r_0\right|\geq\frac{2}{5}r_0-R_V .
\end{align}
For the barrier potential, using the explicit form of the Jost function and the relation $S(k)=F(-k)/F(k)$ one gets
\begin{align}
  \lim_{V_0\to\infty}a(V_0)=-r_2,
\end{align}
which together with Eq.~\eqref{eq:r0_ScatteringLength} shows that Property~\ref{item:RZeroBetaOrder} is satisfied.

\subsection{Property~\ref{item:NuTildeBetaOrder} in Definition~\ref{def:CV}}
According to Definition~\ref{def:SKAndKTildeAndS}, we have that $\nu_{\tilde K}$ is the number of resonances that lie in the stripe $\{z\in\mathbb C\,|\,0\leq\Re z\leq12\|V\|_1, \Im k\leq0\}$. 
Clearly, this number depends on the distribution of the zeros of the Jost function. 
Unfortunately, results from inverse scattering theory like~\citep{Koro,Weikard} suggest that there are little restrictions on this distribution: Korotyaev for example proves in~\cite{Koro} that resonances and potentials are in one-to-one correspondence, so that one can interpret resonances as variables which fix the potential. Hence, they can be put essentially everywhere and the potential just changes accordingly. On the other hand, Zworski proved a formula \cite[Theorem~6 in][]{Zworski1987} for the location of the $n$-th resonance that holds up to an error that becomes small for growing $n$; however, no bound on the error is given. According to this formula we would have
\begin{align}\label{eq:Ntilde}
  \nu_{\tilde K} \approx C\|V\|_1
\end{align}
with the constant $C>0$ depending only on the size of the potential's support and on the behavior of the potential at $r=R_V$  \citep[see][]{Zworski1987}. 
Assuming Property~\ref{item:PotRequirements} of Definition~\ref{def:CV} and Eq.~\eqref{eq:Ntilde} to be satisfied, then $\nu_{\tilde K}$ changes only through $\|V\|_1$ when $\beta\to0$, and we have to study how $\|V\|_1$ behaves in this limit.

A relation between the inverse of the lifetime $\gamma=4\alpha\beta$ and an integral of the potential was famously obtained by G. Gamow~\cite{Gamow1928,Segre1965}. He found the following formula \citep[see][Chapter~7]{Segre1965}
\begin{align}\label{eq:gamow}
  \gamma=4\alpha\beta=\frac{\alpha}{R_N}\exp\left(-2\int_{r_1}^{r_2}\sqrt{V(r)-(\alpha^2-\beta^2)}\,dr\right)  ,
\end{align}
where $V(r)$ is assumed to be shaped like a barrier through which the alpha particle needs to tunnel in order for alpha decay to occur, and $\alpha^2-\beta^2$ is the energy of the alpha particle. The radii $r_1$ and $r_2$ are such that $V(r)-(\alpha^2-\beta^2)\geq0$ if $r\in[r_1,r_2]$ and $R_N$ denotes the nuclear radius. 
Applying Eq.~\eqref{eq:gamow} to the barrier potential shown in Fig.~\ref{fig:potentialET} and  assuming that $V_0-(\alpha^2-\beta^2)\geq V_0/4$ as well as $r_2-r_1\geq\sqrt{r_2-r_1}$, which is justified by the fact that in physical examples the barrier is very thick and much higher than the energy of the alpha particle, we get 
\begin{align}
  \log\frac{1}{\beta}
  =\log(4R_N)+2\sqrt{V_0-(\alpha^2-\beta^2)}(r_2-r_1)
  \geq\sqrt{\|V\|_1},
\end{align}
Hence, Gamow's formula, Eq.~\eqref{eq:gamow}, suggests that
\begin{align}\label{eq:gamow4}
  \|V\|_1\leq\left(\log\frac{1}{\beta}\right)^2=(\log\beta)^2,
\end{align}
that together with Eq.~\eqref{eq:Ntilde} gives exactly Property~\ref{item:NuTildeBetaOrder} in the form of Eq.~\eqref{eq:NTildeET}.

To transform the previous argument into a proof one needs an explicit bound on $\nu_{\tilde K}$ in the form of Eq.~\eqref{eq:Ntilde}, and a rigorous version of Eq.~\eqref{eq:gamow}.
We see two ways to derive the former. 
First, by modification of the proof of Levinson's Theorem \cite[see][]{RS3}, which connects the number of bound states $N$ with the Jost function $F$. In this proof the number of bound states is calculated via the complex contour integral
\begin{align}\label{eq:Levinson}
  N=\frac{1}{2\pi i}\int_{\mathcal C}\frac{\dkp F(z)}{F(z)}\,dz,
\end{align}
where the contour $\mathcal C$ is a closed semi circle in the upper half plane that encloses all bound states. The bound states are zeros of $F$ and thereby poles of the integrand, so that Eq.~\eqref{eq:Levinson} is a direct consequence of the residue theorem. For the purposes of getting a handle on $\nu_{\tilde K},$ we can use Eq.~\eqref{eq:Levinson}, but choose as contour the boundary of the region $\{z\in\mathbb C\,|\,0\leq\Re z\leq12\|V\|_1, \Im k\leq0\}$. The difficulty is now to derive bounds for $\dkp F(z)/F(z)$ along this contour, which yield  bounds for~$\nu_{\tilde K}.$ The second way we can think of to derive a rigorous bound on $\nu_{\tilde K}$ is via a well known result from inverse scattering theory, namely the Marchenko equation \citep[see][]{Koro}. Using this equation, one can calculate the potential from the $S$-matrix and thereby from the resonances. Thus, it might be possible to characterize the potential class for which the resonances in $\{z\in\mathbb C\,|\,0\leq\Re z\leq12\|V\|_1, \Im k\leq0\}$ only have imaginary parts above a certain value. Then one can employ the bounds on the number of resonances $n(r)$ in a ball of radius $r,$ obtained in Lemma~\ref{lem:BoundOnNumberOfZeros}, to get a bound for $\nu_{\tilde K}.$

The proof of Eq.~\eqref{eq:gamow} is, to our knowledge, still an open problem. Moreover, Eq.~\eqref{eq:gamow4} will not hold for general potentials, but only for barrier-like ones as considered by Gamow. For other potentials, Eq.~\eqref{eq:gamow} is not satisfied, so that the relation between $\|V\|_1$ and $\beta$ might be different. For example, in~\cite{Diplom} it was shown that the resonances of the one-dimensional ``potential table'' $V(x)=V_0\1_{[-a,a]}(x)$ satisfy
\begin{align}\label{eq:gamow2}
  k_n=\sqrt{V_0+\frac{(n+1)^2\pi^2}{4a^2}}-i\frac{(n+1)^2\pi^2}{4a^3\sqrt{V_0^2+V_0\frac{(n+1)^2\pi^2}{4a^2}}}+O(V_0^{-3/2}).
\end{align}
Using $\|V\|_1=2aV_0$ and assuming that $V_0$ is large, we have for $\beta=-\Im k_0$ that
\begin{align}\label{eq:gamow3}
	\|V\|_1^2+\|V\|_1\frac{\pi^2}{2a}\approx \frac{\pi^4}{4a^4\beta^2},
\end{align}
which is a completely different relationship between $\|V\|_1$ and $\beta$ than Eq.~\eqref{eq:gamow4}. Note that this difference is not due to the fact that we are looking at a one-dimensional potential rather than a three-dimensional one with rotational symmetry.
Indeed, in the one-dimensional situation the resonances satisfy the equation \cite{Diplom}
\begin{align}
	\exp\left(i4a\sqrt{k^2-V_0}\right)=\left(\frac{k+\sqrt{k^2-V_0}}{k-\sqrt{k^2-V_0}}\right)^2
\end{align}
and in the analogous three-dimensional situation, where the potential reads $V(r)=V_0\1_a(r)$, following \cite{Diplom} it is easy to verify that the resonances satisfy
\begin{align}
  \exp\left(i2a\sqrt{k^2-V_0}\right)=\frac{k+\sqrt{k^2-V_0}}{k-\sqrt{k^2-V_0}}.
\end{align}
Hence, every resonance of the three-dimensional potential appears also in the one-dimensional situation, so that there is a $n$ for which Eq.~\eqref{eq:gamow2} captures the location of the first resonance of the three-dimensional potential.

\vfill\newpage

\section{Energy- and Time-Variance}\label{sec:ETVariances}

In this section we explicitly calculate the variance of energy and time that will be extensively used in the proofs of Theorems~\ref{thm:MainUncertainty} and~\ref{thm:ETvsLL}. 

\begin{lemma}\label{lem:e}
  Let $  E_{\beta,\sigma}$ be defined as in Eq.~\eqref{eq:EBetaSigma}.
Then, for the wave function $\psi$
  \begin{align}\label{eq:Evariance}
	\Var E
	&=\frac{
	2\alpha^2\beta^2E_{\beta,\sigma}^2
	+\frac{\beta^2}{2\sigma^2}(1+E_{\beta,\sigma}^2)
	+\frac{\beta}{2\sigma}\left(\beta^2+4\alpha^2+\frac{3}{2\sigma^2}\right)E_{\beta,\sigma}}
	{\left(1+\beta\sigma E_{\beta,\sigma}\right)^{2}}  .
  \end{align}
\end{lemma}

\begin{proof}
  Note that
  \begin{align}\label{eq:thmE1}
	\Var E=\frac{\langle\psi,H^2\psi\rangle}{\|\psi\|^2}-\frac{\langle\psi,H\psi\rangle^2}{\|\psi\|^4}.
  \end{align}
  First look at
  \begin{align}
	\langle\psi,H\psi\rangle
	&=\langle\psi,\1_{R}H\psi\rangle+\langle\psi,\1_{[R,\infty)}H\psi\rangle\\
	&=k_0^2\|f_{R}\|_2^2-\langle\psi,\1_{[R,\infty)}\ddr\psi\rangle.\label{eq:thmE2}
  \end{align}
  From Lemma~3.1 in~\cite{Skibsted86} we have
  \begin{align}\label{eq:thmE3}
	\|f_{R}\|_2^2=\frac{e^{2\beta R}}{2\beta}
  \end{align}
  and since for $r\geq R$
  \begin{align}\label{eq:thmE4}
	-\ddr\psi(r)
	=\left[\frac{1}{\sigma^2}-\left(ik_0-\frac{r-R}{\sigma^2}\right)^2\right]
	\exp\left(ik_0r-\frac{(r-R)^2}{2\sigma^2}\right)
  \end{align}
in the weak sense,  the second summand in eq.~\eqref{eq:thmE2} is readily calculated. One finds elementary error function integrals, which is why we omit the details and directly give the result
  \begin{align}
	\langle\psi,H\psi\rangle
	=\frac{e^{2\beta R}}{2\beta}\left(\alpha^2
	+\frac{\beta}{2\sigma}(1+2\alpha^2\sigma^2)\sqrt\pi e^{\beta^2\sigma^2}(1+\erf(\beta\sigma))\right).
  \end{align}
  Now,
  \begin{align}
	\langle\psi,H^2\psi\rangle
	&=\langle H\psi,H\psi\rangle
	=\langle H\psi,\1_{R}H\psi\rangle+\langle H\psi,\1_{[R,\infty)}H\psi\rangle\\
	&=|k_0|^4\|f_{R}|_2^2+\langle \ddr\psi,\1_{[R,\infty)}\ddr\psi\rangle.
  \end{align}
  The first summand is again obtained from eq.~\eqref{eq:thmE3} and the second one by integrating the modulus square of eq.~\eqref{eq:thmE4} over $[R,\infty)$, yielding error function integrals again. Omitting the details, one arrives at
  \begin{align}
	\langle\psi,H^2\psi\rangle
	&=\frac{e^{2\beta R}}{2\beta}\bigg[\frac{1}{2\sigma^2}(\beta^2+2\alpha^4\sigma^2)\nonumber\\
   &\quad+\frac{\beta}{4\sigma^3} (3+12\alpha^2\sigma^2+4\alpha^4\sigma^4)\sqrt\pi e^{\beta^2\sigma^2}(1+\erf(\beta\sigma))\bigg].
  \end{align}
  Moreover, using the fact that $f(k_0,r)=e^{ik_0r}$ for $r\geq R,$ we have
  \begin{align}\label{eq:NormGaussianTail}
	\|g_R\|_2^2
	=\int_R^\infty \exp\left(2\beta r-\frac{(r-R)^2}{\sigma^2}\right)\,dr
	=\frac{\sigma}{2}e^{2\beta R}E_{\beta,\sigma}
  \end{align}
  and this together with Eq.~\eqref{eq:thmE3} gives us
  \begin{align}
	\|\psi\|_2^2
	=\|f_R\|_2^2+\|g_R\|_2^2
	=\frac{e^{2\beta R}}{2\beta}\left[1+\beta\sigma E_{\beta,\sigma}\right].
	\label{eq:NormInitialWaveFunction}
  \end{align}
  Plugging $\langle\psi,H^2\psi\rangle,\langle\psi,H\psi\rangle$ and $\|\psi\|^2_2$ into the formula for the variance, Eq.~\eqref{eq:thmE1}, we obtain the assertion of the Lemma.
  \QED
\end{proof}

In contrast to the energy variance, $\Var T$ can not be calculated directly.
We will approximate it by $\Varo T$, which is defined in Eq.~\eqref{eq:SimpleVaro} and determined in the next Lemma. Recall that $\gamma = 4\alpha\beta$.
\begin{lemma}\label{lem:t}
The probability density 
  \begin{align}
	\Pio_T(t)=- \partial_t  \frac {\|\1_{R}   e^{-ik_0^2t}  f_{R_2(t)}\|_2^2}  {\|\1_{R}   \psi\|_2^2},
  \end{align}
  has the mean
\begin{equation}
\meano{t}  =   \frac{1}{\gamma}
\end{equation}
and the variance
  \begin{equation}\label{eq:Tvariance}
 	\Varo{T}= \frac{1}{\gamma^2}.
  \end{equation}
\end{lemma}
\begin{proof}
Due to the fact that
\begin{equation}\label{eq:Exp}
\|\1_{R}   \psi\|_2^2  =  \| f_{R}\|_2^2 ,
\qquad
{\|\1_{R}   e^{-ik_0^2t}  f_{R_2(t)}  \|_2^2}
= e^{-\gamma t}  \| f_{R} \|_2^2
\end{equation}
we have
\begin{equation}\label{eq:SkibstedFactor}
	\Pio_T(t)=-\partial_t\frac  {\|\1_{R}   e^{-ik_0^2t}  f_{R_2(t)}  \|_2^2}  {\|\1_{R}   \psi\|_2^2} = -\partial_te^{-\gamma t}=\gamma e^{-\gamma t}.
\end{equation}
Using integration by parts, the mean then calculates to 
\begin{equation}
\int_0^\infty    t\,   \gamma  e^{-\gamma t}   d t  
=\int_0^\infty    e^{-\gamma t}   d t  
=  \frac{1}{\gamma},
\end{equation}
and the variance
\begin{align}
\Varo{T}
&=\meano{t^2} -\meano{t}^2 
\nonumber\\
&=2 \int_0^\infty  t\,   e^{-\gamma t}   d t  - \frac{1}{\gamma^2}
=  \frac{1}{\gamma^2}.
\end{align}
\QED
\end{proof}

To estimate the error on  the time variance made by using $\Varo T$ as approximation, we start by estimating in the following Lemma the pointwise difference between the true non-escape probability $\Pi_T$ and $\Pio_T$.
For early times, say $t\in(0,A)$, we will control this difference  using the results of Skibsted given in Lemma~\ref{lem:SkibstedET}.
At late times, i.e.\ for $t\in[A,\infty)$, we can use the scattering estimates of  Corollary \ref{cor:main}. 
\begin{lemma}\label{lem:SurvivalProbError}
Let $t>0$  and
\begin{align}
\xi_{(0,A)} (t)
&\coloneqq
\left( 2 + \sqrt{E_{\beta,\sigma} \beta\sigma } \right)
	\nonumber\\
	&\times\left(   \sqrt {54 \,  \beta}  \, t^{1/4}  +  \sqrt{\frac{6}{5\pi}}\frac {\beta^{1/4} } {\alpha^{1/4}}  +  \frac {4\sqrt \beta} {\sqrt{\pi\alpha}}  
	+ \sqrt{E_{\beta,\sigma} \beta\sigma }  \right)  ,  
\\
\xi_{[A,\infty)} (t)
&\coloneqq
2\beta e^{-2\beta R}	\left(   \tilde{c}_3 t^{-3} + \tilde{c}_4 t^{-4}  \right)  +  e^{-\gamma t} .
\end{align}
Then,
\begin{multline}
\left|   \frac {\|\1_{R}   e^{-iHt}  \psi\|_2^2}  {\|\1_{R}   \psi\|_2^2}   
		 -  \frac {\|\1_{R}   e^{-ik_0^2t}  f_{R_2(t)}\|_2^2}  {\|\1_{R}   \psi\|_2^2}   \right|
\\
\leq   \xi_{(0,A)}(t)  \1_{(0,A)}  +  \xi_{[A,\infty)} (t)  \1_{[A,\infty)}  .
\label{eq:BoundProbDifference}
\end{multline}
\end{lemma}
\begin{proof}
We first prove the bound that we will use for $t\in(0,A)$.
Observing that
\begin{equation}
6  \left(  \frac {1+20\sqrt {\beta/\alpha}}  {1+10\sqrt {\beta/\alpha}}  \right)^2  \beta
	\leq 54 \,  \beta  ,
\end{equation}
we get from Lemma~\ref{lem:SkibstedET} that for $t\geq0$
\begin{multline}
\left\|  e^{-iHt}  f_{R}   - e^{-ik_0^2t}  f_{R_2(t)}  \right\|_2
\\
\leq \|f_{R}\|_2 \left( \sqrt{  54\beta \sqrt t  +  \frac {6\sqrt {\beta} } {5\pi\sqrt\alpha}}  +  \frac {4\sqrt {\beta}} {\sqrt{\pi\alpha}}    \right )  .
\end{multline}
Using this together with Eq.~\eqref{eq:thmE3} for $\|f_R\|^2_2$ and Eq.~\eqref{eq:NormGaussianTail} for $\|g_R\|_2^2$ we get
\begin{align}
&\left|   \frac {\|\1_{R}   e^{-iHt}  \psi\|_2^2}  {\|\1_{R}   \psi\|_2^2}   
		 -  \frac {\|\1_{R}   e^{-ik_0^2t}  f_{R_2(t)}\|_2^2}  {\|\1_{R}   \psi\|_2^2}   \right|
\\
&= 
\frac 1  {\|f_{R}\|_2^2}   
\left(   \|\1_{R}   e^{-iHt}  \psi\|_2   +   \|\1_{R}   e^{-ik_0^2t}  f_{R_2(t)}\|_2   \right)  
\nonumber\\
&\qquad\qquad
\times  \left|   \|\1_{R}   e^{-iHt}  \psi\|_2   -   \|\1_{R}   e^{-ik_0^2t}  f_{R_2(t)}\|_2   \right|
\\
&\leq  \frac {\|\psi\|_2 + \|f_{R}\|_2}   {\|f_{R}\|_2^2}   
	\left\|  \1_{R}  \left( e^{-iHt}  \psi   - e^{-ik_0^2t}  f_{R_2(t)}   \right)  \right\|_2
\\
&\leq  \frac {2 \|f_{R}\|_2 + \|g_R\|_2}   {\|f_{R}\|_2^2}   
	\left\|  e^{-iHt}  \psi   - e^{-ik_0^2t}  f_{R_2(t)}  \right\|_2
\\
&\leq  \frac {2 + \sqrt{E_{\beta,\sigma} \beta\sigma }  }   {\|f_{R}\|_2}   
	\left(  \left\|  e^{-iHt}  f_{R}   - e^{-ik_0^2t}  f_{R_2(t)}  \right\|_2
	+  \left\|  e^{-iHt} g_R \right\|_2  \right)
\\
&\leq   \left( 2 + \sqrt{E_{\beta,\sigma} \beta\sigma } \right)
	\left(   \sqrt{  54\beta \sqrt t  +  \frac {6\sqrt {\beta} } {5\pi\sqrt\alpha}}  +  \frac {4\sqrt {\beta}} {\sqrt{\pi\alpha}}  
	+ \sqrt{E_{\beta,\sigma} \beta\sigma }  \right)  .
\end{align}
For $X,Y\geq 0$,
\begin{equation}
X^2 + Y^2 \leq (X+Y)^2 ,
\end{equation}
taking the square root, and choosing $X=\sqrt x$, $Y=\sqrt y$, with $x,y \geq 0$, we have
\begin{equation}
\sqrt{x+y} \leq \sqrt x + \sqrt y ,
\end{equation} 
hence
\begin{equation}
\left|   \frac {\|\1_{R}   e^{-iHt}  \psi\|_2^2}  {\|\1_{R}   \psi\|_2^2}   
		 -  \frac {\|\1_{R}   e^{-ik_0^2t}  f_{R_2(t)}\|_2^2}  {\|\1_{R}   \psi\|_2^2}   \right|
\leq   \xi_{(0,A)} (t) .
\label{eq:BoundProbDifferenceZeroA}
\end{equation}

We now prove the bound used for $t\in[A,\infty)$.
Using Corollary~\ref{cor:main} and Eq.~\eqref{eq:Exp} for $\|\1_{R}   e^{-ik_0^2t}  f_{R_2(t)}\|_2^2$ we get  that
\begin{align}
\hspace{5em}&\hspace{-5em}\left|   \frac {\|\1_{R}   e^{-iHt}  \psi\|_2^2}  {\|\1_{R}   \psi\|_2^2}   
		 -  \frac {\|\1_{R}   e^{-ik_0^2t}  f_{R_2(t)}\|_2^2}  {\|\1_{R}   \psi\|_2^2}   \right|
\\
&\leq    \frac {\|\1_{R}   e^{-iHt}  \psi\|_2^2}  {\|\1_{R}   \psi\|_2^2}   
		 +  \frac {\|\1_{R}   e^{-ik_0^2t}  f_{R_2(t)}\|_2^2}  {\|\1_{R}   \psi\|_2^2}  
\\
&\leq  \xi_{[A,\infty)}.
\label{eq:BoundProbDifferenceAInfty}
\end{align}
\QED
\end{proof}

Having control over the difference between $\Pi_T$ and $\Pio_T$ we can now prove Lemma~\ref{lem:ErrorMeanVariance}, which provides an estimate on the difference between $\Var T$ and $\Varo T$.

\begin{proof}[of Lemma~\ref{lem:ErrorMeanVariance}]
Consider at first the mean.
Recalling Eq.~\eqref{eq:SimpleMean} from Lemma~\ref{lem:SimpleVariance} and Eq.~\eqref{eq:SimpleMeano}, we have 
\begin{align}
\left|  \mean{t} -   \meano{t}    \right|
\leq  \int_0^\infty   \left|   
		\frac {\|\1_{R}   e^{-iHt}  \psi\|_2^2}  {\|\1_{R}   \psi\|_2^2}   
		-  \frac {\|\1_{R}   e^{-ik_0^2t}  f_{R_2(t)}\|_2^2}  {\|\1_{R}   \psi\|_2^2}   
	\right| d t   .
\end{align}
Substituting Eq.~\eqref{eq:BoundProbDifference} from Lemma~\ref{lem:SurvivalProbError} and performing the integral we immediately get  Eq.~\eqref{eq:ErrorMean}.

Now consider the variance. Using Eqs.~\eqref{eq:SimpleVariance} and~\eqref{eq:SimpleVaro} for $\Var T$ and $\Varo T$ we get
\begin{multline}
\left|  \Var{T}- \Varo  T   \right|
\\=
\left| 
 2 \int_0^\infty    t
		  \left( \frac {\|\1_{R}   e^{-iHt}  \psi\|_2^2}  {\|\1_{R}   \psi\|_2^2}   
		 -  \frac {\|\1_{R}   e^{-ik_0^2t}  f_{R_2(t)}\|_2^2}  {\|\1_{R}   \psi\|_2^2}   \right)  d t
+  \meano{t}^2 - \mean{t}^2
  \right|
\end{multline}
From the error estimate on the mean given in  Eq.~\eqref{eq:ErrorMean} we have
\begin{equation}
\left| \meano{t}^2 - \mean{t}^2  \right|
	=  \left | - \bigl( \mean{t} - \meano{t} \bigr)^2  
		+  2 \meano{t} \, \bigl( \meano{t} - \mean{t} \bigr)  \right|
	\leq  \omega^2  +  \frac{2}{\gamma}  \omega   ,
\end{equation}
therefore
\begin{multline}
\left|  \Var{T}- \Varo  T   \right|
\\
\leq
 2 \int_0^\infty    t
		  \left|  \frac {\|\1_{R}   e^{-iHt}  \psi\|_2^2}  {\|\1_{R}   \psi\|_2^2}   
		 -  \frac {\|\1_{R}   e^{-ik_0^2t}  f_{R_2(t)}\|_2^2}  {\|\1_{R}   \psi\|_2^2}   \right|  d t
+ \omega^2  +  \frac{2}{\gamma}  \omega  .
\end{multline}
Using again the bound~\eqref{eq:BoundProbDifference} for the non-escape probability and integrating we get Eq.~\eqref{eq:ErrorVariance}.
\QED
\end{proof}


\section{Proof of  Theorem~\ref{thm:MainUncertainty}}



\subsection{Proof of Statement~\ref{statement:UncertaintyHolds}}

First, we sketch the idea behind the proof.
The approximate time variance $\Varo T=1/\gamma^2$ is independent of $\sigma$, while the energy variance~\eqref{eq:Evariance} can be made very small by making $\sigma$ very big.
Therefore, the same is true for the approximate product $\Po=\Var E /\gamma^2$; this  suggests a possible violation of the uncertainty relation.
On the other side, by increasing $\sigma$ the error $\epsilon_P=\epsilon_T\Var E$ grows very fast and soon becomes too big to make statements on the validity of the uncertainty relation.

The statement of the theorem in symbolical form is
\begin{equation}
\bigl[
	\Po-\epsilon_P\geq 1/4  
	\lor
	\Po+\epsilon_P< 1/4
\bigr]
\Rightarrow
	\Po-\epsilon_P\geq 1/4  ,
\end{equation}
that is equivalent to
\begin{equation}
\Po+\epsilon_P \geq 1/4.
\end{equation}
The quantities $\Po$ and $\epsilon_P$ are functions of the parameters $\alpha$, $\beta$, and $\sigma$, therefore a sufficient condition for this inequality to be true is that the parameter regions corresponding to $\Po<1/4$ and to $\epsilon_P<1/4$ do not intersect.
This sufficient condition stays sufficient if we make the regions bigger by using a $\tPo \leq \Po$ and an $\tilde{\epsilon}_P \leq \epsilon_P$ in place of $\Po$ and $\epsilon_P$.

To find the approximations $\tPo$ and $\tilde{\epsilon}_P$ we will benefit from the fact that the expression~\eqref{eq:Evariance} for $\Var E$ and the expression~\eqref{eq:EpsilonT}  for $\epsilon_T$ are sums of positive terms, therefore we can simply drop some terms from each sum.

We start considering $\Po$.
From the energy variance Eq.~\eqref{eq:Evariance} we get
\begin{align}
	\Po
	&=\frac{
	2\alpha^2\beta^2E_{\beta,\sigma}^2
	+\frac{\beta^2}{2\sigma^2}(1+E_{\beta,\sigma}^2)
	+\frac{\beta}{2\sigma}\left(\beta^2+4\alpha^2+\frac{3}{2\sigma^2}\right)E_{\beta,\sigma}}
	{16\alpha^2\beta^2 \left(1+\beta\sigma E_{\beta,\sigma}\right)^{2}}
	\\
	&=
	\frac{1}
	{8\left(\beta\sigma  +E_{\beta,\sigma}^{-1} \right)^{2}}
	\left[
	{1
	+\frac{1+E_{\beta,\sigma}^{-2}}{4\alpha^2\sigma^2}
	+\frac{E_{\beta,\sigma}^{-1}}{4\alpha^2\beta\sigma^3}\left(\beta^2\sigma^2+4\alpha^2\sigma^2+\frac{3}{2}\right)}
	\right] .
  \end{align}
We can simplify this expression with the change of variables
\begin{equation}
\talpha \coloneqq \alpha\sigma,
\qquad\qquad
\tbeta \coloneqq \beta\sigma,
\end{equation}
and with the definition
\begin{equation}
E_{\tbeta} \coloneqq \sqrt\pi e^{\tbeta^2}\bigl(1+\erf(\tbeta)\bigr) =E_{\beta,\sigma},
\end{equation}
getting
\begin{align}
	\Po
	&=
	\frac{1}
	{8\left(\tbeta  +E_{\tbeta}^{-1} \right)^{2}}
	\left[
	{1
	+\frac{1+E_{\tbeta}^{-2}}{4\talpha^2}
	+\frac{E_{\tbeta}^{-1}}{4\talpha^2\tbeta}
		\left(\tbeta^2+4\talpha^2+\frac{3}{2}\right)}
	\right] .
  \end{align}
These variables are particularly convenient because they transform the parameters of the problem from $(\alpha,\beta,\sigma)$ to only $(\talpha,\tbeta)$.
Notice that 
\begin{equation}
\frac{e^{-\tbeta^2}}{2\sqrt\pi} 
\leq
E_{\tbeta}^{-1} 
\leq
\frac{e^{-\tbeta^2}}{\sqrt\pi} ,
\end{equation}
therefore defining
\begin{equation}
\tPo \coloneqq 
	\frac{\pi}
	{8\left(\sqrt\pi\, \tbeta  + e^{-\tbeta^2} \right)^{2}}
	\left[
	{1
	+\frac{4\pi+e^{-2\tbeta^2} }{16\pi\talpha^2}
	+\frac{e^{-\tbeta^2} }{8\sqrt\pi\talpha^2\tbeta}
		\left(\tbeta^2+4\talpha^2+\frac{3}{2}\right)}
	\right] ,
\end{equation}
we get
\begin{equation}
\Po \geq \tPo.
\end{equation}

\begin{figure}
\centering
\begin{overpic}[width=.65\textwidth
]{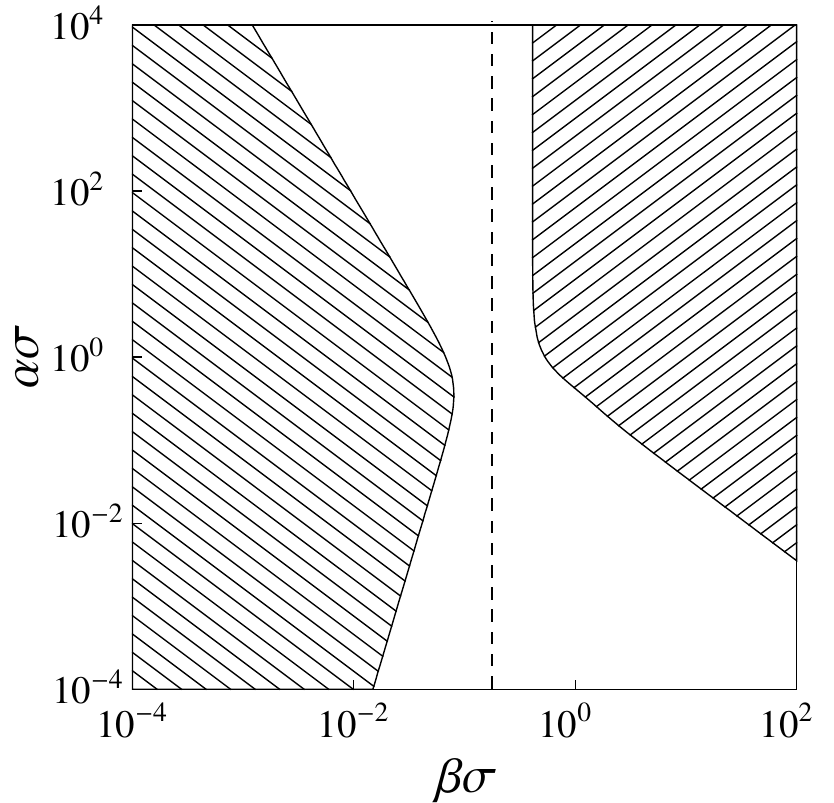}
\put(68,63){\large{\fboxsep0pt\framebox{\fboxsep3pt\colorbox{white}{$\tPo<1/4$}}}}
\put(21,45){\large{\fboxsep0pt\framebox{\fboxsep3pt\colorbox{white}{$\tilde{\epsilon}_P<1/4$}}}}
\end{overpic}
\caption{Regions where $\tPo<1/4$ and where $\tilde{\epsilon}_P<1/4$.
The dashed line corresponds to $\beta\sigma = \tbeta_0 =  10^{-3/4}$.}
\label{fig:Regions}
\end{figure}

We now need to characterize the region of the $(\talpha,\tbeta)$-plane where $\tPo<1/4$.
Figure~\ref{fig:Regions} suggests that this region does not extend in $\tbeta$ further than the value $\tbeta_0 \coloneqq 10^{-3/4}$.
To verify this conjecture we consider the border of this region, that is characterized by the equation
\begin{equation}\label{eq:BorderP}
\tPo(\talpha,\tbeta) = 1/4.
\end{equation}
We solve this equation for $\talpha$, considering $\tbeta$ a parameter.
With the definitions
\begin{align}
a_\tbeta 
	&\coloneqq 
		1
		+\frac {e^{-\tbeta^2}}  {2\sqrt\pi \tbeta}
		-\frac  {2}{\pi}  \left(\sqrt\pi\, \tbeta  + e^{-\tbeta^2} \right)^{2}   ,
\\
b_\tbeta 
	&\coloneqq 
		\frac{4\pi+e^{-2\tbeta^2} }{16\pi}
		+ \frac{e^{-\tbeta^2} }{8\sqrt\pi\tbeta}
			\left(\tbeta^2 +\frac{3}{2}\right)  ,
\end{align}
we can rewrite Eq.~\eqref{eq:BorderP} as
\begin{equation}
a_\tbeta \, \talpha^2 +  b_\tbeta  = 0 .
\end{equation}
Observing that $b_\tbeta >0$, we have that this equations has no solutions in case $a_\tbeta>0$.
For $\tbeta\leq\tbeta_0$, by direct calculation we verify that
\begin{equation}
a_\tbeta 
	\geq 1
		+\frac {e^{-\tbeta_0^2}}  {2\sqrt\pi \tbeta_0}
		-\frac  {2}{\pi}  \left(\sqrt\pi\, \tbeta_0  + 1 \right)^{2} 
	>0,
\end{equation}
therefore in the region $\tbeta\leq\tbeta_0$ there is no $\talpha$ that solves Eq.~\eqref{eq:BorderP}, and $\tPo$ is always greater than $1/4$ there.

We now turn to analyze the error on the product  $\epsilon_P$.
From Eq.~\eqref{eq:EpsilonT} we have $\epsilon_T\geq 2 \zeta$, therefore
\begin{align}
\epsilon_P  
&=  \epsilon_T \Var E 
	= \gamma^2  \Po  \epsilon_T
	\geq \gamma^2 \tPo  \epsilon_T
\\
&\geq
	\frac{\pi \gamma^2  \epsilon_T}
		{16\left(\pi\, \tbeta^2  + 1 \right)}
	\left(  1 + \frac{1 }{4\talpha^2} \right)
\geq
	\frac{\pi \gamma^2  \zeta}
		{8\left(\pi\, \tbeta^2  + 1 \right)}
	\left(  1 + \frac{1 }{4\talpha^2} \right)  .
\label{eq:LowerBoundEpsilonP}
\end{align}
Now for $A\geq0$ let
\begin{equation}
\tilde\zeta_A
\coloneqq
\frac{\sqrt\pi}{2}   e^{\beta^2\sigma^2} \beta\sigma  A^2
+ 2\beta e^{-2\beta R}  \tilde{c}_3 A^{-1}  ,
\end{equation}
then Lemma~\ref{lem:ErrorMeanVariance} gives
 \begin{equation}
\zeta  \geq  \tilde \zeta_A .
\end{equation}
In particular, the last inequality is true for the $A$ that minimizes $\tilde\zeta_A$, that is such that
\begin{equation}
\frac{d\tilde\zeta_A}{dA}  =  
\sqrt\pi e^{\beta^2\sigma^2} \beta\sigma  A
- 2\beta e^{-2\beta R}  \tilde{c}_3 A^{-2}
=0,
\end{equation}
that implies
\begin{equation}
A  =  \left(    
	\frac  
		{2 e^{-2\beta R}  \tilde{c}_3}    
		{\sqrt\pi e^{\beta^2\sigma^2} \sigma}  \right)^{1/3}
	=  \left(    
	\frac  
		{2   \tilde{c}_3}    
		{\sqrt\pi \sigma}  \right)^{1/3}
	e^{-\frac{2}{3}\beta R-\frac{1}{3}\beta^2\sigma^2 }   .
\end{equation}
Substituting that into $\tilde\zeta_A$ we get
\begin{equation}
\frac{3 \pi^{1/6}}{2^{1/3}}
\beta \sigma^{1/3} 
\tilde{c}_3^{2/3}
e^{\frac{1}{3} \beta^2\sigma^2 -  \frac{4}{3} \beta R}  
\eqqcolon
\tilde\zeta  .
\end{equation}
From the definition of $\tilde{c}_3$, Eq.~\eqref{eq:DefinitionC3Tilde}, we see that
\begin{equation}
\tilde{c}_3
\geq
27\frac{2^2}{\alpha}  M_{K,\infty}^2(2)\,  z_{ac,K}^2(0)    .
\end{equation}
From Definition~\ref{def:smallz} and Eq.~\eqref{eq:BoundSDotK} we get
\begin{equation}
z_{ac,K}(0)  \geq 2,
\end{equation}
while Definition~\ref{def:CorollaryConstants} implies
\begin{equation}
M_{K,\infty}(2)
\geq
e^{\beta R}
\frac{\sigma^3}{\sqrt 2}E_{\beta,\sigma/\sqrt 2}
\geq
\frac{\sigma^3}{\sqrt 2}
e^{\beta R + \frac{1}{2}  \beta^2 \sigma^2},
\end{equation}
therefore
\begin{equation}
\tilde{c}_3
\geq
27\frac{8}{\alpha}  \sigma^6  e^{2\beta R + \beta^2 \sigma^2}  .
\end{equation}
Substituting that into $\tilde\zeta$ we get
\begin{equation}
\tilde\zeta
\geq
\frac{54 \cdot 2^{2/3} \pi^{1/6}}{\alpha^{2/3}}
\beta \sigma^{13/3} 
e^{ \beta^2 \sigma^2 }  
=
\frac{54 \cdot 2^{2/3} \pi^{1/6}}{\talpha^{2/3}}
\tbeta \sigma^{12/3} 
e^{ \tbeta^2 }  .
\end{equation}
Using this and setting
\begin{equation}
\tilde{\epsilon}_P (\talpha,\tbeta)
\coloneqq
\left(
27 \cdot 2^{2/3} \pi^{1/6} 
\frac{\pi   \tbeta^2 }
		{\pi\, \tbeta^2  + 1 }
\tbeta
e^{ \tbeta^2 }  
\right)
\left(  4\talpha^2 + 1  \right)
\talpha^{-2/3}
\end{equation}
then Eq.~\eqref{eq:LowerBoundEpsilonP}  gives
\begin{equation}
\epsilon_P \geq \tilde{\epsilon}_P  .
\end{equation}

We now study the region where $\tilde{\epsilon}_P \leq 1/4$.
Figure~\ref{fig:Regions} suggests that this region is completely on the left of $\tbeta_0$.
For $\tbeta=\tbeta_0$ we get by direct calculation
\begin{equation}
\tilde{\epsilon}_P (\talpha,\tbeta_0) 
\geq
\frac{4}{5}
\left(  4\talpha^2 + 1  \right)
\talpha^{-2/3}  .
\end{equation}
The function $\left(  4\talpha^2 + 1  \right)  \talpha^{-2/3} $ is greater than $4 \talpha^{4/3}$, but also than $\talpha^{-2/3}$, and these two bounds cross at $\talpha=1/2$, therefore
\begin{equation}
\tilde{\epsilon}_P (\talpha,\tbeta_0)  \geq  \frac{2^{11/3} }{5} >\frac{1}{4} .
\end{equation}
for all values of $\talpha$.
Observing that $\tilde{\epsilon}_P$ grows with growing $\tbeta$, we can conclude that
\begin{equation}
\epsilon_P \geq \tilde{\epsilon}_P (\talpha,\tbeta) >\frac{1}{4} 
\qquad  \forall \talpha\geq0,  \forall \tbeta\geq\tbeta_0.
\end{equation}

Hence, we have that the region where $\epsilon_P$ is less than $1/4$ and that where $\Po$  is less than $1/4$ do not overlap.
\qed


\subsection{Proof of Statement~\ref{statement:ErrorSomewhenSmall}}

To prove Statement~\ref{statement:ErrorSomewhenSmall} in Theorem~\ref{thm:MainUncertainty} we need to know how $\Po$ and $\epsilon_P$ behave as $\beta$ goes to $0$. 
This implies that we need to know this behavior for $\Var E$, $\Varo T$, and $\epsilon_T$, and therefore also for $\tilde c_3$ and $\tilde c_4$ (recall Lemma~\ref{lem:ErrorMeanVariance}).
This information will be determined in the next Lemmas, to prove which we will make use of the following auxiliary result.

\begin{lemma}\label{lem:NormV}
Let the one-parameter family of potentials $\{V_b\}_{b\in[0,\infty)}$ be in the set $\CV$.
Then, 
\begin{equation}
\exists B>0, \epsilon>0 :
\|V_b\|_1 > \epsilon,  \ \forall b >B.
\end{equation}
\end{lemma}
\begin{proof}
Let us assume that the statement of the Lemma is false.
Then,
\begin{equation}
\forall B>0, \epsilon>0,
\ \exists b_\epsilon > B
:
\|V_{b_\epsilon}\|_1 < \epsilon.
\end{equation}
Form Property~\ref{item:BetoToZero} of Definition~\ref{def:CV} we have that $\lim_{b\to\infty}\beta(b)=0$, i.e.\
\begin{equation}
\forall \epsilon_\beta>0,
\ \exists B_\beta>0 
:
\beta(b) < \epsilon_\beta,
\ \forall b > B_\beta.
\end{equation}
Given $\epsilon$, we choose $\epsilon_\beta=\epsilon$, to which a certain $B_\beta$ corresponds; then, we choose $B=B_\beta$.
All together this gives
\begin{equation}
\forall \epsilon>0,
\ \exists b_\epsilon>0
:
\|V_{b_\epsilon}\|_1 < \epsilon,\
\beta(b_\epsilon) < \epsilon .
\end{equation}
Then,
\begin{equation}
\|rV_{b_\epsilon}\|_1
\leq R_V \|V_{b_\epsilon}\|_1
< R_V \epsilon.
\end{equation}
Consider now the integral equation~\eqref{eq:IntEqF} for the Jost function $F_b$, that is
\begin{equation}
F_b(k)
=
1  +   \int_0^{R_V}    e^{i k r}  V_b(r)  \varphi_b(k,r)   \,  dr  ,
\qquad \forall k\in\C
\end{equation}
and the bound for the generalized eigenfunctions $\varphi_b$ given in Eq.~\eqref{eq:BoundPhi}, i.e.\
\begin{equation}
| \varphi_b(k,r) |  
	\leq 4 e^{4\| r'V_b(r') \|_1} 
	\frac {r}  {1+|k| r}   e^{|\Im k| r}  ,
\qquad \forall k\in\C, \ r\geq0.
\end{equation}
For $b=b_\epsilon$ we can write
\begin{equation}
| \varphi_{b_\epsilon}(k,r) |  
	\leq 4 r e^{4R_V \epsilon + |\Im k| r}  ,
\end{equation}
and
\begin{align}
|F_{b_\epsilon}(k)|
&\geq
\left| 1  -   \left| \int_0^{R_V}    e^{i k r}  V_{b_\epsilon}(r)  \varphi_{b_\epsilon}(k,r)   \,  dr  \right| \right|
\\
&\geq
1  -  \left| \int_0^{R_V}    e^{i k r}  V_{b_\epsilon}(r)  \varphi_{b_\epsilon}(k,r)   \,  dr  \right| 
\\
&\geq
1  -  \int_0^{R_V}    e^{|\Im k| r}  |V_{b_\epsilon}(r)|  |\varphi_{b_\epsilon}(k,r) |  \,  dr  
\\
&\geq
1  - 4R_V e^{2( |\Im k| +2\epsilon) R_V}   \epsilon
.
\end{align}
In particular, for $k=k_0(b_\epsilon)$ we get
\begin{equation}\label{eq:FEpsilon}
|F_{b_\epsilon}(k_0(b_\epsilon))|
\geq
1  - 4R_V e^{2(\beta(b_\epsilon) +2\epsilon) R_V}   \epsilon
\geq
1  - 4R_V e^{6\epsilon R_V}   \epsilon,
\end{equation}
therefore for $\epsilon$ small enough we can make the right hand side of Eq.~\eqref{eq:FEpsilon} as close to one as wanted, therefore we have that
\begin{equation}
\exists b>0
:
|F_{b}(k_0(b))|  \geq  1/2.
\end{equation}
On the other side, by definition of resonance,
\begin{equation}
|F_{b}(k_0(b))|  = 0,
\ \forall b \geq 0,
\end{equation}
hence a contradiction.
\QED
\end{proof}

\begin{lemma}\label{lem:C3C4BetaOrder}
Let $\sigma=\beta$, the assumptions of Corollary~\ref{cor:main} and Hypothesis~\ref{hyp:BetaLimit} be satisfied, then as $\beta\to0$
\begin{align}
  \tilde{c}_3&= O(1),
  \label{eq:B0OrderC3}\\
  \tilde{c}_4&= O\left(\frac{1}{\beta^5}\left[\log\left(\frac{1}{\beta}\right)\right]^{12}\right)
  \label{eq:B0OrderC4}.
  \end{align}
\end{lemma}
\begin{proof}
The quantities $\tilde{c}_3$ and $\tilde{c}_4$ depend on $z_{ac}(n)$, $z_{ac,K}(n)$, $M_{K,\infty}(n)$, and $M_1(n)$, which in turn are combinations of $r_0$, $s_K$, $s$, $C_{n,K}$, $C_n$, and $q= \frac{1}{2\|V\|_1}+6R_V$ (see Definitions~\ref{def:SKAndKTildeAndS} and~\ref{def:smallz}, and the definitions given in Theorems~\ref{th:SBoundsK} and~\ref{th:GlobalSBounds}), so we first determine how the latter quantities behave as $\beta\to0$. 
Wherever we use the order-notation in this proof we always refer to the limit $\beta\to0$.

First, $s_K=1$ because of Eq.~\eqref{eq:SkIsOne}, and $r_0=O(1)$ because of Property~\ref{item:RZeroBetaOrder} of Definition~\ref{def:CV}.
Moreover, Lemma~\ref{lem:NormV} implies that $1/\|V\|_1$ is bounded from above, therefore $q=O(1)$.
Under the assumptions on the potential stated in Section~\ref{sec:AssumptionsPotentialET}, Definition~\ref{def:SKAndKTildeAndS} for $s$ becomes
\begin{equation}
  \frac{1}{s}
  = \sum_{n=0}^{\nu_{\tilde K}-1}   \frac {1}  {\beta_n} ,
\end{equation}
hence
\begin{align}
  s&\leq\beta ,
\\
  \frac{1}{s}
  & \leq \frac {\nu_{\tilde K}}  {\beta} 
  =  O\left( \beta^{-1} \left(\log{\beta}\right)^2 \right)  ,
\end{align}
having used Property~\ref{item:NuTildeBetaOrder} of Definition~\ref{def:CV}.
Using these results in the definition of the constants $C_{n,K}$ given in Theorem~\ref{th:SBoundsK} and of the constants $C_{n}$ given in Theorem~\ref{th:GlobalSBounds}  we get
\begin{equation}
C_{n,K}  = O(1),
\qquad
C_{n}  = O(1),
\qquad n=1,2,3.
\end{equation}
Similarly, from  Definition~\ref{def:smallz} we get
\begin{equation}
z_{ac,K}(n)  = O(1),
\qquad
z_{ac}(n)  = O(1),
\qquad n=1,2,3.
\end{equation}

We now turn to the constants $M_{K,\infty}(n)$ and $M_1(n)$ given in Definition~\ref{def:CorollaryConstants}.
Recalling that $\sigma=\beta$, the only quantity that needs to be determined is
\begin{align}
  E_{\beta,\sigma/\sqrt 2}=E_{\beta,\beta/\sqrt 2}=\sqrt \pi e^{\beta^4/2}\left[1+\erf\left(\frac{\beta^2}{\sqrt 2}\right)\right]=O(1)  .
 \label{eq:OrderEBB}
\end{align}
Then, we have
\begin{equation}
M_{K,\infty}(n)  = O(1),
\qquad
M_1(n) = O\left( \beta^{-n}  (\log\beta)^{2n+1} \right),
\qquad n=0,1,2.
\end{equation}

Substituting these results into the definitions of $\tilde c_3$ and $\tilde c_4$ given in Definition~\ref{def:CorollaryConstants}, we  get the statement of the Lemma.
\QED
\end{proof}

\begin{lemma}\label{lem:BetaOrderPAndRelatives}
Let $\sigma=\beta$,  the assumptions of Corollary~\ref{cor:main} and Hypothesis~\ref{hyp:BetaLimit} be satisfied, then for the wave function $\psi$ we have 
\begin{align}
\Var E &= O\left( \beta^{-2}\right),
\label{eq:B0OrderEnergyVariance}
\\
\Varo T &= O\left( \beta^{-2}\right),
\label{eq:B0OrderVaroT}
\\
\epsilon_T  &=   O\left(\beta^{-2+2/17}(\log\beta)^{12}\right)   ,
\qquad
\text{as } \beta\to0.
\label{eq:B0OrderErrorTimeVariance}
\end{align}
\end{lemma}
\begin{proof}
From Eq.~\eqref{eq:Evariance} for $\Var E$, using  Eq.~\eqref{eq:OrderEBB}, we get immediately
\begin{equation}
\Var E  =  O\left( \beta^{-2}\right),
\quad\text{as } \beta\to0.
\end{equation}
Recalling that $\Varo T = 1/ \gamma^2 = 1/(4 \alpha\beta)^2$, we have also that
\begin{equation}
\Varo T  =  O\left( \beta^{-2}\right),
\quad\text{as } \beta\to0.
\end{equation}

\begin{figure}
\hfill%
\subfloat[\label{subfig:Beta0ErrorMean}]{\includegraphics [width=.45\textwidth]{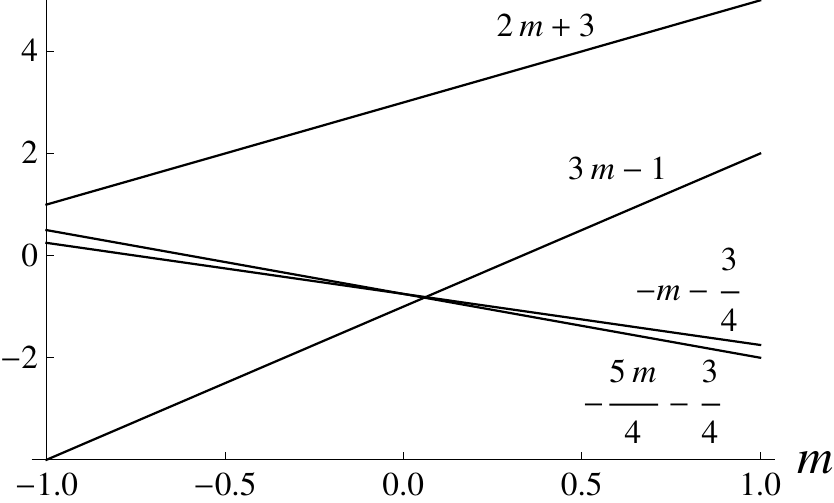}}%
\hfill%
\subfloat[\label{subfig:Beta0ErrorsMeanCloseUp}]{\includegraphics [width=.45\textwidth]{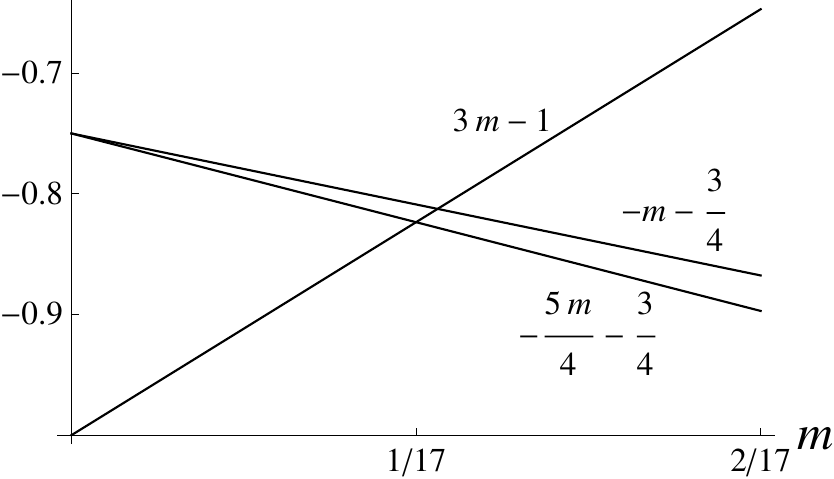}}%
\hfill\mbox{}
\caption[]{
\subref {subfig:Beta0ErrorMean} 
	Plot of the $\beta$-order of $\omega$ given by Eq.~\eqref{eq:DeltaBeta0Orders} as a function of $m$; 
\subref{subfig:Beta0ErrorsMeanCloseUp} Close up view of the optimal region.}
\label{fig:Beta0ErrorMean}
\end{figure}

We now turn to the $\beta$-order of the error $\epsilon_T$.
Substituting the formulas~\eqref{eq:B0OrderC3} and~\eqref{eq:B0OrderC4} for the $\beta$-order of the constants $\tilde{c}_3$ and $\tilde{c}_4$ into Lemma~\ref{lem:ErrorMeanVariance} we get
\begin{align}
\omega_{(0,A)}  &= O\left(\beta^{1/2} A^{5/4} \right)  + O\left(\beta^{1/4} A \right) ,
\\
\omega_{[A,\infty)}  
	&=  O\left(\beta \, A^{-2}\right) + O\left(\beta^{-4}(\log\beta)^{12} A^{-3} \right),
\\
\zeta_{(0,A)} &= O\left(\beta^{1/2}   A^{9/4} \right) + O\left( \beta^{1/4}  A^2 \right) ,
\\
\zeta_{[A,\infty)} 
	&=  O\left(\beta \, A^{-1} \right) + O\left(\beta^{-4}(\log\beta)^{12} A^{-2} \right) ,
\qquad\text{as }\beta\to 0;
\end{align}
therefore,
\begin{align}
\omega 
	&= O\left(\beta^{-4}(\log\beta)^{12} A^{-3} \right) + O\left(\beta \, A^{-2} \right)
	\nonumber\\
	&\qquad	+ O\left( \beta^{1/4}  A \right)   + O\left( \beta^{1/2}   A^{5/4} \right)  ,
\\
\zeta 
	&=  O\left(\beta^{-4}(\log\beta)^{12} A^{-2} \right) +  O\left(\beta \, A^{-1} \right)
	\nonumber\\
	&\qquad	+O\left(  \beta^{1/4}  A^2 \right)  +  O\left( \beta^{1/2}   A^{9/4} \right) ,
\\
&\qquad\text{as }\beta\to 0.
\end{align}
For every $A>0$ we can write $A = \beta^{-1-m}$ with $m\in\R$, hence
\begin{align}
\omega 
	&= O\left(\beta^{-1+3m}(\log\beta)^{12}\right)  
		+ O\left(\beta^{3+2m}\right) 
	\nonumber\\
	&\qquad	+ O\left( \beta^{-3/4-m} \right)    
		+ O\left( \beta^{-3/4-5/4m}  \right)    ,
	\label{eq:DeltaBeta0Orders}
\\
\zeta 
	&=  O\left(\beta^{-2+2m}(\log\beta)^{12}\right) 
		+  O\left(\beta^{2+m}\right) 
	\nonumber\\
	&\qquad	+O\left(  \beta^{-7/4-2m}  \right)  
		+  O\left( \beta^{-7/4-9/4m}  \right)   ,
\\
&\qquad\text{as }\beta\to 0.
\end{align}
Plotting the exponents of every term as functions of $m$ (Fig.~\ref{fig:Beta0ErrorMean}), it is easy to see that the choice that minimizes $\omega$ is such that
\begin{equation}
-1+3m = -\frac{3}{4} - \frac{5}{4} m,
\end{equation}
i.e.\ $m=1/17$.
Then, the value for the parameter $A$ that minimizes in terms of $\beta$-orders the error $\omega$ on the mean  is 
\begin{equation}\label{eq:OptimalA}
A = \beta^{-18/17}.
\end{equation}
In the same way one sees that this value minimizes $\zeta$ too.
Substituting we get
\begin{align}
\omega 
	&= O\left(\beta^{-14/17}(\log\beta)^{12}\right)    ,
\\
\zeta 
	&=  O\left(\beta^{-32/17}(\log\beta)^{12}\right) ,
\qquad\text{as }\beta\to 0,
\end{align}
and recalling that $\gamma=4\alpha\beta$,
\begin{equation}
\epsilon_T 
=  2\zeta  + \omega^2  +  \frac{2}{\gamma}  \omega
=   O\left(\beta^{-32/17}(\log\beta)^{12}\right)   ,	
\qquad\text{as }\beta\to 0.
\end{equation}
\QED
\end{proof}

We are now ready to prove Statement~\ref{statement:ErrorSomewhenSmall} in Theorem~\ref{thm:MainUncertainty}.

\begin{proof}[of Statement~\ref{statement:ErrorSomewhenSmall} in Theorem~\ref{thm:MainUncertainty}]
From Lemma~\ref{lem:BetaOrderPAndRelatives} we can calculate the $\beta$-order of the approximate product $\Po$ and of its error $\epsilon_P$, indeed
\begin{align}
\Po 
& = \Var E \,\Varo T 
=  O \left( \beta^{-4}  \right) ,
\label{eq:B0OrderP}
\\
\epsilon_P 
&= \epsilon_T \Var E
=   O\left(\beta^{-4+2/17}(\log\beta)^{12}\right)   ,
\qquad\text{as }\beta\to 0.
\end{align}
Then,
\begin{equation}
\Po - \epsilon_P 
=  O \left( \beta^{-4}  \right)   \left[  1  - O\left(\beta^{2/17}(\log\beta)^{12}\right)  \right] 
=  O \left( \beta^{-4}  \right) ,
\qquad\text{as }\beta\to 0,
\end{equation}
and the statement of the theorem follows immediately.
\QED
\end{proof}

\section{Proof of Theorem~\ref{thm:ETvsLL}}

To prove Theorem~\ref{thm:ETvsLL} we will use Lemma~\ref{lem:BetaOrderPAndRelatives}, but we also need estimates of $\Gamma$ and $\tau$, for which we need pointwise bounds on  $\Prob(T\leq t)$ and $\Pi_E$.

\begin{lemma}\label{lem:Lifetime}
Let $\sigma=\beta$, the assumptions of Corollary~\ref{cor:main} and Hypothesis~\ref{hyp:BetaLimit} be satisfied, then for the lifetime $\tau$ of the wave function $\psi$ we have
\begin{equation}
\tau  =  
	\frac{1}{\gamma}  
	\left[  1  +  O\left(  \beta^{4/17} \left(\log\beta\right)^{12} \right)  \right],
\qquad\text{as }\beta\to0.
\end{equation}
\end{lemma}
\begin{proof}
At first, notice that from Eq.~\eqref{eq:TDistribution} we have
\begin{equation}
\Prob(T\leq t) 
= \int_0^{t}  \Pi_T(t') d t'
=   1   -   \frac {\|\1_{R}   e^{-iHt}  \psi\|_2^2}  {\|\1_{R}   \psi\|_2^2}  ,
\end{equation}
therefore the lifetime $\tau$ is such that
\begin{equation}\label{eq:TauEq}
\frac {\|\1_{R}   e^{-iH\tau}  \psi\|_2^2}  {\|\1_{R}   \psi\|_2^2} 
=
\frac{1}{e}.
\end{equation}
Moreover, from Eq.~\eqref{eq:Exp} we have
\begin{align}
  \frac {\|\1_{R}   e^{-ik_0^2t}  f_{R_2(t)}\|_2^2}  {\|\1_{R}   \psi\|_2^2}  
  =  e^{-\gamma t},
\end{align}
and using Lemma~\ref{lem:SurvivalProbError} we get that for any $A>0$ 
\begin{equation}
\left|   \frac {\|\1_{R}   e^{-iHt}  \psi\|_2^2}  {\|\1_{R}   \psi\|_2^2}   
		 -  e^{-\gamma t}  \right|
\leq   \xi_{(0,A)}(t)  \1_{(0,A)}  +  \xi_{[A,\infty)} (t)  \1_{[A,\infty)}  .
\end{equation}
We use the fact that
\begin{alignat*}{2}
\xi_{(0,A)} (t)  &\leq  \xi_{(0,A)} (A) ,  \qquad&&\text{for }t\in(0,A),  
\\
\xi_{[A,\infty)} (t)  &\leq  \xi_{[A,\infty)} (A) ,  \qquad&&\text{for }t\in{[A,\infty)},
\end{alignat*}
and define
\begin{equation}
\xi \coloneqq  \xi_{(0,A)}(A)   +  \xi_{[A,\infty)} (A) ,
\end{equation}
getting
\begin{equation}
\left|   \frac {\|\1_{R}   e^{-iHt}  \psi\|_2^2}  {\|\1_{R}   \psi\|_2^2}   
		 -  e^{-\gamma t}  \right|
\leq  \xi .
\label{eq:BoundPointwiseTimeError}
\end{equation}
It is convenient to consider the equation
\begin{equation}
e^{-\gamma t} =  \frac{1}{e}  ,
\end{equation}
whose solution is $1/\gamma$.
Then, the lifetime $\tau$ can not be greater than the solution of the equation
\begin{equation}
e^{-\gamma t}  +  \xi  =  \frac{1}{e}  ,
\end{equation}
nor less than the solution of the equation
\begin{equation}
e^{-\gamma t}  -  \xi  =  \frac{1}{e}  ,
\end{equation}
which are
\begin{align}
&\frac{1}{\gamma}  \left(   1 +  \log  \frac{1}{1-e  \xi}   \right),
\\
&\frac{1}{\gamma}  \left[   1 -  \log  \left(  1+  e  \xi \right)  \right]
\end{align}
respectively.
Using the bounds 
\begin{align}
\log  \frac{1}{1-x}  
&=  \int_0^x   \frac{1}{1-x'}  d x'
\leq 2x,
\qquad\text{for } 0  <  x  \leq \frac 1 2,
\\
\log  (1+x) 
& \leq x  \leq  2x,  
\qquad\text{for } x>0,
\end{align}
we get
\begin{equation}\label{eq:BoundTau}
1 -   2  e \xi  
\leq \gamma \tau \leq
1 +   2  e \xi   ,
\end{equation}
that implies
\begin{equation}
\tau  =  
	\frac{1}{4\alpha\beta}  
	\bigl(  1  +  O\left( \xi \right)  \bigr),
\qquad\text{as }\beta\to0.
\end{equation}

To determine the behavior of $\xi$ as $\beta$ goes to zero we substitute in Lemma~\ref{lem:SurvivalProbError} the formulas~\eqref{eq:B0OrderC3} and~\eqref{eq:B0OrderC4} for the $\beta$-order of the constants $\tilde{c}_3$ and $\tilde{c}_4$ and set $\sigma=\beta$, getting
\begin{align}
\xi_{(0,A)} (t)  &=  O\left( \beta^{1/2}\right)   t^{1/4}  + O\left(\beta^{1/4}  \right),  
\\
\xi_{[A,\infty)} (t)
&=  O(\beta)\,  t^{-3} + O\left(\beta^{-4}(\log\beta)^{12}\right)  t^{-4}  ,
\qquad\text{as }\beta\to 0.
\end{align}
As suggested by Eq.~\eqref{eq:OptimalA}, we set 
\begin{equation}
A = \beta^{-18/17}  ,
\end{equation}
therefore  
\begin{equation}
\xi  =   O\left(  \beta^{4 /17}   \left(\log\beta\right)^{12}\right)  ,
\qquad\text{as }\beta\to 0.
\end{equation}
Recalling Eq.~\eqref{eq:BoundPointwiseTimeError}, we see that the arrival time cumulative distribution function pointwise converges to $1-e^{-\gamma t} $. 
Moreover, using Eq.~\eqref{eq:BoundTau}  we get the proposition.
\QED
\end{proof}

\begin{lemma}\label{lem:Linewidth}
Let $\sigma=\beta$ and let $\Gamma$ denote the linewidth of the wave function $\psi$. Then, 
   \begin{align}\label{eq:Linewidth}
	\Gamma=\gamma+O(\beta^2),
	\qquad\text{as }\beta\to0.
  \end{align}
\end{lemma}
\begin{proof}
  Note that $E=k^2$ implies
  \begin{align}
	\frac{|\hat\psi(k)|^2}{\|\psi\|_2^2}\,dk=\frac{|\hat\psi(E^{1/2})|^2}{\|\psi\|_2^2}\,\frac{1}{2\sqrt E}dE,
  \end{align}
  so that the probability density for energy reads
  \begin{align}
	\Pi_E(E)
	=\frac{|\hat\psi(E^{1/2})|^2}{\|\psi\|_2^2}\,\frac{1}{2\sqrt E}.
  \end{align}
  Let $K=\alpha/4.$ We now look at $\Pi_E$ on $[0,K^2)$ and $[K^2,\infty)$ separately and show that for $\beta$ small enough it attains its maximum on the latter interval. 
  
Corollary~\ref{cor:main} shows that $\hat\psi^{(n)}\in L^\infty_{loc}\cap L^1_w$ for $n=0,1,2$ so that Lemma~\ref{lem:psihat0} applies to it; using the assumption that the Hamiltonian has no zero-resonance we then get  for $E\in [0,K^2)$
  \begin{align}
	\Pi_E(E)=\frac{|\hat\psi(E^{1/2})|^2}{\|\psi\|_2^2}\,\frac{1}{2\sqrt E}\leq\frac{\|\1_K\dkp{\hat\psi}\|_\infty^2}{2\|\psi\|_2^2}  K.
  \end{align}
Plugging $\sigma=\beta$ into Eq.~\eqref{eq:NormInitialWaveFunction} for $\|\psi\|_2^2$ and into the bound on $\|\1_K\dkp{\hat\psi}\|_\infty$ given by Eq.~\eqref{eq:PsiHat1} with $n=1$ and using Property~\ref{item:AlphaConstants} of Definition~\ref{def:CV}, we see that for $E\in [0,K^2)$
\begin{equation}
\label{eq:LemLinewidth0}
	\Pi_E(E)= O(\beta),  
	\qquad\text{as } \beta\to0.
\end{equation}
  
  Let us now look at the probability density for energies in $[K^2,\infty)$. From Eq.~\eqref{eq:SkibstedDistribution} for $\hat f_R$ and setting
  \begin{align}\label{eq:LemLinewidth01}
	\rho(k)=-\frac{1}{2}\frac{e^{i(k_0-k)R}}{k-k_0}\bar S(k)
  \end{align}
we obtain
  \begin{align}
	\hat\psi(k)
	&=\hat f_R(k)+\hat g_R(k)
	=\rho(k) -\frac{1}{2}\frac{e^{i(k_0+k)R}}{k+k_0} +  \hat g_R(k).
  \end{align}
  We will see that $\rho(k)$ gives the main contribution to $\hat\psi(k)$, and therefore to $\Pi_E(E)$.
To this end, consider
  \begin{align}\label{eq:LemLinewidth1}
	&\left|\frac{1}{\sqrt E}|\hat\psi(E^{1/2})|^2-\frac{1}{\alpha}|\rho(E^{1/2})|^2\right|\nonumber\\
	&\quad=\left|k^{-\frac{1}{2}}|\hat\psi(k)|-\alpha^{-\frac{1}{2}}|\rho(k)|\right|\left|k^{-\frac{1}{2}}|\hat\psi(k)|+\alpha^{-\frac{1}{2}}|\rho(k)|\right|.
  \end{align}
  We start bounding the factor with the sum, just by bounding the summands separately.  Since $\hat\psi$ contains $\hat g_R$, we need a bound on it. From Eqs.~\eqref{eq:PsiPlus},~\eqref{eq:phiRBiggerRs} and~\eqref{eq:fboundary} we see that for $r\geq R_V$
  \begin{align}
	f(k_0,r)&=e^{i k_0r},\\
	\bar \psi^+(k,r)&=\frac{1}{2i}(e^{ikr}-S(-k)e^{-ikr})
  \end{align}
  and hence
  \begin{align} 
	|\hat g_R(k)|
	&=\left|\int_R^\infty f(k_0,R)\exp\left(-\frac{(r-R)^2}{2\sigma^2}\right)\bar\psi^+(k,r)\,dr\right|\\
	&\leq\int_R^\infty\exp\left(\beta r-\frac{(r-R)^2}{2\sigma^2}\right)\,dr\\
	&=e^{\beta R}\frac{\sigma}{\sqrt 2}E_{\beta,\sigma/\sqrt 2}.\label{eq:cor1}
  \end{align}
  Using this, the fact that $|S|=1$, and Eq.~\eqref{eq:OrderEBB}  we find that 
  \begin{align}
	k^{-\frac{1}{2}}|\hat\psi(k)|
	&\leq \frac{1}{\sqrt k}\left[\frac{e^{\beta R}}{|k-k_0|}+\frac{\beta e^{\beta R}}{\sqrt{2}}E_{\beta,\beta/\sqrt 2}\right]
	\nonumber\\
	&\leq \frac{e^{\beta R}}{\sqrt K}\left[\frac{1}{\beta}+\frac{\beta}{\sqrt 2}E_{\beta,\beta/\sqrt 2}\right]=O(\beta^{-1}),\\
	|\rho(k)| &\leq\frac{e^{\beta R}}{2\beta}=O(\beta^{-1}),
	\quad\text{as }\beta\to0,
\label{eq:Piece2}
  \end{align}
  
  Let us now estimate the factor with the difference in Eq.~\eqref{eq:LemLinewidth1}. Consider
  \begin{align}
	&\left|k^{-\frac{1}{2}}|\hat\psi(k)|-\alpha^{-\frac{1}{2}}|\rho(k)|\right|\nonumber\\
	&\leq\left|k^{-\frac{1}{2}}\hat\psi(k)-\alpha^{-\frac{1}{2}}\rho(k)\right|\nonumber\\
	&\leq\left|\left(k^{-\frac{1}{2}}-\alpha^{-\frac{1}{2}}\right)\rho(k)+k^{-\frac{1}{2}}\frac{1}{2}\frac{e^{i(k_0+k)R}}{k+k_0}+k^{-\frac{1}{2}}\hat g_R(k)\right|\nonumber\\
	&\leq\left|k^{-\frac{1}{2}}\alpha^{-\frac{1}{2}}\frac{\alpha-k}{k^{\frac{1}{2}}+\alpha^{\frac{1}{2}}}\rho(k)+k^{-\frac{1}{2}}\frac{1}{2}\frac{e^{i(k_0+k)R}}{k+k_0}+k^{-\frac{1}{2}}\hat g_R(k)\right|\nonumber\\
	&\leq \frac{K^{-\frac{1}{2}}\alpha^{-\frac{1}{2}}}{K^{\frac{1}{2}}+\alpha^{\frac{1}{2}}}\frac{e^{\beta R}}{2}\frac{|k-\alpha|}{\sqrt{(k-\alpha)^2+\beta^2}}+\frac{e^{\beta R}}{2\sqrt K(K+\alpha)}+\frac{\beta e^{\beta R}}{\sqrt {2K}}E_{\beta,\beta/\sqrt 2}\nonumber\\
	&\leq \frac{K^{-\frac{1}{2}}\alpha^{-\frac{1}{2}}}{K^{\frac{1}{2}}+\alpha^{\frac{1}{2}}}\frac{e^{\beta R}}{2}+\frac{e^{\beta R}}{2\sqrt K(K+\alpha)}+\frac{\beta e^{\beta R}}{\sqrt {2K}}E_{\beta,\beta/\sqrt 2}
	\nonumber\\
	&=O(1),
	\quad\text{as }\beta\to0.
\label{eq:Piece1}
  \end{align}
  Plugging these inequalities into Eq.~\eqref{eq:LemLinewidth1} and letting
  \begin{align}
	\delta
	&\coloneqq \frac{e^{\beta R}}{2\|\psi\|_2^2}
	\left(\frac{1}{2}\frac{K^{-\frac{1}{2}}\alpha^{-\frac{1}{2}}}{K^{\frac{1}{2}}+\alpha^{\frac{1}{2}}}+\frac{1}{2\sqrt K(K+\alpha)}+\frac{\beta}{\sqrt {2K}}E_{\beta,\beta/\sqrt 2}\right)\nonumber\\
	&\qquad\qquad\times\left(\frac{1}{\sqrt K}\left[\frac{1}{\beta}+\frac{\beta}{\sqrt 2}E_{\beta,\beta/\sqrt 2}\right]+\frac{1}{2\alpha^{\frac{1}{2}}\beta}\right)
  \end{align}
  we obtain that
  \begin{align}\label{eq:LemLinewidth2}
	\left|\Pi_E(E)-\frac{1}{2\alpha}\frac{|\rho(E^{1/2})|^2}{\|\psi\|_2^2}\right|
	\leq\delta.
  \end{align}
  From Eq.~\eqref{eq:NormInitialWaveFunction} for $\|\psi\|_2^2$, we see that 
  \begin{align}\label{eq:LemLinewidth21}
	\|\psi\|_2^{-2}
&={2\beta}{e^{-2\beta R}}\left[1+\sqrt\pi\beta^2 e^{\beta^4}(1+\erf(\beta^2))\right]^{-1}
\\&=O(\beta),
 \quad\text{as }\beta\to0,
  \end{align}
  which together with Eq.~\eqref{eq:Piece1} implies that
  \begin{align}\label{eq:LemLinewidth22}
	\delta=O(1),
	\quad\text{as }\beta\to0.
  \end{align}
Evaluating $\rho$ in $\alpha$ we see that there constants $B,C>0$ such that
 \begin{align}
	\frac{1}{2\alpha}\frac{|\rho(\alpha)|^2}{\|\psi\|_2^2}
	&=\frac{1}{4\alpha\beta\left[1+\sqrt\pi\beta^2 e^{\beta^4}(1+\erf(\beta^2))\right]}
	\\&\geq \frac{C}{\beta},
	\quad\forall \beta<B,
  \end{align}
 which together with Eqs.~\eqref{eq:LemLinewidth2} and~\eqref{eq:LemLinewidth22} shows that
  \begin{align}\label{eq:LemLinewidth3}
	\Pi_E(\alpha^2)\geq \frac{C}{\beta},
	\quad\forall \beta<B.
  \end{align}
  Considering Eq.~\eqref{eq:LemLinewidth0} and the fact that $\alpha^2\in[K^2,\infty)$, we can conclude that the probability density $\Pi_E$ attains its maximum in $[K^2,\infty)$ for $\beta$ small enough.

\begin{figure}
\hfill%
\subfloat[\label{subfig:Linewidth}]{\includegraphics [width=.45\textwidth]{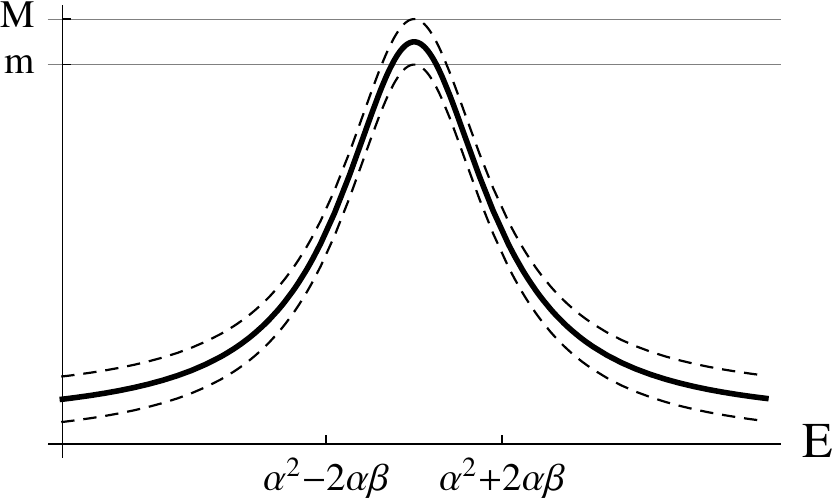}}%
\hfill%
\subfloat[\label{subfig:LinewidthCloseup}]{\includegraphics [width=.45\textwidth]{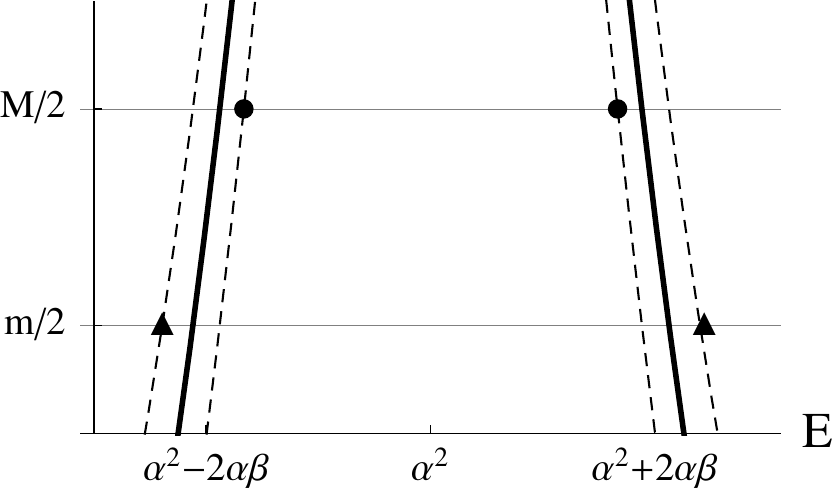}}%
\hfill\mbox{}
\caption[]{
\subref {subfig:Linewidth} 
The thick line is a plot of $|\rho(E^{1/2})|^2/2\alpha\|\psi\|_2^2$ and the dashed lines are a plot of $|\rho(E^{1/2})|^2/2\alpha\|\psi\|_2^2\pm\delta$. 
According to Eq.~\eqref{eq:LemLinewidth2}, the function $\Pi_E$ lies between the dashed lines. 
The constant $M$ is therefore the largest possible maximum of $\Pi_E$ and $m$ is the smallest possible maximum. 
\subref{subfig:LinewidthCloseup} 
A closeup of Fig.~\ref{subfig:Linewidth} is plotted to show that the distance between the two \protect\raisebox{-.1ex}{\large$\bullet$} gives a lower bound on the linewidth $\Gamma$ of $\psi$ and the distance between the two $\blacktriangle$ gives an upper bound.
}
\end{figure}

  We now determine the linewidth of $\Pi_E$. The basic idea is that for small enough $\beta$ the linewidth of $\Pi_E(E)$ is approximately the linewidth of $|\rho(E^{1/2})|^2/2\alpha\|\psi\|_2^2$, because the difference~\eqref{eq:LemLinewidth2} between these two functions is small compared to the maximum of $\Pi_E$, which according to Eq.~\eqref{eq:LemLinewidth3} approximately  $\beta^{-1}$ for $\beta$ small enough.
According to Eq.~\eqref{eq:LemLinewidth2} the function $\Pi_E$ lies between the two functions $|\rho(E^{1/2})|^2/(2\alpha\|\psi\|_2^2)\pm\delta$ (cf.\ Fig.~\ref{subfig:Linewidth}).
These  functions attain their maximum for $E^{1/2}=\alpha$.
Now let
  \begin{align}
	m\coloneqq \frac{|\rho(\alpha)|^2}{2\alpha\|\psi\|_2^2}-\delta=\frac{e^{2\beta R}}{8\alpha\beta^2\|\psi\|_2^2}-\delta,
	\label{eq:LemLinewidth4}
	\\
	M\coloneqq \frac{|\rho(\alpha)|^2}{2\alpha\|\psi\|_2^2}+\delta=\frac{e^{2\beta R}}{8\alpha\beta^2\|\psi\|_2^2}+\delta.
	\label{eq:LemLinewidth5}
  \end{align}
The linewidth of $\Pi_E$ is therefore bounded from above by the distance between the two solutions of (cf.\ Fig.~\ref{subfig:LinewidthCloseup})
  \begin{align}\label{eq:LemLinewidth6}
	\frac{|\rho(E^{1/2})|^2}{2\alpha\|\psi\|_2^2}+\delta=\frac{m}{2}  ,
  \end{align}
  and bounded from below by the distance between the two solutions of
  \begin{align}\label{eq:LemLinewidth7}
	\frac{|\rho(E^{1/2})|^2}{2\alpha\|\psi\|_2^2}-\delta=\frac{M}{2}.
  \end{align}
  First, let us look at Eq.~\eqref{eq:LemLinewidth6}. Using Eq.~\eqref{eq:LemLinewidth01} for $\rho$ it is straightfoward to see that the two solutions of Eq.~\eqref{eq:LemLinewidth6} are
  \begin{align}
	E_U^{\pm}&=\left(\alpha\pm\left[\frac{e^{2\beta R}}{8\alpha\|\psi\|_2^2}\frac{1}{m/2-\delta}-\beta^2\right]^{1/2}\right)^2,
  \end{align}
  so that the upper bound on the linewidth reads
  \begin{align}
	E_U^+-E_U^-
	&=4\alpha\beta\left[\frac{e^{2\beta R}}{8\alpha\beta^2\|\psi\|_2^2}\frac{1}{m/2-\delta}-1\right]^{1/2}\nonumber\\
	&=4\alpha\beta\left[\frac{e^{2\beta R}}{8\alpha\beta^2\|\psi\|_2^2}\left(\frac{e^{2\beta R}}{16\alpha\beta^2\|\psi\|_2^2}-\frac{3}{2}\delta\right)^{-1}-1\right]^{1/2}\nonumber\\
	&=4\alpha\beta\left[\frac{1}{1-24\alpha\beta^2\|\psi\|_2^2e^{-2\beta R}\delta}+\frac{24\alpha\beta^2\|\psi\|_2^2e^{-2\beta R}\delta}{1-24\alpha\beta^2\|\psi\|_2^2e^{-2\beta R}\delta}\right]^{1/2}.
  \end{align}
Similarly to Eq.~\eqref{eq:LemLinewidth21}, we  get $\|\psi\|^2_2=O(\beta^{-1})$ as $\beta\to0$, that together with Eq.~\eqref{eq:LemLinewidth22} for $\delta$ gives
  \begin{align}
	24\alpha\beta^2\|\psi\|_2^2e^{-2\beta R}\delta=O(\beta),
	\quad\text{as }\beta\to0,
  \end{align}
  and hence
  \begin{align}\label{eq:LemLinewidth8}
	\Gamma\leq E_U^+-E_U^-=4\alpha\beta(1+O(\beta)),
	\quad\text{as }\beta\to0.
  \end{align}
  Let us now consider Eq.~\eqref{eq:LemLinewidth7}. In the same way as before we see that its two solutions are
  \begin{align}
	E_L^{\pm}=\left(\alpha\pm\left[\frac{e^{2\beta R}}{8\alpha\|\psi\|_2^2}\frac{1}{M/2+\delta}-\beta^2\right]^{1/2}\right)^2,
  \end{align}
  so that the lower bound on the linewidth satisfies
  \begin{align}\label{eq:LemLinewidth9}
	E_L^+-E_L^-
	&=4\alpha\beta\left[\frac{e^{2\beta R}}{8\alpha\beta^2\|\psi\|_2^2}\frac{1}{M/2+\delta}-1\right]^{1/2}\nonumber\\
	&=4\alpha\beta\left[\frac{1}{1+24\alpha\beta^2\|\psi\|_2^2e^{-2\beta R}\delta}-\frac{24\alpha\beta^2\|\psi\|_2^2e^{-2\beta R}\delta}{1+24\alpha\beta^2\|\psi\|_2^2e^{-2\beta R}\delta}\right]^{1/2}\nonumber\\
	&=4\alpha\beta(1-O(\beta)),
	\qquad\text{as }\beta\to0.
  \end{align}
Collecting Eqs.~\eqref{eq:LemLinewidth8} and~\eqref{eq:LemLinewidth9} we get the assertion of the Lemma.
  \QED
\end{proof}

We can finally prove Theorem~\ref{thm:ETvsLL}.

\begin{proof}[of Theorem~\ref{thm:ETvsLL}]
Lemma~\ref{lem:Lifetime} and Lemma~\ref{lem:Linewidth} together give that
\begin{equation}
\Gamma\tau  =  
 1  +  O\left(  \beta^{4/17} \left(\log\beta\right)^{12} \right)  
 \to1,
\qquad\text{as }\beta\to0,
\end{equation}
while the fact that
\begin{equation}
\Var E\,  \Var T  \geq  \Po - \epsilon_P , 
\end{equation}
togeher with Statement~\ref{statement:ErrorSomewhenSmall} of Theorem~\ref{thm:MainUncertainty}, gives
\begin{equation}
\lim_{\beta\to 0}  \Var E\, \Var T= \infty  .
\end{equation}
\QED
\end{proof}


\section*{Appendix: Proof of Corollary~\ref{cor:main}}
\addcontentsline{toc}{section}{Appendix: Proof of Corollary~\ref{cor:main}} 
By Definition~\ref{def:SKAndKTildeAndS} we immediately have Eq.~\eqref{eq:SkIsOne}.
To prove the estimates on the norms of $\hat\psi^{(n)}$ observe that
  \begin{align}
  	\|\1_K{\hat \psi}^{(n)}\|_\infty\leq\|\1_K{\hat f}_R^{(n)}\|_\infty+\|\1_K{\hat g}_R^{(n)}\|_\infty
  \end{align}
Since $\|\1_K{\hat f}_R^{(n)}\|_\infty$ has been determined in Lemma~\ref{lem:psi}, we are left with calculating $\|\1_K{\hat g}_R^{(n)}\|_\infty$.
We use for $\hat g_R$  the bound given in Eq.~\eqref{eq:cor1}; similarly,
  \begin{align}
	|\dkp{\hat g}_R(k)|
	&=\Biggl|\int_R^\infty \exp\left(ik_0r-\frac{(r-R)^2}{2\sigma^2}\right)
	\nonumber\\
	&\times\frac{1}{2i}\left(ir(e^{ikr}+S(-k)e^{-ikr})+\dkp S(-k)e^{-ikr}\right)\,dr\Biggr|\\
	&\leq\int_R^\infty \exp\left(\beta r-\frac{(r-R)^2}{2\sigma^2}\right)\left(r+\frac{|\dkp S(-k)|}{2}\right)\,dr\\
	&\leq e^{\beta R}\left[\sigma^2+\left(R+\beta\sigma^2+\frac{|\dkp S(-k)|}{2}\right)\frac{\sigma}{\sqrt 2}E_{\beta,\sigma/\sqrt 2}\right]\label{eq:cor2}
  \end{align}
  and
  \begin{align}
	|\ddkp{\hat g}_R(k)|
	&=\bigg|\int_R^\infty \exp\left(ik_0r-\frac{(r-R)^2}{2\sigma^2}\right)\frac{1}{2i}\bigg[-r^2(e^{ikr}-S(-k)e^{-ikr})\nonumber\\
	&\qquad-2ir\dkp S(-k)e^{-ikr}-\ddkp S(-k)e^{-ikr}\bigg]\,dr\bigg|\\
	&\leq e^{\beta R}\bigg[\sigma^2\left(2R+|\dkp S(-k)|+\beta\sigma^2\right)\nonumber\\
	&\qquad+\left(\frac{|\ddkp S(-k)|}{2}+|\dkp S(-k)|(R+\beta\sigma^2)+\sigma^2+(R+\beta\sigma^2)^2\right)\frac{\sigma}{\sqrt 2}E_{\beta,\sigma/\sqrt 2}\bigg].\label{eq:cor3}
  \end{align}
From these inequalities, using the bounds on $\|\1_KS^{(n)}\|_\infty$ from Theorem~\ref{th:SBoundsK}, we obtain
  \begin{align}
	\|\1_K\hat g_R\|_\infty
	&\leq e^{\beta R}\frac{\sigma}{\sqrt 2}E_{\beta,\sigma/\sqrt 2},
	\displaybreak[0]\\
	\|\1_K\dkp{\hat g}_R\|_\infty
	&\leq e^{\beta R}\left[\sigma^2+\left(R+\beta\sigma^2+\frac{C_{1,K}}{2}\right)\frac{\sigma}{\sqrt 2}E_{\beta,\sigma/\sqrt 2}\right],
	\displaybreak[0]\\
	\|\1_K\ddkp{\hat g}_R\|_\infty
	&\leq e^{\beta R}\bigg[\sigma^2\left(2R+C_{1,K}+\beta\sigma^2\right)\nonumber\\
	&\qquad+\left(\frac{C_{2,K}}{2}+C_{1,K}(R+\beta\sigma^2)+\sigma^2+(R+\beta\sigma^2)^2\right)\frac{\sigma}{\sqrt 2}E_{\beta,\sigma/\sqrt 2}\bigg] .
  \end{align}
  These bounds together with the bounds on $\|\1_K{\hat f}_R^{(n)}\|_\infty$ given in Lemma~\ref{lem:psi}  imply Eq.~\eqref{eq:PsiHat1}.

We will now prove Eq.~\eqref{eq:PsiHat2}.
Note that
\begin{align}
	\|{\hat \psi}^{(n)}w\|_1
	&\leq\|{\hat f}_R^{(n)}w\|_1+\|{\hat g}_R^{(n)}\|_\infty\|w\|_1\\
	&=\|{\hat f}_R^{(n)}w\|_1+\|{\hat g}_R^{(n)}\|_\infty\int_0^\infty\frac{1}{1+r^2}\,dr\\
	&=\|{\hat f}_R^{(n)}w\|_1+\frac{\pi}{2}\|{\hat g}_R^{(n)}\|_\infty.
  \end{align}
From the inequalities~\eqref{eq:cor1}-\eqref{eq:cor3}, using the bounds on $\|S^{(n)}\|_\infty$ from Theorem~\ref{th:GlobalSBounds}, we obtain
\begin{align}
	\|\hat g_Rw\|_1
	&\leq\frac{\pi}{2}\|\hat g_R\|_\infty
	\leq e^{\beta R}\frac{\pi\sigma}{2^{3/2}}E_{\beta,\sigma/\sqrt 2},\\
	\|\dkp{\hat g}_Rw\|_1
	&\leq\frac{\pi}{2}\|\dkp{\hat g}_R\|_\infty
	\\&\leq e^{\beta R}\left[\frac{\pi\sigma^2}{2}+\left(R+\beta\sigma^2+\frac{C_{1}}{2s}\right)\frac{\pi\sigma}{2^{3/2}}E_{\beta,\sigma/\sqrt 2}\right],\\
	\|\ddkp{\hat g}_Rw\|_1
	&\leq\frac{\pi}{2}\|\ddkp{\hat g}_R\|_\infty\nonumber\\
	&\leq e^{\beta R}\bigg[\frac{\pi\sigma^2}{2}\left(2R+\frac{C_1}{s}+\beta\sigma^2\right)\nonumber\\
	&\qquad+\left(\frac{C_{2}}{2s^2}+\frac{C_1}{s}(R+\beta\sigma^2)+\sigma^2+(R+\beta\sigma^2)^2\right)\frac{\pi\sigma}{2^{3/2}}E_{\beta,\sigma/\sqrt 2}\bigg].
  \end{align}  
  These bounds together with the bounds on  $\|{\hat f}_R^{(n)}w\|_1$ given in Lemma~\ref{lem:psi}  imply Eq.~\eqref{eq:PsiHat2}.

  Equation~\eqref{eq:CorollaryBound} is then an immediate consequence of Theorem~\ref{thm:main_ac}.
\qed

\chapter{On Quantitative Scattering Estimates}
\label{ch:Scattering}
\blfootnote{The results presented in this chapter are the product of a teamwork with Robert Grummt \citep{GrummtVona2014a}.}
\section{Introduction}

The analysis of the previous chapter rests upon the ability to control, besides the exponential decay regime \citep[see][for a rigorous proof]{Skibsted86}, also the polynomial decay regime that takes over at late times~\citetext{\citealp{SimonRes,Peres}; see \citealp{Rothe} for the recent experimental observation of this change of regime}.
The needed quantitative estimates have been given in Corollary~\ref{cor:main}.
The present chapter is devoted to proving Theorem~\ref{thm:main_ac}, of which Corollary~\ref{cor:main} is a specialization.
As a middle step,  bounds on the derivatives of the $S$-matrix in the form
\begin{equation}
\|\1_K S^{(n)}\|_\infty  \leq  C_{n,K}  ,
\end{equation}
will be also proven, with $n=1,2,3$,  and the constants $C_{n,K}$ explicitly known.
Note that the following analysis is not restricted to the model of alpha decay considered in Chapter~\ref{ch:ET}, but is of general validity.


A quantum mechanical particle with wave function $\Psi$ scattering off a rotationally symmetric, compactly supported potential $V$ in three dimensions is described by the Schrödinger equation
\begin{align}\label{eq:Schroedinger3D}
	i\partial_t\Psi=H\Psi=(-\Delta+V)\Psi,
\end{align}
where $H$ is the Hamiltonian, with domain $\mathcal D(H)$.
A common way to study the scattering behavior of this equation is via dispersive estimates \cite{JensenKato1979,JSS,Rauch1978,Schlag}.
If $P_{ac}$ denotes the projector on the absolutely continuous spectral subspace of the Hamiltonian $H$, $R>0$ and $\1_R\coloneqq\1_{[0,R]}$, then it is well known that these dispersive estimates can be brought in the form%
\footnote{Note that this holds only if $H$ does not have a zero resonance (see Definition \ref{def:ZeroResVirtualStates}), while if it has it then $t^{-3}$ must be replaced by $t^{-1}$.}
\begin{align}\label{eq:main_heur}
  \|\1_R e^{-iHt}P_{ac}\Psi\|_2^2\leq C t^{-3}  ,
\end{align}
but little is known quantitatively about the constant $C$. 
The  main result of this chapter (Theorems~\ref{thm:main_ac} and~\ref{thm:main_e}) are quantitative bounds on the constant $C$, depending on the initial wave function $\Psi,$ the potential $V$ and spectral properties of $H$.

To achieve this we use the well known method of stationary phase applied to the expansion of $e^{-iHt}P_{ac}\Psi$ in generalized eigenfunctions, in combination with a detailed analysis of the $S$-matrix in the complex momentum plane.
Such an analysis is of interest in its own right, and our main result in this regard are Theorems~\ref{th:SBoundsK} and~\ref{th:GlobalSBounds}, which provide quantitive bounds on the $S$-matrix and its derivatives.
To obtain the needed detailed knowledge about the analytic properties of the $S$-matrix, we restrict to  rotationally symmetric, compactly supported potentials. This allows us to employ the scattering theory of Res Jost (see~\cite[Chapter~12]{Newton1966} for a textbook exposition), which in particular expresses the $S$-matrix in terms of one analytic function, the so-called Jost function. 
Expressing the Jost function in terms of its zeros via the Hadamard factorization, and using the fact that the zeros of the Jost function coincide with bound states, virtual states and resonances of $H$ (see Section~\ref{sec:PhysMeaningZeros} for a detailed discussion), we are then able to relate our scattering bounds explicitly to the spectral properties of $H.$ 

A discussion of analytic properties of the $S$-matrix and in particular of the Hadamard factorization of the Jost function is also found in~\cite{Regge,Newton1966}. Regge's paper~\cite{Regge} contains most ideas needed for arriving at the Hadamard factorization, but they are not worked out rigorously. Newton, on the other hand, gives more details in~\cite[Chapter~12]{Newton1966} yet he does not provide a full-fledged proof either. His discussion is our starting point. We work out all details needed for proving the Hadamard factorization of the Jost function in full rigor. 
In particular, we show that although the genus of the zeros of the Jost function is one, it is possible to write its Hadamard factorization as if the genus was zero.
The convergence of the genus-zero factorization of the Jost function is not granted by the general theory of entire functions, and the justification for using it is missing in Regge and Newton.
This was recognized by Boas~\cite{Boas}.
Moreover, we also explicitly show some well known properties of the Jost function of which we have not been able to find proofs, as for instance the fact that it is an analytic function of exponential type.

As a side result to our study of the Jost function, we also obtain an explicit quantitative bound on the number $n(r)$ of zeros of the Jost function within a ball of radius $r$ (Lemma \ref{lem:BoundOnNumberOfZeros}). 
Bounds in any dimension have been given by Zworski~\cite{Zworski1987,Zworski1989}, who proved that $n(r)\leq C_n(r+1)^n$, where $n$ denotes the dimension, but without explicit control over the constant $C_n.$


\section{Statement of main result}\label{sec:MainResult}

To state our main result (Theorems~\ref{thm:main_ac} and~\ref{thm:main_e}) rigorously, we first need to introduce the setting in which we work and the notation that we use.

\subsection{Assumptions on the potential and Definitions}\label{sec:assumptions}
For two functions $f,g : \R\to\C$  we will use the notation
\begin{equation}
f(x) \sim g(x)  \qquad \text{as }x\to x_0
\end{equation}
to mean that the following limit exists and 
\begin{equation}
\lim_{x\to x_0} \frac{f(x)}{g(x)} = 1 .
\end{equation}

Throughout the chapter we consider a non-zero, three-dimensional, rotationally symmetric potential $V=V(r)$, that is real, with support contained in~$[0,R_V]$, such that $\sup( \supp V) = R_V$, and $\|V\|_1<\infty$.
We also assume that the potential admits the asymptotic expansion
\begin{equation}\label{eq:VExpansionAtRs}
V(r)   \sim   \sum_{n=0}^M  d_n  (R_V - r)^{\delta_n},   
	\qquad \text{as }{r\to R_V^-},
\end{equation}
with $0\leq M<\infty$,  $-1<\delta_0<\dots<\delta_N$, $d_n\in\R$, and $d_n$ not all zero.

We will only be concerned with the case of zero angular momentum, to avoid the angular momentum barrier potential, which would not have compact support. In this case the three-dimensional Schrödinger Equation~\eqref{eq:Schroedinger3D} is equivalent to the one dimensional problem
\begin{align}\label{eq:Schroedinger1D}
  i\partial_t\psi=(-\partial_r^2+V)\psi
  \quad\text{with}\quad
  \Psi(r,\theta,\phi)=\frac{\psi(r)}{r}.
\end{align}
The self-adjointness of $(-\partial_r^2+V)$ is ensured by the next

\begin{lemma}\label{lem:Selfadjointness}
  Let the potential $V$ satisfy $\|V\|_1<\infty$ and let
  \begin{align}
	H_0=-\frac{d^2}{dr^2}
  \end{align}
  denote the self-adjoint free Schrödinger operator acting on $\{\phi\in L^2(\mathbb R^+)\,|\,\phi(0)=0\}$. Then, $V$ is infinitesimally form-bounded with respect to $H_0$, $H=H_0+V$ can be constructed by the standard quadratic form technique, and its form domain $Q(H)$ is equal to the form domain $Q(H_0)$ of the free operator.
\end{lemma}
\begin{proof}
  For ease of notation introduce $\dr\psi(r)\coloneqq \frac{d}{dr}\psi(r)$. Then, the form corresponding to $H_0$ and its form domain read
  \begin{align}
	h_0(\phi,\psi)
	&=\langle\dr\phi,\dr\psi\rangle,\\
	Q(H_0)
	&=\{\psi\in L^2(\mathbb R^+)|\dr\psi\in L^2(\mathbb R^+), \psi(0)=0\}.\label{eq:FormDomain}
  \end{align}
  It is well known that $h_0$ is closed on $Q(H_0)$ under the norm
  \begin{align}
  	\|\psi\|_{+1}=\sqrt{h_0(\psi,\psi)+\|\psi\|_2^2}=\sqrt{\|\psi'\|_2^2+\|\psi\|_2^2}.
  \end{align}
  So in order to see that $V$ is infinitesimally form-bounded with respect to $H_0$, we need to show that for all $\epsilon>0$ and $\psi\in Q(H_0)$ there is $c_\epsilon\in\mathbb R$ such that
  \begin{align}
	|\langle\psi,V\psi\rangle|\leq\epsilon h_0(\psi,\psi)+c_\epsilon\|\psi\|_2^2.
  \end{align}

  Now, assume that $\|\psi\|_\infty^2\leq2\|\psi\|_2\|\psi'\|_2$ for all $\psi\in Q(H_0)$, then using the fact that for arbitrary $a,b>0$ and all $\epsilon>0$ there is a $c_\epsilon>0$ such that $ab=a^2\sqrt{b^2/a^2}\leq\epsilon b^2+c_\epsilon a^2$, we get
  \begin{align}
	\|\psi\|_\infty^2\leq\epsilon\|\psi'\|_2^2+c_\epsilon\|\psi\|_2^2.
  \end{align}
  This implies for all $\psi\in Q(H_0)$ that
  \begin{align}
	|\langle\psi,V\psi\rangle|
	&\leq\|V\|_1\|\psi\|_\infty^2
	\\
	&\leq\epsilon\|\psi'\|_2^2+c_\epsilon\|\psi\|_2^2
	=\epsilon h_0(\psi,\psi)+c_\epsilon\|\psi\|_2^2,
  \end{align} 
  thereby proving the infinitesimally form-boundedness of $V$ with respect to $H_0$. The rest of the lemma then follows directly from the KLMN theorem~\cite[Theorem~X.17]{RS2}.

  It remains to prove that $\|\psi\|_\infty^2\leq2\|\psi\|_2\|\psi'\|_2$ for all $\psi\in Q(H_0)$, which will follow from a standard argument given in the proof of Theorem~8.5 in~\cite{LiebLoss}. Due to Theorem~7.6 in~\cite{LiebLoss}, $C^\infty_0\coloneqq\{f\in C^\infty(\mathbb R^+)|f(0)=0\}$ is dense in $Q(H_0)$ with respect to the norm $\|\cdot\|_{+1}$. Hence, there exists a sequence $\psi_m\in C^\infty_0\cap Q(H_0)$ that converges to $\psi\in Q(H_0)$. For this sequence we have
  \begin{align}\label{eq:DensityArgument}
	\psi_m^2(r)=2\int_0^r\psi_m(r')\dr\psi_m(r')\,dr'.
  \end{align}
  The convergence $\psi_m\to\psi$ in the norm $\|\cdot\|_{+1}$ implies that $\psi_m\to\psi$ and $\psi'_m\to\psi'$ in the norm $\|\cdot\|_2$, and thereby we have
  \begin{align}
	&\hspace{-2em}\left|\int_0^r(\psi(r')\dr\psi(r')-\psi_m(r')\dr\psi_m(r'))\,dr'\right|\\
	&\leq\int_0^r   \left|  \psi(\dr\psi-\dr\psi_m)+ \dr\psi_m(\psi-\psi_m)   \right|\,dr'\\
	&\leq  \langle|\psi|,|\dr\psi-\dr\psi_m|\rangle+\langle|\dr\psi_m|,|\psi-\psi_m|\rangle \\
	&\leq \|\psi\|_2\|\dr\psi-\dr\psi_m\|_2+\|\dr\psi\|_2\|\psi-\psi_m\|_2\\
	&\to 0\quad\text{as }m\to\infty,
  \end{align}
  so that the right hand side of Eq.~\eqref{eq:DensityArgument} converges pointwise to
  \begin{align}
	2\int_0^r\psi(r')\dr\psi(r')\,dr'
  \end{align}
  Moreover, via Theorem~2.7 in~\cite{LiebLoss} we can ensure that the left hand side of Eq.~\eqref{eq:DensityArgument} converges pointwise almost everywhere to $\psi^2(r)$ by passing to a subsequence. Thus, we have proven that for all $\psi\in Q(H_0)$
  \begin{align}
	\psi^2(r)=2\int_0^r\psi(r')\dr\psi(r')\,dr'
  \end{align}
  and thereby
  \begin{align}
	\|\psi\|_\infty^2
	&=2\sup_{r\in\mathbb R^+}\left|\int_0^r\psi(r')\dr\psi(r')\,dr'\right|\\
	&\leq2\sup_{r\in\mathbb R^+}\int_0^r |\psi(r')|  \,|\dr\psi(r')|\,dr'\\
	&\leq2\langle|\psi|,|\psi'|\rangle\\
	&\leq2\|\psi\|_2\|\psi'\|_2.\label{eq:DensityArgument2}   \mbox{\qedhere}
  \end{align}
  \QED
\end{proof}

To study the scattering behavior induced by the Schrödinger Eq.~\eqref{eq:Schroedinger1D}, we use generalized eigenfunctions, which solve the time-independent Schrödinger equation
\begin{align}\label{eq:Schroedinger}
  (-\partial_r^2+V(r))\phi(k,r)=k^2\phi(k,r).
\end{align}
To keep notation short, we write
\begin{align}
  \dr\phi(k,r)\coloneqq \partial_r\phi(k,r)
  \quad\text{and}\quad
  \dkp\phi(k,r)\coloneqq \partial_k\phi(k,r).
\end{align}

The following definitions and equations can all be found in Chapter~12 of Newton's book~\cite{Newton1966}. Following his exposition, we define the regular eigenfunctions $\varphi(k,r)$ as the solutions of Schrödinger's Eq.~\eqref{eq:Schroedinger} that satisfy the boundary conditions
\begin{align}
	\varphi(k,0)=0
	\quad\mathrm{as\;well\;as}\quad
	\dr \varphi(k,0)=1
\end{align}
and we define the irregular eigenfunctions $f(k,r)$ as the solutions of Schrödinger's Eq.~\eqref{eq:Schroedinger} that satisfy the boundary condition
\begin{align}\label{eq:fboundary}
	f(k,r)=e^{ikr}
	\quad\mathrm{for}\quad
	r\geq R_V.
\end{align}
Note that this boundary condition as it is formulated hinges on the assumption that the potential $V$ has compact support. Later we will use this property of $f(k,r)$ in an essential way. The Jost function $F$ is defined as the Wronskian of $f$ and $\varphi$, i.e.\
\begin{align}\label{eq:Jost}
  F(k)=W(f(k,r),\varphi(k,r))\coloneqq f(k,r)\dr\varphi(k,r)-\dr f(k,r)\varphi(k,r).
\end{align}
We define 
\begin{align}
  \psi^+(k,r)\coloneqq \frac{k\varphi(k,r)}{F(k)},
  \label{eq:PsiPlus}
\end{align}
and $P_{ac}$ and $P_e$ to be the projections on the subspace of absolute continuity of $H$ and on the span of all eigenvectors of $H$, respectively.
Then the generalized eigenfunction expansion of a funciton $\psi\in L^2(\R^+)$ reads
\begin{align}
  P_{ac}\psi(r)&=\int_0^\infty\hat\psi(k)\psi^+(k,r)\,dk,
  \intertext{with}
  \hat\psi(k)
  	&\coloneqq \mathcal F \psi(k)
    	\coloneqq\int_0^\infty\psi(r)\bar\psi^+(k,r)\,dr.
\label{eq:Fourier}
\end{align}
We also need the relation
\begin{align}\label{eq:phi}
  \varphi(k,r)=\frac{1}{2ik}(F(-k)f(k,r)-F(k)f(-k,r)),
\end{align}
which is an immediate consequence of the Jost function's definition in Eq.~\eqref{eq:Jost} and of the fact that $f(k,r)$ and $f(-k,r)$ span the solution space of Schrödinger's Eq.~\eqref{eq:Schroedinger}. In particular, Eq.\ \eqref{eq:phi} evaluated for $r\geq R_V$ reads
\begin{equation}\label{eq:phiRBiggerRs}
  \varphi(k,r)=\frac{1}{2ik}(F(-k)e^{ikr}-F(k)e^{-ikr}), 
  \qquad r \geq R_V.
\end{equation}
The S-matrix element for zero angular momentum can now be expressed (see~\cite[Chapter~12]{Newton1966} for details) as
\begin{align}\label{eq:SRatioF}
  S(k)=\frac{F(-k)}{F(k)}.
\end{align}

In~\cite[Chapter~12]{Newton1966} it is also shown that the functions $f(k,r),\varphi(k,r)$ and $F(k)$ admit analytic extensions to the whole complex $k$-plane. Therefore, we can make
\begin{definition}\label{def:ZeroResVirtualStates}
	A resonance is a zero of the Jost function $F(k)$ in $\{k\in\mathbb C| \Im k<0,\Re k\neq 0\}.$ We say that the potential has a zero resonance, if and only if $F(0)=0.$ A virtual state is defined to be a zero of $F(k)$ in $\{k\in\mathbb C| \Im k<0,\Re k= 0\}.$
\end{definition}
Moreover, bound states of the potential correspond to the zeros of $F(k)$ in $\{k\in\mathbb C| \Im k>0,\Re k= 0\}.$ 
The resonances  appear in couples symmetric about the imaginary axis and are infinitely many, while there are just finitely many virtual and bound states~\cite{Newton1966,Regge,Rollnik}. 
For further discussion about the zeros of the Jost function and their physical meaning see Section~\ref{sec:PhysMeaningZeros}.

We will also use the symmetry relations \cite[pages 339, 340]{Newton1966}
\begin{align}
F(k) &= \bar F(-\bar k),
\\
f(k,r) &= \bar f (-\bar k,r).
\end{align}

Finally, we introduce the weight function
\begin{align}
  w(x)\coloneqq \frac{1}{1+x^2}
\end{align}
and say $\psi\in L^1_w$ if and only if $\|\psi w\|_1<\infty.$
Moreover, we call $L^\infty_{loc}$ the space of functions $\psi$ such that $\|\1_R \psi\|_\infty$ is finite for every $R>0$.

\subsection{Main result}

The main result (Theorems~\ref{thm:main_ac} and~\ref{thm:main_e}) rests upon bounds on the derivatives of the $S$-matrix given in Theorems~\ref{th:SBoundsK} and~\ref{th:GlobalSBounds}. To state these bounds, we need

\begin{definition}\label{def:SKAndKTildeAndS}
Let $\alpha_n$, $\beta_n$, $\eta_m$, $\kappa_l>0$.
We number the zeros of the Jost function other than $k=0$ with increasing modulus; among them, we denote the bound states by $i \eta_m$, the virtual states by $-i \kappa_l$, and the resonances by $k_n=\alpha_n-i \beta_n$ and $-\bar k_n$. 
Let $N<\infty$ be the number of the bound states and $N'<\infty$ that of the virtual states, then we define
\begin{equation}
\frac {1} {\eta}  \coloneqq \sum_{m=0}^{N-1}    \frac {1}  {\eta_m}     ,
\qquad
\frac {1} {\kappa}  \coloneqq \sum_{l=0}^{N'-1}   \frac {1}  {\kappa_l}
\end{equation}
and for given $K>0$, let  $\nu_K$ be the smallest non-negative integer such that $\alpha_n \geq 2 K$ for all $n\geq\nu_K$. Then
\begin{equation}\label{eq:DefSK}
\frac {1} {s_K}  \coloneqq  
		\frac {1} {\eta}  +  \frac {1} {\kappa}  +  \sum_{n=0}^{\nu_K-1}   \frac {1}  {\beta_n}
\end{equation}
in case the right hand side is not zero, and $s_K\coloneqq 1$ otherwise.
Let $\tilde K\coloneqq 6\|V\|_1$ and
\begin{equation}\label{eq:DefS}
\frac {1} {s}  \coloneqq  \frac {1} {s_{\tilde K}}   .
\end{equation}
\end{definition}

\begin{theorem}\label{th:SBoundsK}
Let $\alpha \coloneqq   \min_{n\in\N^0} \alpha_n,$ $K>0$ and
\begin{align}
  r_0\coloneqq \sum_{n=0}^\infty   \frac {5\beta_n}  { \alpha_n^2 + \beta_n^2}.
\end{align}
Then
\begin{align}
\|\1_K \dkp S\|_\infty  
&\leq 
	\frac {2}  {s_K}     \left[  1+ s_K (R_V+r_0)  \right]
	\eqqcolon \frac{C_{1,K}}{s_K} ,
\label{eq:BoundSDotK} 
\displaybreak[0]\\
\|\1_K \ddkp S\|_\infty  
&\leq 
	\frac {4}  {s_K^2}  \left\{  3 + 2 s_K^2  \left[ \frac{r_0}{\alpha}  +  (R_V+r_0)^2  \right]  \right\}
	\eqqcolon \frac{C_{2,K}}{s_K^2}  ,
\label{eq:BoundSDotDotK} 
\displaybreak[0]\\
\|\1_K \dddkp S\|_\infty  
&\leq 
	\frac {4}  {s_K^3}  \Biggl\{  15  + 6 s_K (R_V+r_0) + 12 s_K^2 \frac{r_0}{\alpha} 
\nonumber\\		
&
		+ s_K^3  \left[ \frac{7 r_0}{\alpha} + \frac{12 r_0}{\alpha} (R_V+r_0) + 8 (R_V+r_0)^3 \right]  \Biggr\}   
		\eqqcolon \frac{C_{3,K}}{s_K^3}.
\label{eq:BoundSDotDotDotK} 
\end{align}
\end{theorem}

Note that $r_0<\infty$ as shown in Lemma~\ref{lem:R0}.
The bounds in Theorem \ref{th:SBoundsK} are valid for any $K$, but for big values of $K$, the bounds given in Theorem~\ref{th:GlobalSBounds} are more convenient.


\begin{theorem}\label{th:GlobalSBounds}
Let
\begin{align}
  q\coloneqq \frac{1}{2\|V\|_1} + 6R_V,
  \quad r_0\coloneqq \sum_{n=0}^\infty   \frac {5\beta_n}  { \alpha_n^2 + \beta_n^2}
  \quad\text{and}\quad  \alpha \coloneqq  \min_{n\in\N^0} \alpha_n.
\end{align}
Then
\begin{align}
\|\dkp S\|_\infty  
&\leq 
	\frac {2}  {s}     \left[  1+ s (3R_V+r_0)  \right]
	\eqqcolon \frac{C_{1}}{s}  ,
\displaybreak[0]\\
\|\ddkp S\|_\infty  
&\leq 
	\frac {4}  {s^2}  \left\{  3 + 2 s^2  \left[ \frac{r_0}{\alpha}  +  (3 R_V+r_0)^2 +R_V q \right]  \right\}
	\eqqcolon \frac{C_{2}}{s^2}  ,
\displaybreak[0]\\
\|\dddkp S\|_\infty  
&\leq 
	\frac {4}  {s^3}  \Biggl\{  15  + 6 s (R_V+r_0) + 12 s^2 \frac{r_0}{\alpha} 
\nonumber\\		
&
		+ s^3  \left[ \frac{7 r_0}{\alpha} + \frac{12 r_0}{\alpha} (R_V+r_0) 
			+ 8 (3R_V+r_0)^3  +18 R_V q^2 \right]  \Biggr\}   
		\eqqcolon \frac{C_{3}}{s^3}.
\end{align}
\end{theorem}


\pagebreak

\begin{figure}
  \centering
  \begin{overpic}[width=.65\textwidth
  ]{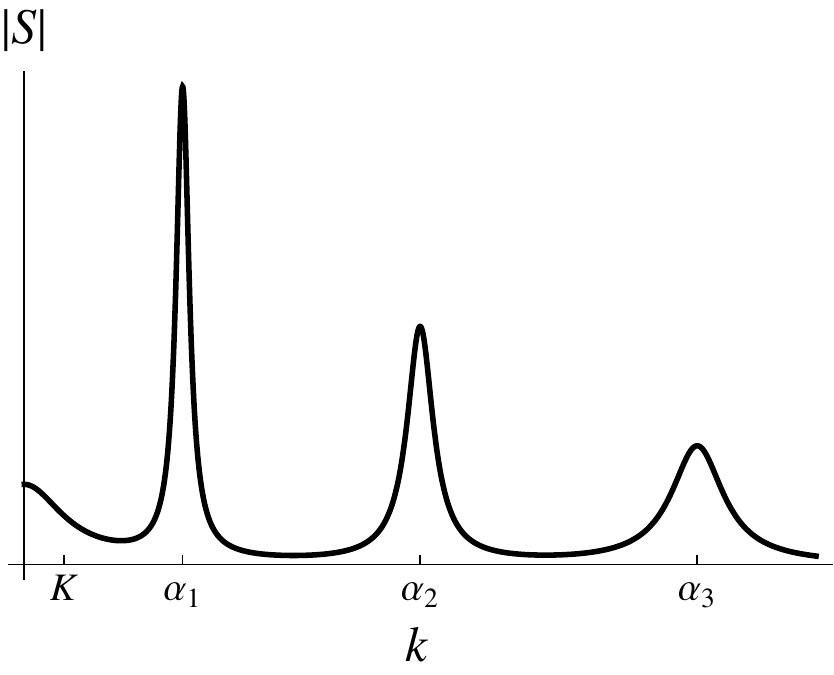}
  \put(3,80.3){.}
  \end{overpic}
  \caption{Schematic plot of $|\dkp S(k)|,$ where resonances are dominant.}
  \label{fig:KRemark}
\end{figure}
\begin{remark}
Let us explain, why we look at bounds for $k\in[0,K)$ and $k\in[0,\infty)$ rather than only for $k\in[0,\infty)$ or for $k\in[0,K)$ and $k\in[K,\infty).$ 
In the proof of Theorem \ref{thm:main_ac} we will find that one gets much tighter bounds treating the region around $k=0$ more carefully.
For this reason, bounds for $k\in[0,K)$ and $k\in[K,\infty)$ would be most useful. 
However, physically interesting situations are mainly those in which resonances dominate the scattering behavior and in such a situation the bounds for $k\in[0,\infty)$ are as good as bounds for $k\in[K,\infty),$ just easier to prove.
Let us briefly explain why they are equally good in case resonances are dominant.
The absolute value of the derivatives of the $S$-matrix have peaks centered around the real part of every resonance and around $k=0$ due to the bound and virtual states (see Lemma~\ref{lem:Hadamard} and Fig.~\ref{fig:KRemark}). 
If the resonances are dominant, then the peak at $k=0$ is smaller than  some of the resonance peaks.
This circumstance is discussed in Section~\ref{sec:Discussion}, where we also explain that in this case a good choice for $K$ is $\alpha_0/4$, hence $\|\1_{[K,\infty)}S^{(n)}\|_\infty=\|S^{(n)}\|_\infty.$ 
In contrast, if bound and virtual states are dominant, then bounds for $k\in[K,\infty)$ rather than for $k\in[0,\infty)$ may be advantageous.
\end{remark}

To keep the statement of our main result as concise as possible we define auxiliary constants.
These constants contain the radius $R$ that appear in our estimate \ref{eq:main_heur}.
The first number in their argument is the order of the derivative of the $S$-matrix in which they are used, the second one being just an index.


\begin{definition}\label{def:smallz}
  Let $R>0.$ 
Using the constants introduced in Theorems \ref{th:SBoundsK} and \ref{th:GlobalSBounds}, define
  \begin{align}
		z_{ac,K}(0,0)
		&\coloneqq \frac{1}{2}\left(2Rs_K+C_{1,K}\right),\displaybreak[0]\\
		z_{ac,K}(0,1)
		&\coloneqq 1,\displaybreak[0]\\
		z_{ac,K}(1,0)
		&\coloneqq \frac{1}{4}\left(2R^2s_K^2+2RC_{1,K}s_K+C_{2,K}\right),\displaybreak[0]\\
		z_{ac,K}(1,1)
		&\coloneqq \frac{1}{2}\left(2Rs_K+C_{1,K}\right),\displaybreak[0]\\
		z_{ac,K}(1,2)
		&\coloneqq 1,\displaybreak[0]\\
		z_{ac,K}(2,0)
		&\coloneqq \frac{1}{6}(2R^3s_K^3+3R^2s_K^2C_{1,K}+3Rs_KC_{2,K}+C_{3,K}),\displaybreak[0]\\
		z_{ac,K}(2,1)
		&\coloneqq \frac{1}{2}(2R^2s_K^2+2Rs_KC_{1,K}+C_{2,K}),\displaybreak[0]\\
		z_{ac,K}(2,2)
		&\coloneqq 2Rs_K+C_{1,K},\displaybreak[0]\\
		z_{ac,K}(2,3)
		&\coloneqq 2  ,
\displaybreak[0]\\
  	z_{ac,K}(n)& \coloneqq \sum_{m=0}^{n+1}z_{ac,K}(n,m) ,
  \end{align}

  and
  \begin{align}
  	z_{e,K}(0)
	&\coloneqq \sqrt 2,\displaybreak[0]\\
	z_{e,K}(1)
	&\coloneqq \frac{1}{\sqrt 2}
	\bigg[2s_K+(2Rs_K+C_{1,K})\eta_0\bigg],\displaybreak[0]\\
	z_{e,K}(2)
	&\coloneqq \frac{1}{\sqrt 2}
	\left(C_{2,K}\eta_0^2+2\eta_0s_K\left(C_{1,K}+Rs_K\right)\left(R\eta_0+1\right)+4s_K^2\right).
  \end{align}
  Define $z_{ac}(n,m)$, $z_{ac}(n)$, and $z_e(n)$ in exactly the same way, but with index $K$ omitted everywhere. 
\end{definition}

Recalling Definition \ref{def:SKAndKTildeAndS}, we now state the main result in Theorems~\ref{thm:main_ac} and~\ref{thm:main_e}.

\newpage

\begin{theorem}\label{thm:main_ac}
  Let $t>0,$ $R\geq R_V$, $K>s$, and 
\begin{equation}
\lambda  \coloneqq
	\begin{cases}
		0,&\mathrm{if}\;F(0)\neq0\\
		1,&\mathrm{if}\;F(0)=0.
	\end{cases}
\end{equation}  
 Then there are constants $c_n,$ such that
  \begin{equation}\label{eq:main_ac}
	\|P_{ac}\1_Re^{-iHt}P_{ac}\psi\|_2^2
	\leq \lambda(c_1t^{-1}+c_2t^{-2})+c_3t^{-3}+c_4t^{-4}
  \end{equation}
  for all $\psi$ with $\hat \psi^{(m)}\in L^\infty_{loc}\cap L^1_w$ where $m=0,1,2.$ If $s,s_K,K\leq 1,$ the constants satisfy
  \begin{align}
	c_1 &\leq \frac{81\pi^2}{K}\frac{|\hat\psi(0)|^2}{s_K^2}z_{ac,K}^2(0),\displaybreak[0]\\
	c_2 &\leq 
	\frac{53\pi^2}{K^3}\frac{|\hat\psi(0)|^2}{s_K^4}z_{ac,K}^2(1)
	+\frac{53\pi^2}{K}\frac{\|\1_K\dkp{\hat\psi}\|_\infty^2}{s_K^2}z_{ac,K}^2(0),\displaybreak[0]\\
	c_3 &\leq \frac{27}{K^5} \frac{\|\1_K{\hat\psi}\|_\infty^2}{s_K^6}z_{ac,K}^2(2)
	+\frac{23\pi^2}{K^3}\frac{\|\1_K\dkp{\hat\psi}\|_\infty^2}{s_K^4}z_{ac,K}^2(1)
	\nonumber\\
	&\quad+\frac{27}{K}\frac{\|\1_K\ddkp{\hat\psi}\|_\infty^2}{s_K^2}z_{ac,K}^2(0),\displaybreak[0]\\
	c_4 &\leq 276 \frac{\|{\hat\psi}w\|_1^2}{s^5}\left(1+\frac{1}{K^2}\right)^4\left(z_{ac}^2(2)+s^2z_{ac}^2(1)+s^4z_{ac}^2(0)\right)\nonumber\\
	&\quad+304\frac{\|\dkp{\hat\psi}w\|_1^2}{s^3}\left(1+\frac{1}{K^2}\right)^3\left(z_{ac}^2(1)+s^2z_{ac}^2(0)\right)\nonumber\\
	&\quad+14 \frac{\|\ddkp{\hat\psi}w\|_1^2}{s}\left(1+\frac{1}{K^2}\right)^2z_{ac}^2(0).
  \end{align}
Bounds on the constants $c_n$ without the assumption $s,s_K,K\leq1$ are given in Eqs.~\eqref{eq:C1Pure}-\eqref{eq:C4Pure}.
\end{theorem}

\begin{remark}
For $t>T$, Eq.~\eqref{eq:main_ac} implies 
\begin{equation}
	\|P_{ac}\1_Re^{-iHt}P_{ac}\psi\|_2^2
	\leq \lambda \left(c_1+\frac{c_2}{T}\right) t^{-1} + (1-\lambda) \left(c_3+\frac{c_4}{T}\right) t^{-3},
\end{equation}
which is of the form of Eq.~\eqref{eq:main_heur}.
However, the bound \eqref{eq:main_ac} is preferable because it allows a higher degree of accuracy on intermediate time scales, for example if $c_3\ll c_4$.
\end{remark}

\begin{remark}
The restriction $K>s$ in Theorem~\ref{thm:main_ac} is set to avoid unessential complications in the proof, where we divide the integration region of several integrals according to Fig.~\ref{fig:regions}. This division is easier if $K>\delta$ and since we fix $\delta=s$ in the course of the proof, we end up with the condition $K>s.$ 
Besides this restriction the value of $K$ can be chosen freely, and it influences the size of the constants in Theorem \ref{thm:main_ac}. 
A choice for $K$ meaningful for many potentials that respects the condition $K>s$ is presented in Section~\ref{sec:Discussion}, however for some potentials a value of $K<s$ might lead to better results.
In this case, the restriction $K>s$ can be removed with slight but cumbersome changes in the proof. 
\end{remark}

\begin{remark}
	It is worth observing that $\hat\psi$ depends not only on the initial state, but also on the potential through the generalized eigenfunctions. So in general $\|\1_K{\hat\psi}^{(n)}\|_\infty$ and $\|{\hat\psi}^{(n)}w\|_1$ will depend on $s_K$ and $s,$ too (see Lemma  \ref{lem:psi} for an example).
\end{remark}

\begin{theorem}\label{thm:main_e}
  Let $t,K>0$, $R\geq R_V$, and
\begin{equation}
\lambda  \coloneqq
	\begin{cases}
		0,&\mathrm{if}\;F(0)\neq0\\
		1,&\mathrm{if}\;F(0)=0.
	\end{cases}
\end{equation}   
Then there are constants $c_n>0,$ such that
  \begin{align}
	\|P_{e}\1_Re^{-iHt}P_{ac}\psi\|_2^2
	&\leq \lambda(c_1t^{-1}+c_2t^{-2})+c_3t^{-3}+c_4t^{-4}
  \end{align}
  for all $\psi$ with $\hat\psi^{(m)}\in L^\infty_{loc}\cap L^1_w$ and $m=0,1,2.$ The constants satisfy
  \begin{align}
	c_1&\leq
	\frac{81\pi^2}{2}\frac{|\hat\psi(0)|^2}{\eta_0}z_{e,K}^2(0)N,\displaybreak[0]\\
	c_2&\leq
	\frac{105\pi^2}{4}\left[\frac{|\hat\psi(0)|^2}{\eta_0^3s_K^2}z_{e,K}^2(1)+\frac{\|\1_K\dkp{\hat\psi}\|_\infty^2}{\eta_0}z_{e,K}^2(0)\right]N ,\displaybreak[0]\\
	c_3&\leq
	\Biggl[9\frac{\|\1_K{\hat\psi}\|_\infty^2}{\eta_0^5s_K^4}z_{e,K}^2(2)
	+166\frac{\|\1_K\dkp{\hat\psi}\|_\infty^2}{\eta_0^3s_K^2}z_{e,K}^2(1)
	\nonumber\\
	&\quad+9\|\1_K\ddkp{\hat\psi}\|_\infty^2\frac{z_{e,K}^2(0)}{\eta_0}\Biggr]N ,\displaybreak[0]\\
	c_4&\leq
	\Bigg[\frac{27}{2}\frac{\|\hat\psi w\|_1^2}{\eta_0^5s^4}\left(1+\frac{1}{K^2}\right)^4\left(z_{e}^2(2)+\eta_0^2s^2z_{e}^2(1)+\eta_0^4s^4z_{e}^2(0)\right)\nonumber\\
	&\quad+12\frac{\|\dkp{\hat\psi} w\|_1^2}{\eta_0^3s^2}\left(1+\frac{1}{K^2}\right)^3\left(z_{e}^2(1)+\eta_0^2s^2z_{e}^2(0)\right)\nonumber\\
	&\quad+\frac{9}{8}\frac{\|\ddkp{\hat\psi} w\|_1^2}{\eta_0}\left(1+\frac{1}{K^2}\right)^2z_{e}^2(0)\Bigg]N.
  \end{align}
\end{theorem}

Together, Theorems~\ref{thm:main_ac} and~\ref{thm:main_e} yield the desired bound on the probability $\|\1_Re^{-iHt}P_{ac}\psi\|_2^2$ to find the particle inside a ball of radius $R,$ because
\begin{align}
   \|\1_Re^{-iHt}P_{ac}\psi\|_2^2=\|P_{ac}\1_Re^{-iHt}P_{ac}\psi\|_2^2+\|P_e\1_Re^{-iHt}P_{ac}\psi\|_2^2.
\end{align}

\begin{remark}
Note that the bounds in our theorems depend on $r_0$ that contains the location of all resonances, and seems therefore difficult to access.
Nevertheless, we will now see that is connected to the scattering length  \citep[see][page 136]{RS3}
\begin{equation}
a=\frac{\dkp S(0)}{2iS(0)},
\end{equation} 
that is experimentally measurable.
From Eq.~\eqref{eq:L} of Lemma \ref{lem:Hadamard}, choosing $k=0$, one immediately gets  \citep[see also][]{Koro}
\begin{equation}\label{eq:ScatteringLengthA}
a 
= - R_V 
	- \sum_{m=0}^{N-1}    \frac {1}  {\eta_m}     
	+ \sum_{l=0}^{N'-1}   \frac {1}  {\kappa_l}
	+ \sum_{n=0}^\infty  \frac {2\beta_n}  { \alpha_n^2 + \beta_n^2}   ,
\end{equation}
which implies
\begin{equation}\label{eq:ScatteringLength}
r_0 
\leq  \frac{5}{2} |a|   
	+ \frac{5}{2} \left | R_V
	+\sum_{m=0}^{N-1}    \frac {1}  {\eta_m}     
	- \sum_{l=0}^{N'-1}   \frac {1}  {\kappa_l}  \right|   .
\end{equation}

Note also that, although the scattering length is physically measured from the scattering cross section at zero energy, it actually depends on all resonances, not only on the first few.\end{remark}

As mentioned in the introduction, we also give a quantitative bound on the number of zeros of the Jost function inside a ball of radius $|k|$.
It is a direct consequence of Lemma \ref{lem:BoundsOnFAndPhi}.
A related result can be found in \cite{Zworski1987,Zworski1989} where the inequality  $n(r)\leq C_n(r+1)^n$ was proven, with $n$ denoting the dimension, but without explicit control over the constant $C_n.$

\begin{lemma}\label{lem:BoundOnNumberOfZeros}
Let  $n(|k|)$ be the number of zeros of the Jost function with modulus  not greater than $|k|$. Then,
\begin{equation}\label{eq:BoundOnNumberOfZeros}
n(|k|)  \leq  
	\frac {1}  {\log 2}
	\left [  4 R_V  |k|   + \log \left(  4\|rV(r)\|_1  e^{4\|rV(r)\|_1}  +1  \right)  \right]   .
\end{equation}
\end{lemma}


\section[{Application of the main result\\\qquad to meta-stable states}]{Application of the main result\\to meta-stable states}\label{sec:Discussion}
We consider as example the alpha-decay of long-lived elements treated by Skibsted in~\cite{Skibsted86}.
There the meta-stable state is modeled via the truncated Gamow function $f_R\coloneqq \1_Rf(k_0,\cdot)$ associated to the first resonance $k_0=\alpha_0-i\beta_0,$ with $\alpha_0,\beta_0>0$ and $f$ defined by Eq.\ \eqref{eq:fboundary}. 
Skibsted showed in \cite{Skibsted86} that the velocity with which the alpha-particle escapes the nucleus is $2\alpha_0$, while the lifetime of the meta-stable state is $(4\alpha_0\beta_0)^{-1}$. 
Comparison with empirical data shows that  $\alpha_0\approx 1$, while the lifetime is very large and therefore $\beta_0\ll 1$.

Let us determine the norms $\|\hat f_R^{(n)}w\|_1$ and $\|\1_K\hat f_R^{(n)}\|_\infty$ that appear in Theorems \ref{thm:main_ac} and \ref{thm:main_e}.
\begin{lemma}\label{lem:psi}
  Let $R\geq R_V,$ $K\in[0,\tfrac{\alpha_0}{2}),$ then the truncated Gamow function $f_R\coloneqq \1_R f(k_0,\cdot)$ satisfies
  \begin{align}
	\|\1_K\hat f_R\|_\infty
	&\leq e^{\beta_0R}\frac{2}{\alpha_0} ,
\label{eq:psiK.0}
	\displaybreak[0]\\
	\|\1_K\dkp{\hat f}_R\|_\infty
	&\leq e^{\beta_0R}\left[\frac{2^2}{\alpha_0^2}+\frac{1}{\alpha_0}\left(2R+\frac{C_{1,K}}{s_K}\right)\right] ,
\label{eq:psiK.1}
	\displaybreak[0]\\
	\|\1_K\ddkp{\hat f}_R\|_\infty
	&\leq e^{\beta_0R}\Biggl[\frac{2^4}{\alpha_0^3}+\left(2R+\frac{C_{1,K}}{s_K}\right)\frac{2^2}{\alpha_0^2}
	\nonumber\\	
	&\qquad+\left(R^2+R\frac{C_{1,K}}{s_K}+\frac{C_{2,K}}{2s_K^2}\right)\frac{2}{\alpha_0}\Biggr] ,
\label{eq:psiK.2}
  \end{align}
  and
  \begin{align}
	\|\hat f_R w\|_1
	&\leq e^{\beta_0R}\left[2\log\left(\frac{2}{\beta_0}\right)+\frac{\pi}{2}\right] ,
\label{eq:psi1.0}
\displaybreak[0]\\
	\|\dkp{\hat f}_R w\|_1
	&\leq e^{\beta_0R}\left[\left(2\log\left(\frac{2}{\beta_0}\right)+\frac{\pi}{2}\right)\left(R+\frac{C_1}{2s}\right)+\frac{\pi}{\beta_0}\right] ,
\label{eq:psi1.1}
	\displaybreak[0]\\
	\|\ddkp{\hat f}_R w\|_1
	&\leq e^{\beta_0R}\Biggl[\left(2\log\left(\frac{2}{\beta_0}\right)+\frac{\pi}{2}\right)\left(R^2+\frac{C_1}{s}R+\frac{C_2}{2s^2}\right)
	\nonumber\\
	&\qquad+\frac{\pi}{\beta_0}\left(2R+\frac{C_1}{s}\right)+\frac{4}{\beta_0^2}\Biggr].
\label{eq:psi1.2}
  \end{align}
  Moreover, if there is a zero resonance, we have
  \begin{align}\label{eq:psi00}
	|\hat f_R(0)|=\frac{e^{\beta_0R}}{\sqrt{\alpha_0^2+\beta_0^2}}.
  \end{align}
\end{lemma}

\begin{proof}
  Lemma 3.2 in~\cite{Skibsted86} shows that
  \begin{align}\label{eq:lem_skib}
	\hat f_R(k)=-\frac{1}{2}\left[\frac{e^{i(k_0-k)R}}{k-k_0}\bar S(k)+\frac{e^{i(k_0+k)R}}{k+k_0}\right]
  \end{align}
  with $k_0=\alpha_0-i\beta_0.$ From this we can already conclude that for $k\in[0,K]$
  \begin{align}
	|\hat f_R|
	\leq\frac{e^{\beta_0R}}{2}\left[\frac{1}{|k-k_0|}+\frac{1}{|k+k_0|}\right]
	\leq e^{\beta_0R}\frac{1}{|k-k_0|}
	\leq e^{\beta_0R}\frac{2}{\alpha_0},
  \end{align}
  which proves Eq.~\eqref{eq:psiK.0}. Moreover, we immediately obtain
  \begin{align}
	\dkp{\hat f}_R(k)
	&=\frac{1}{2}\frac{e^{i(k_0-k)R}}{(k-k_0)^2}\left[(1+iR(k-k_0))\bar S(k)
	-(k-k_0)\dkp{\bar S}(k)\right]
	\nonumber\\
	&\quad+\frac{1}{2}\frac{e^{i(k_0+k)R}}{(k+k_0)^2}(1-iR(k+k_0)),
	\\
	\ddkp{\hat f}_R(k)
	&=-\frac{1}{2}\frac{e^{i(k_0-k)R}}{(k-k_0)^3}\bigg[\left(1+iR(k-k_0)-\frac{R^2}{2}(k-k_0)^2\right)\bar S(k)\nonumber\\
	&\quad-(1+iR(k-k_0))\dkp{\bar S}(k)(k-k_0)
	+\frac{1}{2}\ddkp{\bar S}(k)(k-k_0)^2\bigg]\nonumber\\
	&\quad-\frac{1}{2}\frac{e^{i(k_0+k)R}}{(k+k_0)^3}\left(1-iR(k+k_0)-\frac{R^2}{2}(k+k_0)^2\right),
  \end{align}
  and this implies, along the same lines as before, Eqs.~\eqref{eq:psiK.1} and \eqref{eq:psiK.2}.
  
  
  Now, let us consider
  \begin{align}
	\|\hat f_R w\|_1\leq e^{\beta_0R}\int_0^\infty\frac{1}{|k-k_0|}w(k)\,dk.
  \end{align}
  The weight function $w$ is needed for the integral to converge, while it is unessential in the region around $k=\alpha_0,$ where $|k-k_0|^{-1}$ is biggest. Hence, we split the integral in a region where $|k-k_0|^{-1}\geq 1$, i.e. the interval $[\alpha_0-(1-\beta_0^2)^{1/2},\alpha_0+(1-\beta_0^2)^{1/2}]$, and the rest. If we call the rest $B,$ we have
  \begin{align}
	\|\hat f_R w\|_1
	&\leq e^{\beta_0R}\int_{\alpha_0-\sqrt{1-\beta_0^2}}^{\alpha_0+\sqrt{1-\beta_0^2}}\frac{1}{|k-k_0|}w(k)\,dk
	\nonumber\\
	&\qquad+e^{\beta_0R}\int_B\frac{1}{|k-k_0|}w(k)\,dk\\
	&\leq e^{\beta_0R}\int_{\alpha_0-\sqrt{1-\beta_0^2}}^{\alpha_0+\sqrt{1-\beta_0^2}}\frac{1}{|k-k_0|}\,dk
	+e^{\beta_0R}\int_0^\infty w(k)\,dk\\
	&=e^{\beta_0R}\left[2\log\left(\frac{1}{\beta_0}+\frac{1}{\beta_0}\sqrt{1-\beta_0^2}\right)+\frac{\pi}{2}\right]\\
	&\leq e^{\beta_0R}\left[2\log\left(\frac{2}{\beta_0}\right)+\frac{\pi}{2}\right]
  \end{align}
  confirming Eq.~\eqref{eq:psi1.0}. Similarly, we get
  \begin{align}
	\|\dkp{\hat f}_R w\|_1
	\leq\frac{e^{\beta_0R}}{2}\left[\int_0^\infty\frac{2}{|k-k_0|^2}w(k)\,dk
	+\int_0^\infty\frac{2R+|\dkp S(k)|}{|k-k_0|}w(k)\,dk\right].
  \end{align}
  The second integral can be estimated in the same way as $\|\hat f_R w\|_1$ by using Theorem~\ref{th:GlobalSBounds} on the S-Matrix. The first integral satisfies
  \begin{align}
	\int_0^\infty\frac{2}{|k-k_0|^2}w(k)\,dk
	\leq\int_{-\infty}^\infty\frac{2}{|k-k_0|^2}\,dk
	=\frac{2\pi}{\beta_0}.
  \end{align}
  Hence, Eq.~\eqref{eq:psi1.1} and analogously Eq.~\eqref{eq:psi1.2}.
  
  To derive Eq.~\eqref{eq:psi00}, we only need to evaluate Eq.~\eqref{eq:lem_skib} at $k=0$ and use the fact that $S(0)=-1$ in the presence of a zero resonance, which has been shown in~\cite[page 356]{Newton1966}.
\QED
\end{proof}

\pagebreak

\subsection{Uranium 238}\label{sec:Uranium}

\begin{figure}
  \centering
  \includegraphics[width=.6\textwidth]{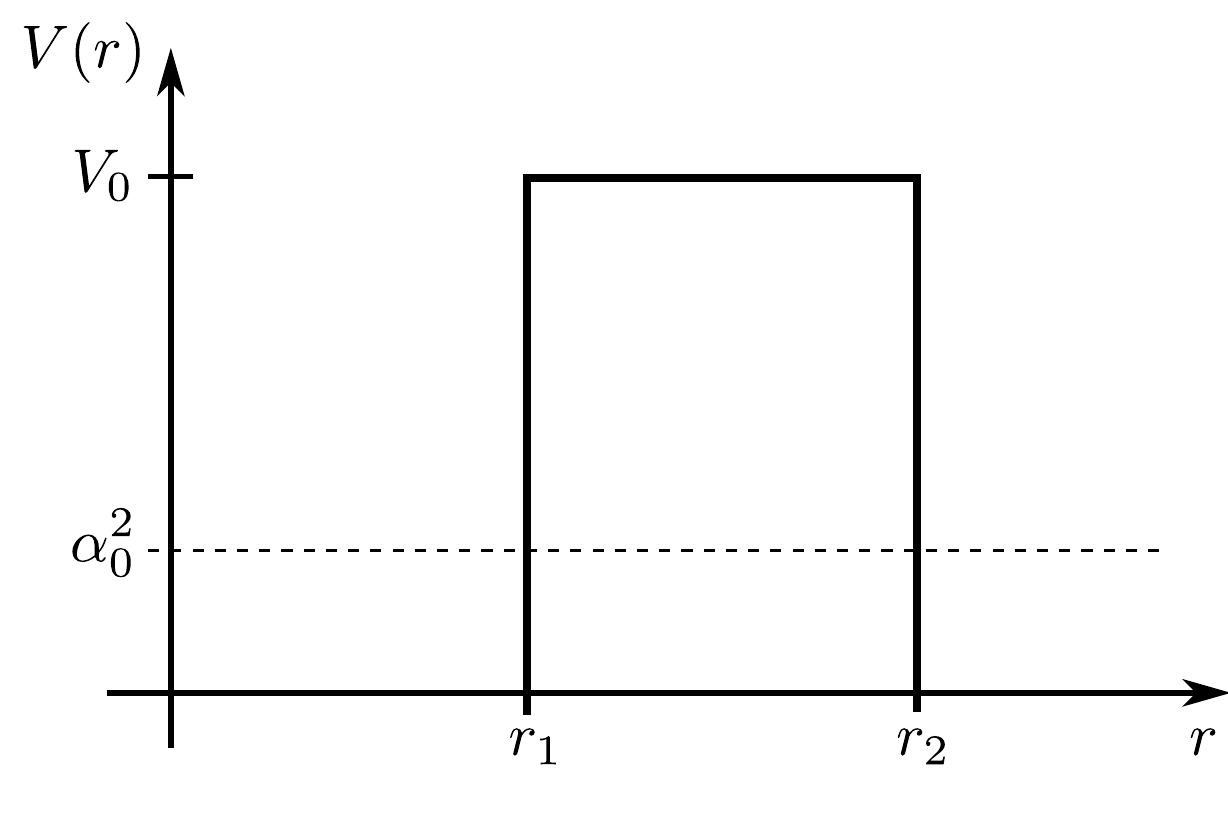}
  \caption{Potential barrier.}
  \label{fig:potential}
\end{figure}

For concreteness we consider now the alpha decay of Uranium~238, which we model using the potential shown in Fig.~\ref{fig:potential}. 
As parameter values we choose $r_1=1,$ $r_2=R_V=3$ and $V_0=480,$ because in natural units they correspond to the nuclear radius of Uranium, three times the nuclear radius and approximately the strength of the Coulomb repulsion $V_{\text{Coulomb}}=36\,\mathrm{MeV}$ experienced by an alpha particle sitting at $r=r_1$ (in SI-units we have $V_0=48\,\mathrm{MeV}$, $r_1=7.2 \,\mathrm{fm}$, and $r_2=21.6\,\mathrm{fm}$). The parameter $V_0$ was chosen, so that the decay rate $4\alpha_0\beta_0$ of the first resonance is in good agreement with the empirically measured decay rate, as is discussed later.

It is clear that the potential does not have any bound states, therefore we have to consider only Theorem~\ref{thm:main_ac}, with $P_{ac} = \1$.
This theorem provides an estimate on the survival probability once the radius $R$ is understood as the radius of a detector waiting for the alpha particle to hit it. Therefore, we use the value $ R=1.4\times 10^{14}$, that corresponds to $1\,\mathrm{m}$.
It should be noted that to have a probability it is necessary to divide both sides of Eq.~\eqref{eq:main_ac} by the $L^2$-norm of the initial wave function.

Due to the simplicity of the potential that we consider, we can determine the Jost function explicitly; it reads
\begin{align}\label{eq:JostBarrier}
  F(k)=
  &e^{ikr_2}\bigg[e^{-ikr_1}\cos((r_2-r_1)\sqrt{k^2-V_0})\nonumber\\
  &-i\frac{k}{\sqrt{k^2-V_0}}\cos(kr_1)\sin((r_2-r_1)\sqrt{k^2-V_0})\nonumber\\
  &-\frac{\sqrt{k^2-V_0}}{k}\sin(kr_1)\sin((r_2-r_1)\sqrt{k^2-V_0})\bigg],
\end{align}
and from this all parameters that appear in the bounds of Theorem~\ref{thm:main_ac} can be determined. 

The first resonance numerically calculates to
\begin{align}
  k_0=3.0040-i\,1.4068\times 10^{-39}.
\end{align}
The decay rate $4\alpha_0\beta_0$ in SI-units is then $2.5682\times 10^{-18}\, \mathrm{s}^{-1},$ which is in good agreement with the experimental value $4.9160\times 10^{-18} \,\mathrm s^{-1}$ and thereby justifies our choice of parameters.

\begin{figure}
\hfill\subfloat[\label{fig:virtualstates}]{\includegraphics [width=.45\textwidth]{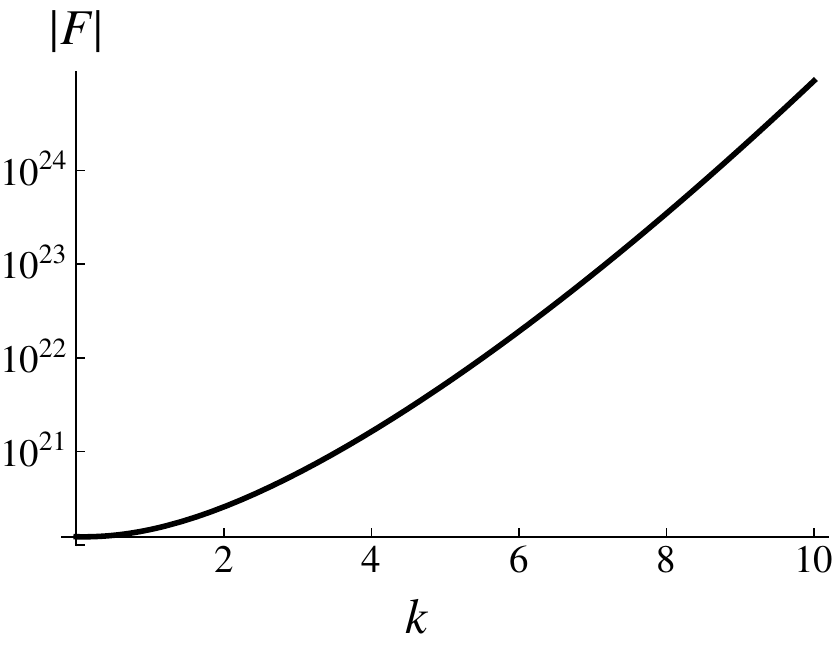}}%
\hfill\mbox{}\\
\hfill\subfloat[\label{fig:Ntilde}]{\includegraphics [width=.4\textwidth]{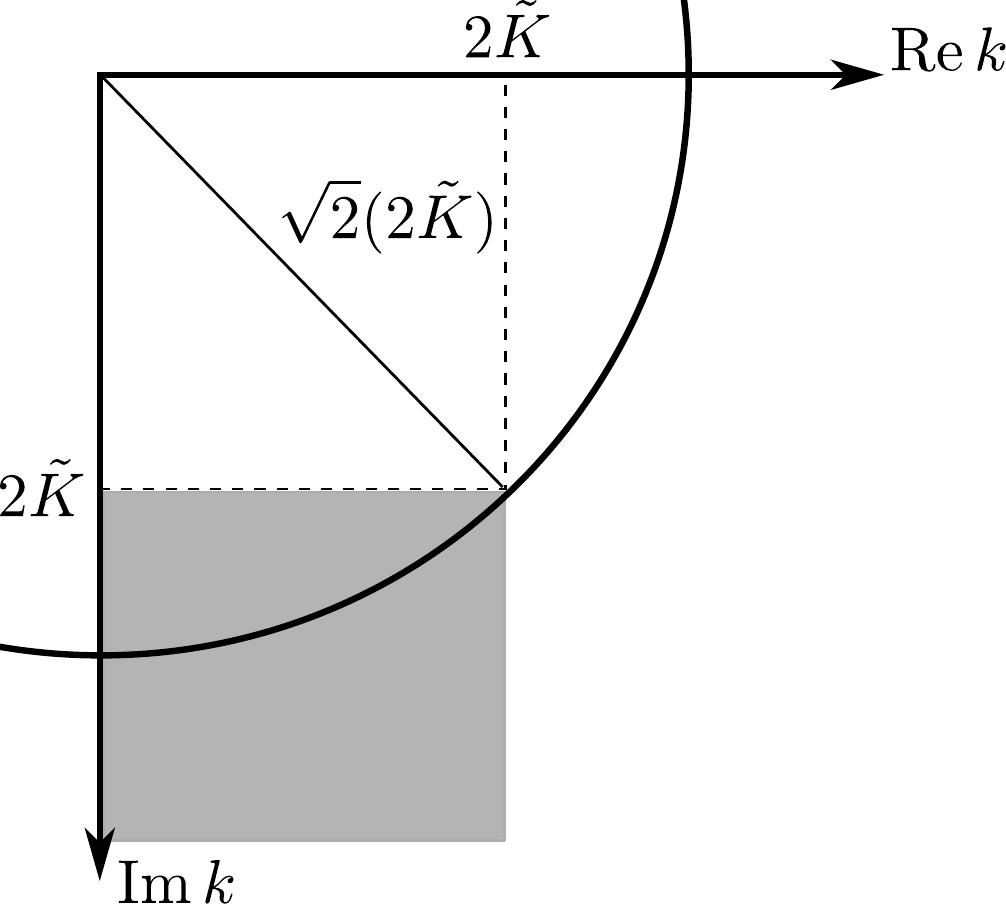}}%
\hfill\subfloat[\label{fig:Ntilde2}]{\includegraphics [width=.5\textwidth]{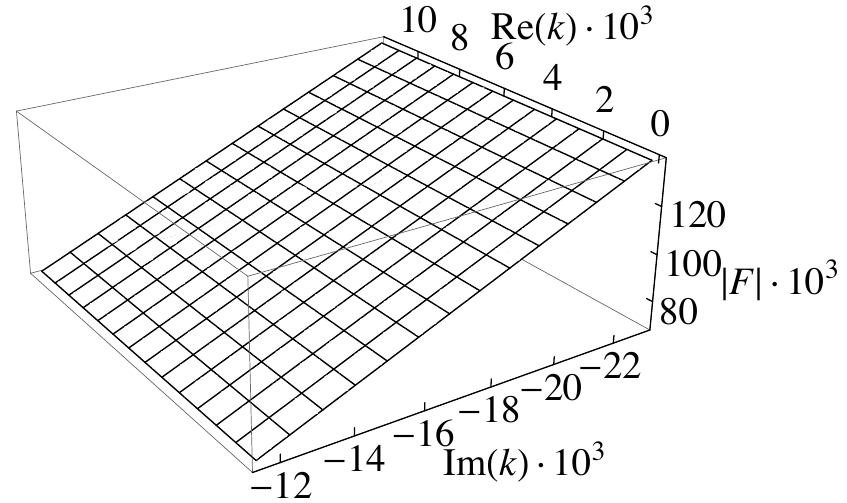}}%
\hfill\mbox{}%
\caption[]{%
\subref{fig:virtualstates}~Plot of $|F(-ik)|$ for $k\geq 0,$ showing that the model potential for Uranium 238 does neither have virtual states nor a zero resonance. \subref{fig:Ntilde}~Plot of the complex $k$-plane, to illustrate how the number of zeros $n(\sqrt 2(2\tilde K))$ in the ball of radius $\sqrt 2(2\tilde K)$ can be used to estimate $\nu_{\tilde K}$.  
\subref{fig:Ntilde2}~Plot of $|F(k)|$ in the shaded region shown in~\subref{fig:Ntilde}.}
\end{figure}

We need to know whether the potential has any virtual states or a zero resonance. For this purpose we plot $|F(-ik)|$ for $k\geq0.$ From Fig.~\ref{fig:virtualstates} it can be seen that the Jost function $F(k)$ does not have zeros on the negative imaginary axis, so that the potential has neither virtual states nor a zero resonance. 

Letting $\nu_K$ be the smallest integer such that $\alpha_n\geq 2K$ for all $n\geq\nu_K$, the parameter $s_K$ that appears in Theorem \ref{thm:main_ac} is in our case
\begin{align}
  \frac{1}{s_K} = 
  \left\{
  	\begin{aligned}
		&1, && \text{if } \nu_K =0,
		\\
		&\sum_{n=0}^{\nu_K-1}   \frac {1}  {\beta_n} , &&\text{otherwise.}
	\end{aligned}
  \right.
\end{align}
To minimize $1/s_K$ we therefore choose $K=\alpha_0/4,$ so $\nu_K=0$ and $1/s_K=1$. 
Similarly, with  $\nu_{\tilde K}$ being the smallest integer such that $\alpha_n\geq 2\tilde K=12\|V\|_1$ for all $n\geq\nu_{\tilde K}$, we have
\begin{align}
  \frac{1}{s}
  = \sum_{n=0}^{\nu_{\tilde K}-1}   \frac {1}  {\beta_n} 
  \leq\frac{\nu_{\tilde K}}{\beta_0},
\label{eq:OneOverS}
\end{align}
under the assumption that $\beta_0\leq\beta_n$ for all $0<n<\nu_{\tilde K}$. 
Therefore, we need a handle on $\nu_{\tilde K}$.
Lemma~\ref{lem:BoundOnNumberOfZeros} is of help here because it gives a bound on the number of zeros $n(r)$ in the ball of radius $r.$ Since $\nu_{\tilde K}$ is the number of zeros in $\{z|\Re z\leq2\tilde K,\Im z\leq0\},$ we need to ensure that there are no zeros in the shaded region shown in Fig.~\ref{fig:Ntilde} below the ball of the radius $\sqrt 2(2\tilde K).$ As can be seen from Fig.~\ref{fig:Ntilde2}, $|F(k)|>0$ in this region, so that
\begin{align}
	\nu_{\tilde K}
	&\leq n(2^{\frac{3}{2}}\tilde K)
	\\
	&\leq\frac{1}{\log 2}\left[4R_V2^{\frac{3}{2}}\tilde K+\log\left(4\|rV(r)\|_1e^{4\|rV(r)\|_1}+1\right)\right]
	\\
	&=2.9314\times 10^{5} ,
\end{align}
and
\begin{equation}
\frac{1}{s} \leq 2.0837\times 10^{44}.
\end{equation}

\begin{figure}
 \centering
 \includegraphics[width=.6\textwidth]{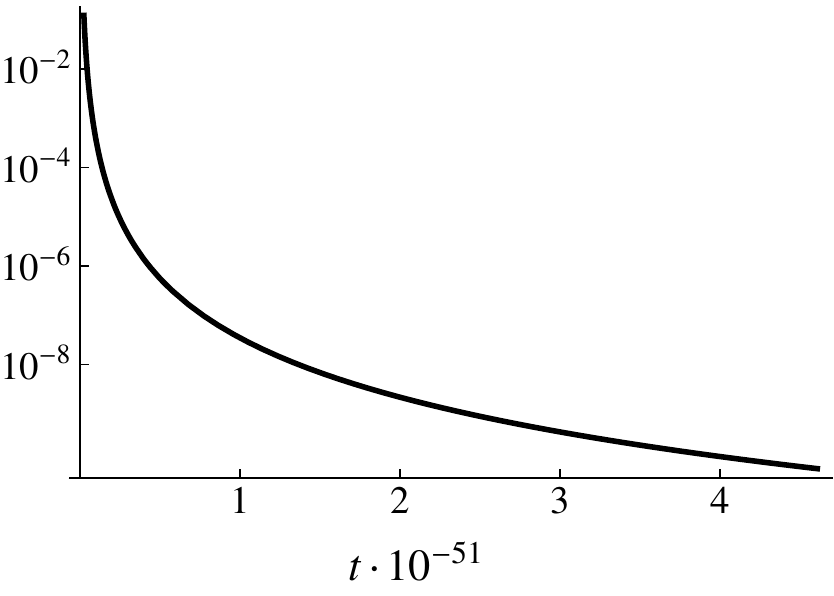}
\caption{\\\mbox{}\qquad Plot of  $\frac{1}{ \| f_R\|_2^2}  (c_3 t^{-3}+c_4 t^{-4})$ for $\frac{1}{10}\,(4\alpha_0\beta_0)^{-4/3}<t<20(4\alpha_0\beta_0)^{-4/3}.$}
 \label{fig:bound}
\end{figure}

To calculate $r_0$, consider Eqs.~\eqref{eq:L} and~\eqref{eq:r0}, from which we get
\begin{align}
  r_0=\frac{5}{2}\left(R_V-\Im\frac{\dkp F(0)}{F(0)}\right)=0.1141.
\end{align}
Together with $\|V\|_1=V_0(r_2-r_1)=960$ we now have everything that is needed to calculate the constants appearing in Theorems~\ref{th:SBoundsK} and \ref{th:GlobalSBounds}. They are given by
\begin{align}
	C_{1,K}&=8.2282
	&C_{2,K}&=89.8853
	&C_{3,K}&=1109.6900\\
	C_1&=2.0000
	&C_2&=12.0000
	&C_3&=60.0000,
\end{align}
where we have assumed that $\alpha_0=\min_{n>0}\{\alpha_n\}_n$. 
Using these values, Definition~\ref{def:smallz}, and Lemma~\ref{lem:psi} we can finally calculate
\begin{align}
	c_3&=3.3519\times 10^{89},
	\\
	c_4&=1.2293\times 10^{235}.
\end{align}
From Theorem \ref{thm:main_ac} we then have the bound on the survival probability 
\begin{align}\label{eq:DiscA}
  \frac{ \|\1_Re^{-iHt}f_R\|_2^2}  { \| f_R\|_2^2}
  	\leq \frac {c_3} { \| f_R\|_2^2}  t^{-3}
		+ \frac {c_4} { \| f_R\|_2^2}   t^{-4},
\end{align}
with \cite[Lemma 3.1]{Skibsted86}
\begin{equation}
{ \| f_R\|_2^2} = \frac {e^{2\beta_0 R}} {2\beta_0} = 3.5541\times 10^{38}.
\end{equation} 
Figure~\ref{fig:bound}  shows that the bound \eqref{eq:DiscA} becomes useful for $t>(4\alpha_0\beta_0)^{-4/3}$ with $(4\alpha_0\beta_0)^{-1}$ being one lifetime.

Note that, in contrast to the fact that $1/s\gg1$, we find that
\begin{align}
  z_{ac}(0)
  &=1+\frac{1}{2}\left(2Rs+2(1+2sR_V+sr_0)\right)
  	=5.1141,
\end{align}
and similarly all other parameters $z_{ac,K}(m),$ $z_{ac}(m)$ and $z_{e}(m)$ are much smaller than $1/s$.
Therefore, the bounds on the constants $c_3$ and $c_4$ are dominated by $1/s$, while the parameters $z_{ac,K}(m),$ $z_{ac}(m)$ and $z_{e}(m)$ play a minor role.

\vfill\newpage


\section{Proof of Theorem~\ref{th:SBoundsK}}\label{sec:SMatrixK}
For  $K>0$, $s_K$ defined in Eq.~\eqref{eq:DefSK}, and $n\in\{1,2,3\}$, we want to establish the bounds
\begin{align}
\|\1_K S^{(n)}\|_\infty  &\leq  C_{n,K}s_K^{-n} .
\end{align}
Our starting point is the expression of the $S$-matrix in terms of the Jost function $F$
\begin{equation}
S(k)=\frac{F(-k)}{F(k)}.
\end{equation}
We will exploit the fact that $F$ is an entire function, which implies that it is possible to write it as a product of factors that depend only on the location of the zeros.
Such a representation is called Hadamard factorization, and it is the main tool we will use to prove Theorem~\ref{th:SBoundsK}.

In order to write the Hadamard factorization of the Jost function, we need to determine some important parameters: the \emph{order},  the \emph{type} \citep[page 8]{Boas1954}, the \emph{convergence exponent of its zeros}, and  the \emph{genus of its zeros} \citep[page 14]{Boas1954}.
We recall their definitions here.
For an entire function $f$, let 
\begin{equation}
M(|k|)\coloneqq \sup_{\theta\in[0,2\pi]} | f( |k| e^{i\theta} ) |.
\end{equation}
The function $f$ is order $\rho$ ($0\leq\rho\leq\infty$) if and only if for every positive $\epsilon$, but for no negative~$\epsilon$
\begin{equation}
M(|k|)  =  O\left(  e^{|k|^{\rho+\epsilon}}  \right)  ,
	\qquad \text{as }{|k|\to\infty}.
\end{equation}
If the order of $f$ is finite and not zero, then $f$ is of finite type $\tau$ ($0\leq\tau\leq\infty$) if and only if for every positive $\epsilon$, but for no negative $\epsilon$
\begin{equation}
M(|k|)  =  O\left(  e^{(\tau+\epsilon) |k|^{\rho}}  \right)  ,
	\qquad \text{as }{|k|\to\infty}.
\end{equation}
For example, the function $e^k$ is of order one and type one.
An entire function $f$ of order one and finite type or of order less than one is  said to be of  \emph{exponential type}.
Let $z_n$ be the zeros of the entire function $f$ not lying on the origin.
Their convergence exponent is defined as the infimum of the positive numbers $\alpha$ such that  
\begin{equation}\label{eq:ConvergenceExponent}
\sum_{n=0}^\infty  \frac{1}{|z_n|^\alpha}  <  \infty  ,
\end{equation}
while their genus is the smallest integer $p\geq0$ such that Eq.\  \eqref{eq:ConvergenceExponent}  is verified for $\alpha=p+1$.
Given the zeros of $f$, consider the products
\begin{align}
\pi_0(k) &\coloneqq  	\prod_{n=0}^\infty   
	\left(   1 - \frac{k}{z_n}  \right) ,
\\
\pi_p(k) &\coloneqq 	\prod_{n=0}^\infty   
	\left(   1 - \frac{k}{z_n}  \right) 
	\exp \left({ \frac{k}{z_n} + \frac{k^2}{2 z_n^2} + \dots + \frac{k^p}{p z_n^p} } \right) ,
\quad p\geq 1 .
\end{align}
If the zeros of $f$ are of genus $p$, then the product $\pi_p$ is called \emph{canonical product of the zeros of f} \citep[page 18]{Boas1954}.

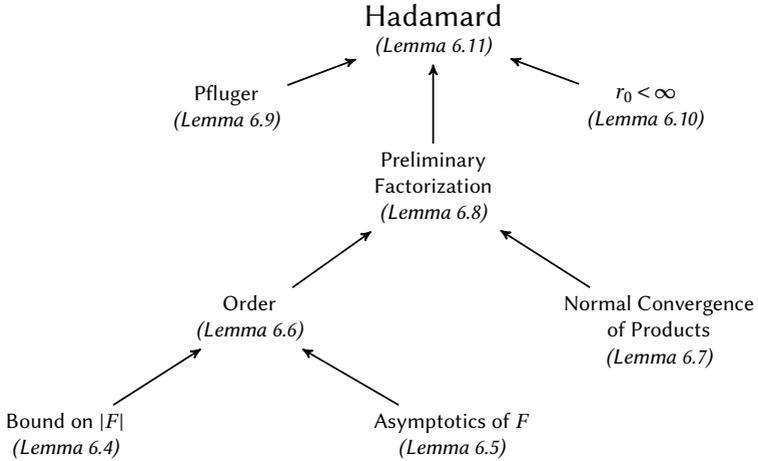
\begin{figure}
\centering 
\newcommand{\lem}[1]{\\\mbox{\textit{(Lemma~\ref{#1})}}}
\newlength{\leveldistance}
\setlength{\leveldistance}{\the\dimexpr\baselineskip*3\relax}
\resizebox{\textwidth}{!}{%
\begin{tikzpicture}[
node distance=\the\leveldistance, 
auto,
arrow/.style={->,>=stealth', thick},
every text node part/.style={align=center}
]
\normalsize
\sffamily
\node (Hadamard) 
	{\mbox{\textlarger[2]{Hadamard}}\lem{lem:GenusAndHadamard}};
\node[below=of Hadamard] (Pre) 
	{Preliminary\\Factorization\lem{lem:PreHadamard}}
	edge[arrow]  (Hadamard);
\node [left=of Pre, yshift={\the\leveldistance}]  
	{Pfluger\lem{lem:Pfluger}}
	edge[arrow]  (Hadamard);
\node [right=of Pre, yshift={\the\leveldistance}]  
	{$r_0<\infty$\lem{lem:R0}}
	edge[arrow]  (Hadamard);
\node[below right=of Pre] (Normal) 
	{Normal Convergence\\of Products\lem{lem:NormalConvergency}}
	edge[arrow]  (Pre);
\node[below left=of Pre] (Order) 
	{Order\lem{lem:ExpType}}
	edge[arrow]  (Pre);
\node[below left=of Order] (Bound) 
	{Bound on $|F|$\lem{lem:BoundsOnFAndPhi}}
	edge[arrow]  (Order);
\node[below right=of Order] (Asymptotics) 
	{Asymptotics of $F$\lem{lem:LogFBigK}}
	edge[arrow]  (Order);
\end{tikzpicture}
}
\caption{Overview of the Lemmas needed to write the Jost function in Hadamard's form.}
\label{fig:LemmaOverview}
\end{figure}

With these definitions, we can write the Hadamard factorization of  $f$.
Let $f$ be of order $\rho$, its zeros of genus $p$, $Q$ a polynomial of degree not greater than $\rho$, and let $f$ have an $m$-fold zero in the origin.
Then $f$ can be written in the form  \citep[2.7.2, page 22; see also page 18]{Boas1954}
\begin{equation}\label{eq:HadamardGeneral}
f(k)  =   k^m  e^{Q(k)}  \pi_p(k).
\end{equation}
We will write such a representation for the  Jost function $F$.
Moreover, we will show that $F$ is of exponential type (Lemma \ref{lem:ExpType}); as a consequence, we will be able to determine the coefficients of $Q$ using a theorem due to Pfluger (see Lemma \ref{lem:Pfluger}).
To arrive at the Hadamard factorization of $F$ we need several intermediate Lemmas, whose structure is depicted in Fig.~\ref{fig:LemmaOverview}.

\subsection{Hadamard factorization of the Jost function}
To determine the order and type of the Jost function the following two lemmas are crucial.
They elaborate some results presented in \citep{Newton1966}.

\begin{lemma}\label{lem:BoundsOnFAndPhi}
Let $\nu=\Im k$.
Then, the Jost function $F$ and the regular eigenfunctions $\varphi$ satisfy the bounds
\begin{align}
| \varphi(k,r) |  
	&\leq 4 e^{4\| r'V(r') \|_1} 
	\frac {r}  {1+|k| r}   e^{|\nu| r}  ,
\label{eq:BoundPhi}
\\
| F(k) |
	&\leq  \left(  4 \| rV(r) \|_1 e^{4\| rV(r) \|_1}  + 1  \right)
	\,  e^{2 R_V |k|}
.
\label{eq:BoundF}
\end{align}
\end{lemma}
\begin{proof}
For $r\in\R^+$ and $k\in\C$, the eigenfunctions $\varphi$ are solutions of the Lippmann-Schwinger equation \citep[Eq.\ 12.4, page 330]{Newton1966}
\begin{equation}\label{eq:IntEqPhi}
\varphi(k,r)
=
\frac{\sin k r}{k}
+\int_0^r    \frac{\sin k (r-r')}{k}   V(r')  \varphi(k,r')   \,  dr'  .
\end{equation}
Writing the solution of this equation as a Born series, it is possible to prove the bound \citep[Eq.\ 12.8, page 332]{Newton1966}
\begin{equation}\label{eq:PhiCBound}
| \varphi(k,r) |  
\leq  e^{q_k(r)} 
	\frac {C r}  {1+|k| r}   e^{|\nu| r} ,
\end{equation}
where  
\begin{equation}
q_k(r) \coloneqq \int_0^r   \frac {C r'}  {1+ |k| r'}    |V(r')|\, dr'  ,
\end{equation}
and the constant $C$ is such that \citep[see][Eq.\ 12.6, page 331]{Newton1966}
\begin{equation}
\left |  \frac{\sin kr} {k}  \right|
\leq
\frac {C r}  {1+ |k| r}  e^{|\nu| r }
,  \qquad  r\geq 0  .
\end{equation}
From the bound \citep[page 139]{RS3}
\begin{equation}
\left |  \frac{\sin k(y-x)} {k}  \right|
\leq
\frac {4 y}  {1+ |k| y}  e^{|\nu| y + \nu x}
,  \qquad  y\geq x\geq 0  ,
\end{equation}
setting $x=0$ and $y=r\geq0$, 
we get
\begin{equation}
\left |  \frac{\sin k r} {k}  \right|
\leq
\frac {4 r}  {1+ |k| r}  e^{|\nu| r }   .
\end{equation}
As a consequence, we can choose $C=4$.
Observing that $q_k(r)  \leq  4 \| r' V(r') \|_1$, from \eqref{eq:PhiCBound} we get \eqref{eq:BoundPhi}.

The integral equation \eqref{eq:IntEqPhi}, together with the relation between $\varphi$ and $F$, Eq.\ \eqref{eq:phiRBiggerRs}, gives the integral equation for $F$ \citep[Eq.\ 12.36, page 341]{Newton1966}
\begin{equation}\label{eq:IntEqF}
F(k)
=
1  +   \int_0^{R_V}    e^{i k r}  V(r)  \varphi(k,r)   \,  dr  .
\end{equation}
Using this and the bound \eqref{eq:BoundPhi} we get
\begin{align}
|F(k)|
&\leq
	1  +  4  e^{4\| r'V(r') \|_1}    \int_0^{R_V}  	\frac {r}  {1+|k| r}   e^{2|\nu| r}   |V(r)| \,  dr  
\\
&\leq
	1  +  4{\| rV(r) \|_1}  e^{4\| rV(r) \|_1}      e^{2|\nu| R_V}
\\
&\leq  \left(  4 \| rV(r) \|_1 e^{4\| rV(r) \|_1}  + 1  \right)
	\,  e^{2 R_V |k|}   .
\end{align}
\QED
\end{proof}

Before giving the next Lemma, we use the results so far obtained to prove Lemma \ref{lem:BoundOnNumberOfZeros}.
\begin{proof}[of Lemma \ref{lem:BoundOnNumberOfZeros}]
The bound \eqref{eq:BoundOnNumberOfZeros} is a direct consequence of the bound \eqref{eq:BoundF}, together with Eq.\ (2.5.11) of \citep{Boas1954}, that says
\begin{equation}
n(|k|) {\log 2}
\leq
\log \max_{\theta\in[0,2\pi)}  |  F(2|k|  e^{i\theta})  | .
\end{equation}
\QED
\end{proof}

\begin{lemma}\label{lem:LogFBigK}
As ${|k|\to\infty}$, the Jost function $F$ satisfies the asymptotic formulas
\begin{align}
\log | F(-i |k| ) |   & \sim    2R_V |k|  ,
\label{eq:LogFDirectionMinusI}
\\
\log | F(i |k| ) |   &\sim  \frac  {1 }  {2|k|}  \int_0^{R_V}  V(r)  \,  dr  
,
\label{eq:LogFDirectionPlusI}
\end{align}
and the limit
\begin{equation}
\lim_{|k|\to\infty} \log | F(\pm|k|) |  =  0 .
\label{eq:LogFDirectionRealAxis}
\end{equation}
\end{lemma}
\begin{proof}
We follow the presentation of \citep[page 361]{Newton1966}.

Let $\nu = \Im k$, then Eqs.\ \eqref{eq:IntEqPhi} and \eqref{eq:BoundPhi} imply that \citep[Eq.\ 12.12]{Newton1966}
\begin{equation}\label{eq:OPhi}
\varphi(k,r) =   \frac  {\sin kr}  {k}   +  o\left(   \frac {e^{|\nu| r}}  {|k|} \right)  ,
\qquad  \text{as } {|k|\to\infty}  .
\end{equation}
Substituting this in \eqref{eq:IntEqF}, and considering only the direction $k=-i |k|$ gives
\begin{equation}
F(-i |k|)  \sim
	\int_0^{R_V}  \frac {e^{2|k| r}}  {2|k|} V(r) \,  dr  
	\qquad  \text{as } {|k|\to\infty}  .
\end{equation}
For ${|k|\to\infty}$ the integral is dominated by $r=R_V$, therefore it is convenient to write
\begin{equation}
F(-i |k|)   \sim
	e^{2R_V |k|}
	\int_0^{R_V}  \frac {e^{-2|k| (R_V-r)}}  {2|k|}  V(r)  \,  dr  ,
	\qquad  \text{as } {|k|\to\infty}  ,
\end{equation}
which implies
\begin{multline}
\log | F(-i |k| ) |   \sim    2R_V |k|  
	+\log  \left( \int_0^{R_V}  \frac {e^{-2|k| (R_V-r)}}  {2|k|}  V(r)  \,  dr  \right) ,
		\\  \text{as } {|k|\to\infty}  .
\end{multline}
This gives \eqref{eq:LogFDirectionMinusI}, provided that the integral does not go to zero as $e^{-2R_V |k|}$ or faster.
This is shown using Watson's lemma  \citep[see e.g.][Lemma 11.1, page 283]{KingBillinghamOtto2003} and  Assumption~\eqref{eq:VExpansionAtRs}, that give
\begin{equation}
\int_0^{R_V}  {e^{-2|k| (R_V-r)}}  V(r)  \,  dr
\sim
\sum_{n=0}^{M} \frac {\Gamma(\delta_n+1) d_n } {|k|^{1+\delta_n}},
	\qquad  \text{as } {|k|\to\infty}  .
\end{equation}
Similarly, for the direction $k=i|k|$ we get
\begin{equation}
F(i |k|)   
=	1  +    \frac  {1 }  {2|k|}  \int_0^{R_V}  V(r)  \,  dr  ,
	\qquad  \text{as } {|k|\to\infty}  ,
\end{equation}
that, using Taylor's expansion, gives \eqref{eq:LogFDirectionPlusI}.

For \eqref{eq:LogFDirectionRealAxis} it is enough to use the fact that $\lim_{|k|\to\infty} F(|k|)=1$ \citep[Th.\ XI.58e, page 140]{RS3} and the symmetry property of $F$ \citep[12.32a, page 340]{Newton1966}
\begin{equation}
\bar F (\bar k) = F(-k).
\end{equation}
\QED
\end{proof}

From the previous Lemmas we get

\begin{lemma}\label{lem:ExpType}
The Jost function has order one and type $2R_V$, and is therefore of exponential type.
Moreover, the convergence exponent of its zeros is one.
\end{lemma}
\begin{proof}

From the bound \eqref{eq:BoundF} we see that the Jost function has order  not greater than one, while from the asymptotic formula \eqref{eq:LogFDirectionMinusI} we get that the order can not be less than one, therefore it must be $\rho=1$.
The same reasoning gives $\tau=2R_V$.

Let $z_n$ denote the zeros of the Jost function $F$ other than $k=0$.
Consider the function
\begin{equation}\label{eq:GDef}
g(k^2)  \coloneqq  F(k)  F(-k),
\end{equation}
that is an entire function of $k^2$, whose zeros are $\{z_n^2\}_{n}$.
Following the proof of the order of $F$, Eqs.\ \eqref{eq:BoundF}, \eqref{eq:LogFDirectionMinusI}, and \eqref{eq:LogFDirectionPlusI} imply that $g$ is of order \nicefrac{1}{2}.
For a function of fractional order the convergence exponent of the zeros is equal to the order \citep[2.8.2, page 24]{Boas1954}, therefore
\begin{equation}
\sum_{n=0}^\infty  \frac{1}{|z_n^2|^\alpha}  
	\begin{cases}
		<  \infty  , & \alpha> 1/2 ,\\
		= \infty  , & \alpha< 1/2 ,
	\end{cases}
\end{equation}
that shows that the convergence exponent of the zeros of $F$ is one.
\QED
\end{proof}

The only parameter missing to write the Hadamard factorization of the Jost function $F$ is the genus of its zeros, that will be determined in Lemma \ref{lem:GenusAndHadamard}. 
However, we will write a product form for $F$ already in Lemma \ref{lem:PreHadamard}.
To that end, we will need to combine different infinite products, for which we will use the following notion of convergence.
A product of continuous functions $\prod_n f_n$ is called \emph{normally convergent} if $\sum_n (f_n-1)$ is normally convergent, i.e.\ if  every point $k\in\C$ has a neighborhood  $U_k$ such that $\sum_n \sup_{k'\in U_k} |f_n(k')-1|<\infty$ (see \citealp[1.2.1, page 7]{Remmert1998}, and \citealp[3.3.1, page 104]{Remmert1991}).
We apply this notion to our case in the next Lemma.

\begin{lemma}\label{lem:NormalConvergency}
Let $k_n=\alpha_n-i\beta_n$, with $\alpha_n,\beta_n>0$, be the resonance zeros of the Jost function $F$.
If the quantity
\begin{equation}
r_0  \coloneqq  \sum_{n=0}^\infty   \frac {5\beta_n}  { \alpha_n^2 + \beta_n^2}    
\end{equation}
is finite, then the products
\begin{align}
&	\prod_{n=0}^\infty  
	\left(   1 - \frac{k}{k_n}  \right) 
	\left(   1 + \frac{k}{\bar k_n}  \right)   ,
\label{eq:NormalConvergencyP0}
\\
&	 \prod_{n=0}^\infty  e^{ 2i \beta_n k/ | k_n |^2}  ,
\label{eq:NormalConvergencyExp}
\\
&	\prod_{n=0}^\infty  
	\left(   1 - \frac{k}{k_n}  \right) 
	\left(   1 + \frac{k}{\bar k_n}  \right)   
	e^{2i \beta_n k/ | k_n |^2} ,
\label{eq:NormalConvergencyP1}
\end{align}
are normally convergent.
Moreover,
\begin{equation}\label{eq:ConnectionProd0Prod1}
\prod_{n=0}^\infty  
	\left(   1 - \frac{k}{k_n}  \right) 
	\left(   1 + \frac{k}{\bar k_n}  \right)   
	e^{2i \beta_n k/ | k_n |^2}
= \exp \left(\frac{2}{5} i r_0 k \right)   \,
	\prod_{n=0}^\infty  
	\left(   1 - \frac{k}{k_n}  \right) 
	\left(   1 + \frac{k}{\bar k_n}  \right) .
\end{equation}
\end{lemma}
\begin{proof}
Let $k=\mu+i\nu$ and  $\gamma_{0,k}$ the straight line connecting the origin to $k$, then we have the bounds:
\begin{multline}
\left| \left(   1 - \frac{k}{k_n}  \right)  
	\left(   1 + \frac{k}{\bar k_n}  \right)   -1  \right |
= 
\left|  \frac {k^2+2i\beta_n k}   {|k_n|^2}  \right |
\leq 
\frac {|k|^2+2\beta_n |k|}   {|k_n|^2} 
\sim
\frac {2\beta_n |k|}   {|k_n|^2} ,
\\\quad\text{as }{n\to\infty},
\end{multline}
\begin{align}
 \left|    e^{ 2i\beta_n k / |k_n|^2 }   -1   \right |
	&= \left|  \int_{\gamma_{0,k}}  \frac{d}{d k'}   \left(  e^{ 2i\beta_n k' / |k_n|^2 }  \right)   \, dk'  \right |
\nonumber\\
	&\leq  \frac {2\beta_n |k|}   {|k_n|^2 } e^{ 2\beta_n |\nu| / |k_n|^2 }  
	\sim
	\frac {2\beta_n |k|}   {|k_n|^2 } 
	   	  ,
\nonumber\\
&\hspace{3em}\text{as }{n\to\infty},
\displaybreak[0]\\
p_n(k)&\coloneqq \left| \left(   1 - \frac{k}{k_n}  \right)  
	\left(   1 + \frac{k}{\bar k_n}  \right)   e^{2i \beta_n k/ | k_n |^2}   -1  \right |
\nonumber
\\ 
\phantom{ \left|    e^{ 2i\beta_n k / |k_n|^2 }   -1   \right |}
&=
\left|  \left(   1 - \frac {k^2+2i\beta_n k}   {|k_n|^2}  \right)   e^{2i\beta_n k / |k_n|^2 }   -1  \right |
\nonumber\\ 
& \leq
\left|   e^{2i\beta_n k / |k_n|^2 }   -1    \right |
	+  \left|   \frac {k^2+2\beta_n k}   {|k_n|^2}   e^{2i\beta_n k / |k_n|^2 }  \right |
\nonumber\\ 
& \leq
\frac {|k|^2+4\beta_n |k|}   {|k_n|^2} 
e^{ 2\beta_n |\nu| / |k_n|^2 } 
	 \sim
\frac {4\beta_n |k|}   {|k_n|^2}  ,
\nonumber\\
&\hspace{3em}\text{as }{n\to\infty}.
\end{align}
For every compact $K\subset\C$ let $h\coloneqq \sup_{k\in K} |k|$, then there is a $n_K\in\N$ and a constant $C_K$ such that
\begin{equation}
\sup_{k\in K}    p_n(k)  < C_K \frac {\beta_n h}   {|k_n|^2} ,  
	\quad  \forall n \geq n_K,
\end{equation}
hence, 
\begin{equation}
\sum_{n=n_K}^\infty \
\sup_{k\in K}    p_n(k)
<
\sum_{n=n_K}^\infty \
C_K\frac {\beta_n h}   {|k_n|^2}
\leq \frac{C_K}{5} r_0
 < \infty .
\end{equation}
Similarly for the other products.

The continuity of the exponential allows us to write
\begin{equation}
\exp \left( \sum_{n=0}^\infty   \frac{2i\beta_n k}{|k_n|^2} \right) 
	=   \exp \left(\frac{2}{5} i r_0 k \right)   
	=  \prod_{n=0}^\infty    \exp  \left( \frac{2i\beta_n k}{|k_n|^2}  \right),
\end{equation}
indeed we can exchange it with the limit of the series.
Moreover, the products \eqref{eq:NormalConvergencyP0} and \eqref{eq:NormalConvergencyExp} are normally convergent and therefore also compactly convergent \citep[page 8]{Remmert1998} and we can multiply them factor by factor \citep[page 6]{Remmert1998} getting \eqref{eq:ConnectionProd0Prod1}.
\QED
\end{proof}

Although we do not know yet the genus of the Jost function, we can already write its Hadamard factorization.
Indeed, the genus $p$ and the convergence exponent $\rho_1$ of the zeros of $F$ satisfy the relation $\rho_1-1\leq p \leq \rho_1$  \cite[see][page 17]{Boas1954}, and Lemma~\ref{lem:ExpType} tells us that $\rho_1$ is one, hence the genus $p$ must be either one or zero.
In the next Lemma we will show that if the genus is zero, one can still write the Jost function in the same factorization form as if it had genus one.
Analogously, in Lemma~\ref{lem:GenusAndHadamard} we will prove that if the genus is one, we can write the Jost function in the same form as if it had genus zero.
Moreover, in Lemma~\ref{lem:GenusAndHadamard} we will be able to determine all constants that appear in the Hadamard factorization of the Jost function.

\begin{lemma}\label{lem:PreHadamard}
Let $a_1,b_1\in\R$, $C\in\C$, 
\begin{align}
\lambda & \coloneqq
	\begin{cases}
		0,&\mathrm{if}\;F(0)\neq0\\
		1,&\mathrm{if}\;F(0)=0,
	\end{cases}
\\
B(k) & \coloneqq   
	\prod_{m=0}^{N-1}  \left(   1 - \frac{k}{i \eta_m}  \right)   
	\
	\prod_{l=0}^{N'-1}  \left(   1 +  \frac{k}{i\kappa_l}  \right)  ,
\\
P_1(k) &\coloneqq  	\prod_{n=0}^\infty   
	\left(   1 - \frac{k}{k_n}  \right)  \left(   1 + \frac{k}{\bar k_n}  \right)  
	e^{  {2i\beta_n k}/{|k_n|^2}}  ,
\end{align}
then for the Jost function $F$ the following representation holds
\begin{equation}\label{eq:PreHadamard}
{F(k)}   =   \left( F(0) + \lambda C k \right)   e^{(a_1+ i b_1)k}   {B(k)}   P_1(k) .
\end{equation}
Moreover, the product $P_1$ is an entire function of exponential type.
\end{lemma}
\begin{proof}
To write the Hadamard factorization given in Eq.\ \eqref{eq:HadamardGeneral} for $F$ we need to know its genus.
The genus $p$ and the convergence exponent $\rho_1$ of the zeros of $F$ satisfy the relation
$\rho_1-1\leq p \leq \rho_1$  \cite[see][page 17]{Boas1954},
hence by Lemma~\ref{lem:ExpType}, $p$ must be either one or zero.
We at first assume that  $p$ is one.
Using the fact that the Jost function can eventually have only a simple zero in $k=0$, and that the number of bound states and virtual states is finite, the Hadamard factorization given in Eq.~\eqref{eq:HadamardGeneral} yields directly Eq.~\eqref{eq:PreHadamard}.
In this case, $P_1$ is the canonical product of the zeros $\{k_n\}_{n\in\N^0}$, therefore it is an entire function of order one thanks to Theorem 2.6.5 of \citep[page 19]{Boas1954}.

Suppose now that  $p$ is zero.
Let $a_0, b_0 \in \R$, and
\begin{equation}
P_0(k) \coloneqq   \prod_{n=0}^\infty   
	\left(   1 - \frac{k}{k_n}  \right)  \left(   1 + \frac{k}{\bar k_n}  \right)  ,
\end{equation}
then Eq.~\eqref{eq:HadamardGeneral} gives
\begin{equation}\label{eq:PreHadamard0}
{F(k)}   =   \left( F(0) + \lambda C k \right)   e^{(a_0+ i b_ 0)k} {B(k)} P_0(k) .
\end{equation}
If the genus of the zeros of the Jost function is zero, then 
\begin{equation}
\sum_{n=0}^\infty \frac{1} {|k_n|}  <  \infty,
\end{equation}
and therefore
\begin{equation}
r_0
\coloneqq  \sum_{n=0}^\infty \frac{5\beta_n} {\alpha_n^2+\beta_n^2}
=  \sum_{n=0}^\infty  \frac{5} {|k_n|} \frac{\beta_n} {\sqrt{\alpha_n^2+\beta_n^2}}
\leq \sum_{n=0}^\infty \frac{5} {|k_n|}  <  \infty  .
\end{equation}
We can then use Eq.~\eqref{eq:ConnectionProd0Prod1} of Lemma~\ref{lem:NormalConvergency}, that gives
\begin{equation}
P_0(k) = e^{-\frac{2}{5} i r_0 k}   P_1(k) .
\end{equation}
Hence, Eq.~\eqref{eq:PreHadamard0} reduces to Eq.~\eqref{eq:PreHadamard} once we set
\begin{equation}
\begin{cases}
b_1 \coloneqq b_0 - \frac{2}{5}  r_0  ,
\\
a_1 \coloneqq a_0.
\end{cases}
\end{equation}
In this case the canonical product of the zeros $\{k_n\}_{n\in\N^0}$ is $P_0$, that is then an entire function of order one again thanks to Theorem 2.6.5 of \citep[page 19]{Boas1954}.
Moreover, $r_0<\infty$ therefore $e^{-\frac{2}{5} i r_0 k}$ is also an entire function of order one, and so $P_1$ is.

We have now only to show that $P_1$ is of finite type.
The Jost function is of order one and of finite type, moreover the function 
that multiplies $P_1$ in Eq.~\eqref{eq:PreHadamard} is clearly an entire function of order not greater than one and of finite type, therefore $P_1$ can only be of finite type, otherwise $F$ could not be so.
%
%
\QED
\end{proof}

To determine the coefficients of the polynomial $Q$ appearing in the Hadamard factorization of $F$, Eq.~\ref{eq:HadamardGeneral}, i.e.\ to determine the constants $a_1$ and $b_1$ in Eq.~\eqref{eq:PreHadamard}, we will use a result by Pfluger  \citep[Th.\ 6B, page 15; see also Th.\ 5, page 11]{Pfluger1943} (see also \citealp[8.4.20, page 147]{Boas1954}, and  \citealp[page 363]{Newton1966}).
We recall here the part of the theorem of interest to us.

\begin{lemma}[Pfluger]\label{lem:Pfluger}
Let  $z,z_n\in\C$  $\forall n\in\N^0$,
\begin{equation}
f(z) \coloneqq  	\prod_{n=0}^\infty   
	\left(   1 - \frac{z}{z_n}  \right)  \left(   1 + \frac{z}{\bar z_n}  \right)  
	e^{  {2i\beta_n z}/{|z_n|^2}}  
\end{equation}
be an entire function of exponential type such that
\begin{equation}\label{eq:PflugerIntCond}
\lim_{r\to\infty}  \int_{-r}^r   \frac {\log  |f( x )|}  {x^2}  \,  d x  <  \infty  ,
\end{equation}
and whose zeros have density $D$,
let $z=|z| e^{i\theta}$, and
\begin{equation}\label{eq:PflugerSigma}
\varsigma \coloneqq \sum_{n=0}^\infty  z_n^{-1} .
\end{equation}
If the density of the zeros of $f$ with positive real part is the same as the density of the zeros with negative real part, then  
\begin{equation}\label{eq:PflugerStatement}
\frac {\log | f(z) |}  {|z|} 
=
\Re  \varsigma  \cos \theta
- \Im \varsigma  \sin\theta
+ \frac{\pi}{2} D  |\sin\theta|
+\varepsilon(z),
\end{equation}
where  $\varepsilon$ is a function of $z$ such that 
\begin{equation}
\limsup_{|z|\to\infty}  \varepsilon(z)  = 0  .
\end{equation}
\end{lemma}

In order to apply this Lemma to the product
\begin{equation}
P_1(k) \coloneqq  	\prod_{n=0}^\infty   
	\left(   1 - \frac{k}{k_n}  \right)  \left(   1 + \frac{k}{\bar k_n}  \right)  
	e^{  {2i\beta_n k}/{|k_n|^2}}  
\end{equation}
introduced in Lemma~\ref{lem:PreHadamard}, we need to prove the following Lemma, which is therefore of technical nature.
Nevertheless, the quantity $r_0$  will appear in the bounds of Theorems \ref{th:SBoundsK} and \ref{th:GlobalSBounds}.

\begin{lemma}\label{lem:R0}
Let $n(|k|)$ be the number of zeros of the Jost function $F$ within a ball of radius $|k|$, and let $k_n=\alpha_n-i\beta_n$, with $\alpha_n,\beta_n>0$, be the resonance zeros of $F$.
The limits
\begin{equation}
\lim_{r\to\infty}  \int_{-r}^r   \frac {\log  |P_1( \mu )|}  {\mu^2}  \,  d \mu  ,
\end{equation}
and
\begin{equation}\label{eq:LimD}
\lim_{|k|\to\infty} \frac {n(|k|)}  {|k|}  
\end{equation}
exist and are finite.
Moreover, 
\begin{equation}\label{eq:r0}
r_0  \coloneqq  \sum_{n=0}^\infty   \frac {5\beta_n}  { \alpha_n^2 + \beta_n^2}   <\infty .
\end{equation}
\end{lemma}
\begin{proof}
We use again the function $g$ defined in Eq.\ \eqref{eq:GDef}.
For $k\geq0$, then $k=|k|$ and the symmetry property $F(-k) = \bar F(\bar k)$ \citep[12.32a, page 340]{Newton1966} implies
\begin{equation}\label{eq:GReal}
g(|k|^2) =  F(|k|)  F(-|k|)
	= F(|k|)  \bar F(|k|)
	=  |F(|k|)|^2  .
\end{equation}

Substituting \eqref{eq:OPhi} in \eqref{eq:IntEqF} we get
\begin{equation}
F(|k|)   
\sim 
	1  -    \frac  {1 }  {|k|}  \int_0^{R_V}  V(r)  \,  dr   ,
	\qquad\text{as }{|k|\to\infty},
\end{equation}
that for $g$ implies
\begin{equation}
g( |k|^2 )  =  |F(|k|)|^2  
\sim	1  -    \frac  {2}  {|k|}  \int_0^{R_V}  V(r)  \,  dr   .
	\qquad\text{as }{|k|\to\infty},
\end{equation}
Therefore  $\log  g( |k|^2 )\to 0$ as $|k|\to\infty$, and
\begin{equation}\label{eq:IntegralLogConditionG}
\int_1^\infty   \frac {\log  g( \mu^2 )}  {\mu^2}  \,  d \mu  <  \infty.
\end{equation}

Consider now the function $g_1(k^2)  \coloneqq  P_1(k)  P_1(-k)$, which is the analogue of the function $g$ for $P_1$.
Note that for real argument $g_1$ is real-valued and positive (cf.\ Eq.~\eqref{eq:GReal}).
From Eq.~\eqref{eq:PreHadamard} we see that
\begin{equation}
{\log  g( k^2 )} = {\log  g_1( k^2 )}  +  O({\log k}),  
\qquad  \text{as } k\to\infty,
\end{equation}
therefore Eq.~\eqref{eq:IntegralLogConditionG} implies that 
\begin{equation}\label{eq:IntegralLogConditionG1}
\int_1^\infty   \frac {\log  g_1( \mu^2 )}  {\mu^2}  \,  d \mu  <  \infty.
\end{equation}
We have $g_1(0)=1$, hence for $\mu\in\R$ it exists an $a_\mu\in\R$  such that
\begin{equation}
\log g_1(\mu^2)  =  \log g_1(0)  
	+ \int_0^{\mu^2}  \frac{d}{dt}  \left( \log g_1(t)  \right)  d t  
=  a_\mu \mu^2,
\end{equation}
with
\begin{equation}
|a_\mu| \leq  \sup_{0\leq t\leq \mu^2}  \frac{d}{dt}  \left( \log g_1(t)  \right) <\infty
\end{equation}
because the function $g_1$ is entire and has no zeros on the real axis.
As  a consequence,
\begin{equation}
\int_{-1}^1   \frac {\log  g_1( \mu^2 )}  {\mu^2}  \,  d \mu  <  \infty
\end{equation}
(for a general argument, see footnote 12 on page 5 of \citep{Pfluger1943} and the comment after Eq.~8.2.2 in \citep[page 136]{Boas1954}).
This, together with Eq.~\eqref{eq:IntegralLogConditionG1} and the fact that $\log g_1(|k|^2) =  2 \log |P_1(|k|)|$, gives
\begin{equation}
\lim_{r\to\infty}  \int_{-r}^r   \frac {\log  |P_1( \mu )|}  {\mu^2}  \,  d \mu  <  \infty .
\end{equation}
The convergence of the latter integral  is equivalent to the convergence of the sum $r_0$ and to the existence  and finiteness of the limit \eqref{eq:LimD} because of Theorem 5 of \citep{Pfluger1943} (see also Theorem 8.4.1 in \citep[page 143]{Boas1954} and the discussion on pages 133-135 of \citep{Boas1954}), together with the fact that the set of the zeros of $P_1$ is equal to that of the zeros of $F$, except for finitely many elements.
\QED
\end{proof}

We remind the reader that the genus of the zeros $z_n$ of an entire function $f$ is the smallest integer $p\geq0$ such that \citep[2.5.4, page 14]{Boas1954}
\begin{equation}
\sum_{n=0}^\infty  \frac{1}{|z_n|^{p+1}}  <  \infty  .
\end{equation}

In the next lemma, we finally write the Hadamard factorization of the Jost function, with all the constants explicitly determined.

\begin{lemma}\label{lem:GenusAndHadamard}
Let
\begin{align}
B(k) & \coloneqq   
	\prod_{m=0}^{N-1}  \left(   1 - \frac{k}{i \eta_m}  \right)   
	\
	\prod_{l=0}^{N'-1}  \left(   1 +  \frac{k}{i\kappa_l}  \right)  ,
\\
P_0(k) & \coloneqq   \prod_{n=0}^\infty   
	\left(   1 - \frac{k}{k_n}  \right)  \left(   1 + \frac{k}{\bar k_n}  \right)  ,
\\
\lambda & \coloneqq
	\begin{cases}
		0,&\mathrm{if}\;F(0)\neq0\\
		1,&\mathrm{if}\;F(0)=0,
	\end{cases}
\end{align}
then for the Jost function $F$ the following decomposition holds
\begin{equation}
F(k) = \left( F(0) + \lambda \dkp F(0) k \right) \, 
	e^{i R_V k} \,  B(k)\,  P_0(k)  .
\label{eq:Hadamard}
\end{equation}
Moreover, the genus of the zeros of $F$ is one and their density is (cf.\ Lemma~\ref{lem:R0})
\begin{equation}
D \coloneqq \lim_{|k|\to\infty} \frac {n(|k|)}  {|k|}  
	=  \frac{2}{\pi}  R_V  .
\end{equation} 
\end{lemma}
\begin{proof}
We already proved in Lemma~\ref{lem:PreHadamard} that 
\begin{equation}\label{eq:HadamardFB1}
{F(k)}   =   \left( F(0) + \lambda C k \right)   e^{(a_1+ i b_1)k}   {B(k)}   P_1(k) .
\end{equation}
We at first notice that Eq.~\eqref{eq:ConnectionProd0Prod1} of Lemma~\ref{lem:NormalConvergency} together with Lemma~\ref{lem:R0} gives
\begin{equation}\label{eq:ConnectionP0P1}
P_0(k) = e^{-\frac{2}{5} i r_0 k}   P_1(k) .
\end{equation}
We can now determine the constants $a_1$ and $b_1$ applying Lemma~\ref{lem:Pfluger} to the function $P_1$.
We can use it because $P_1$ is an entire function (Lemma~\ref{lem:PreHadamard}), the integral condition \eqref{eq:PflugerIntCond} holds because of Lemma~\ref{lem:R0}, and  for every zero of $P_1$ with positive real part $k_n$ there is exactly one zero with negative real part $-\bar k_n$.
The quantity $\varsigma$ of Eq.~\eqref{eq:PflugerSigma} is in this case
\begin{equation}
\varsigma  =  \sum_{n=0}^\infty \left( \frac {1} {k_n}  -  \frac {1} {\bar k_n}  \right)
=  \sum_{n=0}^\infty \frac {2i\beta_n} {|k_n|^2}
= \frac{2}{5} i r_0 .
\end{equation}
Note that the density of the zeros of $P_1$ is equal to the density $D$ of all of the zeros of $F$  (cf.\ Lemma~\ref{lem:R0}).
Equation~\eqref{eq:PflugerStatement} then gives
\begin{equation}\label{eq:Pfluger}
\limsup_{|k|\to\infty}  \frac {\log | P_1(k) |}  {|k|} 
=
- \frac{2}{5}  r_0  \sin\theta
+ \frac{\pi}{2} D  |\sin\theta|.
\end{equation} 
We will specialize this statement for the directions $\theta=\pm\pi/2, \pi$.
From Lemma~\ref{lem:LogFBigK} we have
\begin{align}
\limsup_{|k|\to\infty}  \frac {\log | F(i|k|) |}  {|k|} & = 0,
\nonumber\\
\limsup_{|k|\to\infty}  \frac {\log | F(-|k|) |}  {|k|} & = 0,
\nonumber\\
\limsup_{|k|\to\infty}  \frac {\log | F(-i|k|) |}  {|k|} & = 2R_V .
\end{align}
The factorization \eqref{eq:HadamardFB1} implies
\begin{align}
\frac {\log | P_1(k) |}  {|k|} 
&=
\frac{1}{|k|}   \log  \left|   
	\frac  {e^{-(a_1+ib_1)k} F(k) }   { \left( F(0) + \lambda C k \right)B(k)}   \right|
\nonumber\\
&= 
b_1 \sin\theta  - a_1\cos\theta  
+  \frac {\log | F(k) |}  {|k|} 
-\frac {\log | { \left( F(0) + \lambda C k \right)B(k)} |}  {|k|} 
\nonumber\\
&\sim
b_1 \sin\theta  - a_1\cos\theta  
+  \frac {\log | F(k) |}  {|k|} ,
	\quad\text{as }{|k|\to \infty} .
\label{eq:ConnectionP1F}
\end{align}
Therefore, Eq.~\eqref{eq:Pfluger} gives
\begin{align}
&\limsup_{|k|\to\infty}  \frac {\log | P_1(i|k|) |}  {|k|}  = b_1 
	= - \frac{2}{5}  r_0  + \frac{\pi}{2} D ,
\\
&\limsup_{|k|\to\infty}  \frac {\log | P_1(-|k|) |}  {|k|}  = a_1 = 0,
\\
&\limsup_{|k|\to\infty}  \frac {\log | P_1(-i|k|) |}  {|k|}  = -b_1 + 2R_V 
	=  \frac{2}{5}  r_0  + \frac{\pi}{2} D .
\end{align}
Summing and subtracting the first and last lines
\begin{equation}\label{eq:PflugerCoefficients}
\begin{cases}
b_1 = R_V - \frac{2}{5}  r_0 
\\
a_1 = 0
\\
D =  \frac{2}{\pi}  R_V  .
\end{cases}
\end{equation}

To determine the constant $C$ in Eq.~\eqref{eq:HadamardFB1}, we calculate the logarithmic derivative of $F$.
The infinite product $P_1$ appearing in Eq.~\eqref{eq:HadamardFB1} is normally convergent (Lemma~\ref{lem:NormalConvergency} together with Lemma~\ref{lem:R0}), therefore its logarithmic derivative can be calculated as if it were a finite product \citep[page 10]{Remmert1998}.
This gives 
\begin{equation}
\frac{\dkp F(k)}{F(k)}  =
	\frac {\lambda}  {k}
	+ i R_V
	+ \frac {\dkp B(k)}  {B(k)}
	+ \sum_{n=0}^\infty  
		\left(  \frac {1}  {k-k_n}  +  \frac {1}  {k+\bar k_n}   \right)  ,
\label{eq:LogDerF}
\end{equation}
with
\begin{equation}
\frac {\dkp B(k)}  {B(k)}   =
	\sum_{m=0}^{N-1}    \frac {1}  {k-i \eta_m}     
	+
	\sum_{l=0}^{N'-1}   \frac {1}  {k+i \kappa_l}  ,
\end{equation}
from which we can calculate $\dkp F(k)$ for every $k$ such that $F(k)\neq0$.
Moreover, we can calculate $\dkp F$ at the zeros of $F$ as limit of this result.
In particular, if $F(0)=0$, from Eq.~\eqref{eq:HadamardFB1} we have  for $k>0$
\begin{multline}
\dkp F(k)  =
	\left[ \frac {1}  {k}
	+ i R_V
	+ \frac {\dkp B(k)}  {B(k)}
	+ \sum_{n=0}^\infty  
		\left(  \frac {1}  {k-k_n}  +  \frac {1}  {k+\bar k_n}   \right)  
	\right ]
	\\
	\times
	 C k   e^{(a_1+ i b_1)k}  P_1(k) B(k) .
\end{multline}
In the limit $k\to0$ we then have
\begin{equation}
C=\dkp F(0).
\end{equation}
Substituting this and Eq.~\eqref{eq:PflugerCoefficients} into Eq.\ \eqref{eq:HadamardFB1} and using \eqref{eq:ConnectionP0P1} we get \eqref{eq:Hadamard}.

We now determine the genus of $F$.
The density is nonzero and finite if and only if $|k_n| \sim 2n/D$ as $n\to\infty$ (the two is due to the symmetry of the resonances with respect to the imaginary axis), but that implies 
\begin{equation}\label{eq:GenusContradiction}
\sum_{n=0}^\infty  \frac{1}{|k_n|}  =  \infty  .
\end{equation}
If the genus $p$ of the zeros of $F$ was zero, then the sum in Eq.\ \eqref{eq:GenusContradiction} would be finite, therefore $p$ must be one.
\QED
\end{proof}

\subsection{Bounds}
As a consequence of the product decomposition of the Jost function, we have the following lemma about the differentiability of the $S$-matrix.

\begin{lemma}
The $S$-matrix $S(k)$ is infinitely often differentiable for $k\geq0$.
\end{lemma}
\begin{proof}
Using the Jost function $F$, we introduce the auxiliary function
\begin{equation}
F_0(k)  \coloneqq  \frac {F(k)}  { F(0) + \lambda \dkp F(0) k }  .
\end{equation}
From the product representation \eqref{eq:Hadamard} of $F$ we get a product decomposition for $F_0$, from which we have that $F_0$ is entire \citep[Theorem\ 2.6.5, page 19]{Boas1954}, is such that $F_0(0)=1$, and has no zero on the real axis.

For the $S$-matrix we can write
\begin{equation}
S(k) = \frac {F(-k)} {F(k)}
	= (1-2\lambda) \frac { F_0(-k)}   { F_0(k)}  .
\end{equation}
Since $F_0$ is entire, it is infinitely often differentiable for any $k\geq0$.
Moreover, $F_0(k)$ is never zero for $k\geq0$, therefore $F_0(-k)/F_0(k)$ is analytic for $k\geq0$ and infinitely many derivatives of $S(k)$ exist for $k\geq0$.
\QED
\end{proof}

Taking derivatives of the formula
\begin{equation}
S(k) = \frac {F(-k)} {F(k)},
\end{equation}
and using the definition 
\begin{equation}
L(k)   \coloneqq  \Im \frac {\dkp F(k)}  {F(k)} ,
\end{equation}
one gets
\begin{align}
\dkp S  &=  -2 i \, S L
\label{eq:SDot}
\\
\ddkp S  &= -2 i \, S \dkp L  -4 S L^2 
\label{eq:SDotDot}
\\
\dddkp S  &=   -2 i \, S \ddkp L  - 12 S L \dkp L  + 8i\, S L^3  .
\label{eq:SDotDotDot}
\end{align}
Therefore, we need to bound $L$ and its derivatives.
For this purpose we use the Hadamard factorization of $F$.

\begin{lemma}\label{lem:Hadamard}
If $q\in\N$ then
\begin{align}
\Im \frac {\dkp F(k)}  {F(k)}    
&=R_V 
	+ \Im \frac {\dkp B(k)}  {B(k)}
	- \sum_{n=0}^\infty  
		\left(  \frac {\beta_n}  {| k-k_n |^2}  +  \frac {\beta_n}  {| k+ k_n |^2}   \right)  ,
\label{eq:L}
\\
\frac{d^q }{d k ^q}  \left(    \frac {\dkp F(k)}  {F(k)}  \right)  & =
	 (-1)^q  q!  \Biggl[ \frac {\lambda}  {k^{q+1}} 
		+ \sum_{n=0}^\infty  
		\left(  \frac {1} {(k-k_n)^{q+1}}
			+ \frac {1} {(k+\bar k_n)^{q+1}}  \right )  \Biggr]
		\nonumber\\
	&	\qquad+\frac{d^q }{d k ^q}  \left(    \frac {\dkp B(k)}  {B(k)}  \right)  ,
\label{eq:HigherLogDerF}
\end{align}
and 
\begin{align}
\frac {\dkp B(k)}  {B(k)}   &=
	\sum_{m=0}^{N-1}    \frac {1}  {k-i \eta_m}     
	+
	\sum_{l=0}^{N'-1}   \frac {1}  {k+i \kappa_l}  ,
\displaybreak[0]\\
\Im \frac {\dkp B(k)}  {B(k)}   &=
\sum_{m=0}^{N-1}    \frac {\eta_m}  {k^2 + \eta_m^2}     
	- \sum_{l=0}^{N'-1}   \frac {\kappa_l}  {k^2 + \kappa_l^2}  
\\
\frac{d^q }{d k ^q}  \left(    \frac {\dkp B(k)}  {B(k)}  \right)
  &=   (-1)^q  q!   
  	\Biggl(\sum_{m=0}^{N-1}  \frac {1} {(k-i\eta_m)^{q+1}} 
	\nonumber\\
	&\qquad+ \sum_{l=0}^{N'-1}  \frac {1} {(k+i\kappa_l)^{q+1}} \Biggr)  .
\end{align}
\end{lemma}
\begin{proof}
Equation \eqref{eq:L} is achieved simply taking the imaginary part of Eq.~\eqref{eq:LogDerF}; Eq.~\eqref{eq:HigherLogDerF} is a direct application of Lemma 3.1 of \citep[page 287]{Conway1978}  to the auxiliary function $F(k) / ( F(0) + \lambda \dkp F(0) k )$, that is of order one because of Lemma \ref{lem:ExpType}.
\QED
\end{proof}

\begin{lemma}\label{lem:LBounds}
Let 
\begin{equation}
L(k)   \coloneqq  \Im \frac {\dkp F(k)}  {F(k)} .
\end{equation}
For $K>0$ the following bounds hold:
\begin{align}
\sup_{k<K}  |L(k)|
	&\leq \frac {1}  {s_K}     
	\left[  1+ s_K (R_V+r_0)  \right]    ,
\label{eq:BoundL}
\\
\sup_{k<K}  | \dkp L(k) |  &\leq    
	\frac {2}  {s_K^2}  
	\left(  1  +  s_K^2 \frac {2 r_0 }  {\alpha}    \right)   ,
\label{eq:BoundLDot}
\\
\sup_{k<K}  |\ddkp L(k)|  &\leq    
	\frac {2} {s_K^3} 
	\left(   1  +   s_K^3 \frac {7 r_0 }  {\alpha}   \right)  .
\label{eq:BoundLDotDot}
\end{align}
\end{lemma}
\begin{proof}
Consider  the smallest non-negative integer  $\nu_K$ such that $\alpha_n \geq 2 K$ for all $n\geq\nu_K$, that implies $(\alpha_n-K)^2\geq\alpha_n^2/4$.
Then, from the expansion \eqref{eq:L} for $L$, follows
\begin{align}
\sup_{k<K}   |L(k)|
	&\leq  R_V 
	+ \frac {1}  {\eta}     
	+ \frac {1}  {\kappa}
	+  \sum_{n=0}^{\nu_K-1}  \frac {1}  {\beta_n}  
	+ \sum_{n=\nu_K}^\infty   \frac {\beta_n} {(\alpha_n-K)^2+\beta_n^2}
	+ \sum_{n=0}^\infty   \frac {\beta_n}  {| k_n |^2}  
\nonumber\\
	&\leq   R_V 
	+ \frac {1}  {s_K}     
	+ \sum_{n=\nu_K}^\infty   \frac {4\beta_n} {\alpha_n^2+\beta_n^2}
	+ \sum_{n=0}^\infty   \frac {\beta_n}  {| k_n |^2}  
\nonumber\\
	&\leq   R_V 
	+ \frac {1}  {s_K}     
	+  \sum_{n=0}^\infty   \frac {5\beta_n}  { \alpha_n^2 + \beta_n^2}  
\nonumber\\
	&= \frac {1}  {s_K}     
	\left[  1+ s_K (R_V+r_0)  \right]    .
\end{align}

Analogously, for $\dkp L$ the expansion given in Lemma \ref{lem:Hadamard} implies
\begin{align}
\dkp L(k)
&=  
\sum_{m=0}^{N-1}    \frac {2\eta_m k}  { (k^2 + \eta_m^2)^2 }     
	- \sum_{l=0}^{N'-1}    \frac {2 \kappa_l k}  { (k^2 + \kappa_l^2)^2 }
	\nonumber\\
	&\qquad- \sum_{n=0}^\infty  
		\left[  \frac {2\beta_n (k-\alpha_n)}  {| k-k_n |^4}  
			+  \frac {2\beta_n (k+\alpha_n)}  {| k+ k_n |^4}   \right]  .
\end{align}
Observe that, for $A,B\geq0$,
\begin{equation}
\frac {2AB}  {( A^2 + B^2 )^2} 
\leq
\frac {A^2 + B^2}  {( A^2 + B^2 )^2} 
=
\frac {1}  {A^2 + B^2} 
\leq
\frac {1}  {B^2} ;
\end{equation}
furthermore,
\begin{equation}
\frac {2\beta_n | k\pm\alpha_n | }  {| k\pm k_n |^4} 
\leq
\frac {2\beta_n | k\pm\alpha_n | }  {| k\pm k_n |^3  | k\pm\alpha_n |  } 
=
\frac {2\beta_n }  {| k\pm k_n |^3  }   ,
\end{equation}
therefore
\begin{equation}
| \dkp L(k) |
\leq
	\sum_{m=0}^{N-1}    \frac {1}  {\eta_m^2}     
	+ \sum_{l=0}^{N'-1}    \frac {1}  {\kappa_l^2}
	+ \sum_{n=0}^\infty   \frac {2\beta_n }  {| k - k_n |^3  }
	+ \sum_{n=0}^\infty  \frac {2\beta_n }  {| k + k_n |^3  }   ,
\end{equation}
and
\begin{align}
\sup_{k<K} | \dkp L(k) |
&\leq
	\frac {1}  {\eta^2}  
	+  \frac {1}  {\kappa^2} 
	+  \sum_{n=0}^{\nu_K-1}  \frac {2}  {\beta_n^2}  
	+ \sum_{n=\nu_K}^\infty   \frac {2\beta_n} {[(\alpha_n-K)^2+\beta_n^2 ]^{3/2}}
	+ \sum_{n=0}^\infty   \frac {2\beta_n}  {| k_n |^3}  
\nonumber\\
&\leq
	\frac {1}  {\eta^2}  
	+  \frac {1}  {\kappa^2} 
	+ 2 \left( \sum_{n=0}^{\nu_K-1}  \frac {1}  {\beta_n}  \right)^2  
	+ \sum_{n=\nu_K}^\infty   \frac {16\beta_n} {[\alpha_n^2+\beta_n^2 ]^{3/2}}
	+ \sum_{n=0}^\infty   \frac {2\beta_n}  {| k_n |^3}  
\nonumber\\
&\leq
	\frac {1}  {\eta^2}  
	+  \frac {1}  {\kappa^2} 
	+ 2 \left( \sum_{n=0}^{\nu_K-1}  \frac {1}  {\beta_n}  \right)^2  
	+ \sum_{n=0}^\infty   \frac {18\beta_n}  {| k_n |^3}  
\nonumber\\
&\leq
	\frac {2}  {s_K^2}  
	+  \frac {18}  {5\alpha} \sum_{n=0}^\infty   \frac {5\beta_n}  {| k_n |^2}  
\nonumber\\
&\leq
	\frac {2}  {s_K^2}  
	\left(  1  +  s_K^2 \frac {2 r_0 }  {\alpha}    \right)  .
\end{align}

Moreover, for $\ddkp L$ again from the expansion given in Lemma \ref{lem:Hadamard} we get
\begin{align}
\ddkp L(k)
&=  
-2\sum_{m=0}^{N-1}    \frac {\eta_m^3- \eta_m k^2}  { (k^2 + \eta_m^2)^3 }     
	+2 \sum_{l=0}^{N'-1}    \frac {\kappa_l^3 - \kappa_l k^2}  { (k^2 + \kappa_l^2)^3 }
\nonumber\\
	&\qquad+2 \sum_{n=0}^\infty  
		\left[  \frac {\beta_n^3  -\beta_n (k-\alpha_n)^2}  {| k-k_n |^6}  
			+   \frac {\beta_n^3  -\beta_n (k+\alpha_n)^2}  {| k+k_n |^6}    \right]  .
\end{align}
For $A,B\geq0$,
\begin{equation}
\frac {B^3-A^2 B}  {( A^2 + B^2 )^3}
=
\frac {B(B^2-A^2)}  {( A^2 + B^2 )^3} 
\leq
\frac {B(B^2+A^2)}  {( A^2 + B^2 )^3} 
=
\frac {B}  {( A^2 + B^2 )^2} 
\leq
\frac {1}  {B^3} .
\end{equation}
Furthermore,
\begin{equation}
\frac {| \beta_n^3  -\beta_n (k\pm\alpha_n)^2 | }  {| k\pm k_n |^6}  
\leq
\frac { \beta_n }  {| k\pm k_n |^2}  
\frac { \beta_n^2  + (k\pm\alpha_n)^2  }  {| k\pm k_n |^4}  
=
\frac { \beta_n }  {| k\pm k_n |^4}  
\end{equation}
therefore
\begin{equation}
|\ddkp L(k)|
\leq
\frac {2} {\eta^3}  +\frac {2} {\kappa^3}  
+2 \sum_{n=0}^\infty  \frac { \beta_n }  {| k- k_n |^4}  
+\frac {2} {\alpha^2} \sum_{n=0}^\infty  \frac { \beta_n }  {| k_n |^2}  
\end{equation}
and
\begin{align}
\sup_{k<K}  |\ddkp L(k)|
&\leq
	\frac {2} {\eta^3}  +\frac {2} {\kappa^3}  
	+ \sum_{n=0}^{\nu_K-1}   \frac {2} {\beta_n^3} 
	+ 2 \sum_{n=\nu_K}^\infty  \frac { 16 \beta_n }  {| k_n |^4}  
	+ \frac {2 r_0} {5 \alpha^2}  
\nonumber\\
&\leq
	\frac {2} {s_K^3} 
	\left(   1  +   s_K^3 \frac {7 r_0 }  {\alpha}   \right)
.
\end{align}
\QED
\end{proof}

We can now finally prove Theorem~\ref{th:SBoundsK}.

\begin{proof}[of Theorem~\ref{th:SBoundsK}]
From Eqs.~\eqref{eq:SDot}-\eqref{eq:SDotDotDot} we get the bounds
\begin{align}
| \dkp S |  &\leq  2 | L | ,
\label{eq:SDotLeqL}\\
| \ddkp S |  &\leq 2 |\dkp L|  +4  |L|^2 ,
\label{eq:SDotDotLeqLDot}\\
| \dddkp S |  &\leq  2 |\ddkp L| + 12 |L| \, |\dkp L|  + 8 |L|^3  .
\label{eq:SDotDotDotLeqLDotDot}
\end{align}
Then, Lemma \ref{lem:LBounds} implies the result.
\QED
\end{proof}

\section{Proof of Theorem~\ref{th:GlobalSBounds}}\label{sec:SMatrixSup}
For $s$ given in Definition \ref{def:SKAndKTildeAndS}, and $n\in\{1,2,3\}$, we want to prove that
\begin{align}
\|S^{(n)}\|_\infty  &\leq  C_{n}s^{-n},
\end{align}
We can not simply take the limit $K\to\infty$ in Theorem~\ref{th:SBoundsK}, indeed $s_K^{-1}\to\infty$ in this limit.
To see this, consider
\begin{equation}
\lim_{K\to\infty}  s_K^{-1}  =  
		\frac {1} {\eta}  +  \frac {1} {\kappa}  +  \sum_{n=0}^\infty   \frac {1}  {\beta_n} ,
\end{equation}
which does not converge due to the fact that $\beta_n=o(n)$ as $n\to\infty$, as shown in the next lemma  \citep[see also][page 362]{Newton1966}.
As a consequence, we can use Theorem~\ref{th:SBoundsK} only up to a certain value $\tilde K$, while for $k>\tilde K$ we will devise a different strategy.

\begin{lemma}
Let $k_n=\alpha_n-i\beta_n$ be the zeros of the Jost function $F$ such that $\alpha_n,\beta_n>0$.
Then $\lim_{n\to\infty} {|k_n|}/{n}$ and $\lim_{n\to\infty} {\alpha_n}/{n}$ exist and are such that
\begin{equation}\label{eq:KNOverN}
0 <  \lim_{n\to\infty} \frac{|k_n|}{n} < \infty,
\qquad
0 <  \lim_{n\to\infty} \frac{\alpha_n}{n} < \infty;
\end{equation}
moreover,
\begin{equation}
\beta_n = o(n),
\qquad
\text{as } n\to\infty.
\end{equation}
\end{lemma}
\begin{proof}
Let $n(|k|)$ be the number of zeros of $F$ within a ball of radius $|k|$, and let $\tilde n(|k|)$ the number of resonances in the same ball.
From Lemmas~\ref{lem:R0} and~\ref{lem:GenusAndHadamard}  we know that the following limit exists and that
\begin{equation} 
0 < \lim_{|k|\to\infty} \frac {n(|k|)}  {|k|}  < \infty .
\end{equation}
Clearly, 
\begin{equation}
\lim_{|k|\to\infty} \frac {\tilde n(|k|)}  {|k|}  =  \lim_{|k|\to\infty} \frac {n(|k|)}  {|k|}  ,
\end{equation}
therefore for every sequence $\{\lambda_n\}_{n\in\N^0}$ such that $\lambda_n>0$ and $\lambda_n\to\infty$ as $n\to\infty$,
\begin{equation}
0 < \lim_{n\to\infty} \frac {\tilde n(\lambda_n)}  {\lambda_n}  
	=  \lim_{|k|\to\infty} \frac {\tilde n(|k|)}  {|k|}   < \infty .
\end{equation}
In particular, for $\lambda_n=|k_n|$, we get $\tilde n(|k_n|)=2n$ (the $2$ is due to the symmetry of the resonances), that implies that the limit
\begin{equation}
 \lim_{n\to\infty} \frac{|k_n|}{n} 
\end{equation}
exists and satisfies Eq.~\eqref{eq:KNOverN}.
As a consequence, there is a constant $c\in\R^+$ such that
\begin{equation}\label{eq:KNInf}
|k_n|^2 =  \alpha_n^2 + \beta_n^2 \sim c n^2,
\quad\text{as } n\to\infty,
\end{equation}
therefore
\begin{equation}
\frac{\beta_n}{ \alpha_n^2 + \beta_n^2 } 
\sim 
\frac{\beta_n}{ cn^2 } ,
\quad\text{as } n\to\infty  .
\end{equation}
Together with Eq.~\eqref{eq:r0} of Lemma \ref{lem:R0}, that implies 
\begin{equation}
\frac{\beta_n}{ \alpha_n^2 + \beta_n^2 } 
= o(n^{-1}) ,
\quad\text{as } n\to\infty  ,
\end{equation}
and we get $\beta_n = o(n)$.
Equation~\eqref{eq:KNInf} then implies $\alpha_n\sim \sqrt c \, n$ as $n\to\infty$.
\QED
\end{proof}

For big values of $k$ the product form of the Jost function $F$ is not of help, and we instead use the relation $F(k) = f(k,0)$, where $f$ are the irregular eigenfunctions defined by Eq.\ \eqref{eq:fboundary}.
The eigenfunctions $f$ satisfy the following well known bound (see \citealp[Theorem~XI.57]{RS3}; note that in this reference $\eta(r,-k)$  is equal to our $f(k,r)$, cf.~\citealp[Theorem  XI.57, page 138]{RS3}).
\begin{lemma}\label{lem:f}
Let $k\in[0,\infty)$ and
\begin{align}
  Q_k(r)\coloneqq\int_r^\infty\frac{2r'}{1+kr'}|V(r')|\,dr',
\end{align}
then
\begin{align}
	|f(k,r)|\leq e^{Q_k(r)} .
\end{align}
\end{lemma}

We need to handle up to the third derivative of the Jost function. 
Therefore, we will extend Lemma~\ref{lem:f} to $\partial_k^{n}f(k,r)$ with $n\leq3$, using a similar proof as that of Theorem~XI.57 in \citep{RS3}. 
Starting point is the Lippmann-Schwinger equation
\begin{align}\label{eq:Lipp}
  f(k,r)=e^{ikr}-\int_r^\infty\frac{1}{k}\sin(k(r-r'))V(r')f(k,r')\,dr'.
\end{align}
We want to expand $\partial_k^{n}f(k,r)$ in a Born series, so that once we have a global bound in $r$ for every summand, we get a global bound in $r$ for $\partial_k^{n}f(k,r)$ (assuming the series converges). But this will not work because the first summand of the Born series for $\dkp f$ is $ire^{ikr}$ for which there is no global bound in $r.$ We can solve the problem by looking at 
\begin{align}
	y(k,r)\coloneqq e^{-ikr}f(k,r)
\end{align}	
rather than $f(k,r).$

\begin{lemma}\label{lem:y'}
Let $k\in[0,\infty)$ and $q_k\coloneqq e^{Q_k(0)}\|rV(r)\|_1$,
then
  \begin{align}
	|\dkp y(k,r)|\leq3 \frac{e^{Q_k(r)}}{k}q_k  .
  \end{align}
\end{lemma}

\begin{proof}
 From Eq.~\eqref{eq:Lipp} we see that $y$ satisfies
  \begin{align}
	y(k,r)
	&=1-\int_r^\infty e^{-ik(r-r')}\frac{1}{k}\sin(k(r-r'))V(r')y(k,r')\,dr'\label{eq:lem_y'4}\\
	&\eqqcolon 1-\int_r^\infty g(k,r,r')V(r')y(k,r')\,dr',
  \end{align}
  so that
  \begin{align}
	\dkp y(k,r)
	&=-\int_r^\infty \dkp g(k,r,r')V(r')y(k,r')\,dr'
	\nonumber\\	
	&\qquad	-\int_r^\infty g(k,r,r')V(r')\dkp y(k,r')\,dr'\\
	&\eqqcolon x(k,r)-\int_r^\infty g(k,r,r')V(r')\dkp y(k,r')\,dr'.
  \end{align}
We want to use the inequality 
  \begin{align}\label{eq:lem_y'3}
	\left|\frac{1}{k}\sin\left(k(r-r')\right)\right|\leq\frac{2r'}{1+kr'},
  \end{align}
to prove which consider the following.  
Observe that for $x>0$ 
\begin{align}
	\left|\frac{\sin x}{x}\right|
	\leq\frac{2}{1+x},
	\label{eq:SimpleSinc}
\end{align}
indeed 
\begin{equation}
	\left|\frac{\sin x}{x}\right| (1+x) =	\left|\frac{\sin x}{x}\right|  + |\sin x| \leq 2.
\end{equation}
Choosing $x=k(r'-r)$, with $r'>r$, we get
  \begin{align}
  	\left|\frac{1}{k}\sin\left(k(r-r')\right)\right|
		\leq\frac{2(r'-r)}{1+k(r'-r)},
  \end{align}
that implies \eqref{eq:lem_y'3} because the function $2X/(1+kX)$ is monotonically increasing with $X>0$ for all $k$.
From Eq.~\eqref{eq:lem_y'3} it is then easy to verify
  \begin{align}\label{eq:lem_y'5}
	|\dkp g(k,r,r')|\leq\frac{3r'}{k}.
  \end{align}
  Together with Lemma~\ref{lem:f} and $|y|=|f|$ we thereby obtain
  \begin{align}\label{eq:lem_y'1}
	|x(k,r)|\leq\frac{3}{k}e^{Q_k(0)}\|r'V(r')\|_1.
  \end{align}
  Now, we expand $\dkp y$ in a Born series $\dkp y=\sum_{n=0}^\infty \dkp y_n$ with
  \begin{align}
	\dkp y_0&=x(k,r)\\
	\dkp y_{n+1}&=-\int_r^\infty g(k,r,r')V(r')\dkp y_n(k,r')\,dr'
  \end{align}
  and prove by induction that
  \begin{align}\label{eq:lem_y'2}
	|\dkp y_n(k,r)|\leq\frac{3}{k}e^{Q_k(0)}\|r'V(r')\|_1\frac{Q_k^n(r)}{n!}.
  \end{align}
  Due to Eq.~\eqref{eq:lem_y'1} the induction start is immediately evident. For the induction step assume that Eq.~\eqref{eq:lem_y'2} holds, then
  \begin{align}
	|\dkp y_{n+1}(k,r)|
	&\leq\int_r^\infty|g(k,r,r')V(r')||\dkp y_{n}(k,r)|\,dr'\\
	&\leq\frac{3}{k}e^{Q_k(0)}\|r''V(r'')\|_1\int_r^\infty\frac{2r'}{1+kr'}|V(r')|\frac{Q_k^n(r')}{n!}\,dr',
  \end{align}
  where we have used Eq.~\eqref{eq:lem_y'3}. From the definition of $Q_k(r)$ given in Lemma~\ref{lem:f} it is evident that
  \begin{align}
	\frac{d}{dr'}\frac{Q_k^{n+1}(r')}{(n+1)!}=\frac{2r'}{1+kr'}|V(r')|\frac{Q_k^n(r')}{n!},
  \end{align}
  so that
   \begin{align}
	|\dkp y_{n+1}(k,r)|
	&\leq\frac{3}{k}e^{Q_k(0)}\|r''V(r'')\|_1\int_r^\infty\frac{d}{dr'}\frac{Q_k^{n+1}(r')}{(n+1)!}\,dr'
	\\
	&=\frac{3}{k}e^{Q_k(0)}\|r'V(r')\|_1\frac{Q_k^{n+1}(r)}{(n+1)!}.
  \end{align}
  Hence, Eq.~\eqref{eq:lem_y'2} is proven. Plugging this bound into
  \begin{align}
	|\dkp y(k,r)|\leq\sum_{n=0}^\infty|\dkp y_{n}(k,r)|,
  \end{align}
  we obtain the assertion of the Lemma.
  \QED
\end{proof}

\begin{lemma}\label{lem:y''}
Let $k\in[0,\infty)$ and $q_k\coloneqq e^{Q_k(0)}\|rV(r)\|_1$, then
  \begin{align}
	|\ddkp y(k,r)|
	\leq 6\frac{e^{Q_k(r)}}{k}\left[\frac{1}{k}(1+3q_k)
	+R_V\right]q_k  .
  \end{align}
\end{lemma}

\begin{proof}
We proceed in the same way as in the proof of Lemma~\ref{lem:y'}.
From Eq.~\eqref{eq:lem_y'4} we get
\begin{align}
  \ddkp y(k,r)
  &=-\int_r^\infty\ddkp g(k,r,r')V(r')y(k,r')\,dr'
  \nonumber\\
  &\quad-2\int_r^\infty\dkp g(k,r,r')V(r')\dkp y(k,r')\,dr'\nonumber\\
  &\quad-\int_r^\infty g(k,r,r')V(r')\ddkp y(k,r')\,dr'\\
  &\eqqcolon x(k,r)-\int_r^\infty g(k,r,r')V(r')\ddkp y(k,r')\,dr'
\end{align}
Using Eq.~\eqref{eq:lem_y'3} and Eq.~\eqref{eq:lem_y'5} it is straightforward to check that
\begin{align}
  |\ddkp g(k,r,r')|\leq\frac{6r'}{k^2}(1+kr').
\end{align}
Using this, inequality~\eqref{eq:lem_y'5} for $\dkp g$, and the bound on $y$ obtained from Lemma~\ref{lem:f}, we get
\begin{align}
  |x(k,r)|
  &\leq\frac{6}{k}\Biggl(\frac{3}{k}e^{2Q_k(0)}\|r'V(r')\|_1^2
  +\frac{1}{k}e^{Q_k(0)}\|r'V(r')\|_1
  \nonumber\\
  &\qquad+e^{Q_k(0)}\|r'^2V(r')\|_1\Biggr)\\
  &\leq\frac{6}{k}\left(\frac{3}{k}q_k
  +\frac{1}{k}
  +R_V\right)q_k,
\end{align}
where we have used that $\|r^2V(r)\|_1\geq R_V\|rV(r)\|_1.$ Now, we can proceed in exactly the same way as in the proof of Lemma~\ref{lem:y'} to arrive at the assertion.
\QED
\end{proof}

\begin{lemma}\label{lem:y'''}
  Let $k\in[0,\infty)$ and  $q_k\coloneqq e^{Q_k(0)}\|rV(r)\|_1$, then
  \begin{align}
	|\dddkp y(k,r)|
	\leq 18\frac{e^{Q_k(r)}}{k}\left[\frac{3}{k}(1+3q_k)+R_V\right]^2q_k .
  \end{align}
\end{lemma}
We omit the proof of this Lemma because it runs along the same lines as the proof of Lemma~\ref{lem:y''}.

We can now prove some bounds on the imaginary part of the logarithmic derivative of the Jost function $F$.
Note that the bounds that we got for $|y^{(n)}|$  all depend on powers of $e^{Q_k(0)}\leq e^{2\|V(r)\|_1/k}$ (cf.\ the definition of $Q_k(r)$ in Lemma~\ref{lem:f}), therefore they will be useful only for $k\approx\|V\|_1$ or bigger.
\begin{lemma}\label{lem:LBoundsBigK}
Let  
\begin{equation}
\tilde K\coloneqq 6\|V\|_1,
\quad
L(k)   \coloneqq  \Im \frac {\dkp F(k)}  {F(k)} ,
\quad\text{and}\quad
q\coloneqq\frac{1}{2\|V\|_1}+6R_V .
\end{equation}
Then, for $k\geq \tilde K$,
  \begin{align}
	|L(k)|&
	\leq 2 R_V\frac{\tilde K}{k},
	\label{eq:L_largek}\\
	|\dkp L(k)|
	&\leq 4 R_V\frac{\tilde K}{k} q ,
	\label{eq:L'_largek}\\
	|\ddkp L(k)|&\leq 12 R_V\frac{\tilde K}{k} q^2   ,
	\label{eq:L''_largek}  
\intertext{and}
| \dkp S |  &\leq  4 R_V\frac{\tilde K}{k},
\\
| \ddkp S |  &\leq 8 R_V\frac{\tilde K}{k} q  
		+16 R_V^2 \frac{\tilde K^2}{k^2} ,
\\
| \dddkp S |  &\leq  24 R_V\frac{\tilde K}{k} q^2
		+ 96  R_V^2 \frac{\tilde K^2}{k^2}  q
		+ 64 R_V^3  \frac{\tilde K^3}{k^3}  .
\end{align}
\end{lemma}

\begin{proof}
  Since $|L|\leq |\dkp F|/|F|,$ we need an upper bound on $|\dkp F|$ and a lower bound on $|F|.$ For the upper bound we observe that $\dkp F(k)=\dkp f(k,0)=\dkp y(k,0)$ and hence due to Lemma~\ref{lem:y'}
  \begin{align}
	|\dkp F(k)|\leq3 \frac{e^{Q_k(0)}}{k}q_k=3 \frac{e^{2Q_k(0)}}{k}\|rV(r)\|_1.
  \end{align}
  Upon using
  \begin{align}\label{eq:lemL0}
	Q_k(0)=\int_0^\infty \frac{2r'}{1+kr'}|V(r')|\,dr'\leq\frac{2}{k}\|V\|_1\leq\frac{2}{\tilde K}\|V\|_1=\frac{1}{3},
  \end{align}
  $\|rV(r)\|_1\leq R_V\|V\|_1$ and $e^{2/3}<2,$ we obtain
  \begin{align}\label{eq:lemL2}
	|\dkp F(k)|\leq3 \frac{e^{\frac{2}{3}}}{k}R_V\|V\|_1\leq6R_V\frac{\|V\|_1}{k}.
  \end{align}
  We derive the lower bound for $|F(k)|$ from the Lippmann-Schwinger equation~\eqref{eq:Lipp}. For $F(k)=f(k,0)$ it reads
  \begin{align}
	F(k)=1+\int_0^\infty\frac{1}{k}\sin(kr')V(r')f(k,r').
  \end{align}
  With the help of the bound on $f(k,r)$ from Lemma~\ref{lem:f} and the bound on the sinus term in Eq.~\eqref{eq:lem_y'3} we obtain for the integral
  \begin{align}
	\left|\int_0^\infty\frac{1}{k}\sin(kr')V(r')f(k,r')\right|
	&\leq e^{Q_k(0)}\int_0^\infty|V(r')|\frac{2r'}{1+kr'}\,dr'
	\\
	&\leq e^{2\frac{\|V\|_1}{k}}2\frac{\|V\|_1}{k}.
  \end{align}
  Hence, for $k\geq\tilde K$ we have
  \begin{align}\label{eq:lemL3}
	|F(k)|
	&\geq\left|1-\left|\int_0^\infty\frac{1}{k}\sin(kr')V(r')f(k,r')\right|\right|
	\geq 1-e^{\frac{1}{3}}\frac{1}{3}
	\geq\frac{1}{2}.
  \end{align}
  This and the upper bound on $|\dkp F|$ imply the bound on $L$ in Eq.~\eqref{eq:L_largek}.
  
  To obtain the bound for $|\dkp L|$ in Eq.~\eqref{eq:L'_largek}, observe that with the help of the bounds for $|\dkp F|$ in Eq.~\eqref{eq:lemL2} and for $|F|$ in Eq.~\eqref{eq:lemL3} we have
  \begin{align}\label{eq:lemL1}
	|\dkp L|
	\leq\frac{|\ddkp F|}{|F|}+\frac{|\dkp F|^2}{|F|^2}
	\leq2|\ddkp F|+12^2R_V^2\frac{\|V\|_1^2}{k^2}.
  \end{align}
  We get an upper bound for $|\ddkp F|$ by using Lemma~\ref{lem:y''}, Eq.~\eqref{eq:lemL0} and $\|rV(r)\|_1\leq R_V\|V\|_1$ as follows
  \begin{align}
	|\ddkp F(k)|
	=|\ddkp y(k,0)|
	&\leq 6\frac{e^{Q_k(0)}}{k}\|rV(r)\|_1\left[\frac{1}{k}(1+3e^{Q_k(0)}\|rV(r)\|_1)+R_V\right]\\
	&\leq 12R_V\frac{\|V\|_1}{k}\left[\frac{1}{k}(1+6R_V\|V\|_1)+R_V\right]\\
	&\leq 12R_V\frac{\|V\|_1}{k}\left[\frac{1}{6\|V\|_1}+2R_V\right].\label{eq:lemL4}
  \end{align}
Plugging this into Eq.~\eqref{eq:lemL1} yields
\begin{align}
	|\dkp L|\leq 24R_V\frac{\|V\|_1}{k}\left[\frac{1}{6\|V\|_1}+2R_V+6R_V\frac{\|V\|_1}{k}\right] .
  \end{align}
For $k\geq\tilde K=6\|V\|_1$, we get 
\begin{align}
	|\dkp L|\leq 24R_V\frac{\|V\|_1}{k}\left[\frac{1}{6\|V\|_1}+3R_V\right] ,
\end{align}
therefore Eq.~\eqref{eq:L'_largek} is verified.

We now prove Eq.~\eqref{eq:L''_largek}. Using the bounds for $|\ddkp F|$ in Eq.~\eqref{eq:lemL4}, for $|\dkp F|$ in Eq.~\eqref{eq:lemL2} and for $|F|$ in Eq.~\eqref{eq:lemL3}, we get
  \begin{align}\label{eq:lemL5}
	|\ddkp L|
	&\leq2|\dddkp F|+6\cdot 12^2R_V^2\frac{\|V\|_1^2}{k^2}\left[\frac{1}{6\|V\|_1}+2R_V\right]+2\cdot 12^3R_V^3\frac{\|V\|_1^3}{k^3}.
  \end{align}
  Lemma~\ref{lem:y'''}, Eq.~\eqref{eq:lemL0} and $\|rV(r)\|_1\leq R_V\|V\|_1$ yield
  \begin{align}
	|\dddkp F(k)|
	=|\dddkp y(k,0)|
	&\leq18\frac{e^{2Q_k(0)}}{k}\|rV(r)\|_1\left[\frac{3}{k}(1+3e^{Q_k(0)}\|rV(r)\|_1)+R_V\right]^2\\
	&\leq3\cdot12R_V\frac{\|V\|_1}{k}\left[\frac{1}{2\|V\|_1}+4R_V\right]^2
  \end{align}
  along the same lines as before. Plugging this into Eq.~\eqref{eq:lemL5} we obtain
  \begin{align}
	|\ddkp L|
	&\leq24R_V\frac{\|V\|_1}{k}\Biggl[3\left[\frac{1}{2\|V\|_1}+4R_V\right]^2	
	\nonumber\\
	&\qquad+36R_V\frac{\|V\|_1}{k}\left[\frac{1}{6\|V\|_1}+2R_V\right]
	+12^2R_V^2\frac{\|V\|_1^2}{k^2}\Biggr]\\
	&\leq24R_V\frac{\|V\|_1}{k}\Biggl[3\left[\frac{1}{2\|V\|_1}+4R_V\right]^2
	\nonumber\\
	&\qquad+6R_V\left[\frac{1}{2\|V\|_1}+4R_V\right]+4R_V^2\Biggr]\\
	&\leq72R_V\frac{\|V\|_1}{k}\left[\frac{1}{2\|V\|_1}+6R_V\right]^2,
  \end{align}
  which finishes the proof of Eq.~\eqref{eq:L''_largek}.

From the bounds on $L$ and on its derivatives we get the analogous bounds on the derivatives of the $S$-matrix,  by using the inequalities \eqref{eq:SDotLeqL}, \eqref{eq:SDotDotLeqLDot}, and \eqref{eq:SDotDotDotLeqLDotDot}, that we repeat here:
\begin{align}
| \dkp S |  &\leq  2 | L | ,
\\
| \ddkp S |  &\leq 2 |\dkp L|  +4  |L|^2 ,
\\
| \dddkp S |  &\leq  2 |\ddkp L| + 12 |L| \, |\dkp L|  + 8 |L|^3  .
\end{align}
Substitution of the bounds \eqref{eq:L_largek}-\eqref{eq:L''_largek} completes the proof.
\QED
\end{proof}

We now combine the bounds that we got for $k\geq \tilde K$ with those from Theorem~\ref{th:SBoundsK}, that we will use for $k\leq\tilde K$, to prove Theorem~\ref{th:GlobalSBounds}.

\begin{proof}[of Theorem~\ref{th:GlobalSBounds}]
At first, we substitute $k=\tilde K$ in the bounds of Lemma~\ref{lem:LBoundsBigK}, getting that, for $k\geq \tilde K$,
\begin{align}
| \dkp S |  &\leq  4 R_V ,
\\
| \ddkp S |  &\leq 8 R_V  q  + 16 R_V^2  ,
\\
| \dddkp S |  &\leq  24 R_V q^2
		+ 96  R_V^2 q
		+ 64 R_V^3   
\nonumber
\\
	  &\leq  8 R_V ( 9 q^2  +  14  R_V^2 ) .
\end{align}
To get inequalities valid for any $k\geq0$ we sum the latter bounds and those from theorem  \ref{th:SBoundsK}, choosing there $K=\tilde K$.
In this way we have
\begin{align}
\|\dkp S\|_\infty  
&\leq 
	\frac {2}  {s}     \left[  1+ s (3R_V+r_0)  \right]
\\
\|\ddkp S\|_\infty  
&\leq 
	\frac {4}  {s^2}  \left\{  3 + 2 s^2  \left[ \frac{r_0}{\alpha}  +  (3 R_V+r_0)^2 +R_V q \right]  \right\}
\\
\|\dddkp S\|_\infty  
&\leq 
	\frac {4}  {s^3}  \Biggl\{  15  + 6 s (R_V+r_0) + 12 s^2 \frac{r_0}{\alpha} 
\nonumber\\		
&
		+ s^3  \left[ \frac{7 r_0}{\alpha} + \frac{12 r_0}{\alpha} (R_V+r_0) + 8 (3R_V+r_0)^3  +18 R_V q^2 \right]  \Biggr\}   ,
\end{align}
that is Theorem~\ref{th:GlobalSBounds}.
\QED
\end{proof}

\newpage

\section{Proof of Theorem~\ref{thm:main_ac}}\label{sec:proof_ac}

We want to find an upper bound for
\begin{align}
   \|\1_Re^{-iHt}P_{ac}\psi\|_2^2=\|P_{ac}\1_Re^{-iHt}P_{ac}\psi\|_2^2+\|P_e\1_Re^{-iHt}P_{ac}\psi\|_2^2.
\end{align}
Consider $\|P_{ac}\1_Re^{-iHt}P_{ac}\psi\|_2^2$ first. Using the expansion in generalized eigenfunctions, we get
\begin{align}
  e^{-iHt}P_{ac}\psi(r)=\int_0^\infty\hat\psi(k)\psi^+(k,r)e^{-ik^2t}\,dk.
\end{align}
Due to Eq.~\eqref{eq:fboundary},  $\psi^+$ is known for $r\geq R_V$,  but not for $r<R_V,$ hence this expression can not be used directly. 
However, in the following Lemma we obtain an expression for $\|P_{ac}\1_Re^{-iHt}P_{ac}\psi\|_2$ that does not need explicit knowledge of how the generalized eigenfunctions behave for $r<R_V$.
It is inspired by~\cite{Skibsted86}.

\begin{lemma}\label{lem:IntKernAc}
  Let $\psi\in\mathcal D(H),$ $R\geq R_V$ and
  \begin{align}\label{eq:IntKernAc0}
	Z_{ac}(k,k')&\coloneqq\frac{W(\bar\psi^+(k',R),\psi^+(k,R))}{k'^2-k^2}.
  \end{align}
  Then
  \begin{align}
	 &\|P_{ac}\1_Re^{-iHt}P_{ac}\psi\|_2^2=\left\|\int_0^\infty Z_{ac}(k,\cdot)\hat\psi(k)e^{-ik^2t}dk\right\|_2^2\label{eq:IntKernAc}
  \end{align}
  and
  \begin{align}
	Z_{ac}(k,k')&=\frac{i}{4}\Biggl[\frac{e^{i(k+k')R}S(k)-e^{-i(k+k')R}\bar S(k')}{k+k'}
	\nonumber\\
	&\qquad-\frac{e^{i(k-k')R}\bar S(k')S(k)-e^{-i(k-k')R}}{k-k'}\Biggr].
	\label{eq:IntKernAc2}
  \end{align}
\end{lemma}

\begin{proof}
Recall that the generalized Fourier transform is
\begin{align}
\mathcal F \psi(k)
    	=\int_0^\infty\psi(r)\bar\psi^+(k,r)\,dr,
\end{align}
and that it is  a unitary operator on the subspace of absolute continuity of the Hamiltonian.
Moreover, $P_{ac}=\mathcal F^{-1}\mathcal F$ (see \eqref{eq:Fourier}).
Therefore we can write
\begin{align}
&\|P_{ac}\1_Re^{-iHt}P_{ac}\psi\|_2
=\|P_{ac}\1_RP_{ac}e^{-iHt}\psi\|_2
\\
&\qquad=\|\mathcal F P_{ac}\1_RP_{ac}e^{-iHt}\psi\|_2
=\|\mathcal F\1_R\mathcal F^{-1}\mathcal Fe^{-iHt}\psi\|_2.
\label{eq:NormFOneF}
\end{align}
Now,
\begin{align}
	&\left(\mathcal F\1_R\mathcal F^{-1}\mathcal Fe^{-iHt}\psi\right)(k')
	\\
	&\qquad=\int_0^\infty dr\,\bar\psi^+(k',r)\1_R(r)\int_0^\infty dk\,e^{-ik^2t}\hat\psi(k)\psi^+(k,r)\\
	&\qquad=\int_0^\infty dk\,e^{-ik^2t}\hat\psi(k)\int_0^\infty dr\,\1_R(r)\psi^+(k,r)\bar\psi^+(k',r),
\end{align}
so that, the integral kernel of $\mathcal F\1_R\mathcal F^{-1}$ reads
\begin{equation}
\int_0^\infty dr\,\1_R(r)\psi^+(k,r)\bar\psi^+(k',r).
\end{equation}
This integral kernel can be expressed in terms of $\psi^+(k',R)$ with $R\geq R_V$. 
Observing
  \begin{align}
	\frac{d}{dr}W(\bar\psi^+(k',r),\psi^+(k,r))=(k'^2-k^2)\psi^+(k,r)\bar\psi^+(k',r)
  \end{align}
  and using $\psi^+(k,0)=0,$ we get upon integration
  \begin{align}
	&\int_0^\infty dr\,\1_R(r)\psi^+(k,r)\bar\psi^+(k',r)
	\\
	&\qquad=\frac{W(\bar\psi^+(k',R),\psi^+(k,R))}{k'^2-k^2}
	=Z_{ac}(k,k')
\end{align}
and therefore
\begin{align}
	\left(\mathcal F\1_R\mathcal F^{-1}\mathcal Fe^{-iHt}\psi\right)(k')
	=\int_0^\infty dk\,e^{-ik^2t}\hat\psi(k)Z_{ac}(k,k'),
  \end{align}
  which when plugged into Eq.~\eqref{eq:NormFOneF} proves Eq.~\eqref{eq:IntKernAc}.

To prove Eq.~\eqref{eq:IntKernAc2} we use $\psi^+(k,R)=\frac{1}{2i}(S(k)e^{ikR}-e^{-ikR})$, which is a direct consequence of Eqs.~\eqref{eq:PsiPlus} and \eqref{eq:phiRBiggerRs}.
With this we get
  \begin{align}
	&W(\bar\psi^+(k',R),\psi^+(k,R))
	\\
	&\qquad=\frac{i}{4}(k+k')\left(\bar S(k')S(k)e^{i(k-k')R}-e^{-i(k-k')R}\right)\nonumber\\
	&\qquad\qquad-\frac{i}{4}(k-k')\left(S(k)e^{i(k+k')R}-\bar S(k')e^{-i(k+k')R}\right).
  \end{align}
  Plugging this into Eq.~\eqref{eq:IntKernAc0} finishes the proof.
  \QED
\end{proof}

To extract a time decaying factor from the $k$-integral in Eq.~\eqref{eq:IntKernAc}, we  employ the method of stationary phase.
We use two integrations by parts because one is not enough to obtain the well known $t^{-3}$-factor as leading order. 
Observe that
\begin{equation}\label{eq:exp}
  e^{-ik^2t}=-\frac{\partial_{k}^2e^{-ik^2t}}{2t(2tk^2+i)},
\end{equation}
which when plugged into Eq.~\eqref{eq:IntKernAc} yields upon integration by parts
\begin{align}
  &\|P_{ac}\1_Re^{-iHt}P_{ac}\psi\|_2^2
\\
  &\qquad=\int_0^\infty\left|\int_0^\infty Z_{ac}(k,k')\hat\psi(k)\frac{\partial_{k}^2e^{-ik^2t}}{2t(2tk^2+i)}\,dk\right|^2\,dk'\\
  &\qquad=\int_0^\infty\left|\int_0^\infty \partial_k^2\left[Z_{ac}(k,k')\frac{\hat\psi(k)}{2t(2tk^2+i)}\right]e^{-ik^2t}\,dk\right|^2\,dk'.\label{eq:StatPhase}
\end{align}
The boundary terms vanish because of Lemma \ref{lem:ZacInt}, for the  proof of which we need the auxiliary Lemmas \ref{lem:psihat0}-\ref{lem:Z0ac}.
They provide more knowledge about $Z_{ac}(k,k')$ and $\hat\psi(k)$ as well as their derivatives, especially in the limits $k\to0$ and $k\to\infty.$

\begin{lemma}\label{lem:psihat0}
  Let $K>0$ be finite, $k\in[0,K]$, $\psi$ satisfy the assumptions of Theorem~\ref{thm:main_ac}, and
  \begin{equation}
\lambda  \coloneqq
	\begin{cases}
		0,&\mathrm{if}\;F(0)\neq0\\
		1,&\mathrm{if}\;F(0)=0.
	\end{cases}
\end{equation}   
 Then $\hat\psi(0)=-i\lambda\langle f(0,\cdot),\psi\rangle$ and
  \begin{equation}
	|\hat\psi(k)| \leq \lambda|\hat\psi(0)|+\|\1_{K}\dkp{\hat\psi}\|_\infty k  .
  \end{equation}
\end{lemma}
\begin{proof}
Using the fact that $S(0)=-1$ if $\lambda=1$ and $S(0)=1$ if $\lambda=0$ (see~\cite[page 356]{Newton1966} for the proof) we get
\begin{equation}
\bar\psi^+(0,r) 
= - \frac{1}{2i}(\bar S(0)\bar f(0,r)-\bar f(0,r))
=-i\lambda \bar f(0,r),
\end{equation}
and thereby
\begin{align}
	\hat\psi(0)
	=\int_0^\infty\psi(r)\bar\psi^+(0,r)\,dr
	&=-i\lambda\langle f(0,\cdot),\psi\rangle.
\end{align}
Moreover, observe that
  \begin{align}
	|\hat\psi(k)|  &=  \left|\hat\psi(0)+\int_0^k\dkp{\hat\psi}(\tau)\,d\tau\right|
	\\
	&\leq  |\hat\psi(0)|+\|\1_{K}\dkp{\hat\psi}\|_\infty k
	= \lambda |\hat\psi(0)|+\|\1_{K}\dkp{\hat\psi}\|_\infty k ,
  \end{align}
indeed $\lambda^2=\lambda$.
\QED
\end{proof}

\begin{lemma}\label{lem:Zac}
   Let $K>0$ and $R\geq R_V$ be finite and recall Definition \ref{def:smallz} and the definitions given in Theorem~\ref{th:SBoundsK} and \ref{th:GlobalSBounds}. Then for $k'\in[0,\infty)$ and $k\in[0, K)$
    	\begin{align}
		\left|Z_{ac}(k,k')\right|
		&\leq  \frac{1}{2s_{K}}\left(2Rs_K+C_{1,K}\right)\label{eq:lem_Z1}
		=\frac{z_{ac,K}(0,0)}{s_K},\displaybreak[0]\\
		\left|Z_{ac}(k,k')\right|
		&\leq  \frac{1}{|k-k'|}
		=\frac{z_{ac,K}(0,1)}{|k-k'|}\label{eq:lem_Z2},\displaybreak[0]\\
		\left|\dkp Z_{ac}(k,k')\right|
		&\leq  \frac{1}{4s_K^2}\left(2R^2s_K^2+2RC_{1,K}s_K+C_{2,K}\right)
		\\
		&=\frac{z_{ac,K}(1,0)}{s_K^2}\label{eq:lem_Z3},\displaybreak[0]\\
		\left|\dkp Z_{ac}(k,k')\right|
		&\leq  \frac{2Rs_K+C_{1,K}}{2s_K|k-k'|}+\frac{1}{|k-k'|^2}
		\\
		&=\frac{z_{ac,K}(1,1)}{s_K|k-k'|}+\frac{z_{ac,K}(1,2)}{|k-k'|^2}\label{eq:lem_Z4},\displaybreak[0]\\
		\left|\ddkp Z_{ac}(k,k')\right|
		&\leq  \frac{1}{6s_K^3}(2R^3s_K^3+3R^2s_K^2C_{1,K}+3Rs_KC_{2,K}+C_{3,K})
		\\
		&=\frac{z_{ac,K}(2,0)}{s_K^3}\label{eq:lem_Z5},\displaybreak[0]\\
		\left|\ddkp Z_{ac}(k,k')\right|
		&\leq  \frac{2R^2s_K^2+2Rs_KC_{1,K}+C_{2,K}}{2s_K^2|k-k'|}
		\nonumber\\
		&\qquad+\frac{2Rs_K+C_{1,K}}{s_K|k-k'|^2}+\frac{2}{|k-k'|^3}\\
		&\quad=\frac{z_{ac,K}(2,1)}{s_K^2|k-k'|}+\frac{z_{ac,K}(2,2)}{s_K|k-k'|^2}+\frac{z_{ac,K}(2,3)}{|k-k'|^3}.\label{eq:lem_Z6}
 	\end{align}
For $k\in[0,\infty)$ we have the same bounds with the index $K$ omitted on the right hand side.
\end{lemma}
\begin{proof}
Let $k\in[0,K)$ and
\begin{align}
	Z_{ac}(k,k')
	&=\frac{i}{4}\Biggl[\frac{e^{i(k+k')R}S(k)-e^{-i(k+k')R}\bar S(k')}{k+k'}
	\nonumber\\
	&\qquad-\frac{e^{i(k-k')R}\bar S(k')S(k)-e^{-i(k-k')R}}{k-k'}\Biggr]\label{eq:lem_Z7}\\
	&\eqqcolon \frac{i}{4}\left[\frac{h_1(k,k')}{k+k'}-\frac{h_2(k,k')}{k-k'}\right].
\end{align}
Now, Eq.~\eqref{eq:lem_Z2} follows from the fact that $|S|=1$ and $|k+k'|\geq|k-k'|.$ To prove Eq.~\eqref{eq:lem_Z1}, we use $\bar S(k')=S(-k')$ and observe that via Lipschitz and Theorem~\ref{th:SBoundsK}
\begin{align}
\frac{S(k)-S(-k')}{k-(-k')}\leq \frac{C_{1,K}}{s_K}.
\end{align}
Using this and
\begin{align}
h_1(k,k')  
	&=  \left(   e^{i(k+k')R}  -  e^{-i(k+k')R}  \right)  S(k)
	\nonumber\\
	&\qquad+  e^{-i(k+k')R}  \left( S(k) -\bar S(k') \right)   ,
\end{align}
we get
\begin{align}
  \left|\frac{h_1(k,k')}{k+k'}\right|
  &\leq\left|\frac{e^{i(k+k')R}-e^{-i(k+k')R}}{k+k'}\right|
  +\left|\frac{S(k)-S(-k')}{k-(-k')}\right|\\
  &\leq2R\left|\frac{\sin((k+k')R)}{(k+k')R}\right|
  +\left|\frac{S(k)-S(-k')}{k-(-k')}\right|
  \\
  &\leq 2R+\frac{C_{1,K}}{s_K}.
\end{align}
Together with the analogous bound for the second summand in Eq.~\eqref{eq:lem_Z7} this yields~\eqref{eq:lem_Z1}. To prove Eq.~\eqref{eq:lem_Z4} observe that
\begin{align}\label{eq:lem_Z8}
  \dkp Z_{ac}(k,k')
  =\frac{i}{4}\left[\frac{\dkp h_1(k,k')}{k+k'}-\frac{\dkp h_2(k,k')}{k-k'}-\frac{h_1(k,k')}{(k+k')^2}+\frac{h_2(k,k')}{(k-k')^2}\right].
\end{align}
Using the bounds on the derivatives of the $S$-matrix given in Theorem~\ref{th:SBoundsK} we find
\begin{align}
 \left |\dkp h_1(k,k')\right|
  &=\Biggl|iR\left(e^{i(k+k')R}S(k)+e^{-i(k+k')R}\bar S(k')\right)
  \nonumber\\
 &\qquad +e^{i(k+k')R}\dkp S(k)\Biggr|
   \leq 2R+\frac{C_{1,K}}{s_K}\\
 \left |\dkp h_2(k,k') \right|
  &=\Biggl|iR\left(e^{i(k-k')R}\bar S(k')S(k)+e^{-i(k-k')R}\right)
  \nonumber\\
  &\qquad+e^{i(k-k')R}\bar S(k')\dkp S(k)\Biggr|
  \leq 2R+\frac{C_{1,K}}{s_K}.
\end{align}
This and Eq.~\eqref{eq:lem_Z8} immediately yield Eq.~\eqref{eq:lem_Z4}. To prove Eq.~\eqref{eq:lem_Z3} note that the Taylor expansion of $h_1(x,k')$ and $h_2(x,k')$ in $x$ around $k$ reads
\begin{align}
  h_{1,2}(x,k')=h_{1,2}(k,k')&+\dkp h_{1,2}(k,k')(x-k)
  \nonumber\\
  &+\int_k^x\ddkp h_{1,2}(\tau,k')(x-\tau)\,d\tau.
\end{align}
If we evaluate this at $x=-k'$ for $h_1$ and at $x=k'$ for $h_2$ and observe that $h_1(-k',k')=0=h_2(k',k'),$ we get
\begin{align}
  h_{1,2}(k,k')=\dkp h_{1,2}(k,k')(k\pm k')+\int_k^{\mp k'}\ddkp h_{1,2}(\tau,k')(\tau\pm k')\,d\tau.
\end{align}
Plugging this into Eq.~\eqref{eq:lem_Z8} and employing the variable substitutions $\tau=k-(k+k')\tau'$ for the $\ddkp h_1$-integral and $\tau=k-(k-k')\tau'$ for the $\ddkp h_2$-integral, then yields
\begin{align}
  \dkp Z_{ac}(k,k')
  &=\frac{i}{4}\Bigg[-\int_k^{-k'}\ddkp h_1(\tau,k')\frac{\tau+k'}{(k+k')^2}\,d\tau
  \nonumber\\
  &\quad+\int_k^{k'}\ddkp h_2(\tau,k')\frac{\tau-k'}{(k-k')^2}\,d\tau\Bigg]\\
  &=\frac{i}{4}\Bigg[\int_0^{1}\ddkp h_1(k-(k+k')\tau',k')(1-\tau')\,d\tau'\nonumber\\
  &\quad-\int_0^{1}\ddkp h_2(k-(k-k')\tau',k')(1-\tau')\,d\tau'\Bigg].\label{eq:lem_Z9}
\end{align}
Moreover, the bounds on derivatives of the $S$-matrix due to Theorem~\ref{th:SBoundsK} imply
\begin{align}
 \left |\ddkp h_1(k,k')\right|
  &=\Biggl|-R^2\left(e^{i(k+k')R}S(k)-e^{-i(k+k')R}\bar S(k')\right)
\nonumber\\
&\quad  +2iRe^{i(k+k')R}\dkp S(k)
  +e^{i(k+k')R}\ddkp S(k)\Biggr|\\
  &\leq2R^2+2R\frac{C_{1,K}}{s_K}+\frac{C_{2,K}}{s_K^2}\\
  \left|\ddkp h_2(k,k')\right|
  &=\Biggl|-R^2\left(e^{i(k-k')R}\bar S(k')S(k)-e^{-i(k-k')R}\right)
  \nonumber\\
  &\quad+2iRe^{i(k-k')R}\bar S(k')\dkp S(k)
  \quad+e^{i(k-k')R}\bar S(k')\ddkp S(k)\Biggr|\\
  &\leq2R^2+2R\frac{C_{1,K}}{s_K}+\frac{C_{2,K}}{s_K^2}.
\end{align}
Using this and Eq.~\eqref{eq:lem_Z9}, we obtain Eq.~\eqref{eq:lem_Z3}. Analogously to the proof of Eqs.~\eqref{eq:lem_Z3} and \eqref{eq:lem_Z4}, we arrive at Eqs.~(\ref{eq:lem_Z5}) and~\eqref{eq:lem_Z6}.

For $k\in[0,\infty)$ the proof is the same, except that we use the $S$-matrix bounds provided by Theorem~\ref{th:GlobalSBounds} rather than those of Theorem~\ref{th:SBoundsK}. In effect this amounts to omitting the index $K$ everywhere.\QED
\end{proof}

\begin{lemma}\label{lem:Z0ac}
  Let $R\geq R_V,$ $K>0$, $k\in[0,K]$, and
  \begin{equation}
\lambda  \coloneqq
	\begin{cases}
		0,&\mathrm{if}\;F(0)\neq0\\
		1,&\mathrm{if}\;F(0)=0.
	\end{cases}
\end{equation}   
Then
  \begin{align}
 	\left|Z_{ac}(k,k')\right|  
	&\leq  \lambda \frac{z_{ac,K}(0,0)}{s_K}+\frac{z_{ac,K}(1,0)}{s_K^2}k,
	\nonumber\\
	&\qquad\text{for}\;k'\in[0,2K]\label{eq:lem_Z03},
	\\
	\left|Z_{ac}(k,k')\right|  
	&\leq  \lambda\frac{z_{ac,K}(0,1)}{k'}+\left[\frac{z_{ac,K}(1,1)}{s_K|k-k'|}+\frac{z_{ac,K}(1,2)}{|k-k'|^2}\right]k,
	\nonumber\\
	&\qquad\text{for}\;k'\in[2K,\infty).\label{eq:lem_Z04}
  \end{align}
\end{lemma}

\begin{proof}
Clearly,
\begin{align}\label{eq:lem_Z10}
  |Z_{ac}(k,k')|
  &= \left|Z_{ac}(0,k')+\int_0^k\dkp Z_{ac}(\tau,k')\,d\tau \right|\nonumber\\
  &\leq |Z_{ac}(0,k')|+\int_0^k|\dkp Z_{ac}(\tau,k')|\,d\tau.
\end{align}
First, we prove Eq.~\eqref{eq:lem_Z04}. Observing that
\begin{align}
  Z_{ac}(0,k')
  &=\frac{i}{4k'}\left[e^{ik'R}+e^{-ik'R}\bar S(k')\right](S(0)-1)\label{eq:lem_Z06}
\end{align}
and using the fact that $S(0)=\mp 1$ for $\lambda=1$ and $0$, respectively (see~\cite[page 356]{Newton1966} for the proof), we obtain
\begin{align}\label{eq:lem_Z08}
  |Z_{ac}(0,k')|\leq\lambda\frac{1}{k'}.
\end{align}
If we plug this into Eq.~\eqref{eq:lem_Z10} and employ the bound for $|\dkp Z_{ac}(k,k')|$ provided by Eq.~\eqref{eq:lem_Z4}, we arrive at
\begin{align}
  \left|Z_{ac}(k,k')\right|  \leq  \lambda\frac{z_{ac,K}(0,1)}{k'}+\int_0^k\left[\frac{z_{ac,K}(1,1)}{s_K|\tau-k'|}+\frac{z_{ac,K}(1,2)}{|\tau-k'|^2}\right]\,d\tau.
\end{align}
Since $k\in[0,K]$ and $k'\in[2K,\infty),$ we have $|\tau-k'|\geq|k-k'|$ and this implies
\begin{align}
  \left|Z_{ac}(k,k')\right|  \leq  \lambda\frac{z_{ac,K}(0,1)}{k'}+\left[\frac{z_{ac,K}(1,1)}{s_K|k-k'|}+\frac{z_{ac,K}(1,2)}{|k-k'|^2}\right]k,
\end{align}
which finishes the proof of Eq.~\eqref{eq:lem_Z04}.
As for Eq.~\eqref{eq:lem_Z03}, note that
\begin{align}
  Z_{ac}(0,k')
  &=\frac{i}{4k'}\left[e^{ik'R}(\bar S(k')+1)-\bar S(k')(e^{ik'R}-e^{-ik'R})\right](S(0)-1)\label{eq:lem_Z07}.
\end{align}
Now, via Lipschitz and Theorem~\ref{th:SBoundsK} we see that
\begin{align}
  \left|\frac{\bar S(k')+1}{k'}\right|=\left|\frac{\bar S(k')-\bar S(0)}{k'-0}\right|\leq\|1_K\dkp S\|_\infty\leq\frac{C_{1,K}}{s_K},
\end{align}
which together with
\begin{align}
  \left|\frac{e^{ik'R}-e^{-ik'R}}{k'}\right|=\left|\frac{e^{2ik'R}-e^{0}}{k'-0}\right|\leq 2R  ,
\end{align}
the fact that $S(0)=\mp 1$ for $\lambda=1$ and $0$, respectively, and Eq.~\eqref{eq:lem_Z07} yields
\begin{align}\label{eq:lem_Z09}
  |Z_{ac}(0,k')|\leq\lambda\frac{1}{2}\left[2R+\frac{C_{1,K}}{s_K}\right]
  =\lambda\frac{z_{ac,K}(0,0)}{s_K}.
\end{align}
Plugging this and the bound for $|\dkp Z_{ac}(k,k')|$ provided by Eq.~\eqref{eq:lem_Z3} into Eq.~\eqref{eq:lem_Z10} we have
\begin{align}
  |Z_{ac}(k,k')|
\leq \lambda\frac{z_{ac,K}(0,0)}{s_K}
+ \int_0^k  \frac{z_{ac,K}(1,0)}{s_K^2}  \,d\tau.
\end{align}
Performing the integration in $\tau$ completes the proof of Eq.~\eqref{eq:lem_Z03}.
\QED
\end{proof}

These lemmas allow us to show that the boundary terms due to the integration by parts in Eq.~\eqref{eq:StatPhase} vanish.

\begin{lemma}\label{lem:ZacInt}
Let $\psi$ satisfy the assumptions of Theorem~\ref{thm:main_ac}, then
\begin{align}
\int_0^\infty Z_{ac}(k,k')\hat\psi(k)\frac{\partial_{k}^2e^{-ik^2t}}{2tk^2+i}\,dk=\int_0^\infty \partial_{k}^2\left[Z_{ac}(k,k')\frac{\hat\psi(k)}{2tk^2+i}\right]e^{-ik^2t}\,dk.
\end{align}
\end{lemma}
\begin{proof}
Clearly, integrating by parts twice yields 
\begin{align}
  &\int_0^\infty Z_{ac}(k,k')\hat\psi(k)\frac{\partial_{k}^2e^{-ik^2t}}{2tk^2+i}\,dk
  \\
  &\qquad=  \left[\frac{Z_{ac}(k,k')\hat\psi(k)}{2tk^2+i}\partial_{k}e^{-ik^2t}\right]_0^\infty\label{eq:lem_Zac1}\\
  &\qquad\quad-  \left[\partial_{k}\left(\frac{Z_{ac}(k,k)\hat\psi(k)}{2tk^2+i}\right)e^{-ik^2t}\right]_0^\infty\label{eq:lem_Zac2}\\
  &\qquad\quad+\int_0^\infty \partial_{k}^2\left(\frac{Z_{ac}(k,k')\hat\psi(k)}{2tk^2+i}\right)e^{-ik^2t}\,dk,
\end{align}
so we need to show that the boundary terms vanish. 
We begin with the term~\eqref{eq:lem_Zac1}. 
At infinity it vanishes because $Z_{ac}(k,k')$ is globally bounded (Lemma~\ref{lem:Zac}), $\hat\psi(k)\to 0$ as $k\to\infty$ ($\hat\psi$ is square integrable) and the time dependent factors tend to zero, too. 
At zero the term~\eqref{eq:lem_Zac1} vanishes because $Z_{ac}(0,k')$ and $\hat\psi(0)$ are bounded (Lemmas~\ref{lem:Zac} and \ref{lem:psihat0}), while the time dependent factors are zero for $k=0.$

Now, look at Eq.~\eqref{eq:lem_Zac2} and observe that
\begin{align}
 & \partial_{k}\left(\frac{Z_{ac}(k,k')\hat\psi(k)}{2tk^2+i}\right)
 \\
  &\qquad=\left(\dkp Z_{ac}(k,k')\hat\psi(k)+Z_{ac}(k,k')\dkp{\hat\psi}(k)\right)\frac{1}{2tk^2+i}\label{eq:lem_Zac3}\\
  &\qquad\quad-Z_{ac}(k,k')\hat\psi(k)\frac{4tk}{(2tk^2+i)^2}.
\end{align}
The last summand vanishes as $k\to0$ and $k\to\infty$ for the same reasons \eqref{eq:lem_Zac1} vanished, so let us focus on \eqref{eq:lem_Zac3}. For $k\to\infty$ it vanishes because $Z_{ac}(k,k')$ as well as $\dkp Z_{ac}(k,k')$ are bounded (Lemma~\ref{lem:Zac}), $\hat\psi(k)$ tends to zero ($\hat\psi$ is square integrabel), and $\dkp{\hat\psi}(k)$ can only diverge slower than $k$  ($\|\dkp{\hat\psi}w\|_1<\infty$ by assumption). 
In case there is no zero resonance ($\lambda=0$) the term~\eqref{eq:lem_Zac3} evaluates to zero at $k=0$ because $|Z_{ac}(0,k')|\leq k\,z_{ac,K}(1,0)/s_K^2$ (Lemma~\ref{lem:Z0ac}),  
$\dkp{\hat\psi}(k)$ can only diverge slower than $1/k$ as $k\to0$ ($\|\dkp{\hat\psi}w\|_1<\infty$ by assumption), $|\dkp Z_{ac}(0,k')|$ is bounded (Lemma~\ref{lem:Zac}) and $\hat\psi(0)=0$ (Lemma~\ref{lem:psihat0}). 
In case there is a zero resonance ($\lambda=1$), then $S(0)=-1$ (see~\cite[page 356]{Newton1966} for the proof), hence
\begin{align}\label{eq:lem_Zac4}
  Z_{ac}(0,k')&=-\frac{i}{2k'}(e^{ik'R}+e^{-ik'R}\bar S(k')) .
\end{align}
Furthermore,
\begin{align}
  &\dkp Z_{ac}(k,k')
  =\frac{i}{4}\Bigg[
  -\frac{e^{i(k+k')R}S(k)-e^{-i(k+k')R}\bar S(k')}{(k+k')^2}
  \nonumber\\
 & +\frac{e^{i(k-k')R}\bar S(k')S(k)-e^{-i(k-k')R}}{(k-k')^2}\nonumber\\
  &+\frac{1}{k+k'}\left(iR(e^{i(k+k')R}S(k)+e^{-i(k+k')R}\bar S(k'))+e^{i(k+k')R}\dkp S(k)\right)\nonumber\\
  &-\frac{1}{k-k'}\left(iR(e^{i(k-k')R}\bar S(k')S(k)+e^{-i(k-k')R})+e^{i(k-k')R}\bar S(k')\dkp S(k)\right)
  \Bigg],
\end{align}
implies that
\begin{align}\label{eq:lem_Zac7}
  \dkp Z_{ac}(0,k')
  &=\frac{i}{4k'}(e^{ik'R}+e^{-ik'R}\bar S(k'))\dkp S(0).
\end{align}
Moreover, from $\bar f(k,r)=f(-k,r)$ and $\bar S(k)=S(-k)$ we get
\begin{align}\label{eq:lem_Zac5}
  {\hat\psi}(k)
  &=\int_0^\infty\psi(r)\bar\psi^+(k,r)\,dr
 \\
 & =-\frac{1}{2i}\int_0^\infty\psi(r)(S(-k)f(-k,r)-f(k,r))\,dr.
\end{align}
Employing $\hat\psi(0)=-i\int_0^\infty\psi(r)f(0,r)\,dr$ (Lemma~\ref{lem:psihat0}) we then obtain
\begin{align}\label{eq:lem_Zac6}
  \dkp{\hat\psi}(0)
  &=-\frac{1}{2i}\int_0^\infty\psi(r)(-\dkp{S}(0)f(k,r)-S(0)\dkp {f}(0,r)-\dkp{f}(0,r))\,dr
 \\& =\frac{1}{2}\dkp S(0)\hat\psi(0)
\end{align}
and by plugging Eqs.~\eqref{eq:lem_Zac4}, \eqref{eq:lem_Zac7}, \eqref{eq:lem_Zac5}, and \eqref{eq:lem_Zac6} into the term~\eqref{eq:lem_Zac3}, we see that \eqref{eq:lem_Zac3} evaluates to zero at $k=0.$ This finishes the proof.
\QED
\end{proof}

We are now in the position to prove Theorem~\ref{thm:main_ac}.
\begin{proof}[Theorem~\ref{thm:main_ac}]
  Let $t>0.$ We start from Eq.~\eqref{eq:StatPhase}, which reads
  \begin{align}\label{eq:StatPhase2}
  &  \|P_{ac}\1_Re^{-iHt}P_{ac}\psi\|_2^2
    \\
    &\qquad=  \frac{1}{4t^2}\int_0^\infty\left|\int_0^\infty \partial_{k}^2\left[Z_{ac}(k,k')\hat\psi(k)\frac{1}{2tk^2+i}\right]e^{-ik^2t}\,dk\right|^2\,dk'.
  \end{align}
For $A,B,C\in\R$
\begin{equation}\label{eq:el_ineq}
    (A+B+C)^2\leq3(A^2+B^2+C^2),
  \end{equation}
therefore with the shorthands
  \begin{align}
	&g_1(k,k')  \coloneqq   \ddkp Z_{ac}(k,k')\hat\psi(k)+2\dkp Z_{ac}(k,k')\dkp {\hat\psi}(k)+Z_{ac}(k,k')\ddkp {\hat\psi}(k),\\
	&g_2(k,k')  \coloneqq   \dkp Z_{ac}(k,k')\hat\psi(k)+Z_{ac}(k,k')\dkp{\hat\psi}(k),
  \end{align}
we get
  \begin{align}
 &\|P_{ac}\1_Re^{-iHt}P_{ac}\psi\|_2^2
 \\
    &\leq  \frac{3}{4t^2}\int_0^\infty\left|\int_0^\infty g_1(k,k')\frac{1}{2tk^2+i}e^{-ik^2t}\,dk\right|^2\,dk'\label{eq:thm_main1}\\
    &\quad+  \frac{3}{t^2}\int_0^\infty\left|\int_0^\infty g_2(k,k')\frac{4tk}{(2tk^2+i)^2}e^{-ik^2t}\,dk\right|^2\,dk'\label{eq:thm_main2}\\
    &\quad+  \frac{3}{4t^2}\int_0^\infty\left|\int_0^\infty Z_{ac}(k,k')\hat\psi(k)\frac{4t(i-6tk^2)}{(2tk^2+i)^3}e^{-ik^2t}\,dk\right|^2\,dk'  .
    \label{eq:thm_main3}
  \end{align}
  Note that Eq.~\eqref{eq:el_ineq} as well as $(A+B)^2\leq 2(A^2+B^2)$ will be used repeatedly throughout the proof, sometimes without mention.

  \begin{figure}
  	\centering
  	\includegraphics[width=.5\textwidth]{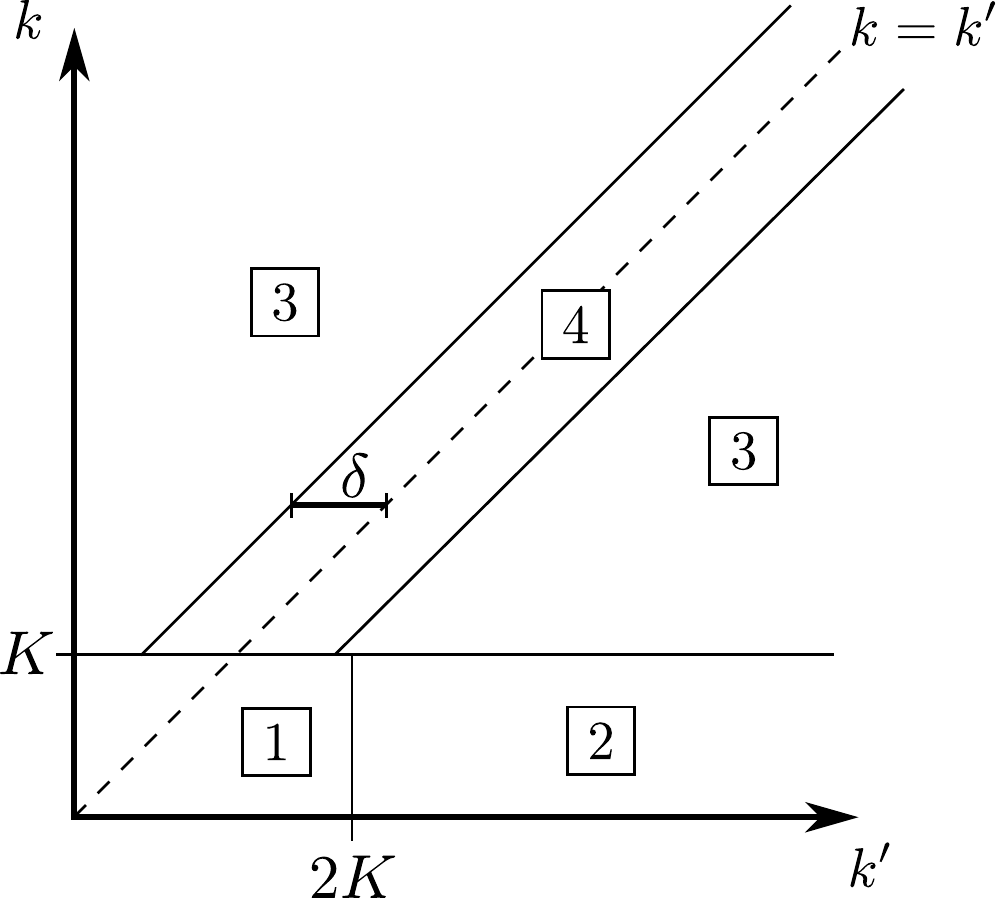}
  	\caption{Division of the $k$-$k'$-plane used to estimate $\|P_{ac}\1_Re^{-iHt}P_{ac}\psi\|_2.$}
  	\label{fig:regions}
  \end{figure}
Although we are not dealing with double integrals, it is useful to think of the \mbox{$k$-$k'$-plane} as it was the integration region, that we will divide as depicted in Fig.~\ref{fig:regions}. 
Let us explain why.
First, for a suitable  $c_K>0$, we can write by a change of variable
  \begin{align}
  	&\left[\int_0^\infty\frac{1}{2t|2tk^2+i|}\,dk\right]^2
	=\left[\frac{1}{2t^{3/2}}\int_0^\infty\frac{1}{|2k^2+i|}\,dk\right]^2
	\leq\frac{1}{2t^3},
\\
  	&\left[\int_K^\infty\frac{1}{2t|2tk^2+i|}\,dk\right]^2
	\leq\frac{1}{t^4}\left[\int_K^\infty\frac{1}{2k^2}\,dk\right]^2
	=\frac{c_K}{t^4},
  \end{align}
which suggests that the term \eqref{eq:thm_main1} contains a  $t^{-3}$ contribution that comes from the part of the integration region where $k<K$, while the $t^{-4}$ contribution comes from $k\geq K.$ Second, $Z_{ac}(k,k')$ has an apparent singularity at $k=k'$ (see Eq.~\eqref{eq:IntKernAc2}). It is apparent in the sense that by performing the limit $k\to k'$ on the right hand side of Eq.~\eqref{eq:IntKernAc2} a finite quantity depending on derivatives of the $S$-matrix is obtained. Therefore, we will use a Taylor expansion in the stripe around $k=k'$, while we use a different strategy in the remaining regions.

  Now, we split the integrals in Eqs.~(\ref{eq:thm_main1}-\ref{eq:thm_main3}) according to Fig.~\ref{fig:regions}. Let $h(k,k')$ be a placeholder for the integrands in Eqs.~(\ref{eq:thm_main1}-\ref{eq:thm_main3}) and let the indicator functions $\1_{3,k'}(k)$ and $\1_{4,k'}(k)$ be one only on the regions 3 and 4, respectively. Then we obtain
  \begin{align}
	\int_0^\infty &\left|\int_0^\infty h(k,k')\,dk\right|^2\,dk'
	\\
	&\leq  2\int_0^{2K}\left[\int_0^{K} |h(k,k')|\,dk\right]^2\,dk'\label{eq:thm_main_r1}\\
	&\quad+  2\int_{2K}^\infty \left[\int_0^{K} |h(k,k')|\,dk\right]^2\,dk'\label{eq:thm_main_r2}\\
	&\quad+  4\int_0^\infty \left[\int_0^\infty  \1_{3,k'} |h(k,k')|\,dk\right]^2\,dk'\label{eq:thm_main_r3}\\
	&\quad+  4\int_0^\infty \left[\int_0^\infty  \1_{4,k'} |h(k,k')|\,dk\right]^2\,dk'.\label{eq:thm_main_r4}
  \end{align}

  First, let us look at integral~\eqref{eq:thm_main1} in region 1. 
Using the bounds on $Z_{ac}(k,k')$ provided by Lemma~\ref{lem:Zac} and 
  \begin{equation}\label{eq:thm_main1.2}
	\left[\int_0^{K}dk \frac{1}{\sqrt{4t^2k^4+1}}\right]^2
	\leq\frac{1}{t}\left[\int_0^{\infty}dk \frac{1}{\sqrt{4k^4+1}}\right]^2\leq\frac{2}{t},
  \end{equation}
obtained by change of variable,  we get
  \begin{align}
	&\frac{3}{2t^2}\int_0^{2K}dk'\left[\int_0^{K}dk |g_1(k,k')|\frac{1}{\sqrt{4t^2k^4+1}}\right]^2\\
	&\quad\leq  \frac{6K}{t^3}\Biggl(\frac{z_{ac,K}(2,0)}{s_K^3}\|\1_{K}\hat\psi\|_\infty
	+2\frac{z_{ac,K}(1,0)}{s_K^2}\|\1_{K}\dkp{\hat\psi}\|_\infty
	\nonumber\\
	&\qquad+\frac{z_{ac,K}(0,0)}{s_K}\|\1_{K}\ddkp{\hat\psi}\|_\infty\Biggr)^2\\
	&\quad\leq  \frac{18K}{t^3s_K^6}\Bigl(z_{ac,K}^2(2,0)\|\1_{K}\hat\psi\|_\infty^2
	\nonumber\\
	&\qquad+4s_K^2z_{ac,K}^2(1,0)\|\1_{K}\dkp{\hat\psi}\|_\infty^2
	+s_K^4z_{ac,K}^2(0,0)\|\1_{K}\ddkp{\hat\psi}\|_\infty^2\Bigr).\label{eq:thm_main1.1}
  \end{align}

  In region 2 integral~\eqref{eq:thm_main1} takes the form
  \begin{align}\label{eq:thm_main1.3}
	\frac{3}{2t^2}&\int_{2K}^\infty dk'\left[\int_0^{K}dk |g_1(k,k')|\frac{1}{\sqrt{4t^2k^4+1}}\right]^2.
  \end{align}
  With the help of the bounds on $Z_{ac}(k,k')$ given in Lemma~\ref{lem:Zac} and the fact that $|k'-k|\geq|k'-K|$ in region 2, we see that
  \begin{align}
	&|g_1(k,k')|
	\nonumber\\
	&\leq \frac{z_{ac,K}(2,1)\|\1_{K}\hat\psi\|_\infty
	+2s_Kz_{ac,K}(1,1)\|\1_{K}\dkp{\hat\psi}\|_\infty
	+s_K^2z_{ac,K}(0,1)\|\1_{K}\ddkp{\hat\psi}\|_\infty}{s_K^2|k'-K|}\nonumber\\
	&+\frac{z_{ac,K}(2,2)\|\1_{K}\hat\psi\|_\infty
	+2s_Kz_{ac,K}(1,2)\|\1_{K}\dkp{\hat\psi}\|_\infty}{s_K|k'-K|^2}
	\nonumber\\
	&+\frac{z_{ac,K}(2,3)\|\1_{K}\hat\psi\|_\infty}{|k'-K|^3}.
  \end{align}
  Employing this in Eq.~\eqref{eq:thm_main1.3} together with the elementary inEq.~\eqref{eq:el_ineq} and Eq.~\eqref{eq:thm_main1.2}, we obtain
  \begin{align}
	&\frac{3}{2t^2}\int_{2K}^\infty dk'\left[\int_0^{K}dk |g_1(k,k')|\frac{1}{\sqrt{4t^2k^4+1}}\right]^2\\
	&\leq  \frac{9}{t^3}\Bigg[\frac{3}{Ks_K^4}\Bigl(z_{ac,K}^2(2,1)\|\1_{K}\hat\psi\|_\infty^2
	\nonumber\\
	&\quad\quad+4s_K^2z_{ac,K}^2(1,1)\|\1_{K}\dkp{\hat\psi}\|_\infty^2
	+s_K^4z_{ac,K}^2(0,1)\|\1_{K}\ddkp{\hat\psi}\|_\infty^2\Bigr)\nonumber\\
	&\quad+\frac{2}{3K^3s_K^2}\left(z_{ac,K}^2(2,2)\|\1_{K}\hat\psi\|_\infty^2
	+4s_K^2z_{ac,K}^2(1,2)\|\1_{K}\dkp{\hat\psi}\|_\infty^2\right)
	\\
	&\quad+\frac{z_{ac,K}^2(2,3)\|\1_{K}\hat\psi\|_\infty^2}{5K^5}\Bigg].\label{eq:thm_main1.4}
  \end{align}

  In region 3 the integral~\eqref{eq:thm_main1} reads (see Eq.~\eqref{eq:thm_main_r3})
  \begin{align}
	&\frac{3}{t^2}  \int_0^\infty dk'  \left[\int_0^\infty dk \1_{3,k'} |g_1(k,k')|\frac{1}{\sqrt{4t^2k^4+1}}\right]^2.
  \end{align}
  Now the $k$-integral ranges up to infinity and we could use $\|\hat\psi^{(n)}\|_\infty$ to handle the $\hat\psi$ dependency of $g_1.$ 
However, the suprema can get big.
In particular, as mentioned in Section~\ref{sec:Discussion}, our bounds are most relevant for wave functions describing meta-stable systems.
In this case, if $\alpha-i\beta$ is the resonance corresponding to the meta-stable state under consideration, then $\hat\psi$  will resemble a Breit Wigner function centered around $k=\alpha$, with width $2\beta$ and height $1/\sqrt{\beta}.$ 
For small $\beta$, i.e.\ for long lifetime, the supremum of such a $\hat\psi$ is big, whereas the integral over $|\hat\psi|$ around $k=\alpha$ will be of order $\sqrt\beta,$ which is small. 
Therefore, for states of physical interest, bounds depending on $L^1$-norms are more convenient than bounds involving suprema.

Since we are now in a region such that $k$ can not be zero, we can pull the time dependency out of the integral via
\begin{equation}\label{eq:thm_main1.5}
	\frac{1}{\sqrt{4t^2k^4+1}}  \leq  \frac{1}{2tk^2}
	=\frac{1}{2t}\left(1+\frac{1}{k^2}\right)w(k)
	\leq  \frac{1}{2t}\left(1+\frac{1}{K^2}\right)w(k)
\end{equation}
with $w(k)=(1+k^2)^{-1}.$ It is useful to keep the weight function $w$ as part of the integrand because $\hat\psi$ might not decay fast enough at infinity for $\|\hat\psi\|_1$ to be finite. Then
  \begin{align}
  	&\frac{3}{t^2}  \int_0^\infty dk'  \left[\int_0^\infty dk \1_{3,k'} |g_1(k,k')|\frac{1}{\sqrt{4t^2k^4+1}}\right]^2\\
	&\quad\leq  \frac{3}{4t^4}\left(1+\frac{1}{K^2}\right)^2  \int_0^\infty dk'  \left[\int_0^\infty dk \1_{3,k'} |g_1(k,k')|w(k)\right]^2\\
	&\quad\leq  \frac{27}{4t^4}\left(1+\frac{1}{K^2}\right)^2 
	\int_0^\infty dk'  \bigg(\bigg[\int_0^\infty dk \1_{3,k'} |\ddkp Z_{ac}(k,k')\hat\psi(k)|w(k)\bigg]^2\nonumber\\
	&\qquad\qquad+\bigg[\int_0^\infty dk \1_{3,k'}|2\dkp Z_{ac}(k,k')\dkp{\hat\psi}(k)|w(k)\bigg]^2\nonumber\\
	&\qquad\qquad+\bigg[\int_0^\infty dk \1_{3,k'}|Z_{ac}(k,k')\ddkp{\hat\psi}(k)|w(k)\bigg]^2\bigg),
  \end{align}
  where we have used the elementary inequality~\eqref{eq:el_ineq}. At this point we employ Jensen's inequality to pull the square into the $k$-integrals, using $|\hat\psi w|, |\dkp{\hat\psi} w|$ and $|\ddkp{\hat\psi} w|,$ respectively, as measures. Hence,
  \begin{align}
	&\frac{3}{t^2}  \int_0^\infty dk'  \left[\int_0^\infty dk \1_{3,k'} |g_1(k,k')|\frac{1}{\sqrt{4t^2k^4+1}}\right]^2
	\displaybreak[0]\\
	&\leq\frac{27}{4t^4}\left(1+\frac{1}{K^2}\right)^2
	\bigg(\|\hat\psi w\|_1\int_0^\infty dk'\int_0^\infty dk \1_{3,k'} |\ddkp Z_{ac}(k,k')|^2|\hat\psi(k)|w(k)\nonumber\\
	&\qquad\qquad+4\|\dkp{\hat\psi} w\|_1\int_0^\infty dk'\int_0^\infty dk \1_{3,k'}|\dkp Z_{ac}(k,k')|^2|\dkp{\hat\psi}(k)|w(k)\nonumber\\
	&\qquad\qquad+\|\ddkp{\hat\psi} w\|_1\int_0^\infty dk'\int_0^\infty dk \1_{3,k'}|Z_{ac}(k,k')|^2|\ddkp{\hat\psi}(k)|w(k)\bigg)
	\displaybreak[0]\\
	&=  \frac{27}{4t^4}\left(1+\frac{1}{K^2}\right)^2
	\int_{K}^\infty dk  \left[\int_0^{k-\delta} dk'+\int_{k+\delta}^\infty dk'\right]
	\nonumber\\
	&\quad\times\bigg(\|\hat\psi w\|_1|\ddkp Z_{ac}(k,k')|^2|\hat\psi(k)|w(k)\nonumber\\
	&\quad+4\|\dkp{\hat\psi} w\|_1|\dkp Z_{ac}(k,k')|^2|\dkp{\hat\psi}(k)|w(k)
	\nonumber\\
	&\quad+\|\ddkp{\hat\psi} w\|_1|Z_{ac}(k,k')|^2|\ddkp{\hat\psi}(k)|w(k)\bigg),\label{eq:thm_main1.6}
  \end{align}
  where $\delta$ is defined in Fig.~\ref{fig:regions} and will be determined later. From the bounds on $Z_{ac}(k,k')$ in Lemma~\ref{lem:Zac} we get 
  \begin{align}
	&|\ddkp Z_{ac}(k,k')|^2
	\leq  3\left(\frac{z_{ac}^2(2,1)}{s^4(k'-k)^2}
	+\frac{z_{ac}^2(2,2)}{s^2(k'-k)^4}
	+\frac{z_{ac}^2(2,3)}{(k'-k)^6}\right),\\
	&|\dkp Z_{ac}(k,k')|^2
	\leq  2\left(\frac{z_{ac}^2(1,1)}{s^2(k'-k)^2}
	+\frac{z_{ac}^2(1,2)}{(k'-k)^4}\right),\label{eq:thm_main1.7}\\
	&|Z_{ac}(k,k')|^2  \leq  \frac{z_{ac}^2(0,1)}{(k'-k)^2},\label{eq:thm_main1.8}
  \end{align}
  which when plugged into Eq.~\eqref{eq:thm_main1.6} together with $\int_{K}^\infty dk |{\hat\psi}^{(n)}(k)|w(k)\leq\|\hat\psi^{(n)} w\|_1$ yield
  \begin{align}
	\frac{3}{t^2}  &\int_0^\infty dk'  \left[\int_0^\infty dk \1_{3,k'} |g_1(k,k')|\frac{1}{\sqrt{4t^2k^4+1}}\right]^2\\
	&\leq  \frac{27}{2t^4}\left(1+\frac{1}{K^2}\right)^2
	\bigg[3\|\hat\psi w\|_1^2\left(\frac{z_{ac}^2(2,1)}{s^4\delta}
	+\frac{z_{ac}^2(2,2)}{3s^2\delta^3}
	+\frac{z_{ac}^2(2,3)}{5\delta^5}\right)+\nonumber\\
	&\quad+8\|\dkp{\hat\psi} w\|_1^2\left(\frac{z_{ac}^2(1,1)}{s^2\delta}
	+\frac{z_{ac}^2(1,2)}{3\delta^3}\right)
	+\|\ddkp{\hat\psi} w\|_1^2\frac{z_{ac}^2(0,1)}{\delta}\bigg].\label{eq:thm_main1.9}
  \end{align}

  In region 4 we  employ Jensen's inequality in the same way as we did for region 3 and we use the bounds on $Z_{ac}(k,k')$ given in Lemma~\ref{lem:Zac}:
  \begin{align}
  	&\frac{3}{t^2}\int_0^\infty dk'\left[\int_0^\infty dk \1_{4,k'} |g_1(k,k')|\frac{1}{\sqrt{4t^2k^4+1}}\right]^2\\
	&\quad\leq  \frac{27}{4t^4}\left(1+\frac{1}{K^2}\right)^2\int_{K}^\infty dk \int_{k-\delta}^{k+\delta} dk'
	\bigg(\|\hat\psi w\|_1|\ddkp Z_{ac}(k,k')|^2|\hat\psi(k)|w(k)
	\nonumber\\
	&\qquad+4\|\dkp{\hat\psi} w\|_1|\dkp Z_{ac}(k,k')|^2|\dkp{\hat\psi}(k)|w(k)
	+\|\ddkp{\hat\psi} w\|_1|Z_{ac}(k,k')|^2|\ddkp{\hat\psi}(k)|w(k)\bigg)\nonumber\\
	&\quad\leq   \frac{27}{2t^4}\left(1+\frac{1}{K^2}\right)^2\delta\bigg[\|{\hat\psi} w\|_1^2\frac{z_{ac}^2(2,0)}{s^6}
	+\|\dkp{\hat\psi} w\|_1^2\frac{z_{ac}^2(1,0)}{s^4}
	+\|\ddkp{\hat\psi} w\|_1^2\frac{z_{ac}^2(0,0)}{s^2}\bigg].\label{eq:thm_main1.10}
  \end{align}
  Summing up Eqs.~\eqref{eq:thm_main1.1}, \eqref{eq:thm_main1.4}, \eqref{eq:thm_main1.9}, and~\eqref{eq:thm_main1.10}, we obtain the following bound for the integral~\eqref{eq:thm_main1}
  \begin{align}\label{eq:thm_main1.11}
	\frac{3}{4t^2}\int_0^\infty\left|\int_0^\infty g_1(k,k')\frac{e^{-ik^2t}}{2tk^2+i}\,dk\right|^2\,dk'  \leq  t^{-3}C_1+t^{-4}C_2
  \end{align}
  with
  \begin{align}
	C_1&\leq
	\bigg(\frac{18K}{s_K^2}z_{ac,K}^2(0,0)
	+\frac{27}{K}z_{ac,K}^2(0,1)\bigg)\|\1_{K}\ddkp{\hat\psi}\|_\infty^2\nonumber\\
	&\quad+\bigg(\frac{72K}{s_K^4}z_{ac,K}^2(1,0)
	+\frac{108}{Ks_K^2}z_{ac,K}^2(1,1)
	+\frac{24}{K^3}z_{ac,K}^2(1,2)\bigg)\|\1_{K}\dkp{\hat\psi}\|_\infty^2\nonumber\\
	&\quad+\bigg(\frac{18K}{s_K^6}z_{ac,K}^2(2,0)
	+\frac{27}{Ks_K^4}z_{ac,K}^2(2,1)
	\nonumber\\
	&\qquad
	+\frac{18}{3K^3s_K^2}z_{ac,K}^2(2,2)
	+\frac{9}{5K^5}z_{ac,K}^2(2,3)\bigg)\|\1_{K}\hat\psi\|_\infty^2,
	\displaybreak[0]\\
	C_2&\leq
	\frac{27}{2}\left(1+\frac{1}{K^2}\right)^2
	\bigg(\delta\frac{z_{ac}^2(0,0)}{s^2}
	+\frac{z_{ac}^2(0,1)}{\delta}\bigg)\|\ddkp{\hat\psi}w\|_1^2\nonumber\\
	&\quad+ \frac{27}{2}\left(1+\frac{1}{K^2}\right)^2
	\bigg(\delta\frac{z_{ac}^2(1,0)}{s^4}
	+8\frac{z_{ac}^2(1,1)}{s^2\delta}
	+8\frac{z_{ac}^2(1,2)}{3\delta^3}\bigg)\|\dkp{\hat\psi}w\|_1^2\nonumber\\
	&\quad+ \frac{27}{2}\left(1+\frac{1}{K^2}\right)^2
	\bigg(\delta\frac{z_{ac}^2(2,0)}{s^6}
	+3\frac{z_{ac}^2(2,1)}{s^4\delta}
	\nonumber\\
	&\qquad
	+\frac{z_{ac}^2(2,2)}{s^2\delta^3}
	+3\frac{z_{ac}^2(2,3)}{5\delta^5}\bigg)\|{\hat\psi}w\|_1^2.
  \end{align}
  Now, $\delta=s$ is seen to be the optimal choice in the sense that $C_2$ is, to leading order, proportional to $s^{-5},$ which is the best possible $s$ dependence if $s\ll1.$

  The strategy we have followed to estimate integral~\eqref{eq:thm_main1} will be repeated for the remaining integrals. For better readability, we give the results now and the proofs later.
  \begin{align}
	&\frac{3}{t^2}\int_0^\infty\left|\int_0^\infty g_2(k,k')\frac{4tk}{(2tk^2+i)^2}e^{-ik^2t}\,dk\right|^2\,dk'
	\\
	&\qquad\leq   \lambda t^{-2}C_3+  t^{-3}C_4+  t^{-4}C_5\label{eq:thm_main2.0}
	\\
  	&\frac{3}{4t^2}\int_0^\infty\left|\int_0^\infty Z_{ac}(k,k')\hat\psi(k)\frac{4t(i-6tk^2)}{(2tk^2+i)^3}e^{-ik^2t}\,dk\right|^2\,dk'
	\\
	&\qquad
	\leq  \lambda (t^{-1}C_6+ t^{-2}C_7)+t^{-3}C_8+t^{-4}C_9\label{eq:thm_main3.0}
  \end{align}
  with
  \begin{align}
	C_3&\leq
	\frac{12\pi^2}{s_K^4}\left(K z_{ac,K}^2(1,0)
	+\frac{s_K^2}{K}z_{ac,K}^2(1,1)
	+\frac{s_K^4}{6K^3}z_{ac,K}^2(1,2)\right)|\hat\psi(0)|^2\nonumber\\
	&\quad+\frac{12\pi^2}{s_K^2}\left(K z_{ac,K}^2(0,0)
	+\frac{s_K^2}{K}z_{ac,K}^2(0,1)\right)\|\1_{K}\dkp{\hat\psi}\|_\infty^2,
	\displaybreak[0]\\
	C_4&\leq
	\frac{6\pi^2}{s_K^4}\left(K z_{ac,K}^2(1,0)
	+\frac{s_K^2}{K}z_{ac,K}^2(1,1)
	+\frac{s_K^4}{3K^3}z_{ac,K}^2(1,2)\right)\|\1_{K}\dkp{\hat\psi}\|_\infty^2
	\displaybreak[0]\\
	C_5&\leq
	24\left(1+\frac{1}{K^2}\right)^3
	\left(\delta\frac{z_{ac}^2(1,0)}{s^4}
	+4\frac{z_{ac}^2(1,1)}{s^2\delta}
	+4\frac{z_{ac}^2(1,2)}{3\delta^3}\right)\|\hat\psi w\|_1^2\nonumber\\
	&\quad +24\left(1+\frac{1}{K^2}\right)^3
	\left(\delta\frac{z_{ac}^2(0,0)}{s^2}
	+2\frac{z_{ac}^2(0,1)}{\delta}\right)\|\dkp{\hat\psi} w\|_1^2
  \end{align}
  and
  \begin{align}
	C_6&\leq\frac{81\pi^2}{s_K^2}
	\left(K z_{ac,K}^2(0,0)
	+\frac{s_K^2}{2K}z_{ac,K}^2(0,1)\right)|\hat\psi(0)|^2
	\displaybreak[0]\\
	C_7&\leq\frac{81\pi^2}{2s_K^4}
	\left(K z_{ac,K}^2(1,0)
	+\frac{s_K^2}{K}z_{ac,K}^2(1,1)
	+\frac{s_K^4}{6K^3}z_{ac,K}^2(1,2)\right)|\hat\psi(0)|^2\nonumber\\
	&\quad+\frac{81\pi^2}{2s_K^2}\left(K z_{ac,K}^2(0,0)
	+\frac{s_K^2}{K}z_{ac,K}^2(0,1)\right)\|\1_{K}\dkp{\hat\psi}\|_\infty^2
	\displaybreak[0]\\
	C_8&\leq\frac{81\pi^2}{16s_K^4}
	\left(K z_{ac,K}^2(1,0)
	+\frac{s_K^2}{K}z_{ac,K}^2(1,1)
	+\frac{s_K^4}{3K^3}z_{ac,K}^2(1,2)\right)\|\1_{K}\dkp{\hat\psi}\|_\infty^2
	\displaybreak[0]\\
	C_9&\leq54\left(1+\frac{1}{K^2}\right)^4
	\left(\delta\frac{z_{ac}^2(0,0)}{s^2}+\frac{z_{ac}^2(0,1)}{\delta}\right)\|\hat\psi w\|_1^2.
  \end{align}
  As before $\delta=s$ is seen to be the optimal choice in the sense $C_5$ and $C_9$ are, to leading order, proportional to $s^{-3}$ and $s^{-1},$ respectively. 

Summing up Eqs.~\eqref{eq:thm_main1.11},\eqref{eq:thm_main2.0}, and \eqref{eq:thm_main3.0} we get
\begin{align}
	&\|P_{ac}\1_Re^{-iHt}P_{ac}\psi\|_2^2
	\\
	&\qquad\leq \lambda C_6 t^{-1}+ \lambda (C_3+C_7)t^{-2}
		+(C_1+C_4+C_8)\,t^{-3}
		+(C_2+C_5+C_9)\,t^{-4} .
\end{align}
Calculating the constants in front of the time factors we find
\begin{align}
	c_1&\leq81\pi^2\frac{|\hat\psi(0)|^2}{s_K^2}
	\left(K z_{ac,K}^2(0,0)
	+\frac{s_K^2}{2K}z_{ac,K}^2(0,1)\right),
	\label{eq:C1Pure}
	\displaybreak[0]\\
	c_2&\leq
	53\pi^2\frac{|\hat\psi(0)|^2}{s_K^4}\left(K z_{ac,K}^2(1,0)
	+\frac{s_K^2}{K}z_{ac,K}^2(1,1)
	+\frac{s_K^4}{6K^3}z_{ac,K}^2(1,2)\right)\nonumber\\
	&\quad+53\pi^2\frac{\|\1_{K}\dkp{\hat\psi}\|_\infty^2}{s_K^2}\left(K z_{ac,K}^2(0,0)
	+\frac{s_K^2}{K}z_{ac,K}^2(0,1)\right),	
	\label{eq:C2Pure}
	\displaybreak[0]\\
	c_3&\leq
	9\frac{\|\1_{K}\ddkp{\hat\psi}\|_\infty^2}{s_K^2}\bigg(2K z_{ac,K}^2(0,0)
	+\frac{3s_K^2}{K}z_{ac,K}^2(0,1)\bigg)\nonumber\\
	&+23\pi^2\frac{\|\1_{K}\dkp{\hat\psi}\|_\infty^2}{s_K^4}\bigg(K z_{ac,K}^2(1,0)
	+\frac{s_K^2}{K}z_{ac,K}^2(1,1)
	+\frac{s_K^4}{3K^3}z_{ac,K}^2(1,2)\bigg)\nonumber\\
	& +9\frac{\|\1_{K}\hat\psi\|_\infty^2}{s_K^6}\bigg(2Kz_{ac,K}^2(2,0)
	+\frac{3s_K^2}{K}z_{ac,K}^2(2,1)
	\nonumber\\
	&\quad+\frac{2s_K^4}{3K^3}z_{ac,K}^2(2,2)
	+\frac{s_K^6}{5K^5}z_{ac,K}^2(2,3)\bigg),
	\label{eq:C3Pure}
	\displaybreak[0]\\
	c_4&\leq
	\frac{27}{2}\frac{\|\ddkp{\hat\psi}w\|_1^2}{s}\left(1+\frac{1}{K^2}\right)^2
	\bigg(z_{ac}^2(0,0)+z_{ac}^2(0,1)\bigg)\nonumber\\
	&\quad+38\frac{\|\dkp{\hat\psi}w\|_1^2}{s^3}\left(1+\frac{1}{K^2}\right)^3
	\bigg[z_{ac}^2(1,0)+8z_{ac}^2(1,1)+\frac{8}{3}z_{ac}^2(1,2)\nonumber\\
	&\qquad\qquad+s^2(z_{ac}^2(0,0)+2z_{ac}^2(0,1))\bigg]\nonumber\\
	&\quad+92\frac{\|{\hat\psi}w\|_1^2}{s^5}\left(1+\frac{1}{K^2}\right)^4
	\bigg[z_{ac}^2(2,0)+3z_{ac}^2(2,1)+z_{ac}^2(2,2)+\frac{3}{5}z_{ac}^2(2,3)\nonumber\\
	&\qquad\qquad+s^2\left(z_{ac}^2(1,0)+4z_{ac}^2(1,1)+\frac{4}{3}z_{ac}^2(1,2)\right)\nonumber\\
	&\qquad\qquad+s^4(z_{ac}^2(0,0)+z_{ac}^2(0,1))\bigg].
	\label{eq:C4Pure}
\end{align}
Using the assumption $s,s_K,K\leq1$ and straightforward simplifications we obtain the proposition.
  \QED
  \end{proof}

  \begin{proof}[of Eq.~(\ref{eq:thm_main2.0})]
We will follow similar lines as for the proof of Eq.~\eqref{eq:thm_main1.11}, with one notable difference, namely the time dependent factor in integral~\eqref{eq:thm_main2.0} is $4k/(2tk^2+i)^2,$ whereas in integral~\eqref{eq:thm_main1.11} it was $t^{-1}(2tk^2+i)^{-1}.$ This difference crucially influences the $t$-behavior coming from regions 1 and 2 because
  \begin{align}
	&\left[\int_0^\infty\left|\frac{4k}{(2tk^2+i)^2}\right|\,dk\right]^2
	=\frac{1}{t^2}\left[\int_0^\infty\frac{4k}{4k^4+1}\,dk\right]^2
	=\frac{\pi^2}{4 t^2}.
\label{eq:thm_main2.1}
\end{align}
The expected $t^{-3}$ behavior can be recovered if $g_2(k,k') \sim c k$, $c\in\C$, as $k\to0$ because
\begin{align}
	&\left[\int_0^\infty k\left|\frac{4k}{(2tk^2+i)^2}\right|\,dk\right]^2
	=\frac{1}{t^3}\left[\int_0^\infty\frac{4k^2}{4k^4+1}\,dk\right]^2
	=\frac{\pi^2}{16t^3}.\label{eq:thm_main2.2}
  \end{align}
The function $g_2$ consists only of $\hat\psi,$ $Z_{ac}$, and their derivatives, whose behavior for $k\to0$ was determined in Lemmas~\ref{lem:Zac}, \ref{lem:Z0ac} and~\ref{lem:psihat0}. These lemmas show that in region 1
  \begin{align}
  	|g_2(k,k')|
	&\leq \lambda\left(\frac{z_{ac,K}(0,0)}{s_K}\|\1_{K}\dkp{\hat\psi}\|_\infty
	+\frac{z_{ac,K}(1,0)}{s_K^2}|\hat\psi(0)|\right)
	\nonumber\\
	&\quad+2\frac{z_{ac,K}(1,0)}{s_K^2}\|\1_{K}\dkp{\hat\psi}\|_\infty k ,
\label{eq:thm_main2.3}
  \end{align}
while in region 2
    \begin{align}
  	|g_2(k,k')|
	&\leq \lambda\left(\frac{\|\1_{K}\dkp{\hat\psi}\|_\infty z_{ac,K}(0,1)s_K+|\hat\psi(0)|z_{ac,K}(1,1)}{s_K|k-k'|}
	+\frac{|\hat\psi(0)|z_{ac,K}(1,2)}{|k-k'|^2}\right)\nonumber\\
	&\quad+2\|\1_{K}\dkp{\hat\psi}\|_\infty\left(\frac{z_{ac,K}(1,1)}{s_K|k-k'|}
	+\frac{z_{ac,K}(1,2)}{|k-k'|^2}\right)k.\label{eq:thm_main2.4}
  \end{align}
We split the integral \eqref{eq:thm_main2} following Eq.~\eqref{eq:thm_main_r1}-\eqref{eq:thm_main_r4}. 
Using Eqs.~\eqref{eq:thm_main2.1}, \eqref{eq:thm_main2.2}, and \eqref{eq:thm_main2.3}, we see that the contribution to the integral~\eqref{eq:thm_main2} from region 1 satisfies 
  \begin{align}
	&\frac{6}{t^2}\int_0^{2K}dk'\left[\int_0^{K}dk |g_2(k,k')|\frac{4tk}{4t^2k^4+1}\right]^2\\
	&\leq  \lambda\frac{12\pi^2K}{t^2}
	\left(\frac{z_{ac,K}^2(0,0)}{s_K^2}\|\1_{K}\dkp{\hat\psi}\|_\infty^2
	+\frac{z_{ac,K}^2(1,0)}{s_K^4}|\hat\psi(0)|^2\right)
	\nonumber\\
	&\qquad+\frac{6\pi^2K}{t^3}\frac{z_{ac,K}^2(1,0)}{s_K^4}\|\1_{K}\dkp{\hat\psi}\|_\infty^2.\label{eq:thm_main2.5}
  \end{align}
In region 2 we use Eq.~\eqref{eq:thm_main2.4} and the fact that $k\leq K$ and $k'\geq 2K,$ to get that
  \begin{align}
	\frac{6}{t^2}&\int_{2K}^\infty dk'\left[\int_0^{K}dk |g_2(k,k')|\frac{4tk}{4t^2k^4+1}\right]^2\\
	&\leq  \lambda\frac{3\pi^2}{t^2}\int_{2K}^\infty dk'
	\Biggl(\frac{\|\1_{K}\dkp{\hat\psi}\|_\infty z_{ac,K}(0,1)s_K+|\hat\psi(0)|z_{ac,K}(1,1)}{s_K|k'-K|}
	\nonumber\\
	&\qquad+\frac{|\hat\psi(0)|z_{ac,K}(1,2)}{|k'-K|^2}\Biggr)^2
	\nonumber\\
	&\quad+\frac{3\pi^2}{t^3}\|\1_{K}\dkp{\hat\psi}\|_\infty^2  \int_{2K}^\infty dk'
	\left(\frac{z_{ac,K}(1,1)}{s_K|k'-K|}
	+\frac{z_{ac,K}(1,2)}{|k'-K|^2}\right)^2\\
	&=  \lambda\frac{6\pi^2}{t^2}
	\Biggl(\frac{1}{Ks_K^2}\left(\|\1_{K}\dkp{\hat\psi}\|_\infty z_{ac,K}(0,1)s_K
	+|\hat\psi(0)|z_{ac,K}(1,1)\right)^2
	\nonumber\\
	&\qquad	+\frac{|\hat\psi(0)|^2z_{ac,K}^2(1,2)}{3K^3}\Biggr)
	\nonumber\\
	&\quad+\frac{6\pi^2}{t^3}\|\1_{K}\dkp{\hat\psi}\|_\infty^2
	\left(\frac{z_{ac,K}^2(1,1)}{Ks_K^2}
	+\frac{z_{ac,K}^2(1,2)}{3K^3}\right).\label{eq:thm_main2.6}
  \end{align}
  Now, we turn to region 3 and observe that for $k\geq K$
  \begin{equation}\label{eq:thm_main2.7}
	\frac{4tk}{4t^2k^4+1}
	\leq  \frac{1}{tk^3}
	\leq\frac{1}{tK}\left(1+\frac{1}{k^2}\right)w(k)
	\leq\frac{1}{tK}\left(1+\frac{1}{K^2}\right)w(k).
  \end{equation}
  Hence, integral~\eqref{eq:thm_main2} in region 3 satisfies
  \begin{align}
	&\frac{12}{t^2}  \int_0^\infty dk'  \left[\int_0^\infty dk \1_{3,k'} |g_2(k,k')|\frac{4tk}{4t^2k^4+1}\right]^2\\
	&\leq\frac{12}{t^4}\frac{1}{K^2}\left(1+\frac{1}{K^2}\right)^2
	\int_0^\infty dk'  \left[\int_0^\infty dk \1_{3,k'} |g_2(k,k')|w(k)\right]^2\\
	&\leq  \frac{24}{t^4}  \left(1+\frac{1}{K^2}\right)^3  \int_0^\infty dk'
	\bigg(\left[\int_0^\infty dk \1_{3,k'} |\dkp Z_{ac}(k,k')\hat\psi(k)|w(k)\right]^2\nonumber\\
	&+\left[\int_0^\infty dk \1_{3,k'} Z_{ac}(k,k')\dkp{\hat\psi}(k)|w(k)\right]^2\bigg).
  \end{align}
  Employing Jensen's inequality with $|\hat\psi w|$ and $|\dkp{\hat\psi} w|$ as measures for the respective $k$-integrals, then yields
  \begin{align}
	&\frac{12}{t^2}  \int_0^\infty dk'  \left[\int_0^\infty dk \1_{3,k'} |g_2(k,k')|\frac{4tk}{4t^2k^4+1}\right]^2\\
	&\leq\frac{24}{t^4}\left(1+\frac{1}{K^2}\right)^3
	\int_{K}^\infty dk  \left[\int_0^{k-\delta} dk'+\int_{k+\delta}^\infty dk'\right]\nonumber\\
	&\times\bigg(\|{\hat\psi} w\|_1|\dkp Z_{ac}(k,k')|^2|{\hat\psi}(k)|w(k)
	+\|\dkp {\hat\psi} w\|_1|Z_{ac}(k,k')|^2|\dkp{\hat\psi}(k)|w(k)\bigg).\label{eq:thm_main2.8}
  \end{align}
  Plugging in the bounds for $|Z_{ac}|^2$ and $|\dkp Z_{ac}|^2$ provided by Eqs.~\eqref{eq:thm_main1.7}, and \eqref{eq:thm_main1.8}, we obtain
  \begin{align}
&	\frac{12}{t^2}  \int_0^\infty dk'  \left[\int_0^\infty dk \1_{3,k'} |g_2(k,k')|\frac{4tk}{4t^2k^4+1}\right]^2\\
	&\leq  \frac{24}{t^4}  \left(1+\frac{1}{K^2}\right)^3
	\left[4\|{\hat\psi} w\|_1^2\left(\frac{z_{ac}^2(1,1)}{s^2\delta}
	+\frac{z_{ac}^2(1,2)}{3\delta^3}\right)
	+\|\dkp {\hat\psi} w\|_1^2\frac{2z_{ac}^2(0,1)}{\delta}\right].\label{eq:thm_main2.9}
  \end{align}
  In region 4, we can again use Jensen's inequality and and the bounds for $Z_{ac}$ given in Lemma~\ref{lem:Zac}. Thereby we see that integral~\eqref{eq:thm_main2} satisfies
  \begin{align}
  &	\frac{12}{t^2}\int_0^\infty dk'\left[\int_0^\infty dk \1_{4,k'} |g_2(k,k')|\frac{4tk}{4t^2k^4+1}\right]^2\\
	&\leq  \frac{12}{t^4}\left(1+\frac{1}{K^2}\right)^3
	\int_{K}^\infty dk\int_{k-\delta}^{k+\delta} dk'\nonumber\\
	&\times\bigg[\|{\hat\psi} w\|_1|\dkp Z_{ac}(k,k')|^2|{\hat\psi}(k)|w(k)
	+\|\dkp{\hat\psi} w\|_1|Z_{ac}(k,k')|^2|\dkp{\hat\psi}(k)|w(k)\bigg]\\
	&\leq  \frac{24}{t^4}\left(1+\frac{1}{K^2}\right)^3\delta
	\left[\|{\hat\psi} w\|_1^2\frac{z_{ac}^2(1,0)}{s^4}
	+\|\dkp{\hat\psi} w\|_1^2\frac{z_{ac}^2(0,0)}{s^2}\right].\label{eq:thm_main2.10}
  \end{align}
  Summing up the contributions for all regions,  Eqs.~\eqref{eq:thm_main2.5}, \eqref{eq:thm_main2.6}, \eqref{eq:thm_main2.9}, and \eqref{eq:thm_main2.10}, we obtain the desired result given in Eq.~\eqref{eq:thm_main2.0}.
\QED
\end{proof}

\begin{proof}[of Eq.~(\ref{eq:thm_main3.0})]
  With the help of the elementary inequality
  \begin{equation}\label{eq:thm_main3.1}
	4\frac{|i-6tk^2|}{|2tk^2+i|^3}  =  \frac{4}{4t^2k^4+1}\frac{\sqrt{(6tk^2)^2+1}}{\sqrt{(2tk^2)+1}}  \leq  \frac{12}{4t^2k^4+1} ,
  \end{equation}
  we see that the time dependent factor in integral~\eqref{eq:thm_main3} satisfies
  \begin{align}
	\left[\int_0^\infty 4\frac{|i-6tk^2|}{|2tk^2+i|^3}\,dk\right]^2
	&\leq\left[\int_0^\infty \frac{12}{4t^2k^4+1}\,dk\right]^2
	\nonumber\\
	&=\frac{1}{t}\left[\int_0^\infty \frac{12}{4k^4+1}\,dk\right]^2
	=\frac{9\pi^2} {t},\label{eq:thm_main3.2}
	\displaybreak[0]\\
	\left[\int_0^\infty 4\frac{|i-6tk^2|}{|2tk^2+i|^3}k\,dk\right]^2
	&\leq\left[\int_0^\infty \frac{12k}{4t^2k^4+1}\,dk\right]^2
		\nonumber\\
	&=\frac{1}{t^2}\left[\int_0^\infty \frac{12k}{4k^4+1}\,dk\right]^2
	=\frac{9\pi^2}{4t^2},\label{eq:thm_main3.3}
	\displaybreak[0]\\
  	\left[\int_0^\infty 4\frac{|i-6tk^2|}{|2tk^2+i|^3}k^2\,dk\right]^2
	&\leq\left[\int_0^\infty \frac{12k^2}{4t^2k^4+1}\,dk\right]^2
	\nonumber\\
	&=\frac{1}{t^3}\left[\int_0^\infty \frac{12k^2}{4k^4+1}\,dk\right]^2
	=\frac{9\pi^2}{16 t^3},
	\label{eq:thm_main3.4}
  \end{align}
therefore we need $Z_{ac}(k,k')\hat\psi(k)\sim c k^2$, $c\in\C$, as $k\to0$ to obtain the expected  $t^{-3}$-decay from integral~\eqref{eq:thm_main3.0} in regions 1 and 2. The behavior of $Z_{ac}$ and $\hat\psi$ for $k\to0$ was determined in Lemmas~\ref{lem:Z0ac} and~\ref{lem:psihat0} and they imply that in region 1
  \begin{align}
	|Z_{ac}(k,k')\hat\psi(k)|
	&\leq  \lambda \frac{z_{ac,K}(0,0)}{s_K}|\hat\psi(0)|
	\nonumber\\
	&+\lambda\left(\frac{z_{ac,K}(0,0)}{s_K}\|\1_{K}\dkp{\hat\psi}\|_\infty+\frac{z_{ac,K}(1,0)}{s_K^2}|\hat\psi(0)|\right)k\nonumber\\
	&+\frac{z_{ac,K}(1,0)}{s_K^2}\|\1_{K}\dkp{\hat\psi}\|_\infty k^2\label{eq:thm_main3.5}
  \end{align}
and in region 2 we get the following bound by using the fact that $k'\leq2K$ whereas $k\geq K$
  \begin{align}
	&|Z_{ac}(k,k')\hat\psi(k)|
	\leq  \lambda\frac{z_{ac,K}(0,1)}{|k'-K|}|\hat\psi(0)|\nonumber\\
	&\quad+\lambda\left(\frac{z_{ac,K}(0,1)s_K\|\1_{K}\dkp{\hat\psi}\|_\infty
	+z_{ac,K}(1,1)|\hat\psi(0)|}{s_K|k'-K|}
	+\frac{z_{ac,K}(1,2)|\hat\psi(0)|}{|k'-K|^2}\right)k\nonumber\\
	&\quad+\left(\frac{z_{ac,K}(1,1)}{s_K|k'-K|}
	+\frac{z_{ac,K}(1,2)}{|k'-K|^2}\right)\|\1_{K}\dkp{\hat\psi}\|_\infty k^2.\label{eq:thm_main3.6}
  \end{align}
  As before we will now follow the strategy used in the proof of Eq.~\eqref{eq:thm_main2.0}. Using Eqs.~\eqref{eq:thm_main3.1}-\eqref{eq:thm_main3.5}, we see that in region~1 integral~\eqref{eq:thm_main3} satisfies
  \begin{align}
	\frac{3}{2t^2}&\int_0^{2K}dk'\left[\int_0^{K}dk |Z_{ac}(k,k'){\hat\psi}(k)|4t\frac{|i-6tk^2|}{|2tk^2+i|^3}\right]^2\\
	&\leq  \frac{81\pi^2K}{t}
	\Bigg[\lambda \frac{z_{ac,K}^2(0,0)}{s_K^2}|\hat\psi(0)|^2
	\nonumber\\
	&+\frac{\lambda}{2ts_K^4}(z_{ac,K}^2(0,0)s_K^2\|\1_{K}\dkp{\hat\psi}\|_\infty^2
	+z_{ac,K}^2(1,0)|\hat\psi(0)|^2)\nonumber\\
	&+\frac{z_{ac,K}^2(1,0)\|\1_{K}\dkp{\hat\psi}\|_\infty^2}{16t^2s_K^4}\Bigg].\label{eq:thm_main3.7}
  \end{align}
Similarly, with the help of Eq.~\eqref{eq:thm_main3.6}, we get for integral~\eqref{eq:thm_main3} in region 2 that
  \begin{align}
	\frac{3}{2t^2}&\int_{2K}^\infty dk'\left[\int_0^{K}dk |Z_{ac}(k,k'){\hat\psi}(k)|4t\frac{|i-6tk^2|}{|2tk^2+i|^3}\right]^2\\
	&\leq  \frac{81\pi^2}{2t}  \int_{2K}^\infty dk'
	\bigg[\lambda\frac{z_{ac,K}^2(0,1)}{|k'-K|^2}|\hat\psi(0)|^2\nonumber\\
	&\quad+\frac{\lambda}{4t}
	\Biggl(\frac{z_{ac,K}(0,1)s_K\|\1_{K}\dkp{\hat\psi}\|_\infty
	+z_{ac,K}(1,1)|\hat\psi(0)|}{s_K|k'-K|}
	\nonumber\\
	&\quad+\frac{z_{ac,K}(1,2)|\hat\psi(0)|}{|k'-K|^2}\Biggr)^2
	+\frac{1}{16t^2}
	\left(\frac{z_{ac,K}(1,1)}{s_K|k'-K|}
	+\frac{z_{ac,K}(1,2)}{|k'-K|^2}\right)^2\|\1_{K}\dkp{\hat\psi}\|_\infty^2\bigg]\\
	&=\frac{81\pi^2}{2t}
	\bigg[\lambda\frac{1}{K}z_{ac,K}^2(0,1)|\hat\psi(0)|^2\nonumber\\
	&\quad+\frac{\lambda}{2t}\Biggl(\frac{2}{Ks_K^2}(z_{ac,K}^2(0,1)s_K^2\|\1_{K}\dkp{\hat\psi}\|_\infty^2
	+z_{ac,K}^2(1,1)|\hat\psi(0)|^2)
	\nonumber\\
	&\quad+\frac{z_{ac,K}^2(1,2)|\hat\psi(0)|^2}{3K^3}\Biggr)
	+\frac{1}{8t^2}\left(\frac{z_{ac,K}^2(1,1)}{Ks_K^2}
	+\frac{z_{ac,K}^2(1,2)}{3K^3}\right)\|\1_{K}\dkp{\hat\psi}\|_\infty^2\bigg].\label{eq:thm_main3.8}
  \end{align}
  Now, observe that due to Eq.~\eqref{eq:thm_main3.1}  and~\eqref{eq:thm_main1.5}
  \begin{equation}\label{eq:thm_main3.9}
	4\frac{|i-6tk^2|}{|2tk^2+i|^3}
	\leq  \frac{3}{t^2}\left(1+\frac{1}{K^2}\right)^2w(k)^2
	\leq  \frac{3}{t^2}\left(1+\frac{1}{K^2}\right)^2w(k).
  \end{equation}
  Employing Eq.~\eqref{eq:thm_main3.9}, Jensen's inequality with $|{\hat\psi} w|$ as measure and the bound for $|Z_{ac}|$ given in Eq.~\eqref{eq:thm_main1.8} we then see that in region 3
  \begin{align}
	&\frac{3}{t^2}  \int_0^\infty dk'  \left[\int_0^\infty dk \1_{3,k'} |Z_{ac}(k,k'){\hat\psi}(k)|4t\frac{|i-6tk^2|}{|2tk^2+i|^3}\right]^2\\
	&\quad\leq\frac{27}{t^4}\left(1+\frac{1}{K^2}\right)^4\|{\hat\psi} w\|_1
	\nonumber\\
	&\quad\times\int_{K}^\infty dk  \left[\int_0^{k-\delta} dk'+\int_{k+\delta}^\infty dk'\right] |Z_{ac}(k,k')|^2|{\hat\psi}(k)|w(k)\\
	&\quad\leq\frac{54}{\delta t^4}\left(1+\frac{1}{K^2}\right)^4z_{ac}^2(0,1)\|{\hat\psi} w\|_1^2\label{eq:thm_main3.10}
  \end{align}
  with $\delta$ to be determined later.
  Again using Jensen's inequality with $|{\hat\psi}w|$ as measure and the bounds for $Z_{ac}$ provided by Lemma~\ref{lem:Zac}, it becomes clear that in region 4 integral~\eqref{eq:thm_main3} satisfies
  \begin{align}
  	&\frac{3}{t^2}\int_0^\infty dk'\left[\int_0^\infty dk \1_{4,k'} |Z_{ac}(k,k'){\hat\psi}(k)|4t\frac{|i-6tk^2|}{|2tk^2+i|^3}\right]^2\\
	&\quad\leq  \frac{27}{t^4}\left(1+\frac{1}{K^2}\right)^4\|{\hat\psi} w\|_1
	\nonumber\\
	&\quad\times\int_{K}^\infty dk\int_{k-\delta}^{k+\delta} dk' |Z_{ac}(k,k')|^2|{\hat\psi}(k)|w(k)\\
	&\quad\leq  \frac{54\delta}{t^4}\frac{z_{ac}^2(0,0)}{s^2}\left(1+\frac{1}{K^2}\right)^4\|{\hat\psi} w\|_1^2.\label{eq:thm_main3.11}
  \end{align}
Summing up the contributions from all regions, Eqs.~\eqref{eq:thm_main3.7}, \eqref{eq:thm_main3.8}, \eqref{eq:thm_main3.10}, and \eqref{eq:thm_main3.11}, we obtain the desired result in Eq.~\eqref{eq:thm_main3.0}.
  \QED
\end{proof}

\vfill
\newpage
\section{Proof of Theorem~\ref{thm:main_e}}\label{sec:proof_e}

We proceed in the same way as in the proof of Theorem~\ref{thm:main_ac}, which was given in the previous section. First we prove the analogue of Lemma~\ref{lem:IntKernAc}.

\begin{lemma}\label{lem:ZeIntKern}
  Let $R\geq R_V$ and $\psi\in\mathcal D(H),$ then
  \begin{align}
	 &\|P_{e}\1_Re^{-iHt}P_{ac}\psi\|_2^2\leq\sum_{n=0}^{N-1}\left|\int_0^\infty Z_{e}(k,n)\hat\psi(k)e^{-ik^2t}dk\right|^2\label{eq:IntKernE}
  \end{align}
  with
  \begin{align}
	Z_{e}(k,n)&
	\coloneqq\sqrt{\frac{\eta_n}{2}}\left[S(k)\frac{e^{ikR}}{k+i\eta_n}+\frac{e^{-ikR}}{k-i\eta_n}\right].\label{eq:IntKernE2}
  \end{align}
\end{lemma}

\begin{proof}
  Let $\phi_n$ denote the bound states. Then
  \begin{align}
	&\|P_{e}\1_Re^{-iHt}P_{ac}\psi\|_2^2
	\\
	&=\sum_{n=0}^{N-1}\frac{1}{\|\phi_n\|_2^2}
		\left\langle   \1_R  e^{-iHt} P_{ac}  \psi ,  \phi_n \right\rangle
		\left\langle   \phi_n  ,  \1_R  e^{-iHt} P_{ac}  \psi  \right\rangle
	\\
	&=\sum_{n=0}^{N-1}\frac{1}{\|\phi_n\|_2^2}\left|\int_0^\infty dk\,e^{-ik^2t}\hat\psi(k)\int_0^\infty dr\1_R\bar\phi_n(r)\psi^+(k,r)\right|^2.
  \end{align}
  Observing
  \begin{align}
	\frac{d}{dr}W(\bar\phi_n(r),\psi^+(k,r))=((i\eta_n)^2-k^2)\bar\phi_n(r)\psi^+(k',r).
  \end{align}
  and using $\psi^+(k,0)=0=\phi_n(0),$ we get upon integration
  \begin{align}
	\int_0^\infty dr\,\1_R(r)\bar \phi_n(r)\psi^+(k,r)=\frac{W(\bar\phi_n(R),\psi^+(k,R))}{(i\eta_n)^2-k^2}.
  \end{align}
  This and the fact that $\|\phi_n\|_2\geq\|\1_{[R,\infty)}\phi_n\|_2$ implies
  \begin{align}
  	&\|P_{e}\1_Re^{-iHt}P_{ac}\psi\|_2^2
	\\
	&\leq\sum_{n=0}^{N-1}\frac{1}{\|\1_{[R,\infty)}\phi_n\|_2^2}\left|\int_0^\infty dk\,e^{-ik^2t}\hat\psi(k)\frac{W(\bar\phi_n(R),\psi^+(k,R))}{(i\eta_n)^2-k^2}\right|^2,\label{eq:lem_IntKern2}
  \end{align}
  To calculate $\|\1_{[R,\infty)}\phi_n\|_2$ observe that $\phi_n(r)=e^{-\eta_nr}$ for $r\geq R_V,$ which yields
  \begin{align}
	\|\1_{[R,\infty)}\phi_n\|_2^2=\int_R^\infty e^{-2\eta_nr}\,dr=\frac{e^{-2\eta_nR}}{2\eta_n}.
  \end{align}
  Using this, $\phi_n(R)=e^{-\eta_nR}$ and $\psi^+(k,R)=\tfrac{1}{2i}(S(k)e^{ikR}-e^{-ikR})$ we calculate
  \begin{align}
	&\frac{1}{\|\1_{[R,\infty)}\phi_n\|_2}\frac{W(\bar\phi_n(R),\psi^+(k,R))}{(i\eta_n)^2-k^2}
	\\
	&\qquad=\sqrt{\frac{\eta_n}{2}}\frac{(k-i\eta_n)S(k)e^{ikR}+(k+i\eta_n)e^{-ikR}}{(i\eta_n)^2-k^2}\\
	&\qquad=-\sqrt{\frac{\eta_n}{2}}\left[S(k)\frac{e^{ikR}}{k+i\eta_n}+\frac{e^{-ikR}}{k-i\eta_n}\right],
  \end{align}
  which when plugged into Eq.~\eqref{eq:lem_IntKern2} yields Eq.~\eqref{eq:IntKernE}.
  \QED
\end{proof}

Next we want to show that the boundary terms due to partial integration in the stationary phase argument vanish, but for this we need more knowledge about how $Z_e(k,n)$ behaves for $k\to0$ and $k\to\infty.$ This is the purpose of the following lemmas.

\begin{lemma}\label{lem:Ze}
  Let $K>0$ and $R\geq R_V.$ For $k\in[0,K),$
  \begin{align}
	|Z_e(k,n)|
	&\leq  \sqrt{\frac{2}{\eta_0}}
	\eqqcolon \frac{z_{e,K}(0)}{\eta_0^{1/2}},\label{eq:lem_Ze1}\\
	|\dkp Z_e(k,n)|
	&\leq \frac{1}{\sqrt 2s_K\eta_0^{3/2}}
	\bigg[2s_K+(2Rs_K+C_{1,K})\eta_0\bigg]
	\eqqcolon \frac{z_{e,K}(1)}{s_K\eta_0^{3/2}},\label{eq:lem_Ze2}\\
	|\ddkp Z_e(k,n)|
	&\leq  \frac{1}{\sqrt 2s_K^2\eta_0^{5/2}}
	\left(C_{2,K}\eta_0^2+2\eta_0s_K\left(C_{1,K}+Rs_K\right)\left(R\eta_0+2\right)+4s_K^2\right)
	\\
	&\eqqcolon \frac{z_{e,K}(2)}{s_K^2\eta_0^{5/2}}.\label{eq:lem_Ze3}
  \end{align}
For $k\in[0,\infty)$ we have the same bounds with the index $K$ omitted on the right hand side.
\end{lemma}

\begin{proof}
Let $k\in[0,K).$ Equation~\eqref{eq:IntKernE2} for $Z_{e}$ then immediately gives
\begin{align}
  |Z_e(k,n)|\leq\sqrt{\frac{2}{\eta_n}},
\end{align}
from which we obtain Eq.~\eqref{eq:lem_Ze1} by using the fact that $\eta_n\geq\eta_0.$ Now, due to the bounds on the derivatives of $S$ given in Theorem~\ref{th:SBoundsK},
\begin{align}
  |\dkp Z_e(k,n)|
  &=\sqrt{\frac{\eta_n}{2}}
  \bigg|(\dkp S(k)+iRS(k))\frac{e^{ikR}}{k+i\eta_n}-iR\frac{e^{-ikR}}{k-i\eta_n}\nonumber\\
  &\quad-\frac{e^{-ikR}}{(k-i\eta_n)^2}-S(k)\frac{e^{ikR}}{(k+i\eta_n)^2}\bigg|\\
  &\leq\frac{1}{\sqrt 2s_K}\bigg[\frac{2s_K}{\eta_n^{3/2}}+(2Rs_K+C_{1,K})\frac{1}{\eta_n^{1/2}}\bigg],
\end{align}
and this implies Eq.~\eqref{eq:lem_Ze2} again using the fact that $\eta_n\geq\eta_0.$ Similarly, we have due to Theorem~\ref{th:SBoundsK}
\begin{align}
  |\ddkp Z_e(k,n)|
  &=\sqrt{\frac{\eta_n}{2}}\bigg|
  -R^2 \bigg(S(k)\frac{e^{ikR}}{k+i\eta_n}+\frac{e^{-ikR}}{k-i\eta_n}\bigg)+2iR\dkp S(k)\frac{e^{ikR}}{k+i\eta_n}\nonumber\\
  &\quad-2iRS(k)\frac{e^{ikR}}{(k+i\eta_n)^2}+2iR\frac{e^{-ikR}}{(k-i\eta_n)^2}-2\dkp S(k)\frac{e^{ikR}}{(k+i\eta_n)^2}\nonumber\\
  &\quad+2S(k)\frac{e^{ikR}}{(k+i\eta_n)^3}+2\frac{e^{-ikR}}{(k-i\eta_n)^3}+\ddkp S(k)\frac{e^{ikR}}{k+i\eta_n}\bigg|\\
  &\leq\frac{1}{\sqrt{2}}\bigg[
  \frac{2R^2}{\eta_n^{\frac{1}{2}}}+2R\frac{C_{1,K}}{s_K\eta_n^{\frac{1}{2}}}+4R\frac{1}{\eta_n^{\frac{3}{2}}}+2\frac{C_{1,K}}{s_K\eta_n^{\frac{3}{2}}}+\frac{4}{\eta_n^{\frac{5}{2}}}+\frac{C_{2,K}}{s_K^2\eta_n^{\frac{3}{2}}}\bigg],
\end{align}
from which we get Eq.~\eqref{eq:lem_Ze3} with the help of $\eta_n\geq\eta_0.$

Let $k\in[0,\infty).$ In this case the proof is exactly the same with the only difference that we use the $S$-matrix bounds provided by Theorem~\ref{th:GlobalSBounds} rather than those in Theorem~\ref{th:SBoundsK}. In effect this amounts to omitting the index $K$ in the bounds \eqref{eq:lem_Ze1}-\eqref{eq:lem_Ze3}.
\QED
\end{proof}

\begin{lemma}\label{lem:Z0e}
  Let $R\geq R_V$ and $K>0$ be finite. Then for $k\in[0,K],$
  \begin{align}
 	&\left|Z_{e}(k,n)\right|  \leq  \lambda\frac{z_{e,K}(0)}{\eta_0^{1/2}}+\frac{z_{e,K}(1)}{s_K\eta_0^{3/2}}k\label{eq:lem_Z05}.
  \end{align}
\end{lemma}

\begin{proof}
  Clearly
  \begin{align}
	|Z_e(k,n)|
	&=\left|Z_e(0,n)+\int_0^k\dkp Z(k',n)\,dk'\right|
	\\
	&\leq|Z_e(0,n)|+\int_0^k|\dkp Z(k',n)|\,dk'.
	\label{eq:lem_Z0e1}
  \end{align}
  From Eq.~\eqref{eq:IntKernE2} for $Z_e$ and the fact that $S(0)=\mp 1$ for $\lambda=1$ and $0$ respectively (see~\cite[page 356]{Newton1966} for the proof) we easily calculate
  \begin{align}
	|Z_e(0,n)|=\frac{1}{\sqrt{2\eta_n}}|S(0)-1|=\lambda\sqrt{\frac{2}{\eta_n}}\leq\lambda\sqrt{\frac{2}{\eta_0}}=\lambda\frac{z_{e,K}(0)}{\eta_0^{1/2}}.
  \end{align}
  Plugging this and the bound for $|\dkp Z_e|$ provided by Eq.~\eqref{eq:lem_Ze2} into Eq.~\eqref{eq:lem_Z0e1} finishes the proof.
\QED
\end{proof}

Now, we are able to prove that the boundary terms due to partial integration in the stationary phase argument vanish.

\begin{lemma}\label{lem:ZeInt}
Let $\psi$ satisfy the assumptions stated in Theorem~\ref{thm:main_e}, then
\begin{align}
\int_0^\infty Z_{e}(k,n)\hat\psi(k)\frac{\partial_{k}^2e^{-ik^2t}}{2tk^2+i}\,dk=\int_0^\infty \partial_{k}^2\left(\frac{Z_{e}(k,n)\hat\psi(k)}{2tk^2+i}\right)e^{-ik^2t}\,dk.
\end{align}
\end{lemma}

\begin{proof}
Clearly,
\begin{align}
  &\int_0^\infty Z_{e}(k,n)\hat\psi(k)\frac{\partial_{k}^2e^{-ik^2t}}{2tk^2+i}\,dk
  \\
  &\quad=  \left[\frac{Z_{e}(k,n)\hat\psi(k)}{2tk^2+i}\partial_{k}e^{-ik^2t}\right]_0^\infty\label{eq:lemp_Ze1}\\
  &\quad-  \left[\partial_{k}\left(\frac{Z_{e}(k,n)\hat\psi(k)}{2tk^2+i}\right)e^{-ik^2t}\right]_0^\infty\label{eq:lemp_Ze2}\\
  &\quad+\int_0^\infty \partial_{k}^2\left(\frac{Z_{e}(k,n)\hat\psi(k)}{2tk^2+i}\right)e^{-ik^2t}\,dk
\end{align}
and
\begin{align}
  &\partial_{k}\left(\frac{Z_{e}(k,n)\hat\psi(k)}{2tk^2+i}\right)
  \\
  &\quad=\left(\dkp Z_{e}(k,n)\hat\psi(k)+Z_{e}(k,n)\dkp{\hat\psi}(k)\right)\frac{1}{2t(2tk^2+i)}\label{eq:lemp_Ze3}\\
  &\quad-Z_{e}(k,n)\hat\psi(k)\frac{4tk}{2t(2tk^2+i)^2}.\label{eq:lemp_Ze4}
\end{align}
The same arguments given in the proof of Lemma~\ref{lem:ZacInt} also apply to Eq.~\eqref{eq:lemp_Ze1} and to Eq.~\eqref{eq:lemp_Ze4}, so we are left with handling Eq.~\eqref{eq:lemp_Ze3}. For $k\to\infty$ Eq.~\eqref{eq:lemp_Ze3} tends to zero because the time dependent factor tends to zero like $k^2$, while $Z_{e}(k,n)$ and $\dkp Z_{e}(k,n)$ are bounded for all $k$ (see Lemma~\ref{lem:Ze}), $\hat\psi(k)\to0$ as $k\to\infty$ ($\hat\psi$ is square integrable) and $\dkp{\hat\psi}(k)$ can only diverge slower than $k$ at infinity ($\|\dkp{\hat\psi}w\|_1<\infty$ by assumption). Let us now look at Eq.~\eqref{eq:lemp_Ze3} for $k\to0$.
In case there is no zero resonance ($\lambda=0$), Eq.~\eqref{eq:lemp_Ze3} tends to zero for $k\to0$ because $\hat\psi(0)=0$ (Lemma~\ref{lem:psihat0}), $|\dkp Z_e(k,n)|$ is bounded (Lemma~\ref{lem:Ze}), $Z_e(k,n)\to0$ at least like $k$ (Lemma~\ref{lem:Z0e}), and $\dkp{\hat\psi}(k)$ can only diverge slower than $1/k$ ($\|\dkp{\hat\psi}w\|_1<\infty$ by assumption). In case there is a zero resonance ($\lambda=1$), $S(0)=-1$ (see~\cite[page 356]{Newton1966} for the proof). Hence,
\begin{align}
	Z_e(0,n)&=-i\frac{1}{\sqrt{2\eta_n}}(S(0)-1)=i\sqrt{\frac{2}{\eta_n}}
	\quad\text{and}\label{eq:lemp_Ze5}\\
	\dkp Z_e(0,n)&=\sqrt{\frac{\eta_n}{2}}\left[\left(\frac{R}{\eta_n}+\frac{1}{\eta_n^2}\right)(S(0)+1)+\dkp S(0)\frac{1}{i\eta_n}\right]
	\\
	&=-i\frac{\dkp S(0)}{\sqrt{2\eta_n}}.\label{eq:lemp_Ze6}
\end{align}
On the other hand, we know from the proof of Lemma~\ref{lem:ZacInt} that  (Eq.~\eqref{eq:lem_Zac6})
\begin{align}\label{eq:lemp_Ze7}
  \dkp{\hat\psi}(0)=\frac{1}{2}\dkp S(0)\hat\psi(0).
\end{align}
Plugging Eqs.~\eqref{eq:lemp_Ze5}, \eqref{eq:lemp_Ze6}, and~\eqref{eq:lemp_Ze7} into Eq.~\eqref{eq:lemp_Ze3} evaluated at $k=0$ shows that it vanishes also when a zero resonance is present.\QED
\end{proof}

Finally we are in the position to prove Theorem~\ref{thm:main_e}.

\begin{proof}[of Theorem~\ref{thm:main_e}]
Combining Lemma~\ref{lem:ZeIntKern} with Lemma~\ref{lem:ZeInt} and using
\begin{align}
  e^{-ik^2t}=-\frac{\partial_{k}^2e^{-ik^2t}}{2t(2tk^2+i)},
\end{align}
we obtain
\begin{align}
  \|P_{e}\1_Re^{-iHt}P_{ac}\psi\|_2^2
  &\leq\sum_{n=0}^{N-1}\left|\int_0^\infty \partial_k^2\left(Z_e(k,n)\hat\psi(k)\frac{1}{2t(2tk^2+i)}\right)e^{-ik^2t}\,dk\right|^2.
\end{align}
For $A,B,C\in\R$
\begin{align}\label{eq:elineq2}
    (A+B+C)^2\leq3(A^2+B^2+C^2),
\end{align}
therefore  with the shorthands
  \begin{align}
	&g_1(k,n)  \coloneqq   \ddkp Z_{e}(k,n)\hat\psi(k)+2\dkp Z_{e}(k,n)\dkp {\hat\psi}(k)+Z_{e}(k,n)\ddkp {\hat\psi}(k)  ,
	\\
	&g_2(k,n)  \coloneqq   \dkp Z_{e}(k,n)\hat\psi(k)+Z_{e}(k,n)\dkp{\hat\psi}(k) ,
  \end{align}
we get
\begin{align}
    &\|P_{e}\1_Re^{-iHt}P_{ac}\psi\|_2^2
    \\
    &\quad\leq  \frac{3}{4t^2}\sum_{n=0}^{N-1}\left[\int_0^\infty |g_1(k,n)|\frac{1}{\sqrt{4t^2k^4+1}}\,dk\right]^2\label{eq:thm_maine1}\\
    &\quad+  \frac{3}{t^2}\sum_{n=0}^{N-1}\left[\int_0^\infty |g_2(k,n)|\frac{4tk}{4t^2k^4+1}\,dk\right]^2\label{eq:thm_maine2}\\
    &\quad+  \frac{3}{4t^2}\sum_{n=0}^{N-1}\left[\int_0^\infty |Z_{e}(k,n)\hat\psi(k)|\frac{4t|i-6tk^2|}{|2tk^2+i|^3}\,dk\right]^2  .
    \label{eq:thm_maine3}
  \end{align}

  Note that we will use Eq.~\eqref{eq:elineq2} and $(A+B)^2\leq 2A^2+2B^2$ throughout the proof often without mentioning it. Let $h(k,n)$ be a placeholder for the integrands in Eqs.~\eqref{eq:thm_maine1}-\eqref{eq:thm_maine3}, then the integration region of each of the above integrals will be divided as follows
  \begin{align}\label{eq:thm_maine4}
	\left|\int_0^\infty h(k,n)\,dk\right|^2
	&\leq  2\left[\int_0^{K}|h(k,n)|\,dk\right]^2
	+  2\left[\int_{K}^\infty |h(k,n)|\,dk\right]^2.
  \end{align}
  In contrast to the proof of Theorem~\ref{thm:main_ac} there is no need to handle the region around the diagonal separately because as one can see from Eq.~\eqref{eq:IntKernE2} $Z_e$ has no apparent singularity for $k\geq0$.
  First consider Eq.~\eqref{eq:thm_maine1}. Using the bounds on $Z_e$ and its derivatives given in Lemma~\ref{lem:Ze} and 
  \begin{equation}\label{eq:thm_maine1.2}
	\left[\int_0^{K}dk \frac{1}{\sqrt{4t^2k^4+1}}\right]^2
	\leq\frac{1}{t}\left[\int_0^{\infty}dk \frac{1}{\sqrt{4k^4+1}}\right]^2\leq\frac{2}{t},
  \end{equation}
  we find
  \begin{align}
	&\left[\int_0^{K}dk |g_1(k,n)|\frac{1}{\sqrt{4t^2k^4+1}}\right]^2\\
	&\quad\leq\frac{6}{t}\left[\frac{z_{e,K}^2(2)}{\eta_0^5s_K^4}\|\1_{K}\hat\psi\|_\infty^2
	+\frac{4z_{e,K}^2(1)}{\eta_0^3s_K^2}\|\1_{K}\dkp{\hat\psi}\|_\infty^2
	+\frac{z_{e,K}^2(0)}{\eta_0}\|\1_{K}\ddkp{\hat\psi}\|_\infty^2\right].\label{eq:thm_maine5}
  \end{align}
Equation \eqref{eq:thm_main1.5} provides the bound
\begin{equation}
	\frac{1}{\sqrt{4t^2k^4+1}}  
	\leq  \frac{1}{2t}\left(1+\frac{1}{K^2}\right)w(k),
\end{equation}
that together with Lemma~\ref{lem:Ze} implies
  \begin{align}
	&\left[\int_{K}^\infty dk |g_1(k,n)|\frac{1}{\sqrt{4t^2k^4+1}}\right]^2\\
	&\quad\leq  \frac{1}{4t^2}\left(1+\frac{1}{K^2}\right)^2
	\nonumber\\
	&\quad\times
	\left[\int_0^\infty dk \left(\|\ddkp Z_e(\cdot,n)\|_\infty|\hat\psi|
	+2\|\dkp Z_e(\cdot,n)\|_\infty|\dkp{\hat\psi}|
	+\|Z_e(\cdot,n)\|_\infty|\ddkp{\hat\psi}|\right)w\right]^2\\
	&\quad\leq  \frac{3}{4t^2}\left(1+\frac{1}{K^2}\right)^2
	\nonumber\\
	&\quad\times
	\left[\|\ddkp Z_e(\cdot,n)\|_\infty^2\|\hat\psi w\|_1^2
	+4\|\dkp Z_e(\cdot,n)\|_\infty^2\|\dkp{\hat\psi}w\|_1^2
	+\|Z_e(\cdot,n)\|_\infty^2\|\ddkp{\hat\psi}w\|_1^2\right]\\
	&\quad\leq  \frac{3}{4t^2}\left(1+\frac{1}{K^2}\right)^2
	\bigg[\frac{z_{e}^2(2)}{s^4\eta_0^5}\|\hat\psi w\|_1^2
	+4\frac{z_{e}^2(1)}{s^2\eta_0^3}\|\dkp{\hat\psi}w\|_1^2
	+\frac{z_{e}^2(0)}{\eta_0}\|\ddkp{\hat\psi} w\|_1^2\bigg].\label{eq:thm_maine6}
  \end{align}
  Now consider Eq.~\eqref{eq:thm_maine2}. We use Lemma~4 that gives a bound on $\hat\psi(k)$ for small $k,$ Lemma~\ref{lem:Z0e} that gives a bound on $Z_e(k,n)$ for small $k,$ the bound on $\dkp Z_e$ provided by Lemma~\ref{lem:Ze} and
  \begin{align}
	&\left[\int_0^\infty\left|\frac{4tk}{(2tk^2+i)^2}\right|\,dk\right]^2
	=\left[\int_0^\infty\frac{4k}{4k^4+1}\,dk\right]^2
	=\frac{\pi^2}{4},
 	\\
	&\left[\int_0^\infty k\left|\frac{4tk}{(2tk^2+i)^2}\right|\,dk\right]^2
	=\frac{1}{t}\left[\int_0^\infty\frac{4k^2}{4k^4+1}\,dk\right]^2
	=\frac{\pi^2}{16t}
  \end{align}
  to obtain
  \begin{align}
	&\left[\int_0^{K} |g_2(k,n)|\frac{4tk}{4t^2k^4+1}\,dk\right]^2\\
	&\quad\leq  \bigg[\int_0^\infty \bigg(\lambda\bigg(\frac{z_{e,K}(1)}{\eta_0^{3/2}s_K}|\hat\psi(0)|
	+\frac{z_{e,K}(0)}{\eta_0^{1/2}}\|\1_{K}\dkp{\hat\psi}\|_\infty\bigg)
	\nonumber\\
	&\qquad+2\frac{z_{e,K}(1)}{\eta_0^{3/2}s_K}\|\1_{K}\dkp{\hat\psi}\|_\infty k\bigg)\frac{4tk}{4t^2k^4+1}\,dk\bigg]^2\\
    &\quad\leq  \frac{\pi^2}{2}\bigg[2\lambda\bigg(\frac{z_{e,K}^2(1)}{\eta_0^3s_K^2}|\hat\psi(0)|^2
    +\frac{z_{e,K}^2(0)}{\eta_0}\|\1_{K}\dkp{\hat\psi}\|_\infty^2\bigg)
	+\frac{z_{e,K}^2(1)}{\eta_0^3s_K^2t}\|\1_{K}\dkp{\hat\psi}\|_\infty^2\bigg].\label{eq:thm_maine7}
  \end{align}
We also use the bound
\begin{equation}
	\frac{4tk}{4t^2k^4+1}
	\leq\frac{1}{tK}\left(1+\frac{1}{K^2}\right)w(k),
\end{equation}
provided by Eq.~\eqref{eq:thm_main2.7}, which gives
  \begin{align}
	&\left[\int_{K}^\infty |g_2(k,n)|\frac{4tk}{4t^2k^4+1}\,dk\right]^2\\
	&\quad\leq  \frac{1}{t^2}\left(1+\frac{1}{K^2}\right)^3\left[\int_{K}^\infty
	\left(|\hat\psi(k)|\frac{z_{e}(1)}{\eta_0^{3/2}s}
	+|\dkp{\hat\psi}(k)|\frac{z_{e}(0)}{\eta_0^{1/2}}\right)w(k)\,dk\right]^2\\
	&\quad\leq  \frac{2}{t^2}\left(1+\frac{1}{K^2}\right)^3
	\left[\|\hat\psi w\|_1^2\frac{z_{e}^2(1)}{\eta_0^3s^2}
	+\|\dkp{\hat\psi}w\|_1^2\frac{z_{e}^2(0)}{\eta_0}\right].\label{eq:thm_maine8}
  \end{align}
  Finally consider Eq.~\eqref{eq:thm_maine3}. In addition to the bounds on $Z_e$ and $\hat\psi$ provided by Lemmas~\ref{lem:Z0e} and \ref{lem:psihat0} respectively, that we have used to treat Eq.~\eqref{eq:thm_maine2}, we now need to use the bound 
  \begin{equation}
	4\frac{|i-6tk^2|}{|2tk^2+i|^3}  \leq  \frac{12}{4t^2k^4+1} ,
  \end{equation}
from Eq.~\eqref{eq:thm_main3.1}, that gives
  \begin{align}
	\left[\int_0^\infty 4t\frac{|i-6tk^2|}{|2tk^2+i|^3}\,dk\right]^2
	&\leq\left[\int_0^\infty \frac{12t}{4t^2k^4+1}\,dk\right]^2
	\\
	&=t\left[\int_0^\infty \frac{12}{4k^4+1}\,dk\right]^2
	=9\pi^2 t,\label{eq:thm_maine3.2}
	\displaybreak[0]\\
	\left[\int_0^\infty 4t\frac{|i-6tk^2|}{|2tk^2+i|^3}k\,dk\right]^2
	&\leq\left[\int_0^\infty \frac{12tk}{4t^2k^4+1}\,dk\right]^2
	\\
	&=\left[\int_0^\infty \frac{12k}{4k^4+1}\,dk\right]^2
	=\frac{9\pi^2}{4},\label{eq:thm_maine3.3}
	\displaybreak[0]\\
  	\left[\int_0^\infty 4t\frac{|i-6tk^2|}{|2tk^2+i|^3}k^2\,dk\right]^2
	&\leq\left[\int_0^\infty \frac{12tk^2}{4t^2k^4+1}\,dk\right]^2
	\\
	&=\frac{1}{t}\left[\int_0^\infty \frac{12k^2}{4k^4+1}\,dk\right]^2
	=\frac{9\pi^2}{16 t},
	\label{eq:thm_maine3.4}
  \end{align}
from which we obtain
\begin{align}
	&\left[\int_0^K |Z_{e}(k,n)\hat\psi(k)|\frac{4t|i-6tk^2|}{|2tk^2+i|^3}\,dk\right]^2\\
	&\quad\leq\bigg[\int_0^\infty \bigg(\lambda|\hat\psi(0)|\frac{z_{e,K}(0)}{\eta_0^{1/2}}
	+\lambda\bigg(\frac{z_{e,K}(0)}{\eta_0^{1/2}}\|\1_{K}\dkp{\hat\psi}\|_\infty
	+\frac{z_{e,K}(1)}{\eta_0^{3/2}s_K}|\hat\psi(0)|\bigg)k\nonumber\\
	&\qquad+\frac{z_{e,K}(1)}{\eta_0^{3/2}s_K}\|\1_{K}\dkp{\hat\psi}\|_\infty k^2\bigg)\frac{12t}{4t^2k^4+1}\,dk\bigg]^2\\
	&\quad\leq  27\pi^2 t\bigg[\lambda|\hat\psi(0)|^2\frac{z_{e,K}^2(0)}{\eta_0}
	+\frac{\lambda}{2t}\bigg(\frac{z_{e,K}^2(0)}{\eta_0}\|\1_{K}\dkp{\hat\psi}\|_\infty^2
	+\frac{z_{e,K}^2(1)}{\eta_0^3s_K^2}|\hat\psi(0)|^2\bigg)\nonumber\\
	&\qquad+\frac{z_{e,K}^2(1)}{4t^2\eta_0^3s_K^2}\|\1_{K}\dkp{\hat\psi}\|_\infty^2\bigg].\label{eq:thm_maine9}
  \end{align} 
If we also use the bounds 
  \begin{equation}
	4t\frac{|i-6tk^2|}{|2tk^2+i|^3}
	\leq  \frac{3}{t}\left(1+\frac{1}{K^2}\right)^2w(k)
  \end{equation}
from Eq.~\eqref{eq:thm_main3.9}, and those on $Z_e$ given in Lemma~\ref{lem:Ze}, we get
  \begin{align}
	&\left[\int_K^\infty |Z_{e}(k,n)\hat\psi(k)|\frac{4t|i-6tk^2|}{|2tk^2+i|^3}\,dk\right]^2
	\leq\frac{9}{t^2}\left(1+\frac{1}{K^2}\right)^4\frac{z_{e}^2(0)}{\eta_0}\|\hat\psi w\|^2_1.\label{eq:thm_maine10}
  \end{align}
Making use of Eq.~\eqref{eq:thm_maine4}, we plug Eq.~\eqref{eq:thm_maine5} and Eq.~\eqref{eq:thm_maine6} into Eq.~\eqref{eq:thm_maine1}, Eq.~\eqref{eq:thm_maine7} and Eq.~\eqref{eq:thm_maine8} into Eq.~\eqref{eq:thm_maine2},  Eq.~\eqref{eq:thm_maine9} and Eq.~\eqref{eq:thm_maine10} into Eq.~\eqref{eq:thm_maine3} respectively.
That completes the proof.
\QED
\end{proof}

\section*{Appendix: Physical Meaning of Resonances, Virtual States, and Zero-Resonance}
\label{sec:PhysMeaningZeros}
\addcontentsline{toc}{section}{Appendix: Physical Meaning of Resonances,\\Virtual States, and Zero-Resonance}

\begin{figure}
  \centering
  \includegraphics[width=.5\textwidth]{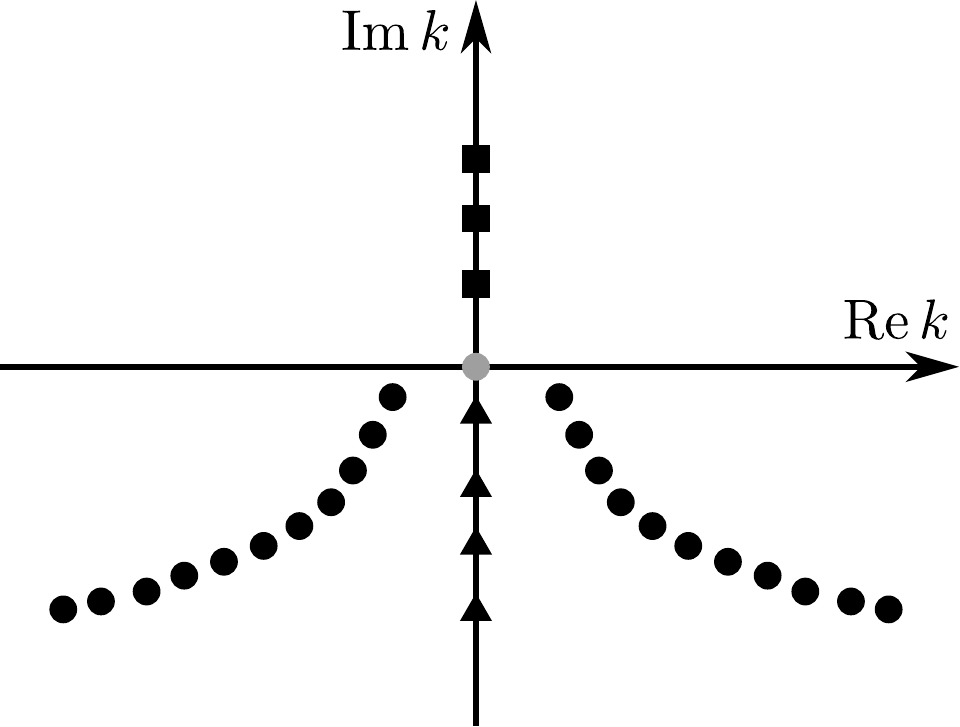}
  \caption{Location of the zeros of the Jost function $F(k)$, distinguished in bound states ($\blacksquare$), virtual states ($\blacktriangle$), resonances (\protect\raisebox{-.3ex}{\large$\bullet$}), and zero-resonance (\protect\raisebox{-.3ex}{\large\textcolor{gray}{$\bullet$}}).}
  \label{fig:Zeros}
\end{figure}

The zeros of the Jost function have important physical meaning, that we will now briefly discuss  (see also Fig.\ \ref{fig:Zeros}).
In the following, we will use the symbols $\mu$ and $\nu$ to denote strictly positive real numbers.

Consider at first a zero of the form $i \nu$.
It corresponds to a \emph{bound state}, indeed the function $f(i\nu,r)$ (see Eq.\ \eqref{eq:fboundary}) is a solution of the Schrödinger equation \eqref{eq:Schroedinger} such that $f(i\nu,r)=e^{-\nu r}$ for $r\geq R_V$, therefore it is square integrable.
In other words, $f(i\nu,r)$ is the eigenfunction corresponding to the eigenvalue $-\nu^2$.
We assumed that the potential had compact support, therefore every state with positive energy can tunnel away, and there can be only bound states with negative energy.
The zeros that  correspond to bound states are simple  \citep[Th.\ XI.58d, page 140]{RS3} and finitely many. 
The latter property can be easily established from Eq.\ \eqref{eq:LogFDirectionPlusI}, that implies that $|F(i\nu)|\to 1$ as $\nu\to\infty$, therefore the Jost function is non-zero from a certain value of $\nu$ on, and  the zeros of a non-zero  entire function can not have finite accumulation points  \citep[see also][page 361]{Newton1966}.
The exact number of eigenstates is given by Levinson's theorem \citep[Theorem\ XI.59, page 142]{RS3}.

A zero of the form $\pm\mu+i\nu$ would  correspond to a square integrable eigenfunction with eigenvalue $(\pm\mu+i\nu)^2\in\C$, but this cannot be the case because the Hamiltonian is self-adjont and has therefore only real eigenvalues.
As a consequence, the bound states are the only zeros in the positive imaginary half-plane.

The zeros in the negative imaginary half-plane correspond to functions $f$ that increase exponentially in $r$ as $r\to\infty$, and are therefore not square integrable.
They are not physical states, but have nevertheless a dynamical meaning.
Consider a zero of the form $\pm\mu-i\nu$; the property $\bar F(\bar k)=F(-k)$ \citep[12.32a, page 340]{Newton1966} implies that these zeros come in couples symmetric with respect to the imaginary axis.
The time evolution of the $f$ corresponding to such a zero is given by the factor 
\begin{equation}
e^{-i(\pm\mu-i\nu)^2t}=e^{-i(\mu^2-\nu^2)t} e^{\mp 2\mu\nu t} ,
\end{equation}
 i.e.\ $f$  exponentially increases and decreases in time for $-\mu$ and $\mu$, respectively.
Therefore, the $f$ corresponding to a zero of the form $\mu-i\nu$ can be a good model for a meta-stable state: a normalizable state in some sense close to this $f$ will have a time evolution similar to it, and can be used to describe a decaying system \citep{Gamow1928,Skibsted86}.
These zeros are called \emph{resonances}.
Given a potential, the resonances can be found through a scattering experiment: when the projectile has energy $\mu^2$ there is a chance that it forms the meta-stable state and is later released in a random direction, generating a peak in the cross section.
The width of the peak can be shown to be related to $\nu$  \citep[see for example][]{Bohm1993}.

Besides the resonances, in the negative imaginary plane there can be zeros of the form $-i\nu$ too, that also correspond to functions $f$ not square integrable. 
They are called \emph{virtual states}.
Their time evolution is expected to be given by a phase, therefore a physical state similar to a virtual state will evolve for some time almost only by a phase.%
\footnote{The time evolution of a physical state similar to a virtual state or to a resonance will after some time diverge from the multiplication by $e^{-ik^2t}$ because of the accumulating error.}
As a consequence, they can  also be considered meta-stable \citep[page 487]{Bohm1993}.
In scattering experiments they manifest as a peak at zero energy.
The virtual states are finitely many, and this can be proven in analogy with the bound states, using Eq.\ \eqref{eq:LogFDirectionMinusI}.

The only place on the real axis where there can be a zero is the origin \citep[page 346]{Newton1966}; such a zero is called \emph{zero-resonance}, and it must be simple \citep[pages 327, 328]{Newton1960}.
The corresponding $f$ is not square integrable, and does not change at all in time.
A zero-resonance is  a meta-stable state, and leads to a peak at zero energy in the scattering cross section, but the presence of a zero in $k=0$ has also a strong influence on the long-time behavior of \emph{any} wave function \citep{JensenKato1979}.
This circumstance can be understood in terms of the stationary phase argument: long time corresponds to $k=0$.
It should be noted, that the presence of a zero-resonance is very untypical.

We observe that the resonances are infinitely many, indeed the function $g$ defined in  \eqref{eq:GDef} is of fractional order, and has therefore infinitely many zeros (see \citep[][page 361]{Newton1966}).
Moreover, there are finitely many resonances below any half-line contained in the negative imaginary half-plane that goes through the origin, and inside any stripe in the negative imaginary half-plane \citep[page 361]{Newton1966}.
That implies that, denoting the resonances by $\alpha_n-i\beta_n$ and ordering them with growing modulus, then as $n\to\infty$
\begin{equation}
\beta_n \to \infty,
\qquad
\alpha_n \to \infty,
\qquad
\frac{\beta_n}{\alpha_n} \to 0  .
\end{equation}
This implies also that the sets $\{\alpha_n\}_{n\in\N^0}$ and $\{\beta_n\}_{n\in\N^0}$ have a minimum.

\begin{remark}
If the potential has a shape like a single barrier, then we expect that $\min_n \beta_n=\beta_0$, indeed states with higher energy impinge more often on the barrier than states with lower energy, and therefore have more occasions to tunnel out.
On the other side, if the potential has a more complicated shape this simple expectation could be wrong; for example if the potential has several barriers, then a state with higher energy after having passed the first barrier has more occasions to go back inside the first barrier than a state with lower energy.
\end{remark}

%
%
\bibliography{biblio}
\bibliographystyle{bibPhDN.bst}
\mbox{}
\cleardoublepage
\pagestyle{empty}%
\mbox{}
\includepdf{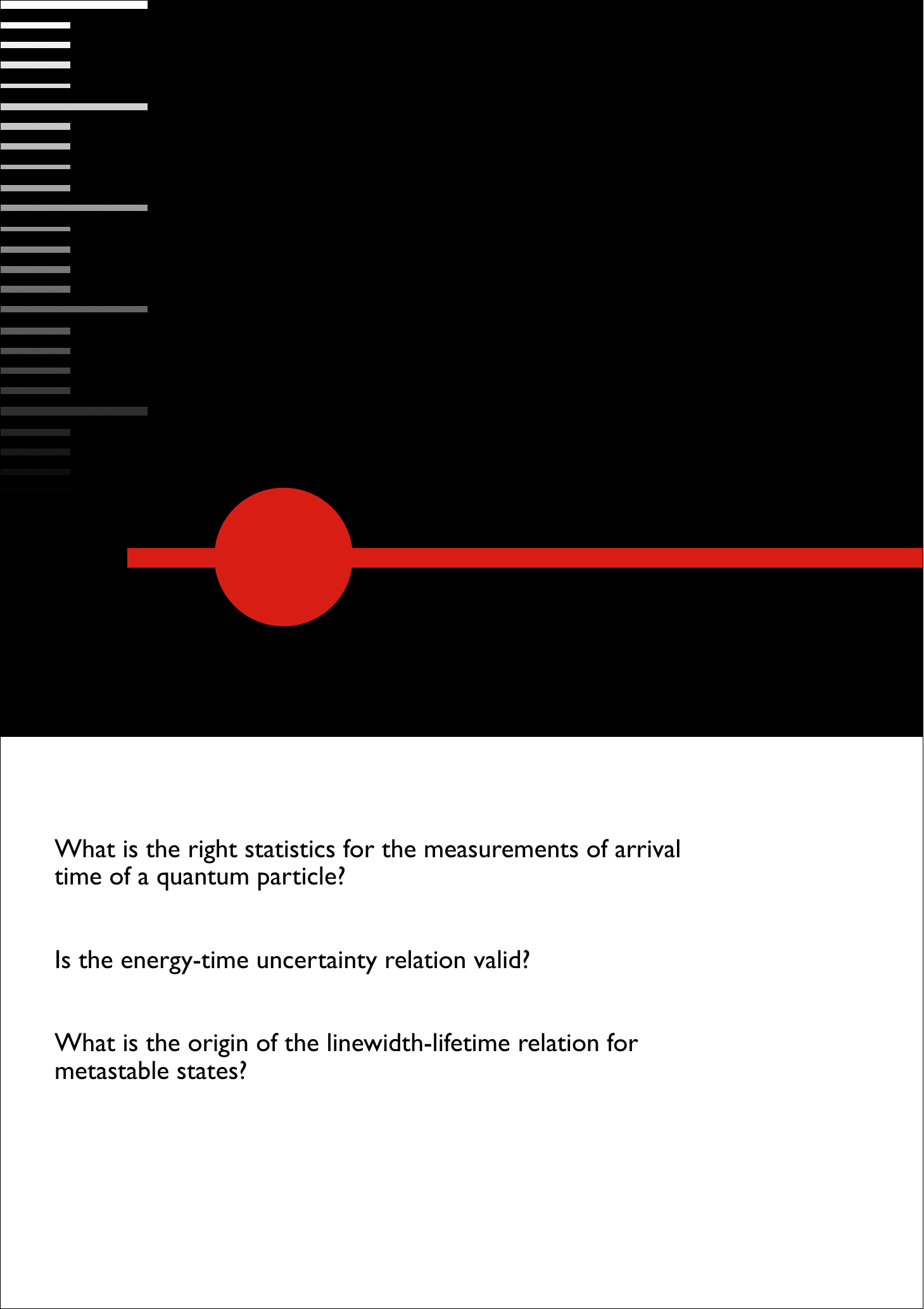}%

\end{document}